\documentclass[a4paper,12 pt]{memoir}

\usepackage[dvips]{color}
\usepackage{hyperref}
\usepackage{pifont}
\usepackage{amsthm,amsfonts}
\usepackage{graphicx,color}
\usepackage{amssymb,amsmath}
\usepackage{aurical}
\usepackage[T1]{fontenc}
\usepackage{epigraph}
\usepackage{fourier}
\usepackage{url}
\makeatletter
\newif\iffelinenonum
\newcommand\MyNumToName[1]{%
\ifcase#1\relax 
\or First\or Second\or Third%
\else Not implemented\fi}
\makechapterstyle{daleif3}{

\setlength\midchapskip{7pt}

}
\makeatother

\newcommand{\cao}{\c c\~ao }

\newcommand{\de}[1]{\left( #1 \right)}

\newcommand{\ket}[1]{\left| #1 \right\rangle}
\newcommand{\bra}[1]{\left\langle #1 \right|}
\newcommand{\braket}[2]{\left\langle #1 \mid #2 \right\rangle}
\newcommand{\norm}[1]{\left\| #1 \right\|}

\newcommand{\sand}[3]{\left\langle #1\right| #2 \left| #3 \right\rangle}

\newcommand{\tr}{\mathrm{Tr}}

\newtheoremstyle{barbarateo}
  {8pt}
  {8pt}
  {\itshape\sffamily}
  {}
  {\bfseries\sffamily}
  {.}
  {5pt}
  {}
\theoremstyle{barbarateo}
\newtheorem{teo}{Theorem}

\newtheorem{lemma}{Lemma}
\newtheorem{cor}{Corollary}
\newtheorem{prop}{Proposition}

\newtheoremstyle{barbaradefi}
  {8pt}
  {8pt}
  {\sffamily}
  {}
  {\bfseries\sffamily}
  {.}
  {5pt}
  {}
\theoremstyle{barbaradefi}
\newtheorem{ex}{\textsf{Example}}

\newtheorem{defi}{\textsf{Definition}}
\newtheorem{prin}{\textsf{Principle}}
\newtheorem{as}{\textsf{Assumption}}
\newtheorem{cons}{\textsf{Constraint}}
\newtheoremstyle{barbaraproof}
  {8pt}
  {8pt}
  {\sffamily}
  {}
  {\itshape\sffamily}
  {.}
  {5pt}
  {}
  \makeatletter  
\def\@endtheorem{\qed\endtrivlist\@endpefalse } 
\makeatother
  \theoremstyle{barbaraproof}
\newtheorem*{dem}{\textit{\textsf{Proof}}}

\def\be{\begin{equation}}
\def\ee{\end{equation}}
\definecolor{violeta}{cmyk}{0.07,0.90,0,0.34}
\definecolor{fresa}{cmyk}{0,1,0.50,0}
\DeclareMathOperator{\cosec}{cosec}

\setsecnumdepth{subsection}
\settocdepth{subsection} 
\maxsecnumdepth{subsection}

\begin{document}
\chapterstyle{daleif3}
\sffamily
\pagestyle{ruled}
\frontmatter

\begin{center}

\thispagestyle{empty}

\vspace*{15.5cm}
{\huge The Exclusivity Principle and the Set of Quantum Correlations}

\vspace{0.5cm}

{\LARGE Bárbara Lopes Amaral}

\vspace{0.5cm}
\begin{Large}
Maio de 2014
\end{Large}

\end{center}

\newpage
\thispagestyle{empty}
\mbox{}

\thispagestyle{empty}

\vfill
\hbox{%
\hspace*{0.2\textwidth}%
\rule{1pt}{\textheight}
\hspace*{0.05\textwidth}%
\parbox[b]{0.75\textwidth}{
\vbox{%
\vspace{2cm}
{\noindent \LARGE The Exclusivity Principle and the Set of Quantum Correlations }\\[2\baselineskip]
{\Large Bárbara Lopes Amaral}\par \vspace{0.3\textheight}
{\Large {Orientador:}\\
Dr. Marcelo Terra Cunha} \vspace*{2cm}
\begin{flushright}
\begin{minipage}{7.8cm}
Tese apresentada \`a UNIVERSIDADE FE\-DE\-RAL DE MINAS GERAIS,
como requisito parcial para a obten\cao do grau de DOUTORA EM
MATEMÁTICA.
\end{minipage}
\end{flushright}
\vspace*{2cm}
{\noindent \large Belo Horizonte, Brasil\\ Maio de 2014}\\[\baselineskip]
}
}
}

\thispagestyle{empty}
\mbox{}
\newpage
\thispagestyle{empty}
\vspace*{\stretch{1}}
\begin{flushright}
\textit{\`A minha casa,  Thales e  Tshabalala.}
\end{flushright}
\vspace{\stretch{4}}
\begin{figure}[h]
\includegraphics[width=0.7\textwidth]{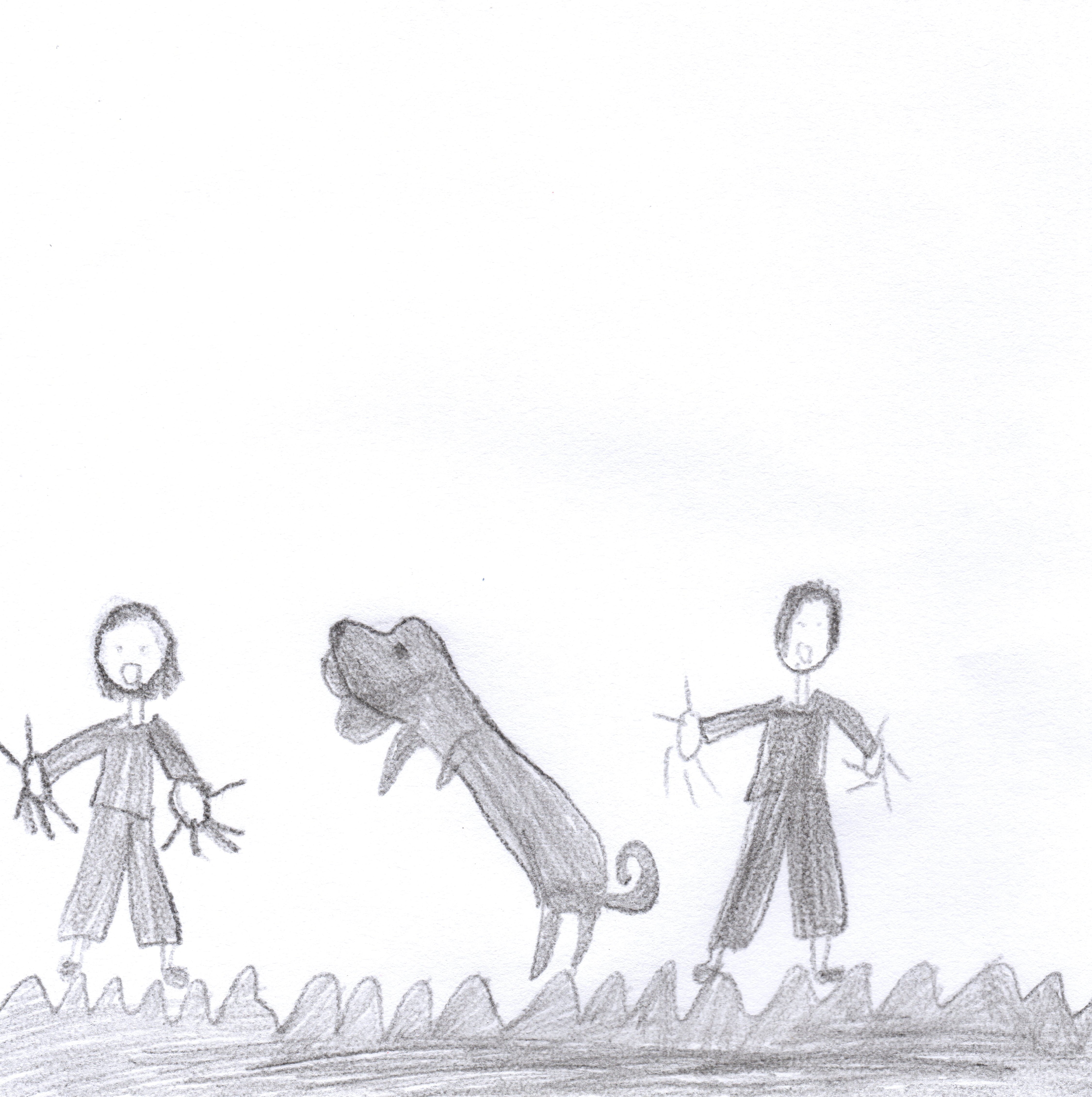}
\end{figure}

\newpage
{\scriptsize
Tenho o direito de ter raiva, de manifest\'a-la, de t\^e-la como motiva\c c\~ao para minha briga tal qual tenho o
direito de amar, de expressar meu amor ao mundo, de t\^e-lo como motiva\c c\~ao de minha briga porque,
hist\'orico, vivo a Hist\'oria como tempo de possibilidade n\~ao de determina\c c\~ao. Se a realidade fosse assim
porque estivesse dito que assim teria de ser n\~ao haveria sequer por que ter raiva. Meu direito \`a raiva
pressup\~oe que, na experi\^encia hist\'orica da qual participo, o amanh\~a n\~ao \'e algo ``pr\'e-dado'', mas um
desafio, um problema. A minha raiva, minha justa ira, se funda na minha revolta em face da nega\c c\~ao do
direito de ``ser mais'' inscrito na natureza dos seres humanos. N\~ao posso, por isso, cruzar os bra\c cos
fatalistamente diante da mis\'eria, esvaziando, desta maneira, minha responsabilidade no discurso c\'inico e
``morno'', que fala da impossibilidade de mudar porque a realidade \'e mesmo assim. O discurso da
acomoda\c c\~ao ou de sua defesa, o discurso da exalta\c c\~ao do sil\^encio imposto de que resulta a imobilidade
dos silenciados, o discurso do elogio da adapta\c c\~ao tomada como fado ou sina \'e um discurso negador da
humaniza\c c\~ao de cuja responsabilidade n\~ao podemos nos eximir. A adapta\c c\~ao a situa\c c\~oes negadoras da
humaniza\c c\~ao s\'o pode ser aceita como consequ\^encia da experi\^encia dominadora, ou como exerc\'icio de
resist\^encia, como t\'atica na luta pol\'itica. Dou a impress\~ao de que aceito hoje a condi\c c\~ao de silenciado para
bem lutar, quando puder, contra a nega\c c\~ao de mim mesmo. Esta quest\~ao, a da legitimidade da raiva
contra a docilidade fatalista diante da nega\c c\~ao das gentes, foi um tema que esteve impl\'icito em toda a
nossa conversa naquela manh\~a.


\vspace{2em}

\'E por isso tamb\'em que n\~ao me parece poss\'ivel nem aceit\'avel a posi\c c\~ao ing\^enua ou, pior,
astutamente neutra de quem estuda, seja o f\'isico, o bi\'ologo, o soci\'ologo, o matem\'atico, ou o pensador da
educa\c c\~ao. Ningu\'em pode estar no mundo, com o mundo e com os outros de forma neutra. N\~ao posso
estar no mundo de luvas nas m\~aos constatando apenas. A acomoda\c c\~ao em mim \'e apenas caminho para a
inser\c c\~ao, que implica decis\~ao, escolha, interven\c c\~ao na realidade. H\'a perguntas a serem feitas
insistentemente por todos n\'os e que nos fazem ver a impossibilidade de estudar por estudar. De estudar
descomprometidamente como se misteriosamente de repente nada tiv\'essemos que ver com o mundo, um
l\'a fora e distante mundo, alheado de n\'os e n\'os dele.

Em favor de que estudo? Em favor de quem? Contra que estudo? Contra quem estudo?

\vspace{2em}

Mas t\~ao decidido quanto antes na luta
por uma educa\c c\~ao que, enquanto ato de conhecimento, n\~ao apenas se centre no ensino dos conte\'udos
mas que desafie o educando a aventurar- se no exerc\'icio de n\~ao s\'o falar da mudan\c ca do mundo, mas de
com ela realmente comprometer- se. Por isso \'e que, para mim, um dos conte\'udos essenciais de qualquer
programa educativo, de sintaxe, de biologia, de f\'isica, de matem\'atica, de ci\^encias sociais \'e o que
possibilita a discuss\~ao da natureza mut\'avel da realidade natural como da hist\'orica e v\^e homens e
mulheres como seres n\~ao apenas capazes de se adaptar ao mundo mas sobretudo de mud\'a-lo. Seres
curiosos, atuantes, falantes, criadores.

\vspace{2em}

Com a vontade enfraquecida, a resist\^encia fr\'agil, a identidade posta em d\'uvida, a auto-estima
esfarrapada, n\~ao se pode lutar. Desta forma, n\~ao se luta contra a explora\c c\~ao das classes dominantes
como n\~ao se luta contra o poder do \'alcool, do fumo ou da maconha. Como n\~ao se pode lutar, por faltar
coragem, vontade, rebeldia, se n\~ao se tem amanh\~a, se n\~ao se tem esperan\c ca. Falta amanh\~a aos
``esfarrapados do mundo'' como falta amanh\~a aos subjugados pelas drogas.
Por isso \'e que toda pr\'atica educativa libertadora, valorizando o exerc\'icio da vontade, da decis\~ao, da
resist\^encia, da escolha; o papel das emo\c c\~oes, dos sentimentos, dos desejos, dos limites; a import\^ancia da
consci\^encia na hist\'oria, o sentido \'etico da presen\c ca humana no mundo, a compreens\~ao da hist\'oria como
possibilidade jamais como determina\c c\~ao, \'e substantivamente esperan\c cosa e, por isso mesmo,
provocadora da esperan\c ca.

\vspace{1em}

\hspace{15em} Paulo Freire, trechos de \emph{Pedagogia da Indigna\c c\~ao.}

\vspace{3em}

\hspace{25em}A esperan\c ca

\hspace{25em}Dan\c ca na corda bamba de sombrinha

\hspace{25em}E em cada passo dessa linha

\hspace{25em}Pode se machucar

\hspace{25em}Azar!

\hspace{25em}A esperan\c ca equilibrista

\hspace{25em}Sabe que o show de todo artista

\hspace{25em}Tem que continuar

 \vspace{1em}
 
 \hspace{15em}Aldir Blanc e Jo\~ao Bosco, \emph{ O B\^ebado e a Equilibrista}. }

 \newpage

\chapter*{\textsf{Agradecimentos}}

Esse \'e o fim de uma era. \'E ineg\'avel que o conhecimento t\'ecnico que eu adquiri durante esses 10 anos \'e imenso, mas o
  que vou guardar de mais  precioso do meu ``tempo de faculdade'' s\~ao as in\'umeras amizades que eu fiz durante esse tempo.
Por esse motivo, as pr\'oximas p\'aginas s\~ao as mais importantes de toda a tese.

Em primeiro lugar, agrade\c co de cora\c c\~ao ao meu querido orientador, Marcelo, que esteve sempre presente em 9 desses anos. 
Muito do que sou hoje \'e fruto do seu trabalho. Agrade\c co tamb\'em por todos os conselhos e discuss\~oes, incluindo 
especialmente as discuss\~oes sobre futebol. Mas essa n\~ao \'e a parte pela qual sou mais grata, porque eu sei que 
nessa parte ele tamb\'em se diverte. Eu devo a ele muitos agradecimentos por todas as horas que ele passou escrevendo projetos,
organizando eventos, cuidando dos v\'arios visitantes
e resolvendo burocracias para que eu e meus colegas de grupo pud\'essemos ter as oportunidades que tivemos
e que ajudaram a transformar o grupo Enlight no que ele \'e hoje. 
Agrade\c co por ter sido compreens\'ivel quando eu decidi trabalhar e por n\~ao ter me deixado desistir nos momentos de
fraqueza. Agrade\c co pelas in\'umeras horas dispensadas na revis\~ao minuciosa
desse texto. Ao Terra e  aos Terr\'aqueos  eu dedico tamb\'em esse trabalho, na esperan\c ca de que ele possa
ser \'util aos Terr\'aqueos futuros. Agrade\c co tamb\'em \`a Mimi e \`a Tat\'a por abrirem m\~ao de  um pouco do seu tempo em fam\'ilia para
que ele pudesse se dedicar \`a nossa orienta\c c\~ao.

Devo tamb\'em meus sinceros agradecimentos ao Professor Ad\'an Cabello, sem o qual esse trabalho n\~ao seria poss\'ivel.
Agrade\c co pelo incentivo e pelas in\'umeras horas de\-di\-ca\-das aos nossos trabalhos em colabora\c c\~ao,  pela aten\c c\~ao e 
pela simpatia de sempre. Agrade\c co a ele e tamb\'em \`a Carmen pela hospitalidade que tornaram meus dias em Sevilha t\~ao agrad\'aveis.

Agrade\c co o Professor Andreas Winter,  Emili Bagan Capella, John Calsamiglia Costa, Ramon Mu\~noz-Tapia,
Anna Sanpera, Marcus Hubber,  Claude Klockl, Alex Monras Blasi, Milan Mosonyi, Rub\'en Quesada, Stefan Baeuml
e a todo pessoal da UAB pela aten\c c\~ao dispensada 
durante minha estadia em 
Barcelona. Agrade\c co especialmente \`a Marionna e ao Elio por me fazerem me sentir em casa a 9 mil quilometros de dist\^ancia.
Agrade\c co  de cora\c c\~ao o Daniel Cavalcanti e   Ariel Bendersky por me ajudarem, especialmente no in\'icio.
Agrade\c co tamb\'em  a todo pessoal do ICFO.

Agrade\c co aos meu pais, \^Angela e Geraldo, por todo amor,
carinho e apoio incondicional. S\'o eles sabem o sacrif\'icio que
fizeram para que eu e minha irm\~a pud\'essemos chegar onde chegamos.
Agrade\c co a  minha irm\~a Luana, especialmente por todas as dicas de ingl\^es, ao \'Atila, especialmente pela obra de arte estilo anos 80 
que ilustra a dedica\'oria dessa tese, ao Di\'ogenes e a toda minha fam\'ilia
querida, epecialmente aqueles que estiveram mais pr\'oximos. Agrade\c co a Los Bochechas pelo apoio que s\'o uma fam\'ilia de verdade
pode nos dar.

Agrade\c co ao DEMAT e DEEST-UFOP por me apoiarem durante a realiza\c c\~ao desse trabalho, epecialmente durante meu afastamento.
Agrade\c co em especial \`a Fufa, \'Eder, J\'ulio, Wenderson, Vin\'icius, Edney, Monique, Isaque,
\'Erica e \'Erica, Tha\'is, Anderson, Fernando e Gra, Claudinha,  Tiago e Di  pelo companherismo, pelos
 momentos de divers\~ao e por dividirem comigo as ang\'ustias inevit\'aveis de quatro anos de doutorado.

Agrade\c co a todos os meus amigos da gradua\c c\~ao,  Diog\~ao e Camila , Samuca, que sempre cuidou de
mim t\~ao bem,  Marquinhos, Breno e Ana, Dudu, E(d)milson e \'Isis.

Agrade\c co a todo pessoal do Enlight. Ao  Professor Marcelo Fran\c ca por todas as discuss\~oes mas principalmente 
por todas as burocracias que ele teve que resolver por n\'os.
Agrade\c co ao Raphael, Pierre, Cristhiano, Pablo, Gl\'aucia e em especial \`a Nadja por tomar conta do suprimento de caf\'e.

Agrade\c co todos os meus professores da f\'isica e da matem\'atica e a todos os 
funcion\'arios dos dois departamentos especialmente ao pessoal da
secretaria da p\'os da matem\'atica e da biblioteca da f\'isica. 

Agrade\c co ao Matthias Kleinmann, Roberto Imbuzeiro, Ernesto Galv\~ao, Raphael Drumond, Remy Sanchis, Gast\~ao Braga,  Bernardo Nunes, Artur Lopes, Alexandre Baraviera, Andreas Winter e Ad\'an
Cabello por aceitarem
o convite de participar da avalia\c c\~ao desse trabalho.

Agrade\c co a todos os Diagonais, em especial ao Leo, ao Pablito e Anderson Silva pela hospedagem,  
ao Carlitos pelo ombro amigo nas
horas de desepero. \`A Ju e ao Robson por mesmo distantes estarem sempre comigo.

Agrade\c co ao meu companheirinho c\~ao, Tshabalala, por estar ao meu lado, literalmente, durante todo o processo de
escrita desse trabalho e tamb\'em ao meu companheiro, Thales, pelo apoio, pelo incentivo, pela paci\^encia, especialmente nessa
reta final que me impediu de estar com voc\^es tanto quanto eu gostaria. Agrade\c co a voc\^es por estarem do meu lado em todos
os aspectos da minha vida, que n\~ao teria a mesma gra\c ca se voc\^es n\~ao estivessem comigo.

Enfim agrade\c co a todas as pessoas maravilhosas que conheci durante
esse tempo e que permanecer\~ao no meu cora\c c\~ao pela vida toda.

\textsf{Ao apoio financeiro das ag\^encias CNPq, FAPEMIG e CAPES,em especial ao Programa Ci\^encias Sem Fronteira por nos possibilitar tantas parcerias de sucesso.}

\newpage
\thispagestyle{empty}
\mbox{}
\newpage
\tableofcontents

\newpage

\mainmatter
\chapter*{\textsf{Introduction}}
\label{chapterintro}
\addcontentsline{toc}{chapter}{Introduction}

\epigraph{Foi preciso que os
fil\'osofos e outros abstractos andassem j\'a meio perdidos na floresta das
suas pr\'oprias elucubra\c c\~es sobre o quase e o zero, que \'e a
maneira plebeia de dizer o ser e o nada, para que o senso comum se
apresentasse prosaicamente, de papel e l\'apis em punho, a demonstrar
por a + b + c que havia quest\~oes muito mais urgentes em que
pensar.}{Jos\'e Saramago, As Intermit\^encias da Morte.}

Quantum theory  provides a set of rules to predict probabilities of different outcomes in different experimental settings. 
While it predicts probabilities which match, with extreme accuracy, the data from actually performed experiments,
it has some peculiar properties which deviate it from how we normally think about 
systems which have a probabilistic description. Two of the ``strange'' characteristics are 
\emph{contextuality} and \emph{ nonlocality}.
The former tells us that we cannot think about a measurement on a quantum system as revealing a property which is 
independent of the set of measurements we chose to make. The later,  describes how measurements made by spatially separated
observers in a multipartite 
quantum system
 can exhibit extremely strong correlations.
Contextuality and nonlocality are the most striking features of quantum theory.
We believe that a complete understanding about these features may be the most important step towards understanding 
the whole theory.

 The necessity of the use of probabilities in the description of an experiment naturally arises
when we do not control all the  parameters involved in it. Our classical intuition leads us to think that
if we could control our devices with perfect accuracy, 
two repetitions of the same procedure with exactly the same value for every possible parameter  had to
provide the same result at the end.
It is natural to imagine that two replicas of the same object will remain identical if they are subjected to the exactly  
same process. If this is not the case, we would have no reason to call them identical in the first place. 

Quantum theory, on the other hand, does not provide definite outcomes for the measurements, even if we have
complete knowledge about the state of the system.
If we have a large set of quantum systems, all prepared in the same state,  we can apply the same measurement to all of them,
obtaining a probability distribution that in general  will exhibit dispersion. This means that for almost all measurements,
at least two outcomes have probability 
larger then zero. If we apply the argument of the previous paragraph,
we would conclude that the systems in this set
could not be identical and hence they could not all
be in the same state. Hence, the state assigned to this preparation by quantum theory can not be everything: there are more
parameters we must use in the description of these systems in order to get definite outcomes for all measurements.
This unknown parameters 
may have different values in our set of systems, and  the probabilistic behavior is due to our lack of knowledge about these
``hidden variables.''

This line of thought led many physicists to believe that quantum theory might be incomplete.  Hence, they conjectured the possibility of 
completing quantum theory, adding extra variables to the quantum description, in a way that with all 
this information (of quantum state plus extra variables) we would be able to predict with certainty the outcome of
all measurements and in a way that when averaging over these extra variables we would get the quantum predictions.
This kind of completion of quantum theory is often called a \emph{hidden-variable model}.

With some very reasonable extra assumptions on these models, we get
a contradiction with the  predictions of quantum theory. If the value associated by the model to a measurement 
 is independent of what other compatible measurements are jointly performed, we say that the model satisfy 
 the \emph{noncontextuality hypothesis}. This demand is consistent with what we expect from classical intuition: 
 physical quantities have predefined values which are only revealed by the measurement process. If
  these values exists prior to the measurement, how can they depend on some choice made at the moment of the  measurement?

 It happens that noncontextual hidden-variable models can not reproduce  quantum statistics. This result 
 is known as the Bell-Kochen-Specker theorem. The result was first proven by Kochen and Specker, and Bell pointed out the 
 assumption of noncontextuality, which was so natural that Kochen and Specker assumed it with no explicit 
 discussion. A huge number of proofs can be found
 on the literature, much simpler then the pioneer proof. 
One of the most common ways to provide a simple proof of this theorem is using the so
called \emph{noncontextuality inequalities}. They are linear inequalities involving the probabilities of 
certain outcomes of the joint measurement of compatible observables that must be obeyed by any
noncontextual hidden-variable model
and can be violated by quantum theory with a particular choice of state and observables.

One of the reasons for studying quantum contextuality  and quantum nonlocality 
is the belief that they are essential for understanding 
quantum theory the same way we understand special relativity. Special relativity can be derived from two simple 
physical principles: the light speed is constant and physics is the same for reference frames in uniform relative motion. 
We cannot do the same for quantum theory and this is one of the most seductive scientific
challenges in recent times. The starting point is assuming general probabilistic theories allowing for probability distributions
that are more general  than those that arise from Kolmogorov's axioms, and even from quantum theory, and
the goal is to find principles that pick out quantum theory from this
landscape of possible theories. There are many ideas on how to do this, and  at least three different approaches to the
problem stand out.

The first one consists of reconstructing quantum theory as a purely
operational probabilistic theory that follows from some
sets of axioms.  Imposing a small number of 
reasonable physical principles, it is possible  to prove that the only consistent probabilistic theory 
is quantum \cite{Hardy01, Hardy11, MM11,CDP11}. 
Although really successful, this approach does not resolve the issue completely, specially because some 
of the principles imposed do not sound so natural. Another drawback is that there is interesting and important quantum effects in simple 
systems (as opposed to composite) that can not be addressed this way.

In the second approach, instead of trying to reconstruct quantum theory, the idea is to 
understand what physical principles explain the nonlocal character of quantum theory.
 Many different principles have been proposed, the most important being non-triviality of communication complexity, 
 Information Causality, Macroscopic Locality and Local Orthogonality \cite{V05, PPKSWZ09, NW09, OW10}. None of them is known to solve the problem completely,
  but many interesting results have been found so far.

The third approach consists of identifying principles that explain the set of quantum contextual correlations  without restrictions imposed by
a specific experimental scenario. 
The belief that identifying the physical principle responsible for quantum contextuality provides a higher probability 
of success than previous
approaches is based on two observations. On one hand, when focusing on quantum contextuality we are just considering a 
natural extension of quantum nonlocality which is free of certain restrictions (composite systems, space-like separated tests 
with multiple observers, entangled states) which play no role in the rules of quantum theory, although they are crucial for many important applications, specially in communication 
protocols (see, for example, references  \cite{wikiquantumcryptography, HHHH09,BBCJPW93} and other references therein), 
and  played an important role in the 
historical debate on whether or not quantum theory is a complete theory.
On the other hand, it is based on the observation that, while calculating the maximum value of quantum correlations for 
nonlocality scenarios is a mathematically complex problem, calculating the maximum contextual value of quantum correlations for an 
{\em arbitrary} scenario  is the solution  of a semidefinite program \cite{CSW14,Lovasz95}.
The difficulties in characterizing quantum nonlocal correlations are due to the 
mathematical difficulties associated to the extra constraints resulting from enforcing a particular labeling of the 
events  in terms of parties, local settings, and outcomes, rather than a 
fundamental difficulty related to the principles of quantum theory.

Within this line of research, the most promising candidate for being {\em the} fundamental principle of quantum contextuality 
is the Exclusivity principle,   which can be stated as follows: 

\begin{center}The sum of the probabilities of a set of pairwise 
exclusive events cannot exceed~1.
\end{center}

By itself, the Exclusivity principle 
singles out the maximum quantum value for some important Bell and 
noncontextuality
inequalities. 
We can get better results if we apply the E principle to more sophisticated scenarios. This happens
because this principle exhibits \emph{activation effects}: a distribution satisfying this principles does not necessarily
satisfies it when combined with other distributions. Activation effects can be used to prove that
the Exclusivity principle singles out the set of quantum distributions for the most simple noncontextuality inequality.
It is still not known if the exclusivity principle
 solves the problem of explaining quantum contextuality completely, but many results have been proven that
support the conjecture that it might.
The main purpose of this thesis is to discuss in detail the situations in which the E principle can be used to rule out distributions outside the quantum set. 

In chapter \ref{chaptergpt} we start the discussion defining the generalized probability theories that are suitable for the description of states and measurements in a physical system
\cite{Barret06, BW12}. 
We will try to  keep the assumptions as general as possible, but for the purposes of this work it is sufficient
to consider a class of theories that satisfy further restrictions that do not have a physical meaning and will be made solely 
to simplify the description. Nonetheless, our framework is general enough to include as special cases the mathematical 
structure of finite dimensional quantum theory and classical probability theory with finite sample spaces.

In chapter \ref{chapterncinequalities} we discuss in detail the assumption of noncontextuality. 
 We present two different approaches, both connected with graph theory: 
 the compatibility-hypergraph approach and the exclusivity-graph approach \cite{CSW14}. 
 The graph-theoretical formulation of quantum contextuality supplies new tools to understand the 
 differences between quantum and classical theories and also the differences between quantum theory and more
general theories \cite{Cabello13, Yan13, ATC14}. 

The pioneer proof of Kochen and Specker is out of the scope of this thesis, but we present it in appendix \ref{chaptercontextuality}.
There the reader can  find a brief discussion on the first attempts to prove the impossibility of hidden-variables models compatible with quantum theory and
 other interesting state-independent proofs of the Kochen-Specker theorem.

In chapter \ref{chapterlovasz} we prove the recent results supporting the conjecture that the E principle might explain the set of quantum distributions in the exclusivity-graph approach
to quantum contextuality. 
The most important results are the ones we have proven in reference \cite{ATC14}. There we show that the Exclusivity principle singles out the \emph{entire set of quantum correlations} 
associated to any exclusivity graph assuming the set of quantum correlations for the complementary graph. 
Moreover,  for self-complementary graphs, the Exclusivity principle, {\em by itself} (i.e., without further assumptions), 
excludes any set of correlations strictly larger than the quantum set. Finally,  for vertex-transitive graphs, 
the Exclusivity principle singles out the maximum value for the quantum correlations assuming only the quantum maximum for 
the complementary graph.  We also show 
 that important results can be proven if we use  graph operations other then complementation and as a consequence
 we show that the exclusivity principle explains the quantum maximum for all vertex-transitive graphs with
 $10$ vertices, except two\footnote{If the E principle explains the quantum bound for one of them, the result of Yan \cite{Yan13} proves that the E principle also explain the quantum bound for the other.}. 
These results show that the Exclusivity principle goes beyond any other proposed principle towards the objective of 
singling out quantum correlations.

Since we made no original contribution to Bell inequalities, the concept of Bell scenarios will only be introduced in appendix \ref{chapternonlocality}. Bell scenarios provide a natural way to enforce  the noncontextuality assumption, since in these situations the experiment is designed in such a way that 
the choice of the different compatible observables to be measured
is made in a different region of the space in a time interval that forbids any signal to be
sent from one region to the other. Since no signal was sent, the choice of what is going to be measured in
one part can not disturb what happens in the other, what guarantees that the model is noncontextual.
In this situation, we say that the model is local and the noncontextuality assumption is usually referred to as the 
\emph{locality assumption}.

Although nowadays we may see quantum nonlocality as a special case of quantum contextuality, historically the discussion of nonlocality in quantum theory preceded the discussion about its noncontextual character.
Quantum nonlocality puzzled the famous trio  Einstein, Podolsky, and Rosen, who discussed this strange property of quantum theory in their  pioneer paper  
``\emph{Can Quantum-Mechanical description of Physical 
Reality Be Considered Complete?}'' in 1935 \cite{EPR35}.
They started one of the greatest debates in foundations of physics and philosophy of science 
in general, that is still fruitful nowadays.

The first one to provide a proof of the impossibility of \emph{local hidden-variable models} was John Bell, in 1964 \cite{Bell64}.
He demonstrated that if the statistics of joint measurements on a pair of two qubits in the singlet state were 
given by a hidden-variable model, a linear inequality involving the corresponding probabilities should be satisfied.
A simple choice a measurements leads to a violation of this inequality, and hence the model can not reproduce
the quantum statistics. 

Many similar inequalities were derived since Bell's work. Because of his pioneer paper, any inequality
derived under the assumption of a local hidden-variable model is called \emph{Bell inequality}.
Quantum theory violates these inequalities in  many situations. Besides the insight given in foundations of 
quantum theory, those violations are also connected to many interesting applications.

The quest for a principle that explains the set of quantum distributions in Bell scenarios has been very fruitful. 
For completeness,  a brief discussion can be found in appendix \ref{chaptertsirelson}.

We will state, and sometimes prove, many results that can be found in the literature. These results will be referred to as \emph{Theorems}.
The original results of the author and collaborators will be referred to as \emph{Propositions}. We will use a huge number of tools from many different areas of mathematics  and physics. This makes
 a proper introduction of some subjects impractical.
Typically, the necessary mathematical definitions will be given in the text, but nor its consequences, nor other previous necessary concepts will find room in the text. We list the concepts we will need, 
along with references where a proper discussion can be found.

\begin{enumerate}
\item Linear algebra: vector spaces, linear maps, matrices, basis, inner products, orthogonal complements, tensor products; Finite dimensional Hilbert spaces. 
An introduction to the the subject can be found in references   \cite{HK61,Lang87};

\item Convex Geometry: we assume that the reader is familiar with the notions of convex sets, convex sums, convex cones, polytopes and H-descriptions. The reader can learn about this subjects in references 
\cite{Rockafellar97};

\item Basic probability theory: finite sample spaces, $\sigma$-algebras and measures. We give a brief introduction in section \ref{sectionclassical} and suggest references \cite{SW95, GS01, James04}
for a more complete treatment.

\item Quantum theory in finite dimension. We present the mathematical aspects in section \ref{sectionquantum}.
We recommend references \cite{Feynman65,Cohen77,Peres95, NC00, Griffiths05, ATB11}.

\item Ordered linear spaces and order unit spaces \cite{Jameson70}.  

\item Category theory, morphisms, opposite category, symmetric monoidal category.  All these definition can be found in reference \cite{Lane98}.

\item Sheaf theory. We define very briefly the  objects we use and recommend reference \cite{MM92} for a complete treatment.

\end{enumerate}

We thank very much all who spent some of their time reading this work. Any comments, questions or suggestions are welcome.

\begin{flushright}\begin{minipage}[r]{2\columnwidth}
B\'arbara Amaral\newline barbaraamaral@gmail.com

 \end{minipage}\end{flushright}

\chapter{\textsf{Generalized Probability Theories}}
\label{chaptergpt}

In this chapter we study generalized probability theories that can be used to describe states and measurements
in a physical system. We will not focus on any particular kind of  system. Our intention is to discuss only the 
abstract mathematical structure behind the description and what the consequences are of assuming a particular type 
of theory. A number of requirements imposed by physical reasoning must be obeyed by all  theories in this framework 
and for now we will try to  keep the assumptions as general as possible. For the purposes of this work it is sufficient
to consider a class of theories that satisfy further restrictions that do not have a physical meaning and will be made solely 
to simplify the description. Nonetheless, our framework is general enough to include as special cases the mathematical 
structure of finite dimensional quantum theory and classical probability theory with finite sample spaces, the subjects 
of the  sections \ref{sectionquantum} and \ref{sectionclassical}, respectively. In section \ref{sectionstates} we define
states and measurements in a physical system and in section \ref{sectionmultipartite} we discuss the mathematical description
of a multipartite system. A  mathematical formalization  of  these concepts is presented in section 
\ref{sectionmoremathematics}. We finish this chapter with general properties of the theories in section 
\ref{sectiongeneral}.

\section{\textsf{States and Measurements}}
\label{sectionstates}

As we said above, our purpose in this chapter is to find a suitable mathematical structure that we can 
apply in the description of experiments carried in a hypothetical physical system. We follow
the ideas presented by Barrett in reference \cite{Barret06}. 

Our first assumption 
is about the nature of the experiments that can be performed in this system. We assume that there are two kinds 
of experiments available: preparations and operations.  Another important requirement is that these experiments be
repeatable: every preparation and every operation can be done as many times as we want and we can use several 
repetitions of a given procedure to count relative frequencies. For each operation there may be several different 
outcomes, each occurring with a well defined probability for a given preparation. Preparations can be compared through 
their statistics in relation to the given operations, and these statistics define a state.

\begin{defi}
Two preparations are equivalent if they give the same probability distribution for all available operations. 
The equivalence class of preparations is called a \emph{state}.
\end{defi}

\begin{defi}
A set of operations is called \emph{informationally complete} or \emph{tomographic} if the list of probabilities
for the outcomes of these operations  
completely specifies the state of the system. 
\end{defi}

For every system there is  a set of tomographic operations. In the worst case scenario, we can take 
the entire set of operations as a tomographic set. This is not the case in general, 
since only a small subset of the available operations is needed to describe the state completely. 
The set of tomographic operations is not unique and we will not assume it to be a minimal set, in the 
sense that it might be the case that removing some operations we still get a tomographic set. This set is not always 
finite, but we will only consider the cases in which a finite tomographic set exists.

\begin{as}
The state of the system can be completely specified by listing the probabilities of the outcomes of a finite set of tomographic 
operations each of them with a finite set of possible outcomes.
\label{asfinitetomographical}
\end{as}

This restriction is not a physical requirement and it is really easy to come up with real physical systems that 
require an infinite set of tomographic operations or tomographic operations with an infinite number of outcomes. 
We are just narrowing down the kind of problems we will deal with in this work.

If we fix the set of tomographic operations $\{M_1,M_2, \ldots, M_n\}$, each $M_i$ with outcomes $\{1, 2, \ldots, m_i\}$, every state can be represented by a list of probabilities:

\be
P=\left[\begin{array}{c}
p(1|M_1)\\
\vdots\\
p(m_1|M_1)\\
p(1|M_2)\\
\vdots\\
p(m_2|M_2)\\
\vdots\\
p(1|M_n)\\
\vdots\\
p(m_n|M_n)\end{array}\right] \in \mathbb{R}^d\ee
in which $p(i|j)$ is the probability of outcome $i$ given that the operation $j$ was applied and 
$d=\sum_{i=1}^{n}m_i$. Since the entries represent probabilities, we have $p(i|j) \geq 0$ and
$$\sum_i p(i|j)=1$$
for every tomographic operation $j$.
Nevertheless, it will be convenient to use also subnormalized states with
\be\sum_i p(i|j)=p \label{unorm}\ee
where $0 \leq p \leq 1$ and $p$ is independent of the tomographic operation $j$. The value $p$ is called the
norm of the state $P$ and will be denoted by $|P|$. 
These subnormalized  states have a physical interpretation: suppose an operation $j$ is performed in a normalized state 
and an outcome $i$ is obtained with probability $p$ less than one. There is a subnormalized state of the 
form \eqref{unorm} associated with this outcome, and each entry $p(k,i|l,j)=p(i|j)\cdot p(k|l)$ of this state corresponds to the probability of 
obtaining outcome $i$ in operation $j$ followed by outcome $k$ in the tomographic operation $l$.

With this interpretation, the vector with all entries equal to zero, denoted by $\overrightarrow{0}$, is an allowed 
(subnormalized) state of every system. This state can be prepared in the following way: suppose we prepare a state for
which outcome $i$ of operation $M$ has probability zero; each entry $p(k|j)$ of the state of the system associated to this
outcome  is the probability of getting $i$ in the first operation and $k$ in the tomographic operation $j$, 
and since outcome $i$ is a zero probability event, all the entries of this vector are zero.

\begin{as}
For each  system the set of allowed normalized states is closed and convex. 
The complete set of states $\mathcal{S}$ is the convex hull of the set of allowed normalized states 
and $\overrightarrow{0}$. The set $\mathcal{S}$ is called the \emph{state space} of the system.
\label{asstatespace}
\end{as}

\begin{defi}
The extremal points of the state space $\mathcal{S}$ are called \emph{pure states}.
The points that are not extremal are called \emph{mixed states}, and can be written as a convex sum of pure states. 
Convex sums are also called \emph{mixtures}.
\end{defi}

\begin{defi}
 We say that a state is \emph{dispersion free} if it provides definite outcomes for all measurements, that is, if 
 for every measurement there is one outcome with probability one.
\end{defi}

If a model admits dispersion free states, then these states are pure. The converse is not always true: some models
may admit pure states that are not dispersion free. This is the case of quantum theory, as we will see 
in section \ref{sectionquantum}.

 When an operation $M$ is performed,  each outcome $i$ is associated to a transformation $f_i$ of the state of the system:
 \be P \mapsto f_i(P).\ee
  The entry $p(k|j)$ of $f_i(P)$ is the probability of obtaining outcome $i$ in operation $M$ followed by outcome $k$ 
  in the tomographic operation $j$.
  Operations with only one outcome preserve normalization.  If the transformation is associated with an outcome 
  that occurs with probability $p< 1$, then it decreases the norm of the state
by a factor of $p$.
\begin{defi}
Operations with more then one outcome are called \emph{measurements}.
\end{defi}

\begin{as}
 We require that the  transformations  preserve mixtures. This means that if
 \begin{subequations}
 \be
 P=\sum_ip_iP_i
 \ee
 then
 \be
  f(P)=
  \sum_ip_if(P_i).
  \label{fmixture}
  \ee

  The physical interpretation of the vector $\overrightarrow{0}$ requires that
  \be
  f\left(\overrightarrow{0}\right)=\overrightarrow{0}.
  \label{fzero}
  \ee
  \label{subeqf}
  \end{subequations}
  \end{as}

In fact, state vector $\overrightarrow{0}$ is prepared when we condition on an outcome $i$ of a measurement $j$
that happens with probability zero. Let $f$ be associated to outcome $k$ of some measurement $l$. Then the entry 
$p(r|s)$ of $f(\overrightarrow{0})$ is the probability of obtaining outcome $i$ in the  measurement $j$, followed by 
outcome $k$ in measurement $l$, followed by outcome $r$ in tomographic measurement $s$. 
Since outcome $i$ is a zero probability event in the first place, all these entries are zero and  
equation \eqref{fzero} follows.

The conditions above imply that we can take $f$ to be linear \cite{Barret06}.

\begin{teo}
The transformation $f$ associated to an operation acting on the state of a physical system can be extended to
a linear operation on $\mathbb{R}^d$.
\end{teo}

\begin{dem}
Equations \eqref{subeqf} imply that $f(rP)=rf(P) \ \forall \ P \in \mathcal{S}$ and $0\leq r \leq 1$. In fact, under these conditions
\be f(rP)=f\left(rP + (1-r)\overrightarrow{0}\right)=rf(P) +(1-r)f(\overrightarrow{0})=rf(P).\label{frP}\ee

Suppose $P\in \mathcal{S}$  and $r>1$. If $rP=P' \in \mathcal{S}$, then $f(rP)= rf(P)$ since $f(P)=f\left(\frac{1}{r}P'\right)$ and 
by equation \eqref{frP}, $f\left(\frac{1}{r}P'\right)=\frac{1}{r}f(P').$ If $rP \notin \mathcal{S}$, we can extend $f$ using the rule
$$ f(rP)=rf(P).$$


Let $\mathcal{S}_+$ be  the set of vectors of the form $rP, \ P \in \mathcal{S}, \ r\geq 0.$ This set 
is a convex cone and $f(rP)=rf(P) \ \forall \ P \in \mathcal{S}_+$ and $ r \geq 0$. It is also true that
\be f\left(\sum_i r_iP_i\right)=\sum_i r_if(P_i), \ \forall \ P_i \in \mathcal{S}_+, \ r_i\geq 0.\label{fs+2}\ee
To prove this, let $P_i=s_iP'_i, \ s_i \geq 0, \ P'_i \in \mathcal{S}$ and $c=\sum_ir_is_i$. Then
$$f\left(\sum_i r_iP_i\right)=f\left(c\sum_i \frac{r_is_i}{c}P'_i\right)$$
and since $\sum_i \frac{r_is_i}{c}P'_i \in \mathcal{S}$
$$f\left(c\sum_i \frac{r_is_i}{c}P'_i\right)=cf\left(\sum_i \frac{r_is_i}{c}P'_i\right)=c\sum_i \frac{r_is_i}{c}f\left(P'_i\right)=
 \sum_i r_is_if\left(P'_i\right)= \sum_i r_if\left(P_i\right).$$

Now we prove that equation \eqref{fs+2} is also true if the coefficients $r_i$ are real.
Let $Q \in \mathcal{S}_+$ such that
$$Q=\sum_i r_iP_i, \ \  P_i \in \mathcal{S}_+, \ r_i \in \mathbb{R}.$$
We can rewrite the above expression as
$$Q + \sum_{r_i <0}|r_i|P_i= \sum_{r_i >0}r_iP_i$$
and applying $f$ to both sides of this equation we get
$$f(Q) + \sum_{r_i <0}|r_i|f(P_i)= \sum_{r_i >0}r_if(P_i)$$
which implies
$$f(Q)=\sum_ir_if(P_i).$$

This proves that $f$ is linear in $\mathcal{S}_+$. If $Q$ belongs to the subspace spanned by $\mathcal{S}_+$, $f(Q)$ can be defined uniquely by linear extension. The action on the orthogonal complement of this subspace is arbitrary and we can define it to be linear.  Then $f$ can be extended linearly to the rest of the vector space $\mathbb{R}^d$.
\end{dem}

This result implies that every transformation can be written as
\be f(P)=MP \ee
where $M$ is a matrix acting on $\mathbb{R}^d$.

An operation is associated to a set of matrices $\{M_i\}$, each $M_i$ corresponding to an outcome $i$ of this operation. The subnormalized state associated to outcome $i$ is $M_iP \in \mathcal{S} $ and the unnormalized probability of $i$ is $|M_iP|$. This means that if $P$ is normalized, the probability of outcome $i$ is $|M_iP|$.

As one should expect, not every set of matrices $\{M_i\}$ corresponds to a valid operation on the system, since some physical requirements must be satisfied.

\begin{cons}
If a set of matrices $\{M_i\}$ represents an operation, the following conditions must hold
\begin{enumerate}
\item Positivity: $0\leq \frac{|M_iP|}{|P|} \leq 1, \ \forall i, \ \forall P\in \mathcal{S}\setminus \{\vec{0}\}$; \label{coperation1}
\item Normalization:  $\sum_i\frac{|M_iP|}{|P|}=1, \ \forall  P \in \mathcal{S}$; \label{coperation2}
\item State preservation: $M_iP \in \mathcal{S}, \ \forall  P \in \mathcal{S}$; \label{coperation3}
\item Complete state preservation: Each transformation $M_i$ must result in allowed states when it acts on a system that is a part of a larger multipartite system. \label{coperation4}
\end{enumerate}
\label{coperation}
\end{cons}

Item \ref{coperation1} of constraint \ref{coperation} must be satisfied because the probability of an outcome is a real number between zero and one. 
Item \ref{coperation2} follows from the fact that the sum of the probability of all outcomes must be one. 
Itens \ref{coperation3} and \ref{coperation4} follow from the fact that any transformation must take an allowed state 
to another allowed state, whether we considerer the system alone or as a part of a larger system composed of several parties. 
We will talk about item \ref{coperation4} again in section \ref{sectionmultipartite}.

\begin{as}
For each system there is a set $\mathcal{T}$ of allowed transformations.  
This set is convex and includes the transformation that takes all $P$ to 
the vector $\overrightarrow{0}$.\label{asallowedtransf}
\end{as}

\begin{defi} An operation is a set of allowed transformations 
$\{M_i\}$, $M_i \in \mathcal{T}$, satisfying constraint \ref{coperation}.
\end{defi}

The set $\mathcal{T}$ can be viewed as a set of possible outcomes for the available operations, each outcome 
represented by a matrix
$M_i \in \mathcal{T}$. Distinct operations may share some outcomes, since a matrix $M_i$ can appear in different measurements. 
The probability of a given outcome does not depend on the measurement in which it appears.

\begin{defi}
 The pair $\left(\mathcal{S}, \mathcal{T}\right)$ is called a \emph{probabilistic model}. 
 A \emph{probability theory} is a collection of
 probabilistic models.
 \label{defimodel}
\end{defi}

The same model can describe different systems. This happens because the description of a real physical system also depends
on how we connect the real experiments with the mathematical objects in the model.
It is also possible that the same system is described by apparently  different models. For example, we could use a different set of 
tomographic measurements and obtain a model in a different vector space and consequently,  a different set of matrices representing 
allowed operations. This difference is irrelevant, since the physics represented by each of them is the same. 

\begin{defi}
 Two probabilistic models $\left(\mathcal{S}_1, \mathcal{T}_1\right)$ and $\left(\mathcal{S}_2, \mathcal{T}_2\right)$
  are \emph{equivalent} if there exist linear bijections
 \begin{eqnarray*}
  \xi : \mathcal{S}_1 &\longrightarrow &\mathcal{S}_2\\
    \zeta : \mathcal{T}_1 &\longrightarrow &\mathcal{T}_2
 \end{eqnarray*}
such that
\begin{equation*} \left|MP\right|= \left|\zeta\left(M\right)\xi\left(P\right)\right|\end{equation*}
for every $M \in \mathcal{T}_1$ and every $P \in \mathcal{S}_1$.\footnote{We do not assume that $\mathcal{S}_1$ and $\mathcal{S}_2$
are subsets of the same real vector space, that is, the number of entries in the vectors representing the states does not have to
be the same.}

\end{defi}

\begin{defi}
If two models belong to the same equivalence class  under the equivalence above, we say that
they describe the same  \emph{type of system}.
\end{defi}

All models describing a given type of system  are equally good. Some of them might be more practical or more appropriate in a 
particular situation, but the choice of 
one instead of the others is just a mater of taste.

\subsection{\textsf{Repeatability}}

In the beginning of this section we mentioned that experiments must be repeatable. 
This means that  every preparation and operation we consider
can be done as many times as we want in the same conditions, what 
allow us to define the statistics of every sequence of experiments. The word \emph{repeatability} will be used again 
with a 
different meaning in the definition of \emph{repeatability of outcomes}. We apologize for the inconvenient use of the same word
for both concepts, but we have no better option in neither case.

\begin{defi}
 A measurement $i$ has \emph{repeatable outcomes} if every time this measurement is performed and an outcome $k$ is obtained, a  subsequent measurement
 of $i$ gives outcome $k$ with probability one.
\end{defi}

In this chapter we still allow measurements with non-repeatable outcomes. In some cases it might be
important to restrict the discussion to the case of repeatable outcomes, and we will do that further when we talk 
about contextuality.

\subsection{\textsf{Compatibility for outcome-repeatable measurements}}

One of the implications of a more general theory for computing probabilities than the usual classical probability theory 
is that in some cases there is not a well defined probability for the results of all measurements in a given set. 
When this global probability distribution exists for all states, we say that the measurements are compatible. This is not new for the reader familiar 
with quantum theory, where non-compatibility is the rule, not the exception. 
 
\begin{defi}
A set of outcome-repeatable measurements $\{j_1, \ldots, j_n\}$  is \emph{compatible} if there is another measurement $j$ with
outcomes 
$\{1, \ldots, m\}$  and functions $f_1, \ldots, f_n$ such that the possible outcomes of each
$j_s$ are $f_s\left(\{1, \ldots, m\}\right)$ and
\be p\left(i|j_s\right)=\sum_{k \in  f^{-1}_s(i)}p\left(k|j\right).\ee
The measurement $j$ is called a \emph{refinement} of each $j_i$, and each $j_i$ is called a \emph{coarse graining} of $j$.
\label{deficompatible}
\end{defi}

If the measurements $\{j_1, \ldots,j_n\}$ are compatible, the probability of a set of outcomes
$i_1, \ldots, i_n|j_1, \ldots, j_n$ is well defined and it is equal to the probability of outcomes 
$\bigcap_kf^{-1}_k(i_k)$ for measurement $j$. 

The notion of compatibility is essential in quantum theory, specially in the problems of non-contextuality we will present 
in chapter \ref{chapterncinequalities}. It is connected to the idea of ``measurements that can be performed at once''. If a set of measurements is 
 compatible, they can be measured jointly on the same individual system without disturbing the results of each other. In practice, to measure all of them
at the same time we apply measurement $M$ in definition  \ref{deficompatible} and then use functions $f_i$ to find out the outcomes of each $M_i$.
 Compatible measurements can be made simultaneously or in any order and can be repeated any number of times in the same system and repeatability
 of the results must be preserved. We will come back to 
this subject many times in the text and in section \ref{sectionquantum} we will see how non-compatible measurements appear in quantum theory.

\section{\textsf{Multipartite systems}}
\label{sectionmultipartite}

In this section we will see how we can describe multipartite systems in general probability theories. 
As for the simple systems, the probability theories used for composite systems must obey some requirements that 
come from natural physical assumptions.

\begin{as}
For every system composed of several parties, we assume that operations that act on only one of the parties are allowed. 
These operations are called \emph{local operations}.
\label{localoperations}
\end{as}

Although the parties do not need to be spatially separated, this is the case most of the times we deal with multipartite
systems. Thais motivates the use  of the word \emph{local} for the operations acting in only one party of the system.

\begin{as}[Local operations commute] Suppose that for each subsystem $i$ of a multipartite system, an operation $M_i$ is 
performed.  Then the state of the composite system  after the sequence of operations $M_i$ does not depend on  the particular order in which the operations were applied.
\label{locommute}
\end{as}

This assumption means that local operations can be regarded as performed simultaneously on each subsystem. This implies that
for each measurement the joint probabilities 
$$p(r_1, \ldots, r_n|M_1, \ldots, M_n)$$ are well defined, where $r_i$ is the outcome of measurement $M_i$ on party $i$.

An important corollary of assumption \ref{locommute} is that for all composite systems  no-signaling  holds \cite{Barret06}. 
This property states that any of the parties  cannot signal its choice of input to the others. Physically, this is a reasonable restriction: since there may be a large spatial separation between the parties, signaling between them would potentially require faster-then-light communication, which would violate the most fundamental principle of special relativity.
 
\begin{cor}[No-signaling]
If an operation was performed on  system $i$, it is not possible to get information about which operation was performed 
by measuring another system $j$.
\end{cor}

\begin{dem}
Suppose an operation $M_i$ was performed on system $i$ and afterwards we apply operation $M_j$ on system $j$.
By assumption \ref{locommute}, the probability of getting outcome $k$ for measurement $M_j$ in this sequence of 
operations is equal to the probability of this outcome if $M_j$ was performed first and then
$$p(k|M_i,M_j)=p(k|M_j),$$
which implies that $p(k|M_i,M_j)$ does not depend on measurement $M_i$. This implies that no 
information on $M_i$ can be gained by any measurement in system $j$.
\end{dem}

\begin{as}[Local Tomographic Principle]
The global state of a multipartite system can be completely 
determined by specifying the joint probabilities of outcomes for local tomographic measurements.
\label{asglobalstate}
\end{as}

Given a system composed of $n$ parts, it follows from the above assumption that a state of the system can be described by a vector with entries of 
the form
$$p(r_1, r_2, \ldots, r_n|M_1,M_2, \ldots,M_n),$$ where $r_i$ is an outcome of a tomographic measurement $M_i$ acting only on party $i$.

The normalized states of the composed system must satisfy
$$\sum_{r_1, \ldots, r_n} p(r_1, r_2, \ldots, r_n|M_1,M_2, \ldots,M_n)=1$$
but, as before, we allow subnormalized states as well.
The no-signaling principle implies that, for any bipartition $\{S, S^{C}\}$ of the set $\{1, \ldots, n\}$, the marginal distribution for the parties in $S$ obtained by summing over all outcomes of the parties $i \in S^{C}$
\be\sum_{r_i, i\in S} p(r_1, r_2, \ldots, r_n|M_1,M_2, \ldots,M_n)\label{reduceds}\ee
does not depend on the measurements $M_i$ with $i \in S$. This means that marginal probability distributions are 
well defined and this allows the definition of the \emph{reduced state} of a subsystem $i$, as the vector with entries given 
by\footnote{\textsf{We can also define the reduced state of a subset $S$ of parties in an analogous form,
using equation \eqref{reduceds}.}}
\be p(r_i|M_i)=\sum_{r_j, j\neq i}p(r_1, r_2, \ldots, r_n|M_1,M_2, \ldots,M_n).\label{reduced}\ee

\begin{defi}
For every state of a multipartite system described by joint probabilities of the 
form $p(r_1, r_2, \ldots, r_n|M_1,M_2, \ldots,M_n)$, the marginal distribution
$p(r_i|M_i)$ is well defined  and is called the \emph{reduced state} of party $i$.
\end{defi}

From now on, every time we refer to a multipartite system we will use only joint probabilities of local tomographic measurements to describe its state and for every subsystem we will use the same set of tomographic measurements to describe its reduced state. The connection is given by equation \ref{reduced}.

As expected, a natural constraint we will impose is that the reduced state of each subsystem is an allowed state of this subsystem.

\begin{cons}
Let $\mathcal{S}$ be the set of allowed states for a multipartite system and $\mathcal{S}^i$ be the set of allowed states for a subsystem $i$. Let $P \in \mathcal{S}$, and $P_i$ be the reduced state of subsystem  $i$. We require that $P_i \in \mathcal{S}^i$.
\end{cons}

The result below gives a connection between the vector spaces associated to the individual systems and the vector space associated to the composite system \cite{Barret06}.
\begin{teo}
Let $P$ be a state of a multipartite system and $P_i$ be the reduced state of party $i$. If 
$P$ belongs to the vector space $V$ and each $P_i$ belongs to the vector space $V^i$, then
$$V\equiv \bigotimes_i V^i.$$
\label{tensor}
\end{teo}
\begin{dem}
We will prove the statement above for the particular case of bipartite systems. 
Since we only consider finite dimensional systems, the general case follows if we apply the particular case several times.

Let $Q^{12}_{ijkl}$ be the vector with entry $1$ for outcome $i$ of tomographic measurement $k$ in party $1$ and 
outcome $j$ for tomographic measurement $l$ in party $2$ and $0$ elsewhere. 
Define the vectors $Q^{1}_{ik}$ and $Q^{2}_{jl}$ analogously. 
Notice that these vectors are not necessarily allowed states of the system. 
Nevertheless, the vectors $Q^{12}_{ijkl}$ generate $V$, the vectors $Q^{1}_{ik}$ generate $V_1$,
the vectors $Q^{2}_{jl}$ generate $V_2$ and
$$Q^{12}_{ijkl}=Q^{1}_{ik} \otimes Q^{2}_{jl},$$
which implies the desired result.
\end{dem}

We can prove that any state of the composite system can be written as a \emph{linear combination} of product states \cite{Barret06}.
\begin{teo}
Any state of a $n$-partite system $P$ can be written in the form
\be P=\sum_i q_iP_i^1\otimes P_i^2\otimes \ldots \otimes P_i^n \label{eqlcproduct}\ee
where $P_i^j$ is a normalized and pure state of the party $j$ and $q_i \in \mathbb{R}$.
\label{theoremlcpure}
\end{teo}

\begin{dem}
We will once more prove the statement for $n=2$, since the general case follows easily from this one.

Consider a composite system consisting of parties $1$ and $2$ in  state $P \in V=V^1 \otimes V^2$. By assumption \ref{localoperations}, for each tomographic measurement $l$ in party $2$ there is  one operation on 
the composite system that corresponds to performing that measurement. Let $\{M_{jl}\}$ be the set of matrices representing this operation, $j$ labeling the
 possible outcomes.

Let $P_{jl}=M_{jl}P$ be the final state after outcome $j$ and let $P^1_{jl}$ be the corresponding reduced state of system $1$. Then
\be P=\sum_{j,l} P_{jl}^1 \otimes Q^2_{jl}\label{tensor1}\ee
where the vector $Q^2_{ij}$ was defined in the proof of theorem \ref{tensor}.

To prove equation \eqref{tensor1}, let us compare the entries of $P$ and $P_{jl}^1 \otimes Q^2_{jl}$. Each entry of $P$ is of the form $p(i,j|k,l)$, which is the probability of  outcome $i$ for tomographic measurement $k$ in system $1$ and  outcome $j$ of tomographic measurement $l$
in system $2$. An entry  of $P_{jl}^1 \otimes Q^2_{jl}$ is non-zero iff it is in position $(i,j|k,l)$ for some outcome $i$ of tomographic measurement $k$ in party $1$.  This entry is equal to the entry $(i|k)$ of $P_{jl}^1$, which is the probability of outcome 
$j$ for tomographic measurement $l$ in system $2$ followed by outcome
$i$ for tomographic measurement $k$ in system $1$. Since local operations commute, equation \eqref{tensor1} follows.

Let $U \otimes W\in V$  with $U \in (\mathcal{S}^1)^\bot$. Then equation \eqref{tensor1} implies that
$$(U \otimes W)P=0.$$

Repeating the same argument but exchanging the parties, we conclude that for any vector of the form $U \otimes W$ with $W \in (\mathcal{S}^2)^\bot$ we have
$$(U \otimes W)P=0.$$

This implies that $P$ belongs to the subspace generated by $U \otimes W$, $U \in \mathcal{S}^1$ and $W \in \mathcal{S}^2$. Since each $\mathcal{S}^i$ is generated by
the states that are normalized and pure, the result follows.
\end{dem}

States of the form $P_i^1\otimes P_i^2\otimes \ldots \otimes P_i^n$ are called \emph{product states}. If a state  can 
be written as a convex combination of product states, that is, if we can choose
the coefficients $q_i$ in equation \eqref{eqlcproduct} such that $0\leq q_i \leq 1$ and $\sum_iq_i=1$,
it is called a \emph{separable state}. States that can not be written in this form are called \emph{entangled}.

Consider a composite system and a transformation $T^1$ acting in subsystem $1$, represented by the matrix $M^1$. We know that this transformation is allowed in the composite system and that the resulting effect  is linear. Hence there is a matrix $\tilde{M}^1$ such that the transformation on the composite system is given by
$$P \mapsto P'= \tilde{M}^1P.$$
We want to find out what  the relation is between $M^1$ and $\tilde{M}^1$ \cite{Barret06}.

\begin{teo}
Consider a multipartite system in a state $P$ and a local transformation $M^1$ on subsystem $1$ , defined by
$$P_1\mapsto P'_1=M^1P_1.$$ 
The joint transformation on the composite system is given by
\be P\mapsto P'=(M^1\otimes I \otimes \ldots \otimes I)P.\label{equationmatrixlocal}\ee
\label{matrixlocal}
\end{teo}

\begin{dem}
We will once more prove the statement for a bipartite system, since the general case follows  from this one.

Let the set of tomographic measurements of systems $1$ and $2$ used to write $P$ and $P'$ be fixed.
Consider the following procedure: apply $T^1$ to system $1$ and then the tomographic measurements  of systems $1$ and $2$. 
The entries of the vector $P'$ give the probability of each possible outcome of this procedure.
By assumption \ref{locommute}, the order of the operations in systems $1$ and $2$ does not matter and this procedure is
equivalent to: first apply the tomographic measurement in system $2$, then apply $T^1$ to system $1$ and then apply 
the tomographic measurement to system $1$. The probabilities for the possible outcomes of this procedure 
are also given by $P'$.

The probability of  outcome $j$ for  tomographic measurement $l$ in system $2$ and  outcome $i$ for tomographic measurement 
$k$ in system $1$, before transformation $T^1$ is 
applied, is given by entry $P_{ijkl}=p(i,j|k,l)$ of vector $P$.  After transformation $T^1$ is applied, the
  outcome $j$ for  tomographic measurement $l$ in system $2$ and  outcome $i$ for tomographic measurement 
$k$ in system $1$ is 
$$P'_{ijkl}=\sum_{i'k'} M^1_{ik,i'k}P_{i'jk'l}=\left[(M^1 \otimes I)P\right]_{ijkl}.$$
This implies that the action of $\tilde{M}^1$ in $\mathcal{S}$ is equal to the action of $M^1\otimes I$. Since the action of $\tilde{M}^1$ outside $\mathcal{S}$ is arbitrary, we can take
$\tilde{M}^1=M^1 \otimes I$.

\end{dem}

Now that we know how   the action of local operations is in the description  of composite systems, 
we can go back to item \ref{coperation4} of constraint \ref{coperation} and see how this restricts the allowed transformations in each subsystem. We have stated that each local transformation $M_i$ on a subsystem $i$ must result in a allowed state of the  multipartite system as well.
This means that not only $M_i$ has to be an allowed transformation of system $i$, $M_i \otimes I\otimes \ldots \otimes I$ has to define an allowed transformation on the composite system. This extra requirement may reduce  even further the set of allowed transformations  in the individual system $i$.

\begin{defi}
A transformation $T$ on a system $1$, represented by matrix $M$, is \emph{well defined} if 
$$(M\otimes I)P^{12} \in \mathcal{S}^{12}$$
for all states $P^{12} \in \mathcal{S}^{12}$, where system 2 can be any other system allowed by the theory. 
\label{wdlocal}
\end{defi}

\begin{cons}
For each system, all transformations in $\mathcal{T}$ must be well defined.
\label{constraintwdlocal}
\end{cons}

System $2$ can itself be a multipartite system, so the general requirement of item \ref{coperation4} of constraint 
\ref{coperation} is implied by 
the special case of bipartite systems of definition \ref{wdlocal} and constraint \ref{constraintwdlocal}. 

Assumption \ref{localoperations} together with theorem \ref{matrixlocal} imply that the allowed transformations of a composite system must include the ones given by equation \eqref{equationmatrixlocal}.

\begin{cor}
If $M^1$ is an allowed transformation on system $1$, then $M^1 \otimes I$ is an allowed transformation of a composed system consisting of system $1$ and another arbitrary system $2$.
\end{cor}

We desire that our description include the possibility of multipartite systems with no correlation among its parties. This is quite natural: imagine that the parties of this system are thousand of kilometers apart and that none of them interacted in the past. We do not expect any correlation among the outcomes obtained in local measurement performed in these subsystems, and this implies that the joint probabilities are independent:
\be p(r_1, r_2, \ldots, r_n|M_1,M_2, \ldots,M_n)=p(r_1|M_1)p(r_2|M_2)\ldots p(r_n|M_n) \label{localindependent}\ee
where $r_i$ is the outcome of local measurement $M_i$ on party $i$.

\begin{as}
If $P^1$ is an allowed state of system $1$ and $P^2$ is an allowed state of system $2$, then $P^1 \otimes P^2$ is an allowed state of the system composed of parties $1$ and $2$.
\label{asindependentstates}
\end{as}

The state $P^1 \otimes P^2$ gives independent probabilities for the bipartite system, in the form of equation 
\eqref{localindependent}. The meaning is that system $1$ is in state $P_1$, system $2$ is in state $P_2$
and they are independent. Again, since system $2$ can itself be a multipartite system, assumption 
\ref{asindependentstates} also implies that any vector of the form \eqref{localindependent} is an allowed state 
of the system composed of parties $1,2,\ldots n$ in which party $i$ is in state given by the probabilities $p(r_i|M_i)$.

The next assumption is another simplification without  physical meaning. We will include in the set  $\mathcal{T}$ all  transformations that are mathematically well defined. There is no physical requirement that guarantees that this is indeed the case. For a particular kind of system, it is possible that nature forbids, for some reason, some of the transformations contained in this set. As our intention is to be general, we will define $\mathcal{T}$ to be the largest set of mathematically allowed transformations.

\begin{defi}
 A probability theory is called \emph{maximal} if the set $\mathcal{T}$ coincides with the 
 set of all  mathematically well defined transformations. 
\end{defi}

\begin{as}
All probability theories considered from now on are maximal.
\label{allowedt}
\end{as}

A number of corollaries follows from this assumption. 
The first one is something we would like to have in our theories: the composition of 
two allowed transformations is an allowed transformation. Mathematically, 
composition of transformation represented by matrices $M$ and $N$ is given by the product $MN$. Then,
if  $M$ and $N$ are matrices associated to allowed transformations of a system, we expect that  $MN$ 
is also an allowed transformation of the same system, and this is indeed the case if $\mathcal{T}$ satisfy 
assumption \ref{allowedt}.

\begin{cor}
\label{corcomposition}
If $M, N \in \mathcal{T}$, then  $MN \in \mathcal{T}$.
\end{cor}

Suppose we start with system $1$ in a state $P_1$ and we append 
another independent system $2$ in state $P_2$. As we know, the state of the system 
composed of subsystems $1$ and $2$ is $P_1 \otimes P_2$. Suppose that we apply an operation 
to the composite system, taking  $P_1 \otimes P_2$ to another state $P'$, not necessarily a product state. This 
state gives a reduced state $P'_1$ that is an allowed state of system $1$. This kind of procedure can be used to 
perform transformations on system $1$ alone, and system $2$ is just used as an ancilla that can be discarded after the process is completed.

\begin{cor}
A procedure consisting on appending an ancilla to system $1$, performing a joint operation on the composed system, 
and then throwing the ancilla away is a well defined transformation on system $1$.
\end{cor}

Physically we already have everything we need in our probabilistic theories. We can add some mathematical structure to our description
without having to restrict it any further. The reader may skip the next section with no prejudice for the understanding of the rest of the text.

\section{\textsf{A little bit of Category Theory}}
\label{sectionmoremathematics}

Previously we have defined a probabilistic model using vectors in $\mathbb{R}^d$ as states and matrices acting in this vector space as 
transformations. We can provide a more formal and general definition.
The point of view we present here is a simplification of the approach of Barnum and Wilce in reference \cite{BW12}.

The first thing we need for our new definition is a \emph{ordered linear space}: a real vector space $E$ equipped with a closed generating cone $E_{+}$.
Such a cone determines a partial ordering, invariant under translation and under positive 
scalar multiplication: if $a,b \in E$ we say that $a \leq b$ iff
$b-a \in E_{+}$. An order unit in $E$ is an element $u \in E_{+}$ such that for every $a \in E$ there is $n \in \mathbb{N}$ 
such that $a \leq nu$.
 We use $\left(E, u\right)$ to denote  an ordered linear space $E$ with an order unit $u$. 
 We say that $\left(E,u\right)$ is an 
\emph{order-unit space}. In this text, we will deal only with finite dimensional ordered linear spaces. In this case,
$E$ always has an order unit.

\begin{defi}
 A \emph{state} on an order-unit space $E$ is a linear functional $\alpha \in E^{*}$ with $\alpha(u)\leq 1$. 
\end{defi}

Once more, our definition  allows subnormalized states, with the same meaning as before. The normalized states
 are the ones with $\alpha(u)= 1$. The set of all states on 
 $E$ is called the state space on $E$ and is denoted by $\mathcal{S}(E)$. This set is a compact and convex set in $E^{*}$.

\begin{defi}
 An \emph{effect} on an order-unit space $E$ is a non-zero element $a \in E$ with $a \leq u$ and $0 \leq \alpha(a) \leq 1,
 \ \forall \ \alpha \ \in \ 
 \mathcal{S}(E).$ 
 \end{defi}
 
 The set of all effects in $E$ will be denoted by
 $\mathcal{E}(E)$. The effects in $E$ play the role of the elements of $\mathcal{T}$. 
 They represent possible outcomes of measurements that can be performed 
 on the system. Each measurement is then given by a set of effects in $E$.  We continue following the lines of assumption \ref{asfinitetomographical}, 
 and this implies that we
 only consider measurements with a finite number of outcomes.

 \begin{defi}
  A \emph{measurement}  on an order-unit space $E$ is a finite set $O=\{a_1,a_2, \ldots, a_n\}$ of  effects $a_i$ with
  $$a_1 + a_2+ \ldots +a_n=u.$$
 \end{defi}

If  $\alpha$ is a normalized state, the probability of obtaining outcome $a_i$ in measurement $O$ is $\alpha(a_i)$. 
Different measurements can share an outcome $a_i$, and  the probability of obtaining this outcome is independent of the 
measurement in which it appears. 

Once a measurement 
is performed and a given outcome is obtained, the state of the system will change, and hence every effect is related to a transformation on 
$\mathcal{S}(E)$, that has to obey  restrictions already discussed in sections \ref{sectionstates} and
\ref{sectionmultipartite}. 

\begin{defi}
 A \emph{probabilistic model} is given by an order-unit space $E$, which determines the state-space $\mathcal{S}(E)$ and 
 the set of effects
 $\mathcal{E}(E)$.
\end{defi}

Here we assume that $\mathcal{S}(E)$  contains all mathematically well defined states and 
 $\mathcal{E}(E)$ contains all mathematically well defined effects. More restrictive models can be considered, but we will
 not deal with them in this text.

Multipartite systems can be represented using composition of models. The composition will be another model, 
together with a way of connecting
states and effects in the single system with  some particular states and effects of the composite system.

  Let $E$ and $F$ be two order-unit spaces, representing systems $1$ and $2$ respectively. The composite system whose parts 
  are $1$ and $2$ is 
 represented in a order-unit space $EF$, together with a positive linear mapping
 
  \begin{eqnarray}
E \times F & \longrightarrow & EF \nonumber\\
(a,b)&\mapsto &ab.
\label{eqmapcomposite}
\end{eqnarray}

This mapping gives the connection between states and effects of $E$ and $F$ and $EF$ we mentioned above. 
Its positivity  implies  that if $a$ is an effect on $E$ and $b$ is an effect on $F$, $ab$ is an effect on $EF$. 
 A number of other requirements must be satisfied by this map and also by the states in $EF^{*}$. All assumptions made in section
 \ref{sectionmultipartite} will  hold for states and effects in $EF$ as well. Since we already provided a detailed discussion
 there, we will not repeat it here. For a different and more mathematical point of view and also for a discussion of the conditions we must 
 impose in the map of equation \eqref{eqmapcomposite}, see reference \cite{BW12}.

 \subsection{Processes and Categories}
 
 A theory aiming to describe physical systems has to provide rules that must be obeyed when a system changes. We already discussed these rules 
 when this change does not alter the type of system we are dealing with, but it might be the case that it does alter the type of the system
 we are trying to describe. We have then to define what are the valid mappings between different types of systems. 
 These mappings are called processes.
 
\begin{defi}
Given two order-unit spaces $\left(E,u\right)$ and $\left(F,v\right)$, a \emph{process} is a positive linear mapping
 $$\phi: E^* \longrightarrow F^*$$
 with $\phi(\alpha)\left(v\right)  \leq \alpha\left( u\right)$ for all states $\alpha $ in $\left(E,u\right)$.
\end{defi}

A process is a map that takes states in $E$ to states in $F$. If $\alpha$ is a normalized state, $\phi(\alpha)\left(v\right)$ is the probability that
$\phi$ occurs given that the initial 
state is $\alpha$. Of course, not every positive linear map counts as a process. The discussion of constraint \ref{coperation}
applies also in this case with very little 
modification.

\begin{defi}
A process $\phi: E^* \rightarrow F^*$  is well defined if 
$$\phi\otimes I: \left(EG\right)^* \longrightarrow \left(FG\right)^*$$
also takes states on $EG^*$ to states on $\left(FG\right)^*$, for every order-unit space $G$, where $\left(EG\right)^*$ is the state space of
the system composed of parties $E$ and $G$, $FG^*$ is the state space of
the system composed of parties $F$ and $G$ and  $\phi\otimes I$ is the extension of $\phi$ to the composite system 
$EG$ (which is defined as applying $\phi$ to system $E$ and doing nothing in system $G$).
\end{defi}

Process must take allowed states of the system  to allowed states  also when 
the system under consideration is a part of a multipartite system. That is why we require that all processes are well defined.

We also assume that convex combinations and composites of processes are also processes, for the obvious reasons.
For every pair of order-unit spaces $E$ 
and $F$ there 
is a null process that takes every states $\alpha \in E^{*}$ to the zero vector in $F^{*}$. The interpretation
of this state is the same as before,
and it can be prepared conditioning in a outcome of a measurement that happens with probability zero.

We postulate the existence of a canonical trivial system $I$ with a single operation, and hence with no measurement. 
For this system, 
$E=E^{*}=\mathbb{R}$. We do not have many options in this case, since the only normalized state  is $1$, 
which gives probability one for the only possible effect.

Given an order-unit space $E$, there are two kinds of  natural processes involving $E$ and the trivial system $I$. 
The first one is a mathematical representation of the experiment that preparates a state. For every normalized state 
$\alpha \in E^{*}$ we define the process $\phi_{\alpha}: \mathbb{R}\longrightarrow E^*$ of \emph{preparation}  of $\alpha$
given by
$$1\mapsto \alpha.$$

The second kind of process is a mathematical representation of obtaining the outcome related to an effect in a measurement.
For every effect $a$ we define the process $\psi_{a}: E^* \longrightarrow \mathbb{R}$ of \emph{registration} of the 
outcome $a$,
taking $\alpha$ to $\alpha(a)$.

\begin{defi}
 A  \emph{probabilistic category} is a category $\mathcal{C}$ such that
 \begin{enumerate}
  \item Every object in $ \mathcal{C}$ is a probabilistic model, including the trivial;
  
  \item The set of morphisms between two objects in  $\mathcal{C}$ is the set of well defined processes between
  the corresponding models.

 \end{enumerate}

\end{defi}

The set of effects on a order-unit space  $E$ can be identified with a subset of $\mathcal{C}(E, I)$ by the injection
$$a \longmapsto \psi_a :E^{*}\rightarrow \mathbb{R}$$
that takes each effect $a\in E$ to the corresponding registration process $\psi_a$, 
and the set of all states on $E$ can be identified 
with a subset of 
$\mathcal{C}(I,E)$ by the injection 
$$\alpha \longmapsto \phi_{\alpha}: \mathbb{R} \rightarrow E^{*}$$
that takes each state $\alpha \in E^{*}$ to the corresponding preparation process $\phi_{\alpha}$.

We must make one more imposition to the kind of categories representing probabilistic theories. 
We already know how to represent bipartite systems, via equation \eqref{eqmapcomposite}, but when  we consider tripartite systems the 
composition may not be associative. This is not a trivial requirement, but it is a very natural one. This property implies that $\mathcal{C}$
has to be a \emph{symmetric monoidal category} \cite{Lane98}.

\begin{defi}
 A \emph{state-complete probabilistic theory} is a probabilistic category $\mathcal{C}$, equipped with a rule of composition $\mathcal{C} \times
 \mathcal{C} \rightarrow \mathcal{C}$ assigning to every pair of models its composition according to
 equation \eqref{eqmapcomposite}, making $\mathcal{C}$ a symmetric-monoidal 
 category.
\end{defi}

This kind of probabilistic theory is called state complete because every mathematically well defined state in 
$E$ is an allowed state on the model. When we deal with real systems, there may be physical constraints that 
forbid some particular states, but we will not deal with this here.

\begin{as}
 We only consider state-complete probabilistic categories.
\end{as}

We will see many other physical impositions we can make on the  system that restricts the  set of allowed states
in chapter \ref{chapterlovasz} and appendix \ref{chaptertsirelson}. 

\subsection{Dual Processes}

The discussion above can be made using maps between effects instead of maps between states.
For every process $\phi: E^* \longrightarrow F^*$, there is a dual process

$$\phi^*: F \longrightarrow E$$
given by $\alpha(\phi^*(b))= \phi(\alpha)(b)$ for all $b \in F$ and $\alpha \in E^{*}$.
Physically, getting  the outcome related to the effect $\phi^*(b)$ in a measurement corresponds to apply process $\phi$ 
first and then obtain outcome $b$ in a measurement.

Given a probabilistic category $\mathcal{C}$ we can define the dual category $\mathcal{C}^{*}$ using the dual processes for $\mathcal{C}(E,F)$
instead of the processes. In physicist's language, $\mathcal{C}$ represents the Schr\"odinger picture while $\mathcal{C}^{*}$ 
represents the
Heisenberg picture \cite{Cohen77}.

The most important probabilistic theories  for us are finite dimensional classical and quantum probability theories. 
They will be presented in detail in
sections \ref{sectionclassical} and \ref{sectionquantum}. Of course, they are not the only examples we can provide. In references \cite{BW12} and \cite{Barret06}, the reader can find
a number of examples differing from these ones. We will not present these examples here, but we emphasize that probabilistic theories beyond 
quantum theory are of great importance in this work.

\section{Classical Probability Theory}
\label{sectionclassical}

Classical probability theory was developed to describe the most elementary random processes we deal with in our everyday life.
The simplest example is a coin toss, where there are two possible outcomes. 
Another familiar example is the throwing of a dice: if we look at the top face of the die, there 
are six possible outcomes: the numbers $\{1,2,3,4,5,6\}$. Of course we can come up with much more complicated examples,
but the most important features are  already present in these simple  cases. The axiomatic system we will present here was 
introduced by the soviet mathematician
 Andrey Kolmogorov in  the 1930s \cite{SW95, GS01, James04}.  
 Although this system can be used to describe a large variety of random phenomena,
 it is not enough to describe the behavior of quantum systems. This leads to other axioms for probability theory
 and an example of such more general formulation is the one present in the previous sections. 
 
 Now we study carefully  classical models and we stress  how the elements of the previous sections
 are represented in this class. All axioms in classical probability theory look very natural and it was indeed
 a shock to many people
 that nature does not always behave in this way. These axioms imply a number of singular properties that make this kind of theory different
 from any 
 other in the framework. In this sense,  classical theory  emerges as a very special exception.

\subsection{\textsf{Sample Spaces}}

A classical probabilistic model consists of three basic elements. The first one is a set  whose elements
represent all possible outcomes in an experiment.  This set is called the \emph{sample space} of the experiment.

\begin{defi}
 The \emph{sample space} $\Omega$ of a random experiment is a set in which every element $\omega \in \Omega$ is associated to a 
 possible outcome of the experiment.
\end{defi}


\begin{ex}[The classical bit]

 The sample space of the game of heads and tails is a set with two elements, corresponding 
 to the two possible outcomes of the experiment of tossing a coin. We could use the set $\{H,T\}$ with the letter $H$ 
 representing outcome \emph{heads} and letter $T$ representing outcome \emph{tails}. It is sometimes easier to work with sample spaces 
 with numerical elements, since this allows the definition of a number of useful quantities we can use to get
 information about the experiment
 we are describing. In this case we generally use the set $\{0,1\}$, but $\{-1,1\}$ is also pretty common. A classical system with sample space with only two elements is
 called a \emph{classical bit}.
\end{ex}
\begin{ex}

 The sample space of the experiment of throwing a dice and looking at its superior face is the set $\{1,2,3,4,5,6\}$, as we already know.
\end{ex}

\begin{ex}
 Sometimes it is not that trivial to define what is the  sample space of an experiment. 
 Think about the possible outcomes of the following experiment: select randomly an inhabitant of a country
 and measure their height. In principle the height of a person is a number in the interval $(0, \infty)$, but of course we 
 know that some values in this set are highly unlikely, such as a height of a billion meters. The interval $(0,3)$ seems a 
 much more reasonable sample space. Nowadays in Brazil we could use the interval $(0, 2.37]$, since the tallest man we have 
 record of, according to a quick search in Google, is Joelisson Fernandes, who claims to be the tallest person in Brazil 
 with $2.37m$ \cite{wikijoelisson}\footnote{If you know anyone taller then Joelisson, let us know.}. If we were in Turkey instead of Brazil we would have to use ate least the interval $(0, 2.51]$, since the tallest man alive in Earth is the Turkish Sultan 
K\"osen with $2.51m$ \cite{wikisultan}.
\end{ex}

The set of all subsets of $\Omega$ will be denoted by $\mathcal{P}(\Omega)$. We would like to assign a probability for
all subsets of $\Omega$, but in general it is not possible to do that in a reasonable manner. Because of this, we need the 
definition  of  measurable sets, the elements of $\mathcal{P}(\Omega)$ for which we can define a probability.
This is the second basic element of a classical probabilistic model.

\begin{defi}
$\Sigma \subset \mathcal{P}\de{\Omega}$ is a $\sigma$-algebra if it satisfies:
\begin{enumerate}
\item $\Omega, \emptyset \ \in \ \Sigma$.
\item $\Sigma$ is closed under complementation: If $A \in \Sigma$, then so is its complement, $\Omega \setminus A$.
\item $\Sigma$ is closed under countable unions: If $\{A_1, A_2, A_3, ... \}$ is a countable sequence 
of elements of $\Sigma$, then  $A = \bigcup_i A_i$ is in $\Sigma$.
\end{enumerate}

\label{defimeasure}
The sets $A \in \Sigma$ are called \emph{measurable sets}.
An ordered pair $\left(\Omega, \Sigma\right)$ where $\Omega$ is a sample space and $\Sigma$ is a 
$\sigma$-algebra over $\Omega$ is called a \emph{measurable space}.
\end{defi}

\begin{ex}

 The trivial $\sigma$-algebra contains only two elements: the entire set $\Omega$ and is complement, the empty set $\emptyset.$

\end{ex}

\begin{ex}[Finite and countable sample space] When $\Omega$ is a finite or a countable set, we 
usually take  $\Sigma=\mathcal{P}(\Omega)$.  This set is a $\sigma$-algebra even if $\Omega$ is not countable, but in this 
case it might not be a good choice. For a classical bit with sample space $\{0,1\}$ we have
$$\Sigma =\left\{\emptyset,\{0\},\{1\}, \{0,1\}\right\}.$$
For the dice, $\Sigma$ has $64$ elements. In the finite case, if $\Omega$ has $n$ elements, $\Sigma$ has $2^n$ elements. 
\end{ex}

\begin{ex}[Continuous sample space] Consider the experiment that consists of 
selecting  a number in the interval $[0,1]$ with equally distributed probability.
In this example, $\Omega=[0,1]$ and if we take 
$\Sigma$ to be $\mathcal{P}\left([0,1]\right)$ the $\sigma$-algebra will be too big and we will not be able to
define a probability for all subsets in it. We have to choose $\Sigma$ in such a way that it   allows the definition of a 
probability for all its elements, respecting the natural properties   probabilities must have, but in 
such a way that it is not too small to live behind some subsets of $[0,1]$ for which the definition of a
probability is almost obvious. For example, consider the subset $A=\left[0,\frac{1}{3}\right]$. If we 
choose a point in $[0,1]$ randomly, and if  all points are equally likely, we expect this point to be in $A$ one 
third of the time. This means that we should define the probability of $A$ as $\frac{1}{3}$, and hence we would 
like to have $A \in \Sigma$. A similar argument holds for all intervals. This means that every interval should belong
to $\Sigma$.  Most of the times the most convenient choice is to take $\Sigma$ as the minimal $\sigma$-algebra  that 
contains all intervals. This  is the Borel $\sigma$-algebra $\mathcal{B}$ and its elements are called \emph{Borelians}.

\end{ex}

The third element we need in a classical probability space is the assignment of a probability to each measurable set
$A \in \Sigma$. We have been using this notion without further consideration, with the interpretation
that this number quantifies the idea of 
relative frequencies of a given outcome. It is related to the ratio
$$ \frac{\mbox{ number of occurrences of A}}{\mbox{number of independent trials of the experiment}}.$$
This definition depends on the assumption of   convergence  of this sequence after many repetitions of the experiment. 

This ratio should not be mistaken with the most naive definition of probabilities, where all atomic elements of $\Sigma$
have the same probability. Here, one is adopting the idea that there is some a priori probability distribution and that 
identically prepared repetitions of the experiment will generate frequencies that converge to such probability distribution. 
For a more precise statement, we have the many versions of the Law of Large Numbers \cite{James04}.

Being practical, we will only focus on the mathematical definition and assume 
the existence of a real number associated to each measurable set in $\Sigma$, its probability.
We assume also that this association is 
done 
in such a way that the properties expected by the interpretation of this number as relative frequencies in a experiment 
should hold.

\begin{defi}
Let $\left(\Omega, \Sigma\right)$ be a measurable space.  A function    $\mu: \Sigma \longrightarrow \overline{\mathbb{R}_+}= 
\mathbb{R}_+ \cup  \{\infty\}$
is called 
a \emph{measure} if it satisfies the following properties:
\begin{enumerate}
\item Non-negativity: $\mu(A) \geq 0 \ \forall A \in \Sigma$;
\item Nullity: $\mu(\emptyset)=0$;
\item Countable additivity (or $\sigma$-additivity): For all countable collections $\left\{A_1, A_2, \ldots\right\}$   
of pairwise disjoint sets $A_i \in \Sigma$:
\be \mu\left(\bigcup_i A_i\right)=\sum_i \mu(A_i). \label{equationmeasureuni}\ee
\end{enumerate}

The measure $\mu$ is called a \emph{probability measure} if $\mu\left(\Omega\right)=1$. If 
$\mu$ is a probability measure over  the measurable space $\left(\Omega, \Sigma\right)$, the   
triple $\left(\Omega, \Sigma, \mu\right)$ is called a \emph{classical probability space}\footnote{Classical mathematicians 
do not need the word \emph{classical} and use the term \emph{probability space} 
for the triple $\left(\Omega, \Sigma, \mu\right)$. We will add a third word to avoid confusion with the general 
theories introduced in section \ref{sectionstates}.}.
\end{defi}

\begin{defi}
 A subset $A$  of $\Omega$ for which a probability can be assigned is called an \emph{event}. If $A=\{w\}$, it is called
 an \emph{elementary event}.
\end{defi}

It follows from definition \ref{defimeasure} that a measure $\mu$ should also satisfy, as expected, the
properties of monotonicity 
and sub-additivity.

\begin{cor}[Monotonicity]
If $A_1$ and $A_2$ are measurable sets with $A_1 \subset A_2$ then
$$\mu(A_1) \leq \mu(A_2).$$
\end{cor}

\begin{cor}[Sub-additivity]
 For any countable sequence $\left\{A_1, A_2, \ldots\right\}$ of sets $A_i \in \Sigma$, not necessarily disjoint, we have
$$\mu\left(\bigcup_i A_i\right)\leq \sum_i \mu(A_i).$$
\end{cor}

\begin{ex}[The classical bit]
A probability  measure in the measurable space of a classical bit is a vector in $\mathbb{R}^2$ of the form
$$\left[\begin{array}{c}
p\\
1-p\end{array}\right]$$
where $p=\mu(0)$, $1-p=\mu(1)$ and  $0 \leq p \leq 1$. 
\label{exbit2}
\end{ex}

\begin{ex}[The discrete case]
In the discrete case, a probability measure $\mu$ in $\left(\Omega, \mathcal{P}(\Omega)\right)$ is defined by a function 
$p: \Omega \rightarrow \overline{\mathbb{R}_+}$ such that
$$\sum_{\omega \in \Omega}p(\omega)=1.$$
The value of $\mu$  in a event $A \in \mathcal{P}(\Omega)$ is then given by equation \eqref{equationmeasureuni}
$$\mu(A)=\sum_{\omega \in A} p(\omega).$$
\end{ex}

\begin{ex}[The Lebesgue measure]
One important measure in $\left([0,1], \mathcal{B}\right)$ is the Lebesgue measure $l$. The value of this measure in a interval $[a,b] \subset [0,1]$ is
$$l\left([a,b]\right)=b-a.$$
This definition can be extended to all elements of the $\sigma$-algebra $\mathcal{B}$ in a unique manner \cite{James04}. 

\end{ex}

We will consider only finite sample spaces, 
which will meet the requirement of assumption \ref{asfinitetomographical}. 
We will always  take $\Sigma= \mathcal{P}(\Omega)$ for simplicity.

\begin{defi}
A \emph{classical probabilistic model}  is a model  in which every normalized state is 
a probability measure in a  measurable space
$\left(\Omega, \Sigma\right)$. The set $\mathcal{T}$ of allowed transformations is the greatest set of linear transformations
in $\mathcal{S}$ satisfying constraint \ref{coperation}.
\end{defi}

Notice that when we assume  $\Sigma = \mathcal{P}\left(\Omega\right)$ the only important information is the number
of elements in the sample space:  two sample spaces with the same number of elements describe
the same type of  system.   

\subsection{\textsf{Transformations}}

An allowed transformation $M \in \mathcal{T}$ must map a state
into another allowed state, according to item  \ref{coperation3} of constraint \ref{coperation}.
This means that every element of
$\mathcal{T}$ is a linear map that takes every probability measure in $\Omega$ to another probability measure in $\Omega$, possibly
multiplied by a constant between zero and one, if the transformation does not preserve normalization. Constraint
\ref{coperation} implies that each entry
of the matrix associated to this transformation must be positive, and the sum of each column must be
a number between zero and one. In the case that $M$ preserves normalization, it is
a stochastic matrix.

There is an important class of transformations in $\mathcal{T}$, given by the indicator functions of elements of the 
$\sigma$-algebra $\Sigma$. Let $\Omega=\{\omega_1, \ldots, \omega_n\}$ and $A \in \Sigma$. Define $I_A $ as the $n\times n$ 
real diagonal matrix with
$$(I_A)_{ii}=\left\{\begin{array}{cc}
                     1& \ \mbox{if} \ \omega_i \ \in \ A \\
                     0& \ \mbox{otherwise.}
                    \end{array}\right.$$
This matrix is an element of $\mathcal{T}$. The matrices in $\mathcal{T}$ that are of this form give rise to an important class
of measurements, given by a partition of the sample space $\Omega$: let $\{A_1, \ldots, A_m\}$ be a partition of $\Omega$
such that every $A_i$ in the partition belongs to $\Sigma$. Then the set of matrices $\{I_{A_1}, \ldots, I_{A_m}\}$ defines a
measurement in the model. Given a normalized state of the system $\mu$, which is, by definition, a measure defined in 
$\left(\Omega, \Sigma\right)$, the probability $p_i$ of outcome $i$, associated to the matrix $I_{A_i}$, is given by
$$p_i=\mu(A_i).$$

To prove that all the  matrices mentioned above indeed belong to $\mathcal{T}$ we still  have to check that item \ref{coperation4}
of constraint
\ref{coperation} is also satisfied. Indeed, one can prove that for classical models, 
all transformations satisfying items \ref{coperation1}, \ref{coperation2} and \ref{coperation3} automatically satisfy  item \ref{coperation4}. 
We will do it in subsection \ref{subsectionclassicalmultipartite}.

\subsection{\textsf{Classical probabilistic theory with finite sample spaces}}

\begin{defi}
A \emph{classical probability theory} is one in which all models are classical. In this text, the sample spaces are all finite.
\end{defi}
 

Since we are dealing with finite sample spaces, without loss of generality we can consider
$\Sigma = \mathcal{P}(\Omega)$ in all models. With this assumption, each model in a classical theory is 
given by a sample space $\Omega.$ We can always use a tomographic set with only one element, 
the measurement associated to the partition in which every subset  contains only one element of $\Omega.$

\begin{cor}
If $\Omega=\{\omega_1, \ldots, \omega_n\}$, the set that contains  only  the measurement 
associated to the partition $\left\{\{\omega_1\}, \{\omega_2\}, \ldots, \{\omega_n\}\right\}$ is a tomographic
measurement for the system given by the measurable space $\left(\Omega, \mathcal{P}(\Omega)\right)$.
\end{cor}

This measurement is called \emph{maximal measurement}. 

The existence of a tomographical set with only one element is a particularity of classical theory, with drastic consequences to 
our way of thinking, as we will see soon. 

\begin{teo}

In a classical probability theory, the state space of the system associated to sample space $\Omega$ is 
a simplex of dimension $|\Omega|$.
\end{teo}

\begin{dem}
Let $\Omega=\{\omega_1, \ldots, \omega_n\}$ and take the tomographic set that consists only of the maximal 
measurement $M$ with outcomes $r_1, \ldots, r_n$. Define the measure $\mu_i$ given by
$$\mu_i\left(\omega_j\right)=\delta_{ij}.$$
Since the states in a classical model are given by probability measures in $\Omega$, all of these $n$ measures 
represent states in the state space of the system $\mathcal{S}.$ 
They are also the only pure states in $\mathcal{S}$, since all other measures in $\Omega$ can be 
written as convex sums of the $\mu_i$. This implies that the set of normalized states is the simplex of dimension $n-1$
in $\mathbb{R}^n$. 

By assumption \ref{asstatespace}, $\mathcal{S}$ is 
the convex hull of the $n$ points $\mu_i$ and $\overrightarrow{0}$, which is homeomorphic to the 
$n$-dimensional simplex in $\mathbb{R}^{n+1}$.
\end{dem}

Although the $n$-dimensional simplex is defined as a subset of $\mathbb{R}^{n+1}$, it can be represented in 
$\mathbb{R}^n$, as the convex hull of the extremal normalized states and $\overrightarrow{0}$. Figures \ref{figcbit} and 
\ref{figctrit} show this for $n=2$ and $n=3$, respectively.

\begin{figure}
\begin{center}
 \includegraphics[scale=0.4]{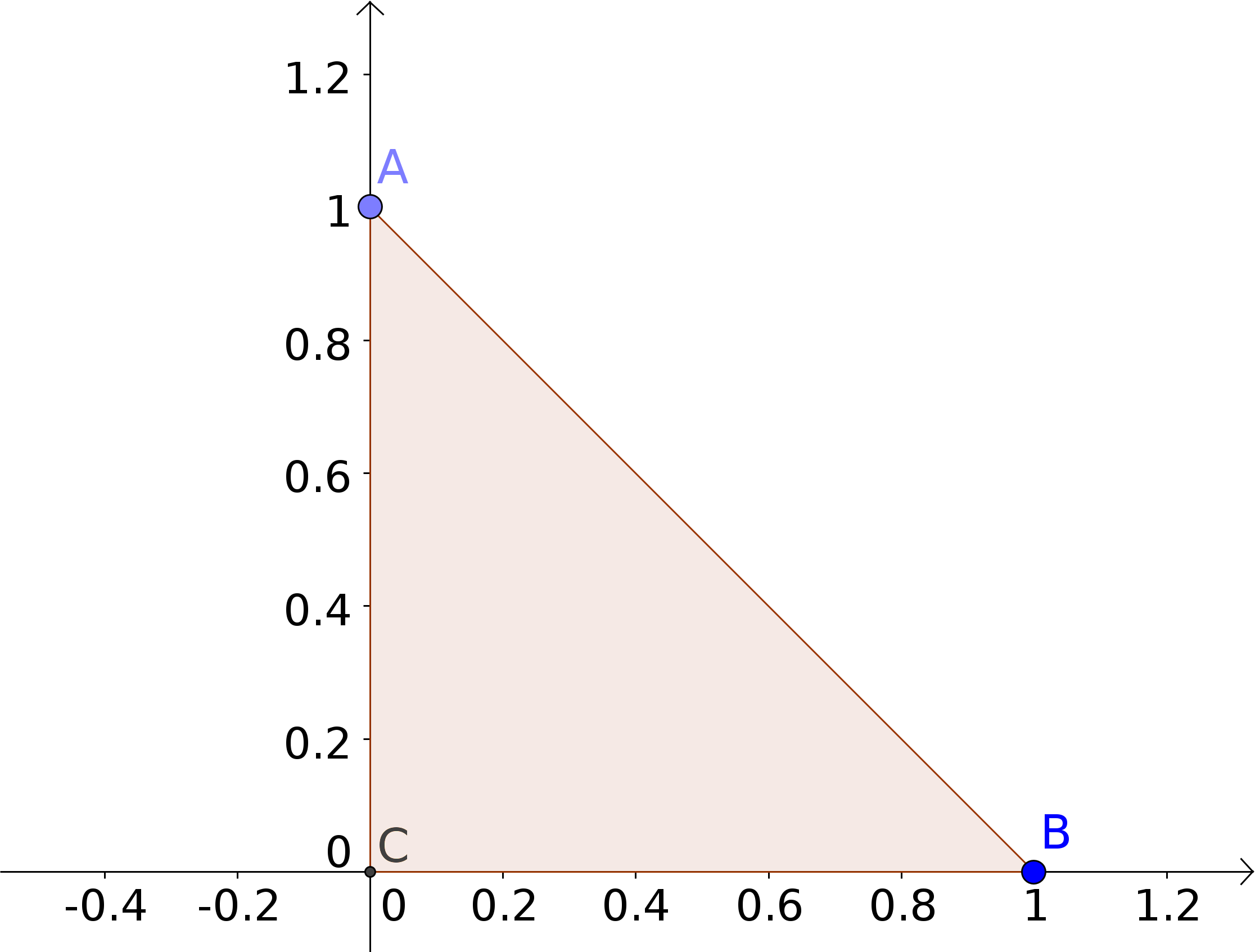}
\caption{\textsf{The state space of a classical bit. The point $A$ represents the normalized extremal state for 
which $\mu(0)=0$ and $\mu(1)=1$ and 
point $B$ represents the normalized extremal state for which $\mu(0)=1$ and $\mu(1)=0$. Point $C$
represents the unnormalized 
state $\vec{0}$.} \label{figcbit}}
 \end{center}
\end{figure} 

\begin{ex}[The state space of a classical bit]
\label{excbit}
We already saw in example \ref{exbit2} that the normalized states of a classical bit are the vectors
$$\left[\begin{array}{c}
p\\
1-p\end{array}\right]$$
where $p=\mu(0)$, $1-p=\mu(1)$ and  $0 \leq p \leq 1$. 
The state space of this system is then given by the convex hull of this set of vectors and $\vec{0}$, which is a triangle 
in $\mathbb{R}^2$. This set is shown in figure \ref{figcbit}.

\end{ex}

\begin{figure}
\begin{center}
 \includegraphics[scale=0.7]{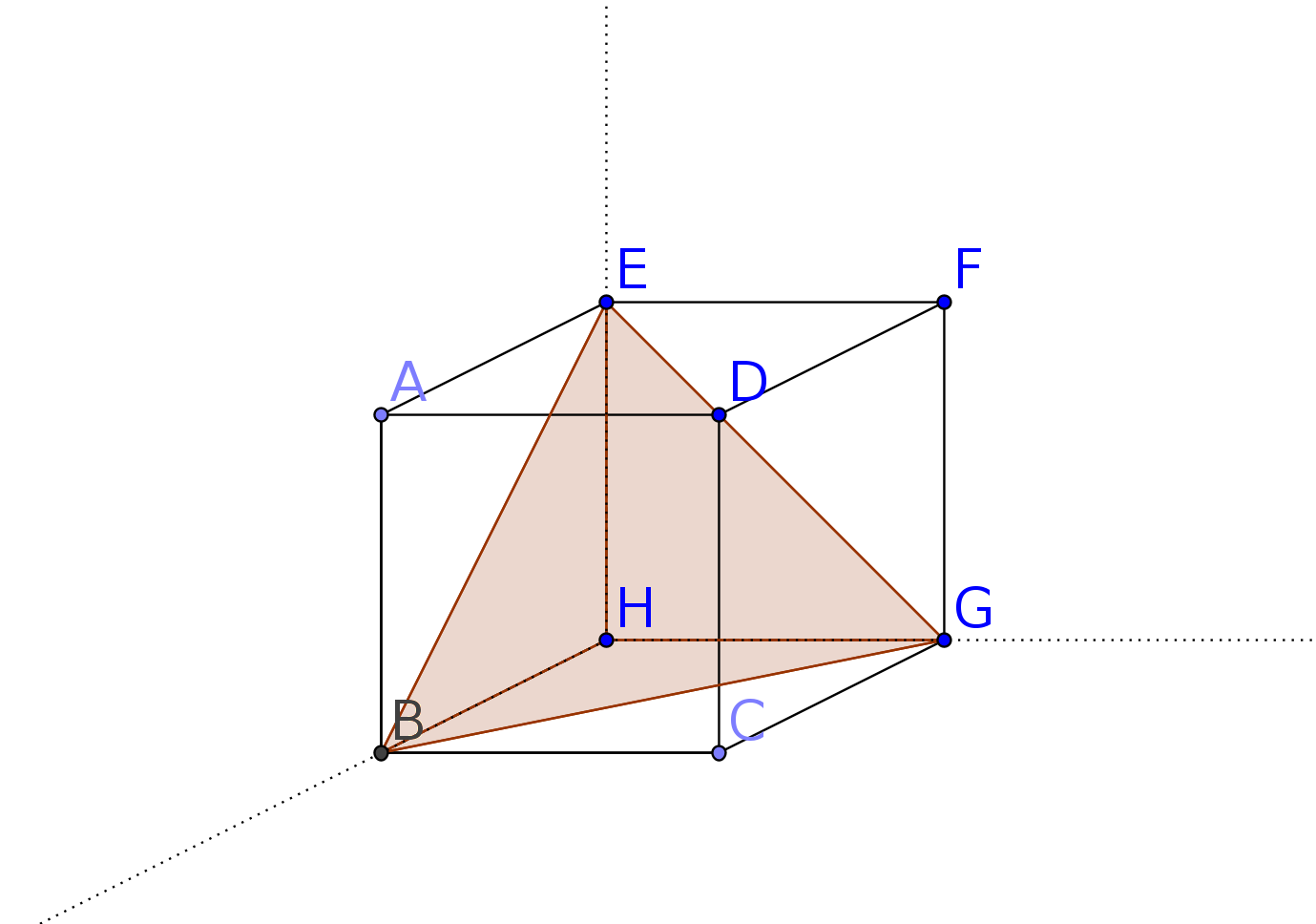}
 \end{center}
\caption{\textsf{The state space of a classical trit. The state space is the tetrahedron in $\mathbb{R}^3$ with extremal points
$B,E,G,H$. The point $B$ represents the normalized extremal state 
for which $\mu(0)=1$ and $\mu(1)=\mu(2)=0$, point $G$ represents the normalized extremal state for which $\mu(0)=\mu(2)=0$ and 
$\mu(1)=1$ and 
point $E$ represents the normalized extremal state for which $\mu(0)=\mu(1)=0$ and $\mu(2)=1$. Point $H$ represents the unnormalized 
state $\vec{0}$.}  \label{figctrit}}
\end{figure} 

\begin{ex}[The state space of a classical trit]
\label{exctrit}
The normalized sates of a classical system with sample space $\{0,1,2\}$ are vectors in $\mathbb{R}^3$ of the form
$$\left[\begin{array}{c}
p\\
q\\
1-p-q\end{array}\right]$$
where $p=\mu(0)$, $q=\mu(1)$, $1-p-q=\mu(2)$ and  $0 \leq p, q, p+q \leq 1$.
The state space of this system is then given by the convex hull of this set of vectors and $\vec{0}$, which is a tetrahedron
in $\mathbb{R}^3$. This set is shown in figure \ref{figctrit}.


\end{ex}

The simplex has a remarkable property that every point can be written uniquely as a convex sum of the extremal points. 
The converse also holds: if in a convex set  every point can be written uniquely as a convex sum of the extremal points, 
then this set is a simplex.  For a proof of this claim, see reference \cite{Rockafellar97}. 
This result has an interesting consequence when the convex set represents the state space of 
a system.

\begin{teo}
If the state space of a system is a simplex, then it can be described by a classical probability space.
\end{teo}

Notice here that the only important thing in the  classical probability spaces we 
consider in this text is the number of elements of $\Omega$, since $\Sigma$ is always equal to
$\mathcal{P}(\Omega)$. This implies that once $|\Omega|$ is fixed, both the state space and the set of measurements are 
determined and it makes no difference which particular  symbols we use to represent the elements of $\Omega$.

\subsection{\textsf{Compatibility}}
In section \ref{sectionstates} we defined the notion of compatibility of measurements,  connected to joint mesurability of them. 
For measurements with repeatable outcomes in classical probability theory there are no incompatible measurements, 
which makes the compatibility concept unnecessary.
This is quite easy to prove: the maximal measurement is a refinement for all other measurements at the same time,
a consequence of the fact that a finite intersection of sets in a $\sigma$-algebra is also an element of
the $\sigma$-algebra.

\begin{cor}
In a classical system, all measurements with repeatable outcomes are compatible.
\end{cor}

One of the central aspects of the generalization presented in section \ref{sectionstates} is that we
no longer demand this property from  our models. 

Incompatibility of measurements is one of the many strange features of non-classical theories, and specially of quantum 
theory. It sounds pretty disturbing that nature forbids us to extract all information from a system by
measuring it. The existence of incompatible measurements has many interesting and intriguing consequences. One of them is the 
noncontextual  character of some non-classical theories, which we will see in chapter \ref{chapterncinequalities}.

\subsection{\textsf{Multipartite systems in classical probability theory}}
\label{subsectionclassicalmultipartite}

In classical probability theory, we require that a multipartite system can also  be described in a
classical probability space. Given the sample spaces of the individual systems, it is
very easy to find the sample space associated to the joint system.

\begin{as}
Given a bipartite system composed of classical parties $1$ and $2$, associated to sample spaces $\Omega_1$ and $\Omega_2$. Then the global system is associated to the sample space $\Omega_1 \times \Omega_2$.
\end{as}

 By assumption \ref{asindependentstates}, all product states are allowed and 
this implies that all measures in $\Omega_1 \times \Omega_2$ are allowed states of the composite system, since every measure 
in this sample space can be written as a convex sum of product states. 
This is a very important statement, and implies the following result:

\begin{teo}
 Every state in a composite classical system can be written as a convex sum of product states.
\end{teo}

 This is not  true for every 
theory. In fact, in theorem \ref{theoremlcpure} we have proved that all  states can be written as a linear combination of 
product states, but there might be states for which it is not possible to find a linear combination of this type
with all coefficients positive. This is the case for quantum theory and also for many other theories in  framework. 
As a corollary of this observation, we can prove the following result:

\begin{teo}
 In a classical model, if a linear map defined in $\mathcal{S}$ satisfies  positivity, normalization  and state preservation,
 it automatically satisfies 
 complete state preservation.
 
\end{teo}

\begin{dem}
 In this proof we use the notation introduced in section \ref{sectionstates}.
 Let $f$ be a map satisfying positivity, normalization  and state preservation. This means that $f$ takes a 
 state in $\mathcal{S}$ to 
 another state in $\mathcal{S}$. Suppose now that our system is  part of a composite system. 
 Let $p$ be a state of the composite system. Since every state of the system 
 can be written as a convex combination of product states, all of them are of the form
 $$p=\sum_i\alpha_i p^1_i  \otimes p^2_i$$
 where each $p^1_i$ is a state in $\mathcal{S}$, each  $p^2_i$ is a state of 
some other arbitrary subsystem,  $0 \leq \alpha_i \leq 1$ for every $i$ and $\sum_i \alpha_i=1$. 
Then, if we apply the map $f \otimes I$ we get
 $$p'=\sum_i \alpha_i f\left(p^1_i\right) \otimes p^2_i.$$
 Since $f\left(p^1_i\right)$ is an allowed state in $\mathcal{S}$ for 
 every $i$, $p'$ is also a convex combination of product states, and hence another valid state of the composite system.
\end{dem}

In this thesis, every time we say that  a system or an experiment is \emph{classical}, we mean 
that it can be described  by a classical probabilistic model.
We stress this fact because the word \emph{classical} can be used in many different situations with different meanings and we do not want 
to create any confusion. In the same way, every time we say that something is not classical we mean that it does not admit a
description through a classical probabilistic model. Many of the models in the framework presented in this chapter
are not classical. One of them is the model obtained with quantum theory, which 
we will present in the next section. 

\section{\textsf{Quantum Probability Theory}}
\label{sectionquantum}

\epigraph{\begin{center}
\textit{\small{Quantum Mechanics deals with nature as She is - absurd.}}\end{center}}{Richard Feynman, \cite{Feynman88}}


Quantum theory is, at the same time, the first physical theory where the probabilistic character is considered intrinsic,
and the first physical theory which does not fit into classical probabilistic models 
under reasonable assumptions. In this section we will see
how states and measurements are described in this theory. For a more complete treatment and to applications on the description
of specific physical systems, see \cite{Feynman65,Cohen77,Peres95, NC00, Griffiths05, ATB11}.

\begin{defi}
A quantum probabilistic model is a model in which the   state  space  is in one-to-one correspondence
with the set of positive operators $\rho$ acting on a 
fixed Hilbert space $\mathcal{H}$ over $\mathbb{C}$ such that $\tr\left(\rho\right) \leq 1.$ This set will be denoted
by $\mathcal{S}\left(\mathcal{H}\right)$. The set 
$\mathcal{T}\left(\mathcal{H}\right)$ of allowed transformations is the
greatest set of linear transformations 
satisfying constraint \ref{coperation}. These transformations correspond to a special type of linear transformations
acting in $\mathcal{S\left(\mathcal{H}\right)}$, as we will see later.
\end{defi}

The normalized states are the ones with $\tr\left(\rho\right) = 1.$ They are called the \emph{density operators} of
 $\mathcal{H}$. Once an orthonormal basis is fixed, each density operator is given by a positive matrix with unit trace. 
 These matrices are called \emph{density matrices}. We will often use the letter $\rho$ to denote both  density operators and 
 density matrices and the specific meaning in each case must be clear from the context. The set of 
 all density operators acting in $\mathcal{H}$ will be denoted by $D\left(\mathcal{H}\right)$. The set of all matrices
 acting on $\mathcal{H}$ will be denoted by $M\left(\mathcal{H}\right)$.
  We will consider only the cases with finite 
dimensional $\mathcal{H}$, to satisfy requirement \ref{asfinitetomographical}. The type of system is determined
by the dimension of $\mathcal{H}$.

\begin{teo}
 The pure states of a quantum model are the unidimensional projectors over $\mathcal{H}$.
\end{teo}

\begin{dem}
 Clearly, the pure states are also normalized states, so we have to worry only with the extremal 
 points of the set $D\left(\mathcal{H}\right)$. Every density matrix can be written in spectral decomposition 
 \be \rho=\sum_i p_i \ket{\psi_i}\bra{\psi_i}, \ p_i \geq 0, \ \sum_i p_i =1,\ee
 where each $\ket{\psi_i}$ is a vector in $\mathcal{H}$ with unit norm. 
 This proves that each density matrix can be written as a convex combination of unidimensional projectors. On the other hand,
 the unidimensional projectors $ \ket{\psi}\bra{\psi}$ themselves can not be written as convex combination of the others,
 because the rank of any convex combination is ate least two. This proves that they are 
 the extremal points of $D\left(\mathcal{H}\right)$, and hence the extremal points of $\mathcal{S}\left(\mathcal{H}\right)$.
\end{dem}

Every mixed state can be written as a convex combination of projectors, but in contrary to what happens
in classical models, this decomposition is not unique. We will shall make this clear in example \ref{exqubit}.

Every unidimensional projector can be associated with its one dimensional image in $\mathcal{H}$.
We can identify this unidimensional space with a class of equivalence of unit vectors in $\mathcal{H}$ under the relation
$$\ket{\psi} \sim e^{i\phi} \ket{\psi}.$$
This means that every pure state is given by a straight line passing through the origin  in $\mathcal{H}$. The set of these
lines is the projective Hilbert space $\mathcal{P}\mathcal{H}$. If in some situation we are restricted to pure states only,
we can use  $\mathcal{P}\mathcal{H}$ instead of $\mathcal{H}$ in the description of the model \cite{BH01, Amaral06}.

It is quite common to use only a unit vector to represent a pure state in quantum theory. This brings no difficulty if we 
keep in mind that each unit vector is only a representative of the equivalence class related to the state and that there 
are many unit vectors representing the same pure state.

\begin{ex}[The quantum bit]
\label{exqubit}
A quantum bit, or  \emph{qubit}, is the system described by a Hilbert space of
dimension two.  It is the quantum analogue of
the classical bit, hence its name. This analogy  justifies the usual notation used for the standard basis in
$\mathcal{H}$: $\{|0\rangle,|1\rangle\}$. Any  pure state of this system can be represented by a unit vector in 
$\mathcal{H}$
$$|\psi\rangle=\alpha |0\rangle + \beta|1\rangle, \ \ \ \alpha, \beta \in \mathbb{C}.$$
The normalized pure states satisfy the further restriction $|\alpha|^2 + |\beta|^2=1$.

General normalized states of a qubit are represented by $2\times 2$ density matrices acting in $\mathcal{H}$.
The set of $2 \times 2$ Hermitian matrices is a real vector space of dimension four, and the set of matrices given by the
three Pauli matrices
$$\sigma_1=\left[\begin{array}{cc}
0&1\\
1&0
\end{array}\right], \ \ \sigma_2=\left[\begin{array}{cc}
0&-i\\
i&0
\end{array}\right], \ \ \sigma_3=\left[\begin{array}{cc}
1&0\\
0&-1
\end{array}\right],$$ 
together with  the identity matrix $I$, is an orthogonal basis.
Hence, every density matrix of a qubit can be written in the form
$$\rho=\frac{1}{2}\left(I + a\sigma_1 + b\sigma_2 +c\sigma_3\right).$$
The coefficient of  $I$ must be $1/2$ because it is the only matrix with non-zero trace, equal to two, and
$Tr(\rho)=1$. The vector
$\left[\begin{array}{ccc} a&b&c
\end{array}\right]$, called the \emph{Bloch vector} of the state,   has to satisfy the condition
$$a^2+b^2+c^2 \leq 1$$
because of the positivity of
$$\rho=\frac{1}{2}\left[\begin{array}{cc}
1+c&a-ib\\
a+ib&1-c
\end{array}\right].$$
This implies that there is a bijective association between normalized states of a qubit and points in the ball of radius one
in $\mathbb{R}^3$, the \emph{Bloch ball}. This bijection preserves mixtures, and points in the sphere $S^2$, the \emph{Bloch sphere},
correspond to the pure states of the 
system. 

Including subnormalized states, the state space $\mathcal{S}\left(\mathcal{H}\right)$ is a cone over the Bloch ball,
which requires four dimensions to be embedded.

From this geometrical representation it is easy to see that  the decomposition of a mixed state in terms of pure state is 
not unique. In fact, any point in the interior of the ball can be written as a convex combination of a finite number of points
in the sphere in many different ways.

\begin{figure}
\label{figurebloch}
\centering
\includegraphics[scale=0.7]{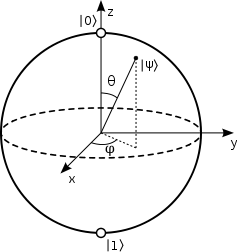}
\caption{\textsf{The Bloch sphere, a geometrical representation of the state space of one qubit.}}
\end{figure}

\end{ex}
 The Bloch sphere is connected to an interesting mathematical object, called the \emph{Hopf fibration}. For more information 
 see \cite{BH01, Amaral06, TerraCunha07, Amaral10}.
 
\subsection{\textsf{Multipartite systems in quantum models}}

A state of a  multipartite system composed of subsystems $1$ and $2$ in  quantum probability theory  
is also given by a positive operator, $\rho$, in a Hilbert space $\mathcal{H}_{12}$, with $\tr(\rho) \leq 1$.
This Hilbert space is
constructed from the Hilbert spaces of the subsystems using the tensor product.

\begin{as} If
the Hilbert spaces of subsystems $1$ and $2$ are $\mathcal{H}_{1}$ and $\mathcal{H}_{2}$, respectively, then 
the Hilbert space of the composite system is given by
\be \mathcal{H}_{12}=\mathcal{H}_{1} \otimes \mathcal{H}_{2}. \ee
\end{as}

The states of the composite systems are matrices in\footnote{The isomorphism we use in this identification is positive and trace 
preserving. For this reason, the density matrices of  the composite system is given by a positive matrix with trace bounded by
one in $M\left(\mathcal{H}_{1}\right) \otimes M\left(\mathcal{H}_{2}\right)$.}  $M(\mathcal{H}_{A}\otimes \mathcal{H}_B)\equiv
M(\mathcal{H}_A) \otimes M(\mathcal{H}_B)$. 

\begin{ex}[Two quantum bits]
 \label{extwoqbits}
The Hilbert space associated to the system of two qubits is 
isomorphic to $\mathbb{C}^2 \otimes \mathbb{C}^2$ and the density matrices of this system are positive
matrices with trace one in $M(\mathbb{C}^2)
\otimes M(\mathbb{C}^2)$. A basis for the real vector space of  $4 \times 4$ Hermitian matrices is the set
of matrices\footnote{We will use the letter $I$ for the identity matrix of every dimension.} 
$\left\{I, I\otimes \sigma_i, \sigma_i \otimes I,
\sigma_i \otimes \sigma_j\right\}$,    and a density 
matrix of the system of two qubits can be written in the form
$$\rho=\frac{1}{4}\left(I + \sum_i R_{0i} \  I \otimes \sigma_i + \sum_i R_{i0} \ \sigma_i \otimes I +
\sum_{ij}R_{ij} \ \sigma_i 
\otimes \sigma_j\right)$$ where
$$R_{ij}=\tr(\sigma_i \otimes \sigma_j \ \rho ).$$
This matrix can also be represented by the matrix $R$, whose entries are the coefficients $R_{ij}$ defined above, with
$R_{00}=1/4$.

Unfortunately, the conditions the positivity of $\rho$ imposes on the  entries of $R$ are not so easily written as 
in the case
of one qubit. Sometimes we can focus on subsets of the set of density matrices, decreasing the number of parameters in the 
problem and simplifying the analysis \cite{Amaral10}.

\end{ex}

In $\mathcal{S}\left(\mathcal{H}_{12}\right)$, we distinguish three kinds of density matrices.

\begin{defi}
 We say that a state $\rho \in D\left(\mathcal{H}_{12}\right)$ is a \emph{product state} if 
 $$\rho=\rho_1 \otimes \rho_2$$
with  $\rho_1  \in D\left(\mathcal{H}_1\right)$ and  $\rho_2  \in
D\left(\mathcal{H}_{2}\right)$.
We say that $\rho$ is a \emph{separable state} if it can be written as a convex combination of product states:
\be \rho=\sum_i p_i \rho_1^i \otimes \rho_2^i,  \ \ \ p_i \geq 0, \ \ \sum_i p_i=1.\label{equationconvexsum}\ee
with $\rho_1^i  \in D\left(\mathcal{H}_1\right)$ and  $\rho_2^i \in
D\left(\mathcal{H}_{2}\right)$.
The density matrices that cannot be written as in \eqref{equationconvexsum} are called \emph{entangled states}.
\end{defi}

\begin{ex}[Entangled states of two qubits]
The simplest non-trivial composite quantum  system is  the system of two qubits.
 The pure separable states of this system are given by vectors of the form
 $$\ket{\Psi}=\ket{\psi_1} \otimes \ket{\psi_2}$$
 where $\ket{\psi_i}$ represent states of a qubit. Hence, every pure separable state is of the form
\be\alpha_1\alpha_2\ket{00}+\alpha_1\beta_2\ket{01} + \beta_1\alpha_2\ket{10}+ \beta_1\beta_2\ket{11},
\label{equationpureseparable}\ee
 $\alpha_i, \beta_i \in \mathbb{C}$, $|\alpha_i|^2+ |\beta_i|^2=1$.
 Very few pure states can be written  this way. Indeed, the set of pure separable states is a quadric of complex dimension
 two in a three
 dimensional complex manifold \cite{BH01, Amaral06, TerraCunha07}. For example, the states
 \begin{eqnarray}
\ket{\Phi_{\pm}}&=&\frac{\ket{00} \pm \ket{11}}{\sqrt{2}}\nonumber\\
 \ket{\Psi_{\pm}}&=&\frac{\ket{01} \pm \ket{10}}{\sqrt{2}},
 \end{eqnarray}
 called the \emph{Bell states}, 
 can not be written in the form \eqref{equationpureseparable}, and hence represent entangled states. 
 
 Deciding if a mixed state is entangled or not is also easy for this system. Let
\begin{eqnarray}
T:M(\mathcal{H})&\longrightarrow &M(\mathcal{H}) \nonumber\\
\rho &\longmapsto &\rho^T
\end{eqnarray}
be the transposition map for a fized basis and  $T\otimes I$ its extension to a composite system, called  \emph{partial transposition}. We have the following result

\begin{teo}[Peres-Horodecki criterion \cite{Peres96, HHH96}]
 A density matrix $\rho$ of a two qubit system is separable iff its partial transposition is a density matrix.
\end{teo}

For other composite system of higher dimension, the partial transposition of every separable density matrix is also a 
density matrix, but the converse does not hold, unless  one of the subsystem has dimension three and the other has dimension
two. In these cases, deciding if a state is separable or not is not easy. For more information on separability criteria, see
\cite{NC00, BZ06, HHHH09, Amaral10} and references therein.
\end{ex}

Entangled states are responsible for many interesting features in quantum theory and  also play an important role
in many protocols that give us strong evidence that quantum information is more powerful than classical information.
For example, entanglement is the key resource in superdense coding \cite{BW92}, teleportation \cite{BBCJPW93}, 
quantum cryptography (see \cite{wikiquantumcryptography, HHHH09} and references therein) 
 Deutsch's and Shor's algorithms \cite{DJ92, Shor99}, just to cite a few examples. Not all entangled states are useful 
 for all tasks: the 
performance of a given state depends on the degree of entanglement it possesses in a very subtle way. Large amounts of
entanglement are
not necessarily good.
Quantifying entanglement
is then very important, but unfortunately it is a very hard task. There are many entanglement quantifiers, and they do not
define the same preorder in the set of density matrices of a system. The reader can find an introduction 
to entanglement
quantifiers in \cite{NC00, BZ06, HHHH09, Amaral10} and references therein.

As a corollary of  assumption \ref{asglobalstate} and the no-signaling principle, given a state of a composite system, 
we can associate a reduced state to every subsystem. In a quantum model, each reduced state is given by a 
density matrix in the corresponding state space.

\begin{defi}[Reduced density matrices]
Given a multipartite system composed of subsystems  $1$ and $2$ in a state $\rho_{12}$, the reduced states of $1$ and $2$ are 
given by
$$\rho_1=\tr_2(\rho_{12}), \ \ \rho_2=\tr_1(\rho_{12}),$$
where $\tr_1= \tr \otimes I$ and $\tr_2=I \otimes \tr$,
$\tr: L\left(\mathcal{H}\right) \rightarrow \mathbb{C}$ denoting the usual trace functional over the space of operators,
are called 
\emph{partial traces}.
Matrix $\rho_i$ is called the \emph{reduced  density matrix} of subsystem $i$. 
\end{defi}

\begin{ex}
 It was proved in example \ref{extwoqbits} that  a density 
matrix of the system of two qubits can be written in the form
$$\rho=\frac{1}{4}\left(I + \sum_i R_{0i} \  I \otimes \sigma_i + \sum_i R_{i0} \ \sigma_i \otimes I +
\sum_{ij}R_{ij} \ \sigma_i 
\otimes \sigma_j\right)$$ where
$$R_{ij}=\tr(\sigma_i \otimes \sigma_j \ \rho ).$$
This state can also be represented by the $4 \times 4$ 
matrix $R$, whose entries are the coefficients $R_{ij}$ defined above, with
$R_{00}=1/4$. 

Using the partial trace, we find that  $$\left[\begin{array}{ccc}
R_{01} &R_{02}& R_{03}
\end{array}\right]$$
is the Bloch vector of the second qubit, while $$\left[\begin{array}{ccc}
R_{10} &R_{20}& R_{30}
\end{array}\right]$$
is the Bloch vector of the first qubit.
\end{ex}

In this section we have discussed results related to bipartite systems, but all of them can be generalized
to system with more parties. The state space has a much richer structure  in those cases and 
finding separability criteria and entanglement quantifiers is even harder \cite{HHHH09}.

\subsection{\textsf{Transformations}}

The set of allowed transformations $\mathcal{T}\left(\mathcal{H}\right)$ in a quantum model 
corresponds to the largest set of linear transformations acting
on the set of operators in $\mathcal{H}$ such that constraint \ref{coperation} is satisfied.

Let \begin{eqnarray*}
\Phi: M(\mathcal{H}) &\longrightarrow &M(\mathcal{H})\\
\rho &\longmapsto &\rho',
\end{eqnarray*}
be a linear map in $M(H)$. Let us now verify what conditions are imposed on $\Phi$ by constraint \ref{coperation}.

Suppose $\mathcal{H}$ is a Hilbert space of complex dimension $d$. The elements of $M(H)$ can be written
as $d \times d$ matrices and the elements of $\mathcal{T}\left(\mathcal{H}\right)$ can be written as $d^2\times d^2$ matrices. We will use two 
indices to write the components of a matrix in $M(H)$ and four indices to write the components of a matrix in $\mathcal{T}\left(\mathcal{H}\right)$.
Then, if $\rho'=\Phi\left(\rho\right)$, we have

$$\rho'_{m\mu}=\sum_{n\nu}\Phi_{\hspace{-0.5em}\tiny\begin{array}{c}
m\mu \vspace{-0.5em}\\
\vspace{-2em} n\nu
 \end{array}}\rho_{n\nu}.$$

Let us see what we can say about the components of the map $\Phi $. We require that $\Phi$ takes states in $\mathcal{S}(H)$
to states in $\mathcal{S}(H)$ and this implies that a number of properties for $\Phi$ must hold. 
The first one is that $\rho '=\Phi(\rho)$ must be
 an Hermitian matrix for every state $\rho$: 
$$\rho'= \left(\rho'\right)^{\dagger}\ \Rightarrow  \ \rho'_{m\mu}=(\rho'_{\mu m})^*.$$
This implies that

\be \sum_{n\nu}\Phi_{\hspace{-0.5em}\tiny\begin{array}{c}
m\mu \vspace{-0.5em}\\
\vspace{-2em} n\nu
 \end{array}}\rho_{n\nu} = \sum_{n\nu}\Phi^*_{\hspace{-0.5em}\tiny\begin{array}{c}
\mu m \vspace{-0.5em}\\
\vspace{-2em} \nu n
 \end{array}}\rho^*_{\nu n}=
\sum_{n\nu} \Phi^*_{\hspace{-0.5em}\tiny\begin{array}{c}
\mu m \vspace{-0.5em}\\
\vspace{-2em}\nu n
 \end{array}}\rho_{n\nu}
 \label{equationphi1}
 \ee
which holds for all choices of $\rho$ only if
 \be\Phi_{\hspace{-0.5em}\tiny\begin{array}{c}
m\mu \vspace{-0.5em}\\
\vspace{-2em} n\nu
 \end{array}}=\Phi^*_{\hspace{-0.5em}\tiny\begin{array}{c}
\mu m \vspace{-0.5em}\\
\vspace{-2em}\nu n
 \end{array}}.
 \label{equationphi2}
\ee


The second condition we have to impose is that $\tr(\rho') \leq 1$. Then
\be\sum_m\rho'_{mm}=\sum_m\sum_{n\nu}\Phi_{\hspace{-0.5em}\tiny\begin{array}{c}
mm \vspace{-0.5em}\\
\vspace{-2em} n\nu
 \end{array}}\rho_{n\nu}\leq 1.\label{equationphi3}\ee
Let $\left\{\ket{1}, \ldots, \ket{d}\right\}$ be an
orthonormal basis for $\mathcal{H}$ and define $\rho_n=\ket{n}\bra{n}$, $\ 1\leq n \leq d.$ Using $\rho=\rho_n$ in
equation \eqref{equationphi3}, we conclude that
\be\sum_m \Phi_{\hspace{-0.5em}\tiny\begin{array}{c}
mm \vspace{-0.5em}\\
\vspace{-2em} nn
 \end{array}} \leq 1.\ee
 Using $\rho$ as the matrix with all components equal to zero except $\rho_{nn}, \rho_{\nu,\nu}, \rho_{n\nu}, \rho_{\nu n}$, that 
 are all equal, we conclude also that
 \be\sum_m \Phi_{\hspace{-0.5em}\tiny\begin{array}{c}
mm \vspace{-0.5em}\\
\vspace{-2em} nn
 \end{array}}+ \sum_m \Phi_{\hspace{-0.5em}\tiny\begin{array}{c}
mm \vspace{-0.5em}\\
\vspace{-2em} n\nu
 \end{array}} + \sum_m \Phi_{\hspace{-0.5em}\tiny\begin{array}{c}
mm \vspace{-0.5em}\\
\vspace{-2em} \nu n
 \end{array}}+ \sum_m \Phi_{\hspace{-0.5em}\tiny\begin{array}{c}
mm \vspace{-0.5em}\\
\vspace{-2em} \nu \nu
 \end{array}}\leq 2.\ee

 When $\Phi$ preserves the norm of the states, the same calculation show that
\be\sum_m\rho'_{mm}=\sum_m\sum_{n\nu}\Phi_{\hspace{-0.5em}\tiny\begin{array}{c}
mm \vspace{-0.5em}\\
\vspace{-2em} n\nu
 \end{array}}\rho_{n\nu}=1\label{equationphi4}.\ee
 Using $\rho$ as the matrix with  $\rho_{nn}= \rho_{\nu\nu}$,  $\rho_{n\nu}= \rho^{*}_{\nu n}=i\rho_{nn}$ and all other
 entries equal to zero, we get one extra constraint that implies the foolowing condition
\be\sum_m \Phi_{\hspace{-0.5em}\tiny\begin{array}{c}
mm \vspace{-0.5em}\\
\vspace{-2em} n\nu
 \end{array}} =\delta_{n\nu}.\ee

 The elements of $\mathcal{T}\left(\mathcal{H}\right)$ that preserve the norm of the states are
called \emph{trace preserving maps}. The maps that  do not increase the norm of some states are
called \emph{trace non-increasing maps}. 
 
 The next requirement we impose is that if $\rho$ is a positive matrix, then  $\Phi(\rho)$ must also be positive.
 
 \begin{defi}
  A map $\Phi: M(\mathcal{H})\rightarrow M(\mathcal{H})$ is called positive if the image of a positive matrix under $\Phi$ is
  also a positive matrix.
 \end{defi}

Every map $\Phi \in \mathcal{T}\left(\mathcal{H}\right)$ is a positive map. The converse does not hold, as we will see in a moment.

To help in the characterization of positive maps, we define the \emph{dynamical matrix} of $\Phi$ as the $d^2 \times d^2$ matrix 
with entries
\be D_{\hspace{-0.5em}\tiny\begin{array}{c}
mn \vspace{-0.5em}\\
\vspace{-2em} \mu \nu
 \end{array}}=\Phi_{\hspace{-0.5em}\tiny\begin{array}{c}
m\mu \vspace{-0.5em}\\
\vspace{-2em} n\nu
 \end{array}}.\ee
 \vspace{1em}

When $\Phi$ is an Hermitian map, its dynamical matrix is also Hermitian. When $\Phi$ is trace non-increasing we have
\begin{eqnarray}
 \sum_m D_{\hspace{-0.5em}\tiny\begin{array}{c}
mn \vspace{-0.5em}\\
\vspace{-2em} m \nu\end{array}}&\leq &1 \nonumber \\
\vspace{1em}\sum_m D_{\hspace{-0.5em}\tiny\begin{array}{c}
mn \vspace{-0.5em}\\
\vspace{-2em} mn
 \end{array}}+ \sum_m D_{\hspace{-0.5em}\tiny\begin{array}{c}
mn \vspace{-0.5em}\\
\vspace{-2em} m\nu
 \end{array}} + \sum_m D_{\hspace{-0.5em}\tiny\begin{array}{c}
m\nu \vspace{-0.5em}\\
\vspace{-2em} m n
 \end{array}}+ \sum_m D_{\hspace{-0.5em}\tiny\begin{array}{c}
m\nu \vspace{-0.5em}\\
\vspace{-2em} m \nu
 \end{array}}&\leq & 2
\end{eqnarray}

and when $\Phi$ is trace preserving we have
 \be\sum_m D_{\hspace{-0.5em}\tiny\begin{array}{c}
mn \vspace{-0.5em}\\
\vspace{-2em} m \nu\end{array}}=\delta_{n\nu}.\ee
 
Let us see now what are the consequences of the positivity of $\Phi$ in the dynamical matrix $D$.
Suppose $\rho$ is  a pure state. Then $\rho=\ket{\phi}\bra{\phi}$ and $\rho_{m\mu}=\phi_{m}\phi_{\mu}^*$. When $\Phi$ is 
positive, $\rho'$ is positive and then, for all $\ket{\psi} \in \mathcal{H}$
$$0 \leq \langle \psi|\rho'|\psi\rangle=\sum_{m\mu}\psi_m^*\rho'_{m\mu}\psi_\mu =\sum_{m\mu n \nu}\psi_m^* \phi_n D_{\hspace{-0.5em}\tiny\begin{array}{c}
mn \vspace{-0.5em}\\
\vspace{-2em} \mu \nu \end{array}}\psi_\mu \phi^*_\nu = \langle \phi^*|\langle \psi |D|\psi\rangle |\phi^*\rangle .$$
Then, if $\Phi$ is a positive map, $\langle \phi^*|\langle \psi |D|\psi\rangle |\phi^*\rangle \geq 0$ for all 
$\ket{\phi}, \ket{\psi} \in \mathcal{H}$. 

\begin{defi}
 A $d^2 \times d^2$ matrix $D$  is called   \emph{block positive } if 
 $$\langle \phi^*|\langle \psi |D|\psi\rangle |\phi^*\rangle \geq 0 \ \forall \
\ket{\phi}, \ket{\psi} \in \mathcal{H}.$$ 
\end{defi}

Then, if $\Phi$ is a positive map, $D$ is a block-positive matrix. This condition is also sufficient.

\begin{teo}[\textsf{Jamio\l kowski}]
A linear map $\Phi: M(\mathcal{H})\longrightarrow M(\mathcal{H})$ is positive iff its dynamical matrix is 
block positive.
\end{teo}

The proof of this result can be found in references \cite{BZ06,Amaral10}.

As we already discussed previously, the condition that $\Phi$ takes states to states in $\mathcal{S}(\mathcal{H})$ is
not sufficient to consider $\Phi$ as an allowed transformation. The constraint of complete state preservation
requires that this must also 
happen when the system is part of a multipartite system. 

\begin{defi}
Let $\Phi$ be a positive map acting on $M(\mathcal{H})$. Let $\mathcal{H}'$  be any other vector space of dimension $k$ and $I$
be the identity map acting on $M(\mathcal{H}')$.
If the map $\Phi \otimes I$, acting on
 $M(\mathcal{H}\otimes \mathcal{H}')$, is positive,  we say that $\Phi$ is   $k$-\emph{positive}.
If $ \Phi$ is $k$-positive for every $k \in \mathbb{N}$, we say that $\Phi$ is  \emph{completely positive}.
\end{defi}

We have seen that in classical theories every state preserving transformation 
is automatically completely state preserving. This is a consequence of the fact that every state is written as convex
combination of product states. The existence of entangled states in quantum theory implies, among many other interesting 
things,
that there are many state preserving maps, namely, the trace non-increasing positive maps acting in $M(\mathcal{H})$, 
that are not completely state preserving.

\begin{ex}
Not every positive map is completely positive. For example, consider the transposition map $T$
acting on the sate space of one qubit.
This map is positive, but 
$$T\otimes I\left( \ket{\Psi_-}\bra{\Psi_-}\right) = T\otimes I \left(\frac{1}{2} \left[\begin{array}{cccc}
0&0&0&0\\
         0&1&-1&0\\
         0&-1&1&0\\
         0&0&0&0
        \end{array}\right]\right)= \frac{1}{2} \left[\begin{array}{cccc}
0&0&0&-1\\
         0&1&0&0\\
         0&0&1&0\\
         -1&0&0&0
        \end{array}\right]$$
which is not  positive.

\end{ex}

If $\Phi$ belongs to $\mathcal{T}(\mathcal{H})$, condition \ref{coperation4} implies that $\Phi \otimes I$ also takes states to states 
in $M(\mathcal{H}\otimes\mathcal{H}')$ for every Hilbert space $\mathcal{H}'$. This means that $\Phi$ 
must be a completely positive
map.

\begin{teo}
 The set of allowed transformations $\mathcal{T}(\mathcal{H})$ is the set of trace non-increasing completely-positive maps acting on
 $M(\mathcal{H})$.
\end{teo}

We can also use the dynamical matrix $D$ to find necessary and sufficient conditions for the complete positivity of $\Phi$.

\begin{teo}[Choi]
 A map $\Phi$ acting on $M(\mathcal{H})$ is completely positive iff the corresponding dynamical matrix $D$ is positive.
\end{teo}

Using this theorem it is possible to prove that completely positive maps can be written in a simple way using
the Kraus representation.

\begin{teo}[Kraus representation]
A linear map $\Phi$ is completely positive iff it can be written in the form
$$\rho\longmapsto \rho'=\sum_iA_i\rho A_i^\dagger,$$
where each $A_i$ is a square matrix of the same size of  $\rho$. Furthermore, $\Phi$ is trace preserving iff the matrices
$A_i$ satisfy
$$\sum_iA_i^\dagger A_i = I.$$
\label{propkraus}
\end{teo}

For  proofs of these results, see \cite{BZ06, Amaral10}.

\subsection{\textsf{Measurements}}

By definition \ref{allowedt},  measurements in quantum models are given by a set of
trace non-increasing completely positive maps
$\{\Phi_1, \Phi_2, \ldots, \Phi_n\}$ such that 
\be\sum_i \tr\left[\Phi_i(\rho)\right]=\tr\left(\rho\right)
\label{equationmeasurement}\ee
for every $\rho \in D\left(\mathcal{H}\right).$

There are two important special cases:  POVM's and  projective measurements.

\begin{defi}
 A \emph{positive-operator valued measurement} (POVM) is a measurement $\{\Phi_1, \Phi_2, \ldots, \Phi_n\}$ in which each
 transformation $\Phi$ is given by
 \begin{subequations}
 \be\Phi_i(\rho)=M_i \rho M_i^{\dagger}\ee
 where the $M_i$ are  matrices in $M(\mathcal{H})$ satisfying 
\be\sum_iM_i^\dagger M_i=I.\label{equationpovm}\ee
 The probability of outcome $i$  for the state $\rho$ is
\be p_i=Tr(M_i^\dagger M_i\rho),\ee
 and the unnormalized state  after outcome  $i$ is
 \be\rho_i=M_i\rho M_i^\dagger.\ee
 \end{subequations}
\end{defi}

A POVM is defined if we give a set of matrices $\{M_1, M_2, \ldots, M_n\}$ satisfying equation
\eqref{equationpovm}. 
Theorem \ref{propkraus} implies that every measurement in quantum mechanics is the coarse graining of a POVM.

\begin{defi}
 A measurement $\{\Phi_1, \Phi_2, \ldots, \Phi_n\}$ is called \emph{projective} if it is a POVM in which the matrices $M_i$ are
 projectors acting on
 $\mathcal{H}$. If every $M_i$ is a unidimensional projector, the measurement is called a \emph{complete projective 
 measurement}.
\end{defi}

A projective measurement is defined if we give a set of  projectors $\{P_1, P_2, \ldots, P_n\}$ satisfying 
$$\sum_iP_i=I.$$
This implies that the $P_i$ are orthogonal projectors.

Projective measurements are the ones satisfying outcome repeatability. A curious feature of quantum theory is that, 
contrary to classical theory, even when we restrict the measurements to outcome repeatable measurements, the pure states
are not dispersion free states. Indeed, given a projective measurement $\{P_1, P_2, \ldots, P_n\}$,
a pure state $\ket{\psi}\bra{\psi}$ gives outcome $i$ with probability one iff
$$P_i\ket{\psi}=\ket{\psi}$$
and this happens iff $\ket{\psi}$ belongs to the subspace in which $P_i$ projects. Of course, most of the pure states do not
satisfy this property, and hence there are different outcomes with non-zero probability.
Nevertheless, there is a difference in the behavior of pure and mixed states when it comes to outcome definiteness.

\begin{teo}
The density matrix $\rho$ represents a pure state if, and only if, there is a complete
projective measurement with probability $p_i=1$ for some outcome $i$. 
\end{teo}

\begin{dem}
Let $\rho=\ket{\psi}\bra{\psi}$. Take a complete projective measurement such that outcome $i$ is associated to
the one-dimensional projector $P_i=\ket{\psi}\bra{\psi}$. Then we have that $p_i=1$.

Suppose now that  
 $$\rho=\sum_j \lambda_i \ket{\psi_j}\bra{\psi_j}$$
 is a mixed state and $\{P_1, P_2, \ldots, P_n\}$ is a complete projective
 measurement. This means that $P_i=\ket{\phi_i}\bra{\phi_i}$ and that $\{\ket{\phi_i}\}$ is an orthonormal basis for 
 $\mathcal{H}$.
 If the probability of outcome $i$ is $p_i=1$ for the state $\rho$, $|\braket{\phi_i}{\psi_j}|=1$ for every $j$, which means
 that $\rho=\ket{\phi_i}\bra{\phi_i}$ is a pure state, a contradiction.
\end{dem}

\subsection{\textsf{Compatibility of projective measurements}}

Compatibility of two outcome-repeatable measurements can be easily decided in quantum models from the matrices 
defining the measurements.

\begin{teo}
 Two projective measurements $\{P_1, P_2, \ldots, P_n\}$ and $\{Q_1, Q_2, \ldots, Q_m\}$ are compatible iff
 $P_i$ and $Q_j$ commute for every $1 \leq i \leq n$ and $1 \leq j \leq m$.
\end{teo}

\begin{dem}
 The measurements are compatible if they  are both coarse grainings of the same complete projective measurement.
 This happens iff all $P_i$ and $Q_j$ are simultaneously diagonalized, and hence, iff they commute. 
\end{dem}

\subsection{\textsf{Expectation value of a measurement}}

In classical probability theory, the concept of \emph{random variable}, 
a real-valued function defined on the sample space $\Omega$, is a useful tool that allows
the definition of many important quantities such as expectation values and  variances. Something similar
can be done
in generalized probabilistic theories. We  simply label the outcomes of a measurement  by real numbers, and then we are
able to 
define the same quantities, related to the value of each outcome and the corresponding  
probabilities.

\begin{defi}
\label{defiexpectation}
 The \emph{expectation value} of a measurement $\mathrm{ M}$ with outcomes $a_i \in \mathbb{R}$ in a state $\rho$ is
 \be \left\langle\  \mathrm{ M} \ \right\rangle= \sum_i a_ip_i\ee
 where $p_i$ is the probability of obtaining $a_i$ when measurement $\mathrm{M}$ is applied on state $\rho$.
\end{defi}

For projective measurement in a quantum model, each outcome $a_i$ is associated
to a projector $P_i$ and the probability $p_i$ is given by 
\be p_i=\tr\left(P_i \rho\right)\ee
where $\rho$ is the operator corresponding to the state of the system. Hence, the expectation value of a projective measurement
 $\mathrm{ P}$ can be easily calculated
$$\left\langle \ \mathrm{ P} \ \right\rangle= \sum_i a_ip_i=\sum_i a_i\tr\left(P_i \rho\right)$$
and by the linearity of the trace
\be\left\langle \ \mathrm{P} \  \right\rangle=\tr\left(\sum_i a_iP_i \rho\right)=\tr\left(O\rho\right)\ee
where $O=\sum_ia_iP_i$ is an Hermitian operator with eigenvalues $a_i$. The eigenspace associated to $a_i$ is the
 subspace in which $P_i$ projects. This operator is called the \emph{observable} associated to the measurement.
 This proves the following
 
 \begin{teo}
  The expectation value of an observable $O$, associated to a  projective measurement $\mathrm{ P}$,  for a 
  given state is
 \be\left\langle \ \mathrm{ P} \  \right\rangle=\tr\left(O\rho\right),\label{eqqexpvalue}\ee
 where $\rho$ is the  density operator associated to the state.
 \end{teo}

 When the state is pure, $\rho=\ket{\psi}\bra{\psi}$ and equation \eqref{eqqexpvalue} reduces to
 $$\left\langle \ \mathrm{ P} \  \right\rangle=\sand{\psi}{O}{\psi}.$$

\subsection{\textsf{Processes}}

The same results we presented above for transformations in $\mathcal{T}\left(\mathcal{H}\right)$ can be proven for processes, maps that change
the type of system under consideration. The processes must also obey physical requirements similar to the ones imposed
to the elements of $\mathcal{T}\left(\mathcal{H}\right)$.  Let $\mathcal{H}_1$ and $\mathcal{H}_2$ be two Hilbert space, not necessarily of 
the same dimension and let
$$\Lambda:M(\mathcal{H}_1)\rightarrow M(\mathcal{H}_2)$$
be a linear map. The definitions of
\emph{positive}, \emph{k-positivity} and \emph{completely positivity}
can be generalized to this kind of map.


\begin{defi}
A map $\Phi: M(\mathcal{H}_1)\longrightarrow M(\mathcal{H}_2)$ is called \emph{positive} if $\Phi(\rho)$ 
is positive for every positive $\rho \in M(\mathcal{H}_1)$.
If 
$$\Phi \otimes I:M(\mathcal{H}_1\otimes \mathcal{H}')\rightarrow M(\mathcal{H}_2\otimes \mathcal{H}')$$
is positive, where  $\mathcal{H}'$ is a Hilbert space of dimension $k$, $\Phi$ is a $k$-\emph{positive} map.
$\Phi$ is called \emph{completely positive} if it is a $k$-positive map for every $k$.
\end{defi}

When the bases of $\mathcal{H}_1$ and $\mathcal{H}_2$ are fixed, we can represent the map $\Phi$ by a matrix, which
we will also denote by $\Phi$. Once more, since $\Phi$ acts  in $M(\mathcal{H}_1)$, the entries of the 
corresponding matrix will carry four indices. The action of $\Phi$ in a density matrix $\rho \in \mathcal{H}_1$ is a
density matrix $\rho' \in \mathcal{H}_2$, whose entries are given by
$$\rho'_{mn}=\sum_{\mu \nu} \Phi_{\hspace{-0.5em}\tiny\begin{array}{c}
m\mu \vspace{-0.5em}\\
\vspace{-2em} n\nu
 \end{array}} \rho_{\mu\nu}.$$

We can also define the dynamical matrix $D$ associated to the process $\Phi$
$$D_{\hspace{-0.5em}\tiny\begin{array}{c}
mn \vspace{-0.5em}\\
\vspace{-2em} \mu \nu
 \end{array}}=\Phi_{\hspace{-0.5em}\tiny\begin{array}{c}
m\mu \vspace{-0.5em}\\
\vspace{-2em} n\nu
 \end{array}}.$$
If $\mathcal{H}_1$ and $\mathcal{H}_2$ do not have the same dimension, the matrix of $\Phi$ is not a square matrix but the
associated dynamical matrix $D$ is.
If $\dim(\mathcal{H}_1)=k$
 and $\dim(\mathcal{H}_2)=l$, then the matrix of $\Phi$ is a $k^2 \times l^2$ matrix, while $D$ is a square matrix of size
  $kl \times kl$. The version of  Jamio\l kowski's and Choi's theorems for processes can also be proven.


 \begin{teo}
A linear map $\Phi: M(\mathcal{H}_1)\longrightarrow M(\mathcal{H}_2)$ 
is positive iff the associated dynamical matrix $D$ is block-positive.
 It is completely positive iff $D$ is positive.
 \end{teo}

The dynamical matrix can be writen in terms of the action of  $\Lambda \otimes I:M(\mathcal{H}_1 \otimes
\mathcal{H}_2)\rightarrow M(\mathcal{H}_2 \otimes \mathcal{H}_2)$ in the state
$P_+=|\Phi_+\rangle \langle \Phi_+| \in M(\mathcal{H}_A \otimes
\mathcal{H}_A)$ where $$|\Phi_+\rangle=\frac{1}{d}\sum_i |ii\rangle,$$
$d$ being the dimension of $\mathcal{H}_1$.


\begin{teo}[Choi-Jamio\l kowski's Isomorphism]
Given a linear map $\Lambda:M(\mathcal{H}_1)\rightarrow M(\mathcal{H}_2),$
$$D_\Lambda=\Lambda \otimes I(|\Phi_+\rangle \langle \Phi_+|).$$
\end{teo}

A proof of this result can be found in references \cite{BZ06, Amaral10}.

\section{\textsf{Final Remarks}}
\label{sectiongeneral}

In this section we will discuss briefly  general properties that follow from the assumptions we have made about the structure
of  general probability theories. A number of properties are satisfied by all of them but others
 are present only in specific kinds of models. Classical probability theory, for
example, has a number of characteristics that distinguish it from all others. 
Some properties thought as special features of quantum theory 
are in fact general, and in many aspects it is classical probability theory that emerges as a very particular case. In this 
sense,
many of these properties can be seen as a signature of the ``non-classicality'' of the theory, rather than a signature of
its ``quantumness''.  For more detailed discussion and for
the proofs of the results presented below, see reference \cite{Barret06}.

The first one, that we already mentioned, is the fact that classical theory is the only one in which every 
mixed state can be decomposed uniquely as a convex combination of pure states. This is due to the fact that
the state space of a classical model is a simplex, and this is the only convex body with this property.

Another interesting property of classical theories is the effect of an  
outcome-repeatable measurement in  the system. The definition of measurement
we gave includes a transformation of the state of the system. 
Note that this fact by itself should not create any panic, since even in classical probability
theory the state of the system can change after a measurement. 
What is special about quantum theory is that pure states can change after a measurement,
whereas in classical probability theory only  mixed states can change, as we saw in 
section \ref{sectionclassical}. This is not the case for most theories in this framework. 
The same questions of interpretation of the change of the state after a measurement that bother quantum theory 
for so many years may show up once again. We will not jump into the quicksand of philosophical debate here and 
we will assume a clear practical position when it comes to interpretation of our assumptions and  their consequences. 
Nevertheless we mention that there is room for a lot of different points of view in this subject and that the reader 
should feel free to think about it as much as (s)he wants \cite{Robba13}.

In theorem \ref{theoremlcpure} we proved that any state of a 
composite system can be written as a linear combination of product states. This does not 
imply, and we also did not assume, that every state can be written as a \emph{convex combination} of product states. 
States with this property are called \emph{separable}, and the states that are not separable are called \emph{entangled}.
As we saw in section \ref{sectionquantum}, in some models there may be entangled states. 
Entangled states are closely related to an interesting feature of quantum theory called nonlocality, 
that we will define properly in appendix \ref{chapternonlocality}, although they are not always equivalent \cite{VB14,BCPSW13}. 
Classical probability theories do not allow entangled states and do not exhibit nonlocality,
but quantum theory and many other theories do.

Another feature of all classical theories is that they are the only ones allowing cloning of an arbitrary pure state. 
A probabilistic cloning procedure is given by the following steps: begin with a system  in a pure state $\rho$; introduce an ancilla system 
of the same type, prepared in a fixed pure state $\rho_0$;  apply a joint transformation  on the pair of systems
such that the final state is
$$\rho\times \rho$$
with probability larger than zero.

\begin{teo}
 If in a given probability theory there is a probabilistic cloning procedure to every model, then the theory is
 classical.
\end{teo}

The proof of this result can be found in reference \cite{Barret06}.

We can recognize many properties exclusive of classical theories. This allows us to arrive in this kind of theory
if we make all the assumptions done in section \ref{sectionstates} and \ref{sectionmultipartite} and postulate
also any one of this properties that single out classical theories among the other ones in this framework. The main
question motivating this work is if we can do the same for quantum theory: \emph{ is there any physical principle
that singles out quantum theory in the universe of all generalized probability theories? What different ways
are there of uniquely identifying quantum theory from the
other theories in the framework by adding as few extra assumptions as possible?}

The features connected to the quantum character of the theories are still not completely understood, but we believe
that the study of quantum contextuality is shedding light upon this quest.  

\chapter{\textsf{Non-contextuality inequalities}}
\label{chapterncinequalities}

Quantum theory   has an intrinsic   statistical  character. It does not provide the exact value of all measurements
 for any state  of the system, but rather the probabilities of the occurrence of each possible outcome,
even when the state of the system is pure. We have seen in section \ref{sectionquantum} that the expectation
value of a projective
measurement $\mathrm{P}$   in a state $\rho$ is given by
\be\left\langle \mathrm{ P}\right\rangle=\mbox{Tr}\left(O\rho\right)\label{eqmeanvalue}\ee
where $O$ is the observable associated to the measurement. We have seen also that there is no dispersion for $\mathrm{P}$ 
iff the support of $\rho$ is contained in an eigenspace of $O$.
This means that in general, there is a statistical distribution for the outcomes of $\mathrm{P}$, even if the state of the 
system is of the form $\ket{\psi}\bra{\psi}$. In this chapter we want to discuss this 
probabilistic character of quantum theory, focusing only in outcome-repeatable measurements, which means that
we will work with projective measurements from now on.

Consider a set with a huge number of copies of the same system, all prepared in the same way. Such a set will be called
an \emph{ensemble}.
To calculate the probability distribution of a given measurement for this preparation one can 
  perform this measurement in several copies, and count the relative frequencies of each outcome.
 For most measurements, this 
distribution has dispersion larger then zero.
Two possible explanations for this indeterminacy on the outcomes of the measurements are \emph{a priori} 
conceivable:

\begin{enumerate}
 \item[I.] The individual systems of the ensemble are in different states, in such a way that we could separate the copies
 in a number of sub-ensembles, each of them consisting in a definite state that is dispersion-free for all the measurements.
 The probabilistic character of the experiments is, in this case, explained by our lack of information: we do not know  
 everything about the system we are measuring and hence we can not predict the results.
 
  \item[II.] All individual systems are in the same pure state and that is all the information we can get. The laws of
  nature allow that different outcomes are possible even when we perform the same measurement in two 
  identically prepared systems.
\end{enumerate}
In this chapter we present a number of attempts to find objective criteria which allow us  to decide between these two options.
We will see that, under some very reasonable circumstances, there is no way out but to accept option II.

Before we enter the specific details of the proofs of the impossibility of option I, let us think about why 
option I seems so logical to our classical minds, modeled by our daily experience with macroscopic systems.
The necessity of the use of probabilities in the description of an experiment naturally arises
from the incompleteness of our knowledge about the parameters involved in it. Due to our classical intuition, 
we are used to think that if we knew everything about our experiment, 
two repetitions of the same procedure with exactly the same value for every possible parameter involved had to
provide the same result at the end.
It is reasonable to imagine that two replicas of the same object will remain identical if they are subjected to the exactly  
same process. If this is not the case, we would have no reason to call them identical in the first place.

Let us focus now in quantum theory and apply this reasoning to an ensemble of systems in the same state $\ket{\psi}\bra{\psi}$.
Since this ensemble will exhibit dispersion for most measurements, the elements of the ensemble
could not be identical and hence they could not all
be in the same state. Hence, the state assigned to this preparation by quantum theory can not be everything: there are more
parameters we must use in the description of these systems in order to get dispersion-free states. This unknown parameters 
may have different values in our ensemble, and  the probabilistic behavior is due to our lack of knowledge on these
``hidden variables.''

This line of thought lead many physicists to believe that quantum theory might be wrong, or at least, incomplete. Since
quantum theory is capable of reproducing every experimental data people could get in the laboratory up to these days, 
we have absolutely no evidence that it might be wrong. Hence, our best shot is to suppose the possibility of 
completing quantum theory, adding extra variables to the description of pure states, in a way that with all 
this information (of pure quantum state plus extra variables) we would be able to predict with certainty the outcome of
all measurements and in a way that when averaging over these extra variables we would get the quantum predictions.
This kind of completion of quantum theory is often called a \emph{hidden-variable model}.

A good example in which a similar argument applies is classical thermodynamics, which states physical laws involving
macroscopic aspects 
of matter, such as pressure, volume and temperature. These laws do not provide all the information about the systems
studied, since they appear when we average over a large number of atoms and we do not take into account the individual
parameter such as position and velocity of each atom. Although very useful for many applications, classical thermodynamics
does not explain phenomena such as Brownian motion, which require a more complete treatment,  provided by
statistical physics.

It happens that under the assumption of \emph{noncontextuality}, hidden-variable models compatible with 
quantum theory are not possible. This result is known as the \emph{Bell-Kochen-Specker theorem.} The noncontextuality hypothesis states that the value assigned by the model to a measurement can not depend on other compatible measurements performed jointly.

The first proof of this result was provided by Kochen and Specker  \cite{KS67}. It is based on a set of $117$ observables with possible outcomes
$0$ or $1$. This set is constructed in such a way that if we assign one of this values to each of them noncontextualy, we reach a contradiction with what we expect from quantum theory. The assumption of noncontextuality was so natural that it was only pointed out after by Bell \cite{Bell66}. Many other proofs using the same idea have been provided, using sets with a smaller number of observables.
They have an important  common feature: they are all \emph{state-independent}. This means 
that if we choose the set of observables as in any of these proofs, the assignment of definite values for the corresponding projective  measurements can not reproduce the statistics given
by any quantum state when we average over all possible values of the hidden variables. The reader interested in such proofs may find a number of examples in 
appendix \ref{chaptercontextuality}.

It is possible to provide simpler state-dependent proofs of the impossibility of hidden variables compatible with quantum
theory. The idea behind this kind of proof is to show that no hidden-variable model can reproduce the statistics of some
measurements for a given state of the corresponding system.  Some of this proofs use a very small number of vectors and hence
are much simpler than the state-independent ones.

One of the most common ways to provide a state-dependent proof of the Kochen-Specker theorem is using the so
called \emph{noncontextuality inequalities}. They are linear inequalities involving the probabilities of 
certain outcomes of the joint measurement of compatible observables that must be obeyed by any hidden-variable model
and can be violated by quantum theory with a particular choice of state and observables. In this 
chapter we study  noncontextuality inequalities 
and some different ways to approach the subject.

One advantage of the impossibility proofs using noncontextuality inequalities is that many of them use a small number of
observables, which may make them much more suitable for experimental implementations.  The experimental verification
of quantum violations was already performed for a number of 
inequalities, specially in the particular case of Bell inequalities, which are introduced in appendix \ref{chapternonlocality}.

Here we discuss two approaches to noncontextuality inequalities: the compatibility hypergraph approach, in section 
\ref{sectioncompatibility} and the Exclusivity graph approach, in section \ref{sectionexclusivity}. 
In section \ref{sectioncontextuality} we discuss the assumption of noncontextuality.
In section \ref{sectionsheaftheory} we explain the connection between the first approach and Sheaf theory.
In section \ref{sectionphysical} we discuss the probability distributions obtained with classical and quantum theories.
In section \ref{sectionncinequalities} we define  noncontextuality inequalities. The important examples of 
the KCBS inequality and the $n$-cycle inequalities are discussed in sections \ref{sectionKCBS} and
 \ref{sectionncycle}, respectively. In section \ref{sectionexclusivitygraph} we introduce the exclusivity graph,
 which is an important tool for both approaches. In section \ref{sectionexclusivityinequalities} we define 
 noncontextuality in the second approach and review the examples given before in this new perspective. 
 The graph theoretical formulation of quantum contextuality supplies new tools to understand the 
 differences between quantum and classical theories. In section \ref{sectionlargestcontextuality} we use some of
 these tools to find the scenarios exhibiting the largest quantum contextuality. We close the chapter with
 some final remarks.

\section{\textsf{The assumption of noncontextuality}}
\label{sectioncontextuality}

Let $\{O_1, O_2, \ldots, O_m\}$ be a set of compatible measurements. Such a set will be called a \emph{context}. 
Let $\{O_1, O'_2, \ldots, O'_n\}$ be another context containing $O_1$ and such that  $O_i$ and $O'_j$ are not necessarily
compatible. 
The compatibility between the elements of each context implies that they have a common refinement, which
allows us to design an experiment in which all of them can be 
jointly measured.  A hidden-variable model must provide a definite outcome for this measurement
and hence the model provides a set of definite outcomes for each context.

\begin{defi}
A \emph{hidden-variable model} for a system is a set of extra variables $\Lambda$ and a rule that 
 specifies for each pair $(\rho, \lambda)$, where $\rho$ is a pure state of the system and $\lambda \in \Lambda$,
 a definite set of outcomes for every maximal context\footnote{We say that a context is maximal if there is no
 other set of compatible measurements that contains it properly.}
 $\{O_1, O_2, \ldots, O_m\}$.
\end{defi}

Some authors consider hidden-variable models that are not deterministic, that is, the measurements may not have
definite outcomes for every state. Nonetheless, the ``non-determinism'' in those
models comes from the fact that we do not know everything about the system, and hence they can be completed to give a
deterministic model. We will not consider this kind
of model in this text.

Suppose now that a hidden-variable model is provided for the system. Such a model 
assigns a string of definite values to  both $\{O_1, O_2, \ldots, O_m\}$ and $\{O_1, O'_2, \ldots, O'_n\}$.
We demand that the value assigned to $O_1$ be independent of the context in which it appears: if the outcome of $O_1$
according to the model is $o_1$ when a joint measurement of $\{O_1, O_2, \ldots, O_m\}$ is performed, 
the same outcome $o_1$ must be assigned to $O_1$ by the model if we jointly measure
$\{O_1, O'_2, \ldots, O'_n\}$.

\begin{defi}
 We say that a hidden-variable model is \emph{noncontextual} if the value associated by the model to an observable $O$
 is independent of which and which  compatible measurements are performed jointly. 
\end{defi}

This observation was first pointed out by Bell \cite{Bell66}, who argued that there is no \emph{a priori} reason to require
noncontextuality from a hidden-variable model. Suppose we perform the measurement of an observable $O_1$ and together one may
choose to measure either $\{O_2, \ldots, O_m\}$ or $\{O'_2, \ldots, O'_n\}$, both compatible with $O_1$ but not to one another. These different possibilities 
may require completely different experimental arrangements, and hence to demand that the values associated to $O_1$
be the same can not be physically justified. The outcome of a measurement may depend not only on the state of the system, but 
also on the apparatus used to measure it. 

Although the measurement process and the interaction between system and apparatus are important issues
in quantum theory, this is not the problem here, since we could include all variables of the apparatus in
the model, and apply the same reasoning again. The point that makes the noncontextuality assumption plausible is that
there is no need to measure the compatible observables simultaneously. Suppose we measure $O_1$ and then we choose
what else we are going to measure, $\{O_2, \ldots, O_m\}$ or $\{O'_2, \ldots, O'_n\}$ or even if we are not
measuring anything else. The hidden-variable model 
should predict the outcome of $O_1$, but if this model is contextual this value would depend on a measurement 
that will be performed in the future or, even worst, on a decision to measure or nor, yet to be made!

Another way to enforce naturally the noncontextuality assumption is to design the experiment in such a way that 
the choice of $\{O_2, \ldots, O_m\}$ or $\{O'_2, \ldots, O'_n\}$ 
is made in a different region of the space in a time interval that forbids any signal to be
sent from one region to the other. Since no signal was sent, the choice of what is going to be measured in
one part can not disturb what happens in the other, what demands the model to be noncontextual.
In this situation, we say that the model is local and the noncontextuality assumption is usually referred to as the  \emph{locality assumption}. We talk about  this special case in appendix \ref{chapternonlocality}.

\section{\textsf{Contextuality: the compatibility hypergraph approach}}
\label{sectioncompatibility}

Suppose an experimentalist has many possible measurements to carry out in a physical system. Each measurement has a number of 
possible outcomes, that occur with a certain probability for a given state of the system.

\begin{defi}
\label{defimeasurementcover}
Let $X$ denote the set of possible measurements available.
A \emph{compatibility cover} $\mathcal{C}$ is a family of subsets  of $X$ such that
\begin{enumerate}
\item Each $C \in \mathcal{C}$ is a set of compatible measurements;
 \item $\cup_{C \in \mathcal{C}} C =X$;
\item   $C,C' \in \mathcal{C}$ and $C \subseteq C'$
implies $C = C'$. \label {antichain}
\end{enumerate}
\end{defi}

 As we mentioned previously, each $C \in \mathcal{C}$ is called a 
\emph{context}.
Condition \ref{antichain} is called \emph{anti-chain}
condition and it
guarantees that  all contexts in $\mathcal{C}$ are maximal.

We will assume without loss of generality that all measurements have the same number of outcomes. The set of possible outcomes
will be denoted by $O$. We remark here that the actual labels given to the outcomes are not important. The only important thing in what 
follows is the number of elements in $O$. 

\begin{defi}
 A triple $\left(X, \mathcal{C}, O \right)$ is called a \emph{compatibility scenario}\footnote{In this thesis, we will often use 
  the word \emph{scenario} instead of \emph{compatibility scenario}.}.
\end{defi}

The compatibility relations among the elements of $X$ can be represented with the help of a hypergraph.

\begin{defi}
The \emph{compatibility hypergraph} of a scenario $\left(X, \mathcal{C}, O \right)$ is a hypergraph such that the vertices are 
the measurements in $X$ and the hyperedges are the contexts $C \in \mathcal{C}$.
\end{defi}

Notice that the compatibility hypergraph does not suffice to identify the scenario, since the number of outcomes for 
each measurement is not determined.
For a given subset $C \in \mathcal{C}$, consider the set of possible outcomes for a joint measurement of the elements of $C$.
This set is the Cartesian product of $|C|$ copies of $O$ and will be denoted by $O^{C}$. This set can  be identified with the set of functions
$$\lambda: C \longrightarrow O.$$
Each function $\lambda \in O^C$ is called a \emph{section over} $C$.

When a system is prepared in a given state and the measurements in $C$ are performed subsequently, a set of outcomes in $O^C$
will be observed. This individual run of the experiment will be called an \emph{event}. Each event is an element of
$O^C$ and hence is represented by a section over $C$.

\begin{defi}
A \emph{probability distribution} $p$ for $\mathcal{C}$ is a family of functions
$p_C:  O^C\rightarrow [0,1]$ such that $\sum_{s\in O^C} p_C(s)=1$, $C \in \mathcal{C}$.
\end{defi}

Each probability distribution can be associated to a vector $p \in \mathbb{R}^n, 
n=\displaystyle{\sum_{C \in \mathcal{C}}\left|O^C\right|}$.
If we have $\mathcal{C}=\left\{C_1, C_2, \ldots, C_n\right\}$ and for each $C_i$ we have
$O^{C_i}= \left\{s_i^1, s_i ^2, \ldots , s_i^{m_i}\right\}$, we define
\be\small{
 p=\left[\begin{array}{ccccccccc}
             p_{C_1}\left(s_1^1\right)&p_{C_1}\left(s_1^2\right)&\ldots&p_{C_1}\left(s_1^{m_1}\right)&  \ldots&p_{C_n}\left(s_n^1\right)&p_{C_n}\left(s_n^2\right)&\ldots&p_{C_n}\left(s_n^{m_n}\right)
            \end{array}\right] }\ee
This association is discussed in more detail in reference \cite{AQBTC13}.

For a given compatibility cover, the set of possible probability distributions is a polytope with
$\displaystyle{\prod_{C \in \mathcal{C}}\left|O^{C}\right|}$ vertices. Each vertex corresponds to 
probability one for one of the outcomes
$s \in O^C$ for each context $C \in \mathcal{C}$.
All other distributions are convex combinations of these vertices.

Let $C=\{M_1, \ldots, M_n\}$ be a context in $\mathcal{C}$. Each element of $O^C$ is a string $s= \left(a_1, \ldots, a_n\right)$ with
 $n$ elements of $O$. For each $U \subset C$, there is a natural restriction
\begin{eqnarray}
 r^C_U:O^C &\rightarrow& O^U \\
s=\left(a_i\right)_{M_i \in C} &\mapsto &s|_U=\left(a_i\right)_{M_i \in U}.
\label{eqrestriction}
\end{eqnarray}
This operation corresponds to dropping the elements in the string $s$ that do not correspond to measurements in $U$.

Given a  probability distribution in $C \in \mathcal{C}$ we can also naturally define  marginal distributions for each  $U\subset C$:
\begin{eqnarray}
 p^C_{U}\ :\ O^{U}&\rightarrow &[0,1] \nonumber\\
p^C_{U}(s)&=& \sum_{s' \in O^C; r^C_U(s') = s} p_C(s').
\end{eqnarray}
The superscript $C$ in $p^C_U$ is necessary  because the marginals may depend on the context $C$. 
 
 \begin{ex}
Consider 
 the situation
where $$X=\{M_1, M_2, M_3\} \  \ \mbox{and} \ \  \mathcal{C}=\left\{C_1=\{M_1, M_2\}, C_2=\{M_2, M_3\}\right\},$$ 
each measurement with two possible outcomes $\pm 1$.
The extreme distribution with $p_{C_1}(1,1)=1$ and $p_{C_2}(-1,-1)=1$ gives the marginals $p_{M_2}^{C_1}(1)=1$ and $p_{M_2}^{C_2}(1)=0$.
 \end{ex}

We will reject  distributions with this property: we require that if
 two contexts $C_1$ and $C_2$ overlap, the marginals defined by $p_{C_1}$ and $p_{C_2}$  in the intersection be the same.

\begin{defi}
\label{definondisturbance}
The \emph{non-disturbance} set $\mathcal{X}\left(\Gamma\right)$ is the set of probability distributions such that if the intersection of two contexts $C$ and $C'$
is non-empty, then $p^C_{C\cap C'}= p^{C'}_{C\cap C'}$. A probability distribution $p \in \mathcal{X}\left(\Gamma\right)$ is
called an \emph{empirical model}.
\end{defi}

The non-disturbance set is a polytope, since it is defined by a finite number of linear inequalities and equalities: 
the inequalities
imposed by the fact that its elements represent probabilities and the equalities imposed by definition \ref{definondisturbance}.

After imposing conditions on the restriction of the probability distributions, we
ask now if it is possible to extend the distributions $p_c$ to larger sets containing $C$. The naive ultimate 
goal would be to define a distribution on
the set $O^X$, which specifies assignment of outcome to all measurements, in a way that the restrictions
 yield the probabilities specified by the empirical model on all
 contexts in $\mathcal{C}$. A more subtle and adequate question is to decide when it is possible to 
achieve this goal. This question was first studied by Fine in reference  \cite{Fine82}, for the 
restricted case of Bell scenarios (see appendix \ref{chapternonlocality}) and 
generalized by Brandenburger
 and Abramsky  in reference  \cite{AB11} .

\begin{defi}
A \emph{global section} for $X$ is a probability distribution $p_X: O^X \rightarrow  [0,1]$. A \emph{global section for a distribution} 
$p \in \mathcal{X}\left(\Gamma\right)$
is a global section for $X$
such that the restriction of $p_X$ to each context $C \in \mathcal{C}$ is equal to $p_C$. 
The distributions with global section
are called \emph{noncontextual}.
\end{defi}

A global
section for a distribution $p$ corresponds exactly to the existence of a distribution defined
on all measurements, which marginalizes to yield the  probabilities determined by the empirical model. If a global section for $p$ exists, $p$ is called noncontextual 
because this global section  is
deeply connected to the existence of a noncontextual hidden-variable model reproducing the statistics of $p$.
In fact, if there is a global section for $p$ we can construct the hidden-variable model in the following way:
as hidden variable we use an element of the  classical probability space $O^X$, and the value assigned by $\lambda \in O^X$ to
a measurement $M$ is $\lambda(M)$. Then, the global section $p_X$ for $p$ provides a probability distribution in the set of
hidden variables with the property that if we average over all hidden variables according to this function we recover the 
quantum predictions. A proof of the converse can be found in section 8 of reference  \cite{AB11}, and this  gives:

\begin{teo}[Brandenburger and  Abramsky, 2011]
A probability distribution $p \in \mathcal{X}\left(\Gamma\right)$ has a global section if and only if there
is a  noncontextual hidden-variable model recovering its statistics.
\label{propglobalsection}
\end{teo}

%

Some distributions do not admit global sections. They are called \emph{contextual}.

\begin{ex}[Contextual non-disturbing distribution]
  Consider the scenario $\left(X, \mathcal{C}, O\right)$, where
  $$X=\{M_1,M_2, M_3\}, \ \mathcal{C}=\left\{\{M_1,M_2\}, \{M_2,M_3\}, \{M_1,M_3\}\right\}\ \mbox{and} \ O=\{-1,1\}.$$
 The distribution

\begin{center}
\begin{tabular}{c|c|c|c|c}
   &$(1,1)$&$(1,-1)$&$(-1,1)$&$(-1,-1)$\\ \hline
$M_1M_2$&$\frac{1}{2}$&$0$&$0$&$\frac{1}{2}$\\ \hline
$M_2M_3$&$\frac{1}{2}$&$0$&$0$&$\frac{1}{2}$\\ \hline
$M_1M_3$&$0$&$\frac{1}{2}$&$\frac{1}{2}$&$0$
  \end{tabular}
\end{center}
where entry $ij$ of the table is the probability of obtaining outcome $j$ when measurement $i$ is performed, 
is a non-disturbing distribution, but it does not have a global section. This distribution is the one that
appears in the famous Specker's parable of the Over-protective Seer \cite{LSW11}.

\end{ex}

\section{\textsf{Sheaf-theory and contextuality}}
\label{sectionsheaftheory}

It is possible to provide a more formal mathematical formulation of contextuality using categories and sheaf theory,
as pioneered by Abramsky and co-workers \cite{AD05, AB11}. 
This approach provides a direct and unified characterization of both contextuality and non-locality, along with different
new tools, insights and results. We provide a brief introduction to the sheaf theoretical aspects of contextuality in this
section and we refer to \cite{AB11} for more detailed definitions and discussions. We use some terminology of category
theory, which are  explained in references \cite{MM92, Lane98}.

We start once again with a set $X$ of possible measurements.  The set of possible 
outcomes for each measurement is $O$, and when a set of compatible measurements $U \subset X $ is performed, a set of outcomes
in $O^U$ will be observed. Each individual run of the experiment is what we called an \emph{event}. 

Events in $O^U$ and sections over $U$ are in bijective correspondence. Let $s: U\rightarrow O$ be  a section. The event
associated to $s$ is the event in which the measurements in $U$ were performed and for each $M \in U$ outcome
$s(M)$ was obtained.

Define the function $\varepsilon$ that takes each subset $U \subset X$ to $O^U$, the set  of sections over $U$. We can also
define  a natural action by restriction according to equation \eqref{eqrestriction}: if $U \subset U'$
\begin{eqnarray}
  r^{U'}_U:\varepsilon(U')&\longrightarrow &\varepsilon(U) \nonumber\\
  s&\longmapsto &s|_U.
\end{eqnarray}
This restriction is such that 
\be r^{U}_U=id_U \label{eqrestrictionid}\ee
and if $U\subset U'\subset U''$, 
\be r^{U'}_{U'}\circ r^{U'}_U=r^{U''}_U.\label{eqrestrictioncomp}\ee

Let \textbf{Set} be the category  whose objects are sets and arrows are functions between sets. 
Let $\mathcal{P}(X)$ 
be the category whose objects are the subsets of $X$ and there is a unique arrow from $U$ to $U'$ if and only if
$U \subset U'$.
Let $\mathcal{P}(X)^{OP}$ 
be the category whose objects are the subsets of $X$ and there is a unique arrow from\footnote{The \emph{opposite category} 
or \emph{dual category} $\mathcal{C}^{op}$ of a given category $\mathcal{C}$ is formed by reversing the morphisms, that is,
interchanging the source and target of each morphism \cite{MM92, Lane98}.} $U'$ to $U$ if and only if
$U \subset U'$. Then, we can use the function $\varepsilon$ defined above as a functor 
$$\varepsilon: \mathcal{P}(X)^{OP}\longrightarrow \mbox{\textbf{Set}}$$
that takes each $U \subset X$ to $\varepsilon(U)=O^U$ and the unique arrow $U'\rightarrow U$ to the restriction $r^{U'}_U$,
when $U \subset U'$. Equations \eqref{eqrestrictionid} and \eqref{eqrestrictioncomp} prove that $\varepsilon$ is in fact a 
functor and hence $\varepsilon$ is a \emph{presheaf}.

\begin{defi}
 Given a category $C$, a functor $F:C^{OP}\rightarrow \mbox{\textbf{Set}}$ is called a \emph{presheaf}.
\end{defi}

The functor $\varepsilon$ has another distinguished property. Let $\{U_i\}_{i \in I}$ be a family of subsets of $U$ such that 
$\bigcup_i U_i=U$ and $\{s_i \in \varepsilon(U_i)\}_{i \in I}$ a family of sections that agree in all intersections, that is
$$s_i|_{U_i \bigcap U_j}=s_j|_{U_i \bigcap U_j}$$
for every $i, j \in I$. Then there is a unique section $s \in \varepsilon(U)$ such that $s|_{U_i}=s_i$. In fact, given $M \in U$ 
there is at least one $i \in I$ such that $M \in U_i$. Let $m=s_i(M)$. Since all sections $s_i$ agree on the overlaps, $m$
does not depend on the index $i$ chosen. We define then $s(M)=m$.

This distinguished property is called the \emph{sheaf condition} and $\varepsilon$ is called the \emph{sheaf of events}
\cite{MM92}.

\begin{defi}
Let $F: \mathcal{P}(X)^{OP}\rightarrow \mbox{\textbf{Set}}$ be a presheaf and $f^{U'}_U: F(U') \rightarrow F(U)$ be the
arrow in $\mbox{\textbf{Set}}$ associated to the unique arrow $U'\rightarrow U$ if $U \subset U'$. 
If $s\in F(U')$, let $s|_U = f^{U'}_U(s)$. We say that $F$ is a \emph{sheaf} if it  satisfies the following two conditions:
\begin{enumerate}
 \item Locality: If $(U_i \subset X)$ is a covering of  $U \in X$, and if $s,t \in F(U)$ are such that $s|U_i = t|U_i$ 
 for each set $U_i$, then $s = t$; 
\item Gluing: If $(U_i)$ is a covering of  $U$, and if for each $i$ there is a section $s_i$
 over $U_i$ such that for each pair $U_i,U_j$, the restrictions of $s_i$ and $s_j$
agree on the overlaps, that is 
$$s_i|_{U_i\cap U_j} = s_j|_{U_i \cap U_j},$$ then there is a section 
$s \in F(U)$ such that $s|_{U_i} = s_i$ for each $i$. 
 \end{enumerate}
\end{defi}

Sections correspond to definite outcomes, but most of the times it is not  possible to predict
with certainty the outcome of every measurement. When probabilistic theories enter the game we must use 
\emph{probability distributions} over
the set of sections $O^U$. To make definitions more general, we will consider distributions taking values over a 
 commutative semiring $R$ \cite{AB11}. 

\begin{defi}
 An $R$-\emph{distribution} on $U$ is a function $d:U\rightarrow R$ such that $\sum_{M \in U} d(U)=1$.
\end{defi}

When we are interested in probability distributions, $R$ is the semiring of positive real numbers.
Nonetheless, it is quite instructive to keep $R$ general, even when we are working with probabilities in  a
compatibility scenario.

We write $\mathcal{D}_R(U)$ for the set of $R$-distributions on $U$.

Let $f: U'\rightarrow U$ be a function among two sets $U'$ and $U$. We define
$$ \mathcal{D}_R(f): \mathcal{D}_R(U')\longrightarrow \mathcal{D}_R(U)$$
that takes each distribution $d$ to the distribution $\mathcal{D}_R(f)(d)=d': Y\longrightarrow R$ defined by
$$d'(y)=\sum_{x; f(x)=y} d(x).$$
This definition is functorial since $\mathcal{D}_R(id)=id$ and
$\mathcal{D}_R(g\circ f)=\mathcal{D}_R(g)\circ \mathcal{D}_R(f)$.

With the definitions above we can construct the functor
$$\mathcal{D}_R: \mbox{\textbf{Set}}\rightarrow \mbox{\textbf{Set}}$$
that takes each set $U$ to
the set of $R$-distributions on $U$ and each function $f: U'\rightarrow U$ to the function 
$\mathcal{D}_R(f): \mathcal{D}_R(U')\longrightarrow \mathcal{D}_R(U)$.

We can compose this functor with the sheaf $\varepsilon$ to define the presheaf
$$\mathcal{D}_R \circ \varepsilon :\mathcal{P}(X)^{OP}\longrightarrow \mbox{\textbf{Set}}$$
which assigns to each subset $U \subset X$ the set of $R$-distributions on the sections over $U$.
If $U \subset U'$, the unique arrow $U'\rightarrow U$ is taken by this presheaf to the  map
$\mathcal{D}_R\left(r_U^{U'}\right)$
acting on the set of
$R$-distribution on $O^{U'}$: if $d \in \mathcal{D}_R(\varepsilon(U'))$, then
$$\mathcal{D}_R\left(r_U^{U'}\right)(d)= d|_U$$
where $d|_U(s)=\sum_{s'; s'|_U=s}d(s)$.

The restriction $d|_U$ is the \emph{marginal} distribution of $d$, which assigns to each section $s$ in the smaller set $U$
the sum of the weights of all sections $s'$ in the larger set that restrict to $s$.

We now take into count the fact that not all measurements can be performed together, what can be done
by considering a compatibility cover   $\mathcal{C}$ of $X$ (see definition \ref{defimeasurementcover}). 
Each subset of $X$ that belongs to $\mathcal{C}$ is a maximal set of compatible measurements.

With the language of categories introduced above, an empirical model for the scenario $(X, \mathcal{C}, O)$ is
a family of $R$-distributions $e_C \in \mathcal{D}_R(\varepsilon(C))$, $C \in \mathcal{C}$.  Once more, we 
consider only non-disturbing models, that is, we demand that
$$e_C|_{C\cap C'}=e_{C'}|_{C\cap C'}$$
whenever $C\cap C' \neq \emptyset$.

We have already observed that the presheaf  $\varepsilon$ is indeed a sheaf. It is natural to ask if the same holds
for the presheaf $\mathcal{D}_R \circ \varepsilon$. The no-disturbance condition corresponds precisely to the first
condition required for a presheaf to be a sheaf, and hence
the sheaf condition
for $\mathcal{D}_R \circ \varepsilon$ is equivalent to the existence of a global distribution $d \in 
\mathcal{D}_R \circ \varepsilon(X)$ such that $d|_C=e_C$ to each context $C$.

Theorem  \ref{propglobalsection} implies that such a distribution $d$ exists if and only if there is a hidden
variable model reproducing the statistics of the empirical model. Hence, we have:

\begin{teo}
 The empirical model $(e_C)$ satisfies the sheaf condition if and only if there is a hidden-variable model
 reproducing its statistics.
\end{teo}

A proof of this result can be found in reference \cite{AB11}.

Thus, we have a characterization of the phenomena of 
contextuality in terms of \emph{obstructions to the existence of global sections in a presheaf}, which opens the door
to the use of the methods of sheaf theory to the study of contextuality.

\section{\textsf{Probability Distributions and Physical Theories}}
\label{sectionphysical}

\subsection{\textsf{Classical Non-Contextual Realizations}}

Given a hypergraph $\Gamma$, a classical realization for $\Gamma$
is a  probability space $(\Omega,\Sigma, \mu )$, where $\Omega$ is a sample space, 
$\Sigma$ a $\sigma -$algebra and $\mu$ a probability measure in $\Sigma$, and for each $i \in V$ a partition
of $\Omega$ into $|O|$ disjoint subsets $A_j^i \in \Sigma, \ \ j \in O$, where $V$ is the set of vertices of\footnote{Equivalently 
we can say that a distribution is non-contextual if for each $i\in V$ there is a random variable $R_i : \Omega \rightarrow O$
and $p(a_1, \ldots, a_n|M_1, \ldots, M_n)= \mu\left(R_i=a_i\right).$} $\Gamma$. 
For each context $C= \{M_1, \ldots, M_n\}$, the probability of the outcome $a_1, \ldots, a_n$  is 
$$p(a_1, \ldots, a_n|M_1, \ldots, M_n)=\mu\left(\bigcap_kA_{a_k}^{k}\right).$$
 The probability distributions that can be written in this form are called \emph{classical distributions}.
The set of  classical distributions\footnote{This set depends also on the set of possible outcomes $O$, but
we will not write this explicit to simplify the notation.}
$\mathcal{NC}\left(\Gamma\right)$ is a polytope with $\left|O^X\right|$ vertices, all of them  noncontextual.

As an immediate consequence of theorem \ref{propglobalsection}, we have the following result:

\begin{cor}
\label{propclassicalnoncontextual}
 A distribution has a global section if and only if it is classical.\footnote{This result shows that is possible to use the 
 notion of global section to define non-contextual distributions: we say that a distribution is non-contextual if
 it has a global section.}
\end{cor}

In fact, once a classical realization is given, the construction of the global section is guaranteed by 
the fact that the intersection of a finite number of sets in a $\sigma$-algebra also belongs to the $\sigma$-algebra. Conversely, given the 
global section, we can construct the classical realization using the same argument present in the paragraph preceding 
theorem \ref{propglobalsection}. 

\subsection{\textsf{Quantum Realizations}}

A quantum realization is given by a Hilbert space $\mathcal{H}$, for each $i \in V$ a
Hermitian matrix $O_i$ in this Hilbert 
space,
  and a density matrix $\rho$ acting on $\mathcal{H}$. For a given context $C \in \mathcal{C}$,
the compatibility condition demands the existence of a basis for $\mathcal{H}$ in
which all  $O_i$ belonging to $C$ are diagonal.
For each context $C= \{M_1, \ldots, M_n\}$, the probability of the outcome $a_1, \ldots, a_n$  is 
$$p(a_1, \ldots, a_n|M_1, \ldots, M_n)=\tr\left(\prod_kP_{a_k}\rho\right)$$
where $P_{a_k}$ is the projector over the eigenspace corresponding to outcome $a_k$ of observable $O_k$.
The probability distributions that can be written in this form  are called \emph{quantum  distributions}. Notice
that the Hilbert space is not fixed and the set of quantum distributions contains realizations in all dimensions.
This set, which we denote by $\mathcal{Q}(\Gamma)$, is a convex set but is not a polytope in general.

\begin{teo}
 The set of quantum distributions $\mathcal{Q}(\Gamma)$ is a convex set.
\end{teo}

\begin{dem}
 Let $p^1$ and $p^2$ be two quantum distributions. We want to prove that any convex combination
 $$\alpha p^1 + \beta p^2, \ \ 0\leq \alpha,  \beta\leq 1,\  \  \alpha +
\beta =1$$
is a quantum distribution.
 
 Let $\rho^1$ and observables $\left\{O^1_{i}\right\}$ be a quantum 
 realization for $p^1$ and $\rho^2$ and observables $\left\{O^2_{j}\right\}$ be a quantum 
 realization for $p^2$, that is 
$$p^1\left(a_1, \ldots, a_n|M_1, \ldots, M_n\right)=\tr\left(\prod_kP^1_{a_k}\rho^1\right)$$
and similar for $p_2$
$$p^2\left(a_1, \ldots, a_n|M_1, \ldots, M_n\right)=\tr\left(\prod_kP^2_{a_k}\rho^2\right)$$
where $P^1_{a_k}$ is the projector over the eigenspace corresponding to outcome $a_k$ of observable $O^1_k$ and
analogously for $P^2_{a_k}$. 

It is important to notice here that the density matrices and projectors in the quantum realizations for $p_1$ and $p_2$ 
given above do not have necessarily  the same dimension. Nonetheless, it is always possible to extend one of them to
a Hilbert
space of higher dimension, so without loss of generality we will consider that all density matrices and projectors
act in the same Hilbert space $\mathcal{H}$.

Let $\left\{\ket{1}, \ket{2}\right\}$ be an orthonormal basis for $\mathbb{C}^2$ and define the 
density matrix 
$$\rho=\alpha \rho^1\otimes \ket{1}\bra{1} + \beta \rho^2 \otimes \ket{2}\bra{2}$$ and the projectors
$$P_{a_k}=  P^1_{a_k} \otimes \ket{1}\bra{1} + P^2_{a_k}\ket{2}\bra{2},$$ acting on $\mathcal{H} \otimes \mathbb{C}^2$.
Then we have that 
$$p\left(a_1, \ldots, a_n|M_1, \ldots, M_n\right):= \tr \left(\prod_kP_{a_k}\rho\right)=
\alpha\left(\prod_kP^1_{a_k}\rho^1\right) + \beta\left(\prod_kP^2_{a_k}\rho^2\right)$$
which implies that 
$$p=\alpha p^1 + \beta p^2.$$
Hence, any convex combination of quantum distributions is also a quantum distribution.
\end{dem}

It is important to mention that the use of a Hilbert space of higher dimension than $\mathcal{H}$ can not be avoided. 
In fact, if we bound the dimension
of the quantum realizations, we get a set that is not convex, as shown by  P\'al and  V\'ertesi in reference 
\cite{PV09}.

The set of classical distributions is contained in the set of quantum distributions. To prove that, we just have to notice that 
 the set of distributions obtained from a  probability space with $n$ elements is equivalent to the set of distributions
 obtained with diagonal projectors and density matrices in a Hilbert  space of dimension $n$ with a fixed basis.
 The set of elements in the sample space $\Omega$ is the set of unidimensional projectors and the measure is given by
 $\mu(P)=Tr\left(\rho P\right).$

\section{\textsf{Non-Contextuality Inequalities}}
\label{sectionncinequalities}

We would like to find 
simple criteria to decide whether a probability distribution $p$ is noncontextual or not.
According to theorem \ref{propclassicalnoncontextual}, this is equivalent to test
if $p \in \mathcal{NC}\left(\Gamma\right)$. We will use the fact  that $\mathcal{NC}\left(\Gamma\right)$
is a polytope to derive a finite number of inequalities 
that provide
 necessary and sufficient conditions for membership in this set. 

A convex polytope may be defined as an intersection of a finite number of half-spaces. 
Such definition is called a \emph{half-space representa\-tion} 
(H-representation or H-description). There exist infinitely many H-descriptions of a convex polytope. However, 
for a full-dimensional 
convex polytope, the minimal H-description is in fact unique and is given by the set of facet-defining halfspaces.

Since $\mathcal{NC}(\Gamma)$ is a polytope, there is a minimal set  of inequalities giving a H-representation. Some of this
inequalities  are the trivial inequalities related to the definition of  probability distributions
(positivity and normalization),
 but others are not and in general are not satisfied by all
quantum distributions. These inequalities are called \emph{noncontextuality inequalities}.

\begin{defi}
 A \emph{noncontextuality inequality} is a linear inequality 
 \be S:=\sum \gamma_{a_1, \ldots, a_n| M_1, \ldots, M_n} p(a_1, \ldots, a_n| M_1, \ldots, M_n) \leq b,
 \label{eqncinequalities}\ee
 where all $\gamma_{a_1, \ldots, a_n| M_1, \ldots, M_n}$ and $b$ are real numbers, which is 
 satisfied by all elements of  the classical polytope 
 $\mathcal{NC}\left(\Gamma\right)$ and violated by some contextual distribution.
 A \emph{tight noncontextuality inequality} is a linear inequality defining a non-trivial facet of the classical polytope 
 $\mathcal{NC}\left(\Gamma\right)$.
\end{defi}

Any H-description provides a necessary and sufficient condition for membership in $\mathcal{NC}\left(\Gamma\right)$: 
a distribution $p$ is classical if and only if
it satisfies all noncontextuality inequalities for this scenario. Although verifying if a distributions satisfies or not the inequalities
is very simple, finding  all inequalities that provide an H-description for a general scenario
is a very difficult computational task, related to the max-cut problem, which belongs to the  
NP-hard class of computational complexity \cite{BM86, DL97, AIT06}.

\section{\textsf{The KCBS inequality}}
\label{sectionKCBS}

The KCBS scenario was introduced by Klyachko, Can, Binicio\u{g}lu, and Shumovsky in reference \cite{KCBS08}.
It consists of  five measurements $X=\{M_0,M_1, M_2, M_3, M_4\}$, 
with compatibility structure given by
$$\mathcal{C}=\{\{M_0,M_1\}, \{M_1,M_2\}, \{M_2,M_3\}, \{M_3,M_4\}, \{M_0,M_4\}\}.$$
The set of possible outcomes is $O=\{\pm 1\}$. The hypergraph $\Gamma$ in this case is a familiar simple graph: the pentagon.

\begin{figure}[h]
\centering
 \includegraphics[scale=1.5]{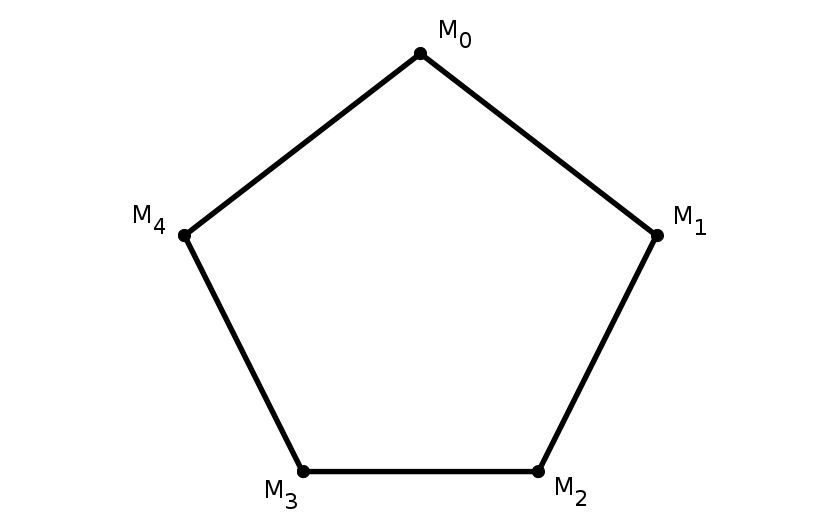}
 \caption{The compatibility hypergraph of the KCBS scenario. \label{figpentagon}}
\end{figure}

This scenario was completely characterized in references \cite{Araujo12, AQBTC13}. There are $2^4$ tight
 noncontextuality inequality
and all of them can be written in the form
\begin{equation}
 \sum_{i=0}^{4}\gamma_i\left\langle M_i M_{i+1}\right\rangle \leq 3,
\label{KCBS1}
\end{equation}
where 
$\left\langle M_i M_i\right\rangle=p\left(M_i=M_{i+1}\right) -  p\left(M_i\neq M_{i+1}\right)$,  $\gamma_i \in \{\pm 1\}$ 
and the number of $\gamma_i = - 1$ is odd.

The inequality obtained when all $\gamma_i = - 1$ is the famous KCBS inequality, presented in the seminal paper \cite{KCBS08}. 
It is equivalent to the inequality
\begin{equation}
\sum_{i=0}^{4}\left\langle P_i\right\rangle  \leq 2,
\label{KCBS2}
\end{equation}
where $P_i = 1-2M_i$ and $\left\langle P_i \right\rangle=p\left(P_i=1\right) -  p\left(P_i= - 1\right)$.

These inequalities are violated by some quantum distributions in dimension three or higher. 
The maximal violation for inequality \eqref{KCBS1} for
quantum distributions is $5-4\sqrt{5}$, which corresponds 
to a maximal quantum value of $\sqrt{5}$ for inequality \eqref{KCBS2}.
These violations can be   obtained with the state $\ket{\psi}=(1,0,0)$ and with
projectors
$$P_i= \left(\cos(\theta), \sin(\theta)\cos\left( \frac{ 4i\pi}{5}\right), \sin(\theta)\sin\left( \frac{ 4i\pi}{5}\right)\right),$$
where $\cos^2(\theta) = \frac{\cos\left(\frac{\pi}{5}\right)}{\left(1 + \cos\left(\frac{\pi}{5}\right)\right)}.$

An interesting property of these projectors is that they are orthogonal if $\left(i, j\right) \in E(\Gamma)$. 
This implies that the outcome
$11$ can never occur in a measurement of $M_i$ and $M_j$.

Some non-disturbing distributions can achieve the algebraic maximum violation of $5$ for inequality \eqref{KCBS1}.

\begin{ex}
The no-disturbing distribution 
\begin{center}
\begin{tabular}{c|c|c|c|c}
   &$(1,1)$&$(1,-1)$&$(-1,1)$&$(-1,-1)$\\ \hline
$M_0M_1$&$\frac{1}{2}$&$0$&$0$&$\frac{1}{2}$\\ \hline  
$M_1M_2$&$\frac{1}{2}$&$0$&$0$&$\frac{1}{2}$\\ \hline
$M_2M_3$&$\frac{1}{2}$&$0$&$0$&$\frac{1}{2}$\\ \hline
$M_3M_4$&$\frac{1}{2}$&$0$&$0$&$\frac{1}{2}$\\ \hline
$M_4M_0$&$0$&$\frac{1}{2}$&$\frac{1}{2}$&$0$
  \end{tabular}
\end{center}
gives
$$\sum_{i=0}^{4}\gamma_i\left\langle M_i M_{i+1}\right\rangle =5,$$
where $\gamma_i =1$ for $i=0,1,2,3$ and $\gamma_4=-1$,
reaching the algebraic maximum for  the KCBS inequality \eqref{KCBS1}.
 
\end{ex}

This shows that, in general, the violation obtained with no-disturbing distributions is higher
than the quantum maximum, and hence, that the non-disturbance
polytope contains properly the quantum set.

\section{\textsf{The $n$-cycle inequalities}}
\label{sectionncycle}

A simple generalization of the KCBS inequality is obtained when we use as the compatibility hypergraph an $n$-cycle: a graph with
$n$ vertices $0, 1, \ldots, n-1$ and such that two vertices $i,j$ are connected iff $\left|i-j\right|=1 \ \mbox{mod} \ n$. 
The corresponding
scenario has $n$ measurements $X=\{M_0,M_1, \ldots, M_{n-1}\}$, 
with compatibility structure given by
$$\mathcal{C}=\{\{M_0,M_1\}, \{M_1,M_2\}, \ldots, \{M_{n-2},M_{n-1}\}, \{M_{n-1},M_{0}\}\}.$$
The set of possible outcomes is also $O=\{\pm 1\}$. The complete set of noncontextuality inequalities for this scenario
was found in reference \cite{AQBTC13}.

\begin{teo}
 There are $2^{n-1}$ tight noncontextuality inequalities for the $n$-cycle scenario, and they are of the form
 \be \sum_{i=0}^{n-1}\gamma_i\left\langle X_i X_{i+1}\right\rangle \leq n-2,\ee
 where the sum is taken modulo $n$, $\gamma_i= \pm 1$, and the number of indices $i$ such that $\gamma_i=-1$ is odd.
\end{teo}

Some quantum distributions violate this bound if $n \geq 4$. The maximum quantum violation is given by
\be  \left\{\begin{array}{cc}
                \frac{3n\cos\left(\frac{\pi}{n}\right)-n}{1 + \cos\left(\frac{\pi}{n}\right)} & \ \mbox{if} \ n \ \mbox{is odd,}\\
               n\cos\left(\frac{\pi}{n}\right) & \ \mbox{if} \ n \ \mbox{is even}.
               \end{array}\right.
\ee

For $n$ odd, the quantum bound can be achieved already in a three-dimensional system,  with the state
$\left(\begin{array}{ccc}
 1&0&0
\end{array}\right)$ and measurements $M_i= 2\ket{v_i}\bra{v_i}-I$, where
 $$\ket{v_i}=\left(\begin{array}{ccc}
 \cos(\theta)& \sin(\theta)\cos\left(\frac{i\pi(n-1)}{n}\right)& \sin(\theta)\sin\left(\frac{i\pi(n-1)}{n}\right)
\end{array}\right)$$ and $\cos^2(\theta)=\frac{\cos\left(\frac{\pi}{n}\right)}{\left(1+ \cos\left(\frac{\pi}{n}\right)\right)}.$ 

For $n$ even, the quantum bound can be achieved in a four-dimensional system, with the state
$\left(\begin{array}{cccc}
 0&\frac{1}{\sqrt{2}}&-\frac{1}{\sqrt{2}}&0
\end{array}\right)$ and measurements $M_i= X_i \otimes I$ for odd $i$ and $M_i= I \otimes X_i$ for even $i$, where
 $X_i= \cos\left(\frac{i\pi}{n}\right)\sigma_x + \sin\left(\frac{i\pi}{n}\right)\sigma_z$.

These bounds were calculated with the help of the tools we will introduce in the next section.

The interest in this scenario comes from the fact that all distributions in scenarios where the compatibility 
graph has no closed loop are noncontextual. 

\begin{teo}
There  is a quantum noncontextual distribution  if and only if $\Gamma$ has an $n$-cycle as induced subgraph with 
$n > 3$.
\end{teo}

  Equivalently, we may say that there is quantum violation of some noncontextuality inequality for the scenario if, and 
  only if $\Gamma$ has an $n$-cycle as induced subgraph with 
$n > 3$. In this sense, the $n$-cycle scenarios are the simplest ones where it is possible to find quantum violations of 
noncontextuality
inequalities. For a proof of this result, see reference \cite{BM10}.

\section{\textsf{The Exclusivity Graph}}
\label{sectionexclusivitygraph}

Given a scenario it is possible to define another graph related to it that allows the calculation of several bounds for the
associated inequalities. We introduce some definitions first. In what follows
 $$a_1, \ldots, a_n| M_1, \ldots, M_n$$
will denote the event where compatible measurements $M_1, \ldots, M_n$ were performed and 
outcomes $a_1, \ldots, a_n$ were obtained.

Since each outcome $a_i$ in  measurement $M_i$ is associated to an element of $\mathcal{T}$, the event 
$a_1, \ldots, a_n| M_1, \ldots, M_n$ is associated 
to a composition of transformations, which is also a transformation according  to corollary \ref{corcomposition}. 

\begin{defi}
 We say that two events are \emph{exclusive} if
 the corresponding transformations represent different outcomes of the same measurement.
\end{defi}

\begin{defi}
 Given a scenario, the \emph{exclusivity graph} $\mathcal{G}$ of this scenario is the simple graph whose vertices are
 labeled by 
 all possible events
 $$a_1, \ldots, a_n| M_1, \ldots, M_n$$
 in this scenario. Two vertices  are connected by an edge if and only if the corresponding events are exclusive.
\end{defi}

Generally,  not all possible  events are involved in a given inequality. The ones involved define an induced subgraph of 
$\mathcal{G}$
from which we can get a lot of information about the inequality.

\begin{defi}
The \emph{exclusivity graph} $G$ of  a noncontextuality inequality is the induced subgraph of $\mathcal{G}$ defined by
the vertices that correspond to events appearing in the inequality.
\end{defi}

\begin{ex}[The exclusivity graphs of the $n$-cycle inequalities]
\label{exncycle}
 Since $$\left\langle M_iM_j \right\rangle = 2\left(p(11| M_iM_j)+p(-1-1| M_iM_j)\right)-1$$ and 
 $$-\left\langle M_iM_j \right\rangle = 2\left(p(1-1| M_iM_j)+p(-11| M_iM_j)\right)-1,$$ there are
 $2n$ events in each noncontextuality inequality for the $n$-cycle scenario. If $n$ is odd, the corresponding 
 exclusivity graph is 
 the \emph{prism graph} of order $n$, $Y_n$,  and if $n$ is even, the exclusivity graph is 
 the \emph{M\"obius ladder} of order $2n$, $M_{2n}$. The first four
 of these graphs are depicted in figure \ref{figncycle}. 
 
 \begin{figure}
 \centering
  \includegraphics[scale=0.4]{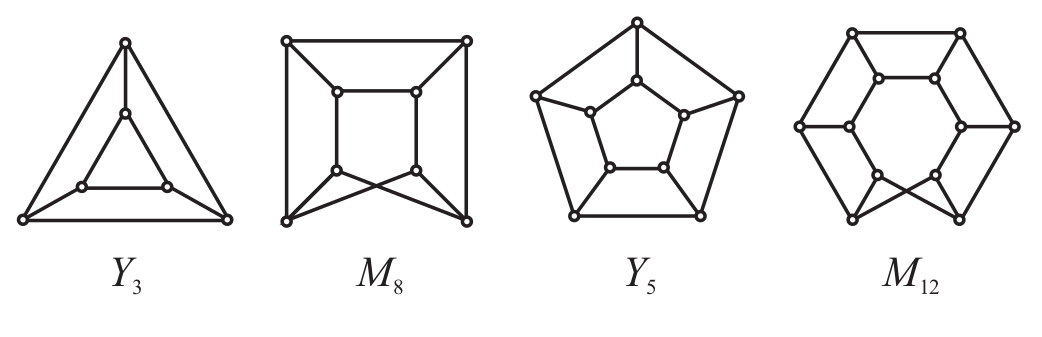}
\caption{Exclusivity graphs for the $n$-cycle inequalities for $n=3,4,5,6$.}
\label{figncycle}
 \end{figure} 
\end{ex}

We restrict ourselves now to the case where all coefficients $ \gamma_{a_0, \ldots, a_n| M_1, \ldots, M_n}$ in 
equation \eqref{eqncinequalities} are equal to one. Many important inequalities can be written in this
form, including the $n$-cycle inequalities. In this case, we can use the exclusivity graph $G$ and 
some graph functions to get information about
the maximal bounds for the quantity $S$ in different probabilistic theories. First, a few definitions from graph theory.

\begin{defi}
An \emph{independent set} or \emph{stable set} in a graph $G$ 
is a set of vertices of $G$, no two of which are adjacent. 
A \emph{maximum independent set} is an independent set of largest possible size for  G.
\end{defi}

\begin{defi}
 The \emph{independence number} $\alpha(G)$ of a graph $G$ is the cardinality of a maximum  independent set of $G$.
\end{defi}

\begin{defi}

Let $\{1, \ldots , n\}$ be the set of vertices of a graph $G$.
An \emph{orthonormal representation} for $G$ in a finite-dimensional vector space with inner product 
$V$ is a set of unit vectors $\{\ket{u_1}, \ldots , \ket{u_n}\}$ such that 
$\ket{u_i}$ and $\ket{u_j}$ are orthogonal whenever $i$ and $j$ are not connected in $G$. 
\end{defi}

\begin{defi}
 The \emph{Lov\'asz number} of a graph $G$ is 
 $$\vartheta(G)=\max \sum_i \braket{u_i}{\psi}$$
 where the maximum is taken over all $V$ and over all orthogonal representations 
 $\{\ket{u_1}, \ldots , \ket{u_n}\}$
 for $\overline{G}$ and all unit vectors $\ket{\psi}$ in  $V$.
 An orthonormal representation achieving the maximum, called
 an \emph{optimal orthonormal representation}, always exists. 
\end{defi}

 Both $\alpha(G)$ and
$\vartheta(G)$ are extremely important for the study of classical and quantum bounds of 
noncontextuality inequalities. For a  more detailed discussion about these graph functions,
see \cite{Lovasz79, Lovasz95, Knuth93, Rosenfeld67, Bollobas98}.

\begin{teo}[Cabello, Severini and Winter, 2010]
 The classical bound of a noncontextuality inequality is the independence number $\alpha(G)$ of the exclusivity
 graph $G$ of the inequality.
\end{teo}

\begin{dem}
 Since the noncontextuality inequalities are linear, the maximum classical bound is achieved in a vertex of the 
 noncontextual polytope. For such a vertex, the probability of each event is either zero or one and the value 
 of the sum $S$ for this distribution is equal to the number of events with probability one.
 Since the sum of the 
 probabilities of two exclusive events can not be higher than one, two connected vertices can not have probability equal to
 one at the same time. Hence, the set of vertices whose probabilities are one is an independent set, and hence can not
 have more than $\alpha(G)$ elements. 

To prove that equality holds, it suffices to take any maximum independent set and 
use the classical distribution that assigns probability one to each vertex in this set.

\end{dem}

\begin{teo}[Cabello, Severini and Winter, 2010]
\label{teoqbound}
 The quantum bound of a noncontextuality inequality is upper bounded by the Lov\'asz  number $\vartheta(G)$ of the exclusivity
 graph $G$ of the inequality.\footnote{If the coefficients of the inequality are not all equal to one, we use the wighted 
 versions of $\alpha$ and $\vartheta$.}
\end{teo}

\begin{dem}
The maximal quantum value for $S$ is obtained for a pure state $\rho=\ket{\psi}\bra{\psi}$. 
Let $\left\{e_i\right\}$, $e_i=a^i_0, \ldots, a^i_{n_i}\left| M^i_1\right.\, \ldots, M^i_{n_i}$, be 
the set of events present in the inequality and 
 $P_i=\prod_k P_{a^i_k}^{M_k^i}$ be the projector corresponding to $e_i$, where $P_{a^i_k}^{M_k^i}$ is 
 the projector associated to outcome $a_k^i$ for measurement $M^i_k$ . Define
 $$\ket{v_i}=\frac{P_i\ket{\psi}}{\left|P_i\ket{\psi}\right|}.$$
 Then we have
 $$S=\sum_i p(a^i_0, \ldots, a^i_{n_i}| M^i_1, \ldots, M^i_{n_i})=\sum_i|\braket{\psi}{v_i}|^2.$$
 If $e_i$ and $e_j$ are exclusive events, the corresponding projectors $P_i$ and $P_j$ are orthogonal, and 
 hence $\ket{v_i}$ and $\ket{v_j}$ are also orthogonal. The set of vectors $\ket{v_i}$ and the state $\ket{\psi}$
 provide an orthogonal representation for $\overline{G}$ and 
 $$\sum_i|\braket{\psi}{v_i}|^2 \leq \vartheta(G).$$
\end{dem}

\begin{ex}[Quantum bound for the $n$-cycle inequalities]
 The observation that $\vartheta(Y_n)=\frac{3n\cos\left(\frac{\pi}{n}\right)-n}{1 + \cos\left(\frac{\pi}{n}\right)}$,
 $\vartheta(M_{2n})= n\cos\left(\frac{\pi}{n}\right)$, and theorem \ref{teoqbound} were used by the authors in 
 reference \cite{AQBTC13} to 
 find the 
 quantum maximum violation of the $n$-cycle inequalities.
\end{ex}

Although in the previous example the bound was tight, this is not true in general. This can happen when the scenario imposes
extra constraints that make the Lov\'asz optimal representations for the graph 
unattainable for quantum systems. 

\begin{ex}
 In reference  \cite{SBBC13} we find three inequalities for which $\vartheta(G)$ is larger then the 
 quantum maximum. Consider the scenario where the system is composed by two spatially separated parties. In the first subsystem
 there are two measurements available, denoted by $A_0$ and $ A_1$, and in the second subsystem we also have
 two measurements available, denoted by $B_0$ and $B_1$. All measurements have two possible outputs, $0$ and $1$.
 In this case, the compatibility of the 
 measurements in different systems is guaranteed by spatial separation
 (for more details, see appendix \ref{chapternonlocality}). The compatibility hypergraph is a square,
 with edges linking measurements in different parties, as shown in figure \ref{figsquare}.

  \begin{figure}[h]
  \centering
  \includegraphics[scale=1.5]{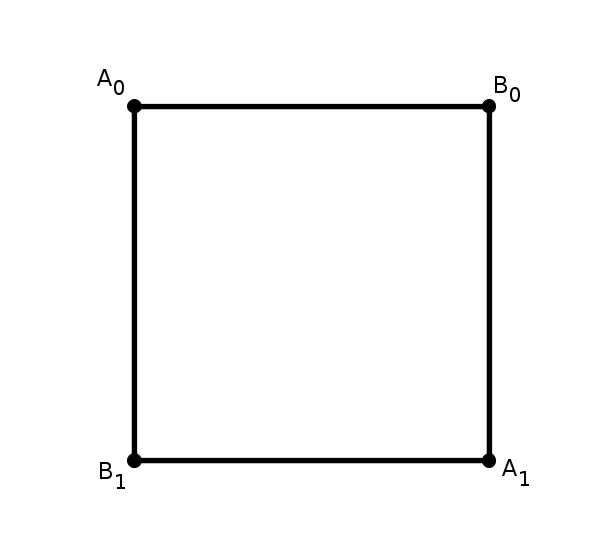}
  \caption{The compatibility hypergraph for the bipartite scenario with measurements $\{A_0, A_1\}$ for the first party
  and measurements $\{B_0, B_1\}$ for the second party. \label{figsquare}}
  \end{figure}

 This scenario admits two noncontextuality inequalities with quantum bound larger than the classical bound for which
 the exclusivity graph is a pentagon:
 \begin{eqnarray*}
  p(00|00) + p(11|01) + p (10|11) + p(00|10)+ p (11|00) & \leq & 2, \\
  p(00|00) + p(11|01) + p(10|11) + p(00|10) +p(\underline{\hspace{0.6em} }1|\underline{\hspace{0.6em} }0)
  &\leq  &2.
 \end{eqnarray*}
In the inequalities above, $ ab| xy$ denotes the event where the first party applies measurement $A_x$ and gets outcome
  $a$ and the second party applies measurement $B_y$ and gets outcome $b$;
  $\underline{\hspace{0.6em} }1|\underline{\hspace{0.6em} }0$ corresponds to
  the event where the second party applies measurement $B_0$ and gets outcome $1$, irrespectively of the first party's action.

  The quantum bound for the first inequality  is approximately $2.178$, while for the second it is approximately $2.207$.
The events appearing in these inequalities and their exclusivity structures
 are shown figure \ref{figbellpentagon}  (a) and (b).
 
 Consider also the scenario where the first party has three measurements, instead of two. 
 The compatibility hypergraph of this scenario is shown in figure \ref{figbell32}.
 
   \begin{figure}[h]
  \centering
  \includegraphics[scale=1.5]{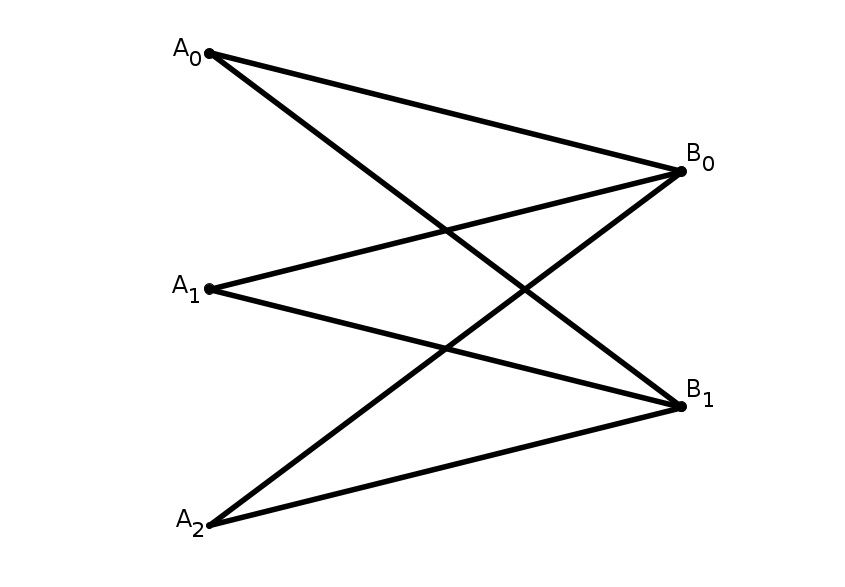}
  \caption{The compatibility hypergraph for the bipartite scenario with measurements $\{A_1, A_2, A_3\}$ for the first party
  and measurements $\{B_1, B_2\}$ for the second party. \label{figbell32}}
  \end{figure}

 This scenario admits one noncontextuality inequality with quantum bound larger than the classical bound for which
 the exclusivity graph is a pentagon:
 $$p(00|00) + p(11|01) + p(10|11) + p(00|10) +p(11|20) \leq 2.$$ 
 The quantum bound for this inequality is approximately $2.207$.
 The events appearing in these inequalities and their exclusivity structure
 are shown in figure \ref{figbellpentagon} (c).

 \begin{figure}[h]
  \centering
  \includegraphics[scale=0.4]{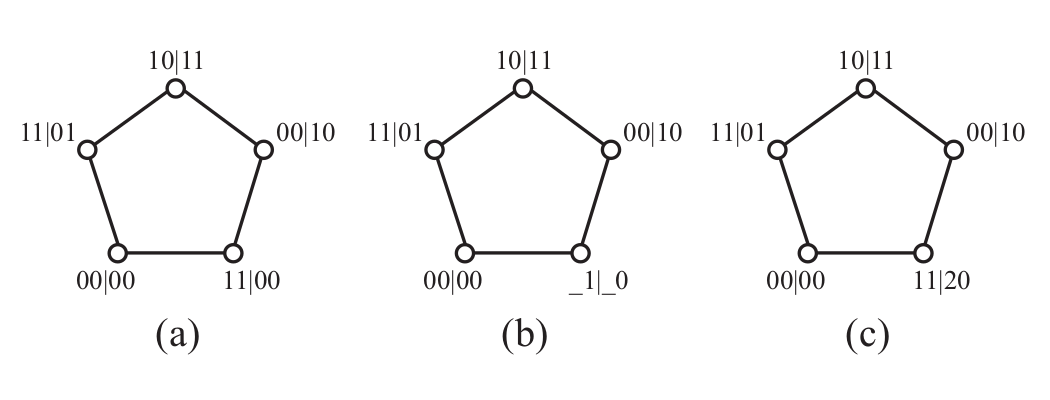}
  \caption{The labeling of the exclusivity graph for the three noncontextuality inequalities with 
  pentagonal exclusivity structure.\label{figbellpentagon}}
 \end{figure}

 For each of these inequalities, the quantum bound is strictly smaller then the Lov\'asz number of the pentagon
 $\vartheta(C_5)= \sqrt{5}\approx 2.236$. This 
 proves that, in general, $\vartheta(G)$ gives only a loose upper bound for the maximum quantum  value of the
 inequality.
 
 \end{ex}

\section{\textsf{Contextuality: the Exclusivity-Graph Approach}}
\label{sectionexclusivity}

\subsection{\textsf{A graph approach to the Bell-Kochen-Specker Theorem}}

The mathematical  content of the original proof of the  Bell-Kochen-Specker theorem  is that there are sets of one dimensional 
projectors for which it is not possible to 
assign definite values $0$ or $1$ noncontextually in such  a way that, if a set of mutually orthogonal
projectors add to identity, then the value $1$ must be assigned to one, and only one, of them 
(for more details, see section \ref{sectionks} of appendix \ref{chaptercontextuality}). 

The usual physical interpretation of this result connects each projector to a measurement in a quantum system
with possible outcomes $0$ and $1$. The noncontextuality assumption translates into the observation that the value
assigned to each measurement is independent of other compatible measurements performed simultaneously.
With this association, the theorem implies the impossibility of noncontextual assignment of definite values
to all measurements in a quantum system consistently with the quantum statistics, proving the impossibility of 
noncontextual hidden-variable models.

The set of one-dimensional projectors in a proof of the Bell-Kochen-Specker theorem can be represented using a graph, known as the \emph{Kochen-Specker diagram}. The vertices of the graph are the projectors in the set and two of 
them are joined by an edge whenever they are compatible.

We  can look at this result from a different perspective. Instead of associating each
projector with a measurement, we will use the fact that any projector $P$ belongs to the set of allowed transformations 
$\mathcal{T}$
of the system  and associate it with a possible outcome of a measurement.
With this interpretation,  each  vertex in a Kochen-Specker diagram corresponds to an element of $\mathcal{T}$
and two vertices are connected by an edge if the corresponding transformations can be associated 
to two different outcomes of one and the same measurement. 

Suppose now that a hidden-variable model is given. This model provides definite values to all measurements, and hence,
given a transformation $P$, we know if the outcome it corresponds to occurs or not. If  the outcome associated to the measurement
by the hidden-variable model is the one that corresponds to $P$, we associate the value $1$ to $P$. Otherwise, we associate
the value $0$ to $P$.

If we have a set of projectors $\left\{P_1, \ldots, P_n\right\}$ summing up to identity, we know that there is a measurement for which the outcomes are associated
to these projectors. Hence, since one, and only one, outcome must occur, one, and only one, of these projectors
is associated to the value $1$. Hence,  we have
\be \sum_iv\left(P_i\right)=1\ee 
where $v(P_i)$ is the value assigned by the model to projector $P_i$.

In this new perspective, the noncontextuality assumption means that the value associated to a projector $P$
by the hidden-variable model is independent of the other projectors used to define the measurement. As we have seen, 
the same transformation corresponds to an outcome of several different measurements. 
Then, whenever $P$ corresponds to an outcome of different measurements $\mathrm{ M}_1, \mathrm{ M}_2, \ldots, 
\mathrm{ M}_n$, a noncontextual hidden-variable model assigns the  outcome corresponding to $P$ to some
$\mathrm{ M}_i$ if and only if it does for all other $\mathrm{M}_j$.

We can also see the KCBS inequality in this new perspective. 
The compatibility graph $G$ of this scenario is a pentagon and the 
maximum quantum violation
is obtained with projectors $P_i$ such that $P_i$ and $P_j$ are orthogonal if $(i,j) \in E(G)$. 
This observation leads to two different interpretations of the graph $G$ 
in quantum realizations in this particular case.
 First, each vertex $i$ of $G$ 
can  be viewed as the observable associated to the projector $P_i$. 
The second way to interpret $G$  is associating a measurement  $M$ to every edge $(i,j) \in E(G)$  
which includes outcomes associated to $P_i$ and $P_j$.

In the  exclusivity graph approach, we start with  a graph $G$ with vertices $V(G)$ and edges $E(G)$. 
For each $i \in V$ there is a transformation $P_i \in \mathcal{T}$ in a
probabilistic model and for each
$ (i,j) \in E(G)$ a measurement among whose outcomes are the
$P_i$ and $P_j$. Hence, the events represented by each vertex are mutually exclusive.

Given a graph $G$, a physical model for $G$ is a set of measurements in a physical system, one
 for each edge in $ E\left(G\right)$. For a given state of the system, there is a probability associated to each 
event $i \in V$. We collect these probabilities in a vector $p \in \mathbb{R}^{|V|}$. The set of possible vectors depends on 
the physical theory used to describe the system and we will study this set for classical probability theories,
 quantum theory and generalized probabilistic theories with certain properties, as  explained below.

\subsection{\textsf{Classical Non-Contextual Realizations}}

A classical realization for $G$ is given by a probability space $(\Omega,\Sigma, \mu )$, 
where $\Omega$ is a sample space, 
$\Sigma$ a $\sigma -$algebra and $\mu$ a probability measure in $\Sigma$ and 
for each $i \in V$ a set $A_i \in \Sigma$
 such that $A_i \cap A_j = \varnothing$ if $(i,j)$ belongs to $E\left(G\right)$. 
 For each $i$ the probability of outcome $i$ is 
$$p_i=\mu\left(A_i\right).$$
The set of probability vectors obtained with classical models $\mathcal{E}_C(G)$ is a polytope. 
Distributions that belong to
this set are
called \emph{noncontextual distributions}. Incidentally, this set is a well-known convex polytope in computer science literature, where it is denoted by $STAB(G)$ \cite{Knuth93,Rosenfeld67}. 



\subsection{\textsf{Quantum Realizations}}

A quantum realization for $G$ is given by a density matrix $\rho$ acting  in a Hilbert space $\mathcal{H}$ and for each $i \in V$ a projector $P_i$ acting
in $\mathcal{H}$ such that $P_i$ and $P_j$ are orthogonal  if $(i,j)$ belongs to  $ E\left(G\right).$ 
For each $i$ the probability of the 
outcome $i$ is 
$$p_i=\mbox{Tr}\left(P_i\rho\right).$$
The set of probability vectors obtained with quantum realizations will be denoted by
$\mathcal{E}_Q(G)$ and it  is not a polytope in general. This set is a well-known convex body in computer science literature, where it is denoted by $TH(G)$ \cite{Knuth93,Rosenfeld67}.
Distributions that belong to
this set are
called \emph{quantum distributions}.

If we fix a basis for $\mathcal{H}$ and consider all matrices diagonal in this basis we recover the classical distributions.
Hence 
$$\mathcal{E}_C(G)\subset \mathcal{E}_Q(G).$$

\subsection{\textsf{The Exclusivity Principle}}

The main point of this work is to provide physical principles that 
single out quantum theory in the landscape of theories presented in chapter \ref{chaptergpt}.
With this purpose in mind, we will also  consider probability distributions obtained when we use generalized probability theories, 
but we demand that  they satisfy the following principle:

\begin{prin}[The Exclusivity Principle]
\label{defiexclusivityprinciple}
Given a set  $\{e_k\}$  of pairwise exclusive events, the corresponding probabilities $p_k$  satisfy the 
 following equation:
\begin{equation}
 \sum_{k} p_k \leq 1.
\end{equation}
\end{prin}

From now on, we refer to the Exclusivity principle simply as the \emph{E-principle}.

From the graph theoretical point of view, this restriction is equivalent to impose the condition that
whenever the set of vertices $\{v_k\}$ is a clique\footnote{A clique in $G$ is a complete induced subgraph of $G$.}
in $G$, the sum of the corresponding probabilities $p_k$ can
not exceed one.

Specker pointed out that, in
quantum theory, pairwise joint measurability of a set $\mathcal{M}$
of observables implies joint measurability of $\mathcal{M}$, while in
other theories this implication does not need to hold \cite{Specker60}. This property is known as the \emph{Specker principle}.
Later, Specker conjectured that this is \emph{the fundamental
theorem} of quantum theory \cite{Specker09}. The $E$ principle is a consequence of the Specker principle, as shown in reference 
\cite{NBAASBC13}. 

The E principle can be used to explain why (some) distributions outside the quantum set are forbidden. 
Many promising results where found so far, as we discuss in chapter \ref{chapterlovasz}.

\subsection{\textsf{E-Principle Realizations}}

An E-principle realization for $G$ is given by a state in a probabilistic model and for each $i \in V$ a  transformation
$T_i \in \mathcal{T}$, such that the corresponding  probability distribution  satisfies the E principle.
 
 The distributions obtained in this way are called \emph{E-principle distributions}.
 The set of all E-principle distributions, denoted by $\mathcal{E}_E(\Gamma)$,   is also a polytope. 
This set is a well
known convex polytope in computer science literature, where it is denoted by $QSTAB(G)$ \cite{Knuth93,Rosenfeld67}.

It is a known fact from computer science literature that $TH(G) \subset QSTAB(G)$, which is equivalent
to $\mathcal{E}_Q(G) \subset \mathcal{E}_E(G)$. This was also proven in references \cite{CSW14, FSABCLA12}.

\begin{teo}
\label{teoquantume}
 The quantum distributions satisfy the E principle.
\end{teo}

\begin{dem}
 In quantum theory, exclusive events are associated to orthogonal projectors. Hence, if $\{e_i\}$ is a set of mutually
 exclusive events, a quantum realization will provide a set $\{P_i\}$ of mutually orthogonal projectors. As a consequence
 we have 
 $$\sum_iP_i \leq I$$
 and hence
  $$\sum_ip_i=\sum_iTr\left(P_i\rho\right) \leq Tr\left(\rho\right) \leq 1.$$
\end{dem}

\section{\textsf{Non-contextuality inequalities in the exclusivity-graph approach}}
\label{sectionexclusivityinequalities}

Once more, since the set $\mathcal{E}_C(G)$ is a polytope, it admits an H-description: a finite set
of linear inequalities which provide necessary and sufficient conditions for membership in this set.

\begin{defi}
\label{defiexclusivityinequality}
 A \emph{noncontextuality inequality} is a linear inequality 
 \be \sum \gamma_i p_i \leq b,\ee
 where all $\gamma_i$ and $b$ are real numbers, which is 
 satisfied by all elements of  the classical polytope 
 $\mathcal{E}_C(G)$ and violated by some contextual distribution.
 A \emph{tight noncontextuality inequality} is a linear inequality defining a non-trivial facet of the classical polytope 
 $\mathcal{E}_C(G)$.
\end{defi}

To obtain necessary and sufficient conditions for membership in $\mathcal{E}_C(G)$, we have to find all 
tight noncontextuality inequalities for $G$. This is a difficult problem, in general, and sometimes it is useful to
concentrate in one particular inequality and find out what information it can give.

Given a  graph $G=(V,E)$, consider, for example, the sum of probabilities
\begin{equation}
\beta  = \sum_{i \in V} p_i.
\end{equation}
We can use this sum to provide  necessary conditions to membership in $\mathcal{E}_C(G)$,
$\mathcal{E}_Q(G)$ and $\mathcal{E}_E(G)$. To derive these
conditions we need to identify what are the maximum values of $\beta$ for each of classical, quantum and E-principle 
realizations, which will be denoted respectively by $\beta_C$, $\beta_Q$ and $\beta_E$.
Naturally, by theorem \ref{teoquantume} and the fact that $\mathcal{E}_C(G)\subset \mathcal{E}_Q(G)$, we have
$$\beta_C \leq \beta_Q \leq \beta_E.$$

The inequality
\begin{equation}
 \sum_{i \in V} p_i \leq \beta_C
\end{equation}
is  a \emph{noncontextuality inequality} as long as $\beta_C < \beta_E$ and
\begin{equation}
 \sum_{i \in V} p_i \leq \beta_Q
\end{equation}
is a necessary condition for membership in $\mathcal{E}_Q(G)$.

Also in the exclusivity-graph approach, the graph functions $\alpha(G)$ and $\vartheta(G)$ can be used to
calculate $\beta_C$ and $\beta_Q$. The bound $\beta_E$ can be calculated with the help of another graph function, known
as the \emph{fractional packing number} of $G$.

\begin{defi}
 The \emph{fractional packing number} $\alpha^*(G)$ of a graph $G$ is defined by
 $$\alpha^*(G)=
 \max\left\{\sum_i p_i \left|0 \leq  p_i \leq 1 \ \mbox{and} \ \sum_{i \in C} p_i \leq 1, C \ \mbox{any clique of} 
 \ G\right.\right\}.$$
\end{defi}

\begin{teo}[Cabello, Severini, and Winter, 2010]
Given a graph $G$,
$$\beta_C=\alpha(G), \ \beta_Q= \vartheta(G),  \ \beta_E=\alpha^*(G)$$
where $\alpha(G)$ is the independence number of $G$, $\vartheta(G)$ is  the Lov\'asz number of $G$ and 
$\alpha^*(G)$ is the fractional-packing number of $G$. 
\label{teoqboundexclusivity}
\end{teo}

This result follows directly from the observation that $\mathcal{E}_C(G)=STAB(G), \ \mathcal{E}_Q(G)=TH(G)$ and 
$\mathcal{E}_E(G)=QSTAB(G)$ and the well known fact from computer science literature that 
$\alpha(G), \ \vartheta(G),  \ \alpha^*(G)$ are the maximum values of $\sum_i p_i$
over $STAB(G), \ TH(G),$ and $ \ QSTAB(G)$ respectively \cite{Knuth93,Rosenfeld67}.
Nonetheless, we provide a proof here because it may help
us to understand the physical significance of these graph functions.

\begin{dem}
 The classical bound is achieved in a vertex of the noncontextual polytope. For such a distribution, each $p_i$ is equal
 to zero or one. If $i$ and $j$ are connected by an edge in $G$ they represent different outcomes of the same measurement and 
 hence $p_i$ and $p_j$ can not be both equal to one. Hence the set of indices $i$ such that $p_i=1$ is an independent set and 
 can have at most $\alpha(G)$ elements. This implies that $\beta_C \leq \alpha(G)$ and equality is achieved if we choose
 any independent set $I \subset V(G)$ with $\alpha(G)$ elements and define $p_i=1$ if and only if $i \in I$.
 
 The quantum bound is achieved when we use a pure state $\ket{\psi}$. Let $P_i$ be the projector associated to vertex $i$
 and
 $$\ket{v_i}=\frac{P_i\ket{\psi}}{\left|P_i\ket{\psi}\right|}.$$
 If $i$ and $j$ are connected by an edge in $G$, the corresponding projectors are orthogonal and the vectors $\ket{v_i}$ and
 $\ket{v_j}$ are also orthogonal. Hence, the set of vectors $\ket{v_i}$ and $\ket{\psi}$ provide an orthogonal representation
 for $\overline{G}$ and hence
 $$\sum_ip_i=\sum_i \sand{\psi}{P_i}{\psi}=\sum_i |\braket{\psi}{v_i}|^2 \leq \vartheta(G).$$
 On the other hand, given a orthogonal representation $\left\{\ket{v_i}\right\}$ for $\overline{G}$
 and a state $\ket{\psi}$, let
 $P_i=\ket{v_i}\bra{v_i}$. The projectors $P_i$ and $P_j$ are orthogonal if $i$ and $j$ are connected in $G$ and
 hence $P_i$ and $\ket{\psi}$ provide a quantum realization achieving the upper bound $\vartheta(G)$.

 The equality $\beta_E=\alpha^*(G)$ follows directly from the definition of
 $\alpha^*$: the  restriction $0 \leq  p_i \leq 1$ is satisfied if and only if the $p_i$ represent probabilities 
 and  the condition that $\sum_{i \in C} p_i \leq 1$ for  any clique $C$ of  $G$ is exactly the demand that the E principle
 be satisfied by  
  the distribution. 
 
\end{dem}

We can also calculate the maximum of general linear functions 
\begin{equation}
S_w=\sum_i w_i p_i, \ \ w_i \geq 0
\label{wsum}
\end{equation}
using the weighted versions of the $\alpha$, $\vartheta$ and $\alpha^*$ \cite{Knuth93, Rosenfeld67},
as shown 
by Cabello, Severini, and Winter in reference \cite{CSW14}.


\begin{ex}[A new version of the $n$-cycle inequalities]
\label{exnewncycle}
 The simplest exclusivity  graph for which $\beta_C < \beta_Q$ is the pentagon \cite{CDLP13}.
 It can be proven by inspection that $\beta_C=2$. The quantum bound  is $\beta_Q=\sqrt{5}$, as shown by 
 Lov\'asz original calculation of $\vartheta(C_5)$ \cite{Lovasz79}.
  The maximum value obtained  with E-distributions  is $\frac{5}{2}$, which can be reached when all events 
  have probability equal
  to $\frac{1}{2}$.
 
When $G$ is any $n$-cycle with $n$  odd, we can also prove by inspection that the classical bound is $\beta_C=\frac{n-1}{2}$.
The quantum bound can also be explicitly calculated, and we have that 
 $\beta_Q=\frac{n\cos\left(\frac{\pi}{n}\right)}{1+\cos\left(\frac{\pi}{n}\right)}$, which is equal to $\sqrt{5}$ for $n=5$. 
 The maximum obtained with E-distributions   is $\frac{n}{2}$, which can be reached when all events 
  have probability equal
  to $\frac{1}{2}$.
 
 If $n$ is even, $C_n$ is a bipartite graph, and the vertices in one bipartition define a maximal  independent set.
 The parts have the same size, and hence
 the classical bound is $\frac{n}{2}$. The distribution that assigns probability $\frac{1}{2}$ to all 
 vertices realizes the bound $\beta_E$, which is then equal to $\beta_C$.
 The quantum bound $\beta_Q$ is sandwiched between $\beta_C$ and $\beta_E$ and hence we conclude that
 $\beta_Q$ is also equal to $\frac{n}{2}$.
\end{ex}

\section{\textsf{The quest for the largest contextuality in nature}}
\label{sectionlargestcontextuality}

The connection of the classical and quantum bounds for noncontextuality inequalities and graph theory
allows one to study the violation of such inequalities focusing only on the graph itself.
To study how quantum representations may differ from classical ones we seek for graphs with 
 ``large''  violations. In this section we show some families of graphs with this behavior and 
present the known results about the growth of both $\alpha(G)$ and 
$\theta(G)$ with the number of vertices of $G$.





The measure of violation we propose is the ratio $\frac{\vartheta(G)}{\alpha(G)}$ as a function of the number of
vertices in the graph $G$, 
which represent the number of possible outcomes (elements of $\mathcal{T}$ in the experiment).  

\subsection{\textsf{The quantum gambler}}
\label{subsectiongambler}

A famous bookmaker accepts all kinds of bets. A gambler brings a preparation device
and a set of measurement devices. The preparation device works on demand, always 
preparing the same  known state. The compatibility structure of the measurement devices is also known, 
and exclusiveness can be directly verified.

A set of events with $n$ vertex-transitive exclusivity graph $G$ is picked. The state is such that all events in this set 
have equal probability $p$. The gambler chooses one of the events 
and bets $c$ units of money that this event will happen. If this is the case, the bookmaker
agrees to pay her
\be\frac{c}{p +\epsilon}\ee
units of money. The value of $\epsilon$ is chosen in such a way that the bookmaker guarantees his profit after many rounds
of the game.

If the bookmaker believes the system to be classical,  the prize will be calculated using 
$p=\frac{\alpha}{n}$. If the gambler is able to arrange the same scenario in a quantum system, $p=\frac{\vartheta(G)}{n}$.
This means that a quantum gambler, playing against a classical bookmaker will increase her profit after many rounds by 
a factor of $\frac{\vartheta(G)}{\alpha(G)}$. Hence the gambler will seek for the scenario where this ratio is as large as 
possible.

\subsection{\textsf{The growth of the ratio $\frac{\vartheta}{\alpha}$}}

An important family of noncontextuality inequalities is the $n$-cycle inequalities, presented in example \ref{exncycle}.
In this case, the compatibility graph  is a cycle with $n$ vertices. If $n$ is odd, the exclusivity graph $G$ is 
 the prism graph of order $n$, $Y_n$,  and if $n$ is even, the exclusivity graph is 
 the M\"obius ladder of order $2n$, $M_{2n}$. These graphs are shown in figure \ref{figncycle}.
If $n$ is odd, 
$$\frac{\vartheta(Y_n)}{\alpha(Y_n)}=\frac{2n\cos\left(\frac{2\pi}{n}\right)}{\left(1+\cos\left(\frac{2\pi}{n}\right)\right)
\left(n-2\right)},$$
 and for $n$ even 
 $$\frac{\vartheta(M_{2n})}{\alpha(M_{2n})}=\frac{2n\left(1 + \cos\left(\frac{2\pi}{n}\right)\right)}{n-2}.$$
 The quantum maximum can be obtained in a system of dimension three for $n$ odd, and four for 
 $n$ even \cite{AQBTC13}. 
 In this case, the quantum maximum approaches the 
classical maximum as the number of vertices $n$ grows. 

Something similar happens for the inequalities shown in example
\ref{exnewncycle}, when the $n$-cycle is used as exclusivity graph, with $n$ odd. In this case
$$\frac{\vartheta(C_n)}{\alpha(C_n)}=\frac{2n\cos\left(\frac{2\pi}{n}\right)}{\left(1+\cos\left(\frac{2\pi}{n}\right)\right)
\left(n-2\right)},$$
and the quantum maximum also approaches the classical bound.

In both cases, the differences between  classical and quantum distributions become
smaller when $n$ grows. We want to find families of graphs with the opposite behavior.
We seek for situations in which the ratio $\frac{\vartheta(G)}{\alpha(G)}$ grows as fast as possible.


First we notice that if we fix $\alpha(G)<k$, there is a limit for the ratio  $\frac{\vartheta(G)}{\alpha(G)}$ .

\begin{teo} For every $k \in \mathbb{N}$ there exists an absolute constant $M_k$ such that 
for any graph $G$ on $v$ vertices with $\alpha(G)< k$, $\vartheta(G) \leq M_kv^{1-2/k}.$

\end{teo}

The result above is Theorem 5.1 of reference \cite{AK98}. It generalizes the result of
\cite{Konyagin83} for $k=3$, for which $M_3=2^{\frac{2}{3}}$. 
Although there is  no explicit constructions for general $k$,  in \cite{Alon94} the author shows a family of graphs 
with $\alpha=2$
 approaching $\frac{\vartheta }{ \alpha} =v^{1/3}$. 
 The graphs depend on a parameter $r$ that can not be a multiple of $3$. 
The number of vertices is $2^{3k}$. For $r=2$ it is a graph with $64$ vertices and its complement is a graph formed by $16$ 
unconnected squares.
 In this case $\vartheta=\alpha$ and it does not exhibit quantum violation. We
 have computed the adjacency matrix for the complement of the graph we want for $r=4$.
 It has over 2 million edges. We don't know if for $r=4,5$ the corresponding inequalities have quantum violation.
For $r > 6$ we have $\vartheta>\alpha$. These graphs are Cayley graphs and, as a consequence, regular  and vertex-transitive.

If we do not fix the noncontextual bound we can obtain larger violations with simpler graphs, 
for which the number of vertices does not grow so fast.

 \begin{teo}For every $\epsilon > 0$ there is an explicit family of graphs for which $\vartheta \geq \left(\frac{1}{2} - \epsilon\right)v$ and $\alpha <v^{\delta (\epsilon)}$, $\delta (\epsilon)< 1$.
 \end{teo}


This is Theorem 6.1 in \cite{AK98}. For a pair of integers $q>s>0$, $G(q,s)$ will be the graph on
 $v=\left(\begin{array}{c} 2q\\ q\end{array}\right)$ vertices, each vertex corresponding to a $q$-subset of $\{1,2,\ldots , 2q\}$. 
Two vertices are adjacent iff their intersection has exactly $s$ elements.
For small values of $q$ and $s$ we have:

\begin{center}
\begin{tabular}{cccc}
$q$&$s$&$\alpha$&$\vartheta$\\
\hline
$2$&$1$&$2$&$2$\\
$3$&$1$&$4$&$5$\\
$3$&$2$&$4$&$5$\\
$4$&$1$&$17$&$23$\\
$4$&$2$&$10$&$10$\\
$4$&$3$&$14$&$14$\\
$5$&$1$&$\geq 55$&$94,5$\\
$5$&$2$&$\geq 27$&$42$\\
$5$&$3$&$\geq 12$&$18,67$\\
$5$&$4$&$\geq 28$&$42$
\end{tabular}
\end{center}


For this family, the authors provide an orthonormal representation that achieves the lower bound on $\vartheta$ in dimension $2q$. This orthonormal representation provides a state and  measurements that we can use to achieve this amount of violation.


Although these are the best explicit constructions, it is already known that they do not reach the maximum violation $\frac{\vartheta}{\alpha}$ as a 
function of the number of vertices in the graph \cite{Feige95}.

\begin{teo} For every $\epsilon > 0$ there is a graph $G$ on $v$ vertices such that $\frac{\vartheta(G)}{\alpha(G)}> v^{1-\epsilon}$.
\end{teo}

\begin{teo} There exists an infinite family of graphs on $v$ vertices for which $\frac{\vartheta(G)}{\alpha(G)}> \frac{v}{2^{c\sqrt{\log(v)}}}.$
\end{teo}

Although the results above prove the existence of families with larger ratio then the ones considered above, 
its proofs are based on the probabilistic method and there is no explicit construction approaching these lower bounds
\cite{AS04}. 
It is also not known if these bounds are tight.

It is interesting to notice that the large growth of the ratio $\frac{\vartheta}{\alpha}$ was bad news for research in
graph theory. While $\vartheta(G)$ is easy to compute, other quantities such as the independence number and 
the Shannon capacity of the graph are hard to calculate in general and both are upper bounded by $\vartheta(G)$
\cite{Lovasz79}.
A large growth of  $\frac{\vartheta}{\alpha}$ shows that the bound for $\alpha$ 
is far from being tight, and hence this number can not be 
used in general as a good approximation to the independence number.

As the study of these families may help us to understand how quantum distributions can go beyond the noncontextual ones, 
we believe that there may be some practical applications to high violations of noncontextuality inequalities. 
As an example, we conjecture that there may be a connection between these large violations and the certification of 
randomness in the data obtained in the experiments \cite{PAMGMMOHLMM10,UZZWYDDK13}.

\section{\textsf{Final Remarks}}
\label{sectionfinalncinequalities}

In this chapter we have discussed a  way of proving the impossibility of noncontextual hidden-variable models.
The  set of noncontextual distributions is a polytope and hence can be described by a finite set of 
linear inequalities, violated by some quantum distributions, which proves that the quantum statistics can not be reproduced by 
these models in all situations.

The first approach to noncontextuality we have discussed is through the compatibility graph, which coincides 
with the usual approach to quantum contextuality (as can be seen in appendix \ref{chaptercontextuality}).
In this case, an experimentalist is given
a set of possible measurements to perform in a physical system, and the compatibility structure of this set is encoded in 
the \emph{compatibility graph} of the scenario. The probability distributions for each context are collected to form
an empirical model with the no-disturbance property. The set of noncontextual distributions is a polytope
and the quantum set is in general larger, as proven by the fact that some quantum distributions do not satisfy all 
noncontextuality inequalities
in the H-description of the noncontextual set. 

The mathematical formalism of this scenario can be translated into a 
sheaf-theoretic language, which provides a characterization of the phenomena of 
contextuality in terms of \emph{obstructions to the existence of global sections in a presheaf}, which opens the door
to the use of the methods of sheaf theory to the study of contextuality.

When all coefficients of the inequality are equal to one, 
the local and quantum bounds for a noncontextuality inequality can be found with the help of another graph, the exclusivity
graph of the inequality. The classical bound is equal to the independence number of the exclusivity graph and the quantum bound 
is upper bounded by the Lov\'asz number of this graph. Many important inequalities can be written in this form, including the $n$-cycle inequalities of 
example \ref{exncycle}. The weighted versions of these graph functions can be used to
calculate the classical and quantum bound when the coefficients are not all equal to one, but we will not consider this case here.
 We refer to \cite{CSW14, Knuth93} for more details.
 
 Another perspective to contextuality is given by the exclusivity graph approach. We start with  the exclusivity 
graph $G$, where  each  vertex $i$  represents an event, a transformation $P_i \in \mathcal{T}$ in a
probabilistic model. If
$ (i,j) \in E(G)$  the events $i$ and $j$ are exclusive, that is, there is a measurement among whose outcomes are 
$P_i$ and $P_j$. 
The main difference between this approach and the compatibility graph approach is that in this case we make no restriction in 
the compatibility scenario leading to the exclusivity structure of the events.

In this new perspective, the noncontextuality assumption means that the value associated to a projector $P$
by the hidden-variable model is independent of the other projectors used to define the measurement. As we have seen, 
the same transformation corresponds to an outcome of several different measurements. 
Then, whenever $P$ corresponds to an outcome of different measurements $\mathrm{ M}_1, \mathrm{ M}_2, \ldots, 
\mathrm{ M}_n$, a noncontextual hidden-variable model assigns the  outcome corresponding to $P$ to some
$\mathrm{ M}_i$ if and only if it does for all other $\mathrm{M}_j$.

The set of noncontextual distributions is once more a polytope, contained in the set of quantum distributions which is 
generally larger. It can be described by a finite set of noncontextuality inequalities, violated by quantum distributions in 
many situations.

When all coefficients of the inequality are equal to one, 
the local, quantum and generalized bounds for the noncontextuality inequality can be found using only the exclusivity
graph of the inequality. The classical bound is equal to the independence number of the exclusivity graph and the quantum bound 
is \emph{equal} to the Lov\'asz number of this graph. In this case we have an equality between the
quantum bound and the Lov\'asz number because we do not have
 extra restrictions imposed by
a specific compatibility structure.

The most general distributions we consider have to satisfy the 
Exclusivity principle, and for this kind of distribution the bound is equal to the fractional packing number of
the exclusivity graph. This principle will be used later on in chapter \ref{chapterlovasz} in our attempt to understand why 
quantum theory is not \emph{more noncontextual} then it is. 

Many important inequalities can be written in this form, including the $n$-cycle inequalities of 
example \ref{exnewncycle}.  Once more, the weighted versions of these graph functions can be used to
calculate the bounds when the coefficients are not all equal to one, but we will also not consider this case here.
 We refer to \cite{CSW14, Knuth93} for more details.
 
 We believe that besides the importance for the foundations of quantum theory,
 large violations of noncontextuality inequalities may have practical applications such as
 amplification of randomness. We have presented the known results about the growth of the ratio
 $\frac{\alpha(G)}{\vartheta(G)}$, seeking for the families of graphs for which this ratio grows as fast as 
 possible. Unfortunately, many of the known results are based on the probabilistic method and there is no explicit 
 construction of the graphs or the explicit construction is so complicated that it makes any experimental implementation
 impossible.

\chapter{\textsf{What  explains the  Lov\'asz bound?}}
\label{chapterlovasz}

\epigraph{\begin{center}
\textit{\small{If the truth be told, few physicists have ever really felt comfortable with quantum theory.}}\end{center}}{Philip Ball, \cite{Ball13}}

The mathematical formulation of quantum theory is almost one century old and during this time 
a number of brilliant scientists around
the world have built a quite good knowledge about it, 
both on the theoretical aspects and experimental control of quantum systems.
``Physicists are capable of making stunningly accurate calculations about
molecular structure, high-energy particle collisions, semiconductor behavior, spectral emissions and much more'' \cite{Ball13}.
They learned how to manipulate quantum systems for information processing.
They know a lot about the structure of matter and how to use it for our purposes. This certainly has a great
impact on the development of current technology.

From the practical point of view we may say that physicist have a good relationship with quantum theory.
But, just as Einstein, Podolsky and Rosen in 1935, you can get in serious trouble when you try to understand the 
meaning of the mathematical objects, specially if you try to apply the reasoning of classical physics we are used to.

This situation led many people to adopt the way of thinking known as \emph{Copenhagen interpretation}.
According to this line of thought,  the weirdness of quantum theory reflects fundamental limits on what can be known 
about nature
and we just have to  accept  it. Quantum theory should
not be understood but seen just as a tool to get practical results. As famously phrased by  
 David Mermin, physicist should ``shut up and calculate''\cite{Mermim89}.

Not everyone is happy with this interpretation, including Mermim himself \cite{Mermim14}. Physics is not just about getting practical results, it is 
also about \emph{understanding} how nature behaves.
Since the EPR vs Bohr debate, many have tried to understand (or question, like EPR) 
the abstract formulation of quantum theory from more
compelling physical arguments. This is one of the most seductive scientific
challenges in recent times: deriving quantum theory
 from simple physical principles.

 The starting point is assuming general probabilistic theories allowing for probability distributions
that are more general than those that arise in quantum theory, and
the goal is to find principles that pick out quantum theory from this
landscape of possible theories. There are diverse ideas on how to do this, and  at least three different approaches to the
problem stand out.

The first one consists of reconstructing quantum theory as a purely
operational probabilistic theory that follows from some
sets of axioms. The idea is to demolish the abstract entities and start again. Imposing a small number of 
reasonable physical principles, they manage to prove that the only consistent probabilistic theory 
is quantum. 
Although really successful, this approach does not resolves the issue completely, specially because some 
of the principles imposed do not sound so natural. This ``unsatisfaction'' is very well phased by Chris Fuchs \cite{Fuchs11}:

\begin{quote}There is no doubt that this is invaluable work, particularly for our understanding of the intricate 
connections between so many quantum information protocols. But to me, it seems to miss the mark for an ultimate 
understanding of quantum theory; I am left hungry. I still want to know what strange property of matter forces this 
formalism upon our information accounting. 
I would like to see an axiomatic system that goes for the weirdest part of quantum theory.\end{quote}

The second approach to the problem goes in this direction. Instead of trying to reconstruct quantum theory, the idea is to 
understand what physical principles explain one of the weirdest part of quantum theory: nonlocality.
 Many different principles have been proposed, which we left for Appendix  \ref{chaptertsirelson}.

The third approach consists of identifying principles that explain the set of quantum contextual correlations  without restrictions imposed by
a specific experimental scenario. 
The belief that identifying the physical principle responsible for quantum contextuality can be more successful than previous
approaches is based on two observations. On one hand, when focusing on quantum contextuality we are just considering a 
natural extension of quantum nonlocality which is free of certain restrictions (composite systems, space-like separated tests 
with multiple observers, entangled states) which play no role in the rules of quantum theory, although they are crucial for many important applications, specially in communication 
protocols (see, for example, references  \cite{wikiquantumcryptography, HHHH09,BBCJPW93} and other references therein), 
and  played an important role in the 
historical debate on whether or not quantum theory is a complete theory.

On the other hand, it is based on the observation that, while calculating the maximum value of quantum correlations for 
nonlocality scenarios is a mathematically complex problem (see \cite{PV10} to see how complex is to get the quantum maximum 
for  a simple inequality like $I_{3,3,2,2}$), calculating the maximum contextual value of quantum correlations for an 
{\em arbitrary} scenario characterized by its exclusivity graph  is simple: as we proved in section 
\ref{sectionexclusivityinequalities}, the maximum quantum 
contextuality is given by the Lov\'asz number  of its exclusivity graph, which is the solution  of a semidefinite program
\cite{Lovasz95}.
Indeed, from the graph 
approach perspective, the difficulties in characterizing quantum nonlocal correlations are due to the 
mathematical difficulties associated to the extra constraints resulting form enforcing a particular labeling of the 
events of a exclusivity structure in terms of parties, local settings, and outcomes \cite{SBBC13}, rather than a 
fundamental difficulty related to the principles of quantum theory.

Within this line of research, the most promising candidate for being {\em the} fundamental principle of quantum contextuality 
is the Exclusivity principle,   which can be stated as follows (see principle \ref{defiexclusivityprinciple}): 

\begin{center}The sum of the probabilities of a set of pairwise 
exclusive events cannot exceed~1.
\end{center}

The Exclusivity principle was suggested by the works of Specker \cite{Specker60} and Wright 
\cite{Wright78} and used in \cite{CSW10} as an upper bound for quantum contextuality. 
 However, its fundamental importance for QM 
was conjectured long before \cite{Specker09}. It was promoted to a possible fundamental principle 
by the observation that it explains the maximum quantum violation of the simplest noncontextuality inequality, as we will
see in section \ref{sectionpentagon}. It also explain the quantum maximum for many other 
inequalities and rules out nonlocal boxes  in some important Bell scenarios (see section \ref{sectionlocalorthogonality}). 
The Exclusivity principle, when applied only to Bell scenarios is called \emph{local 
orthogonality} \cite{FSABCLA12}. However, {\em with this extra restriction}, the Exclusivity principle cannot single out some 
quantum nonlocal correlations \cite{FSABCLA12}.

By itself, the Exclusivity principle 
singles out the maximum quantum value for some Bell and 
noncontextuality
inequalities \cite{Cabello13}. According to the results of section \ref{sectionexclusivityprinciple} this happens
whenever $\vartheta(G)=\alpha^*(G)$. We can get better bounds if we apply the E principle to more sophisticated scenarios.
When applied to the OR product of two copies of the exclusivity graph, which physically may be seen as two independent realizations of the same experiment,
the Exclusivity principle singles out
the maximum quantum value for experiments whose exclusivity graphs are vertex-transitive and self-complementary 
\cite{Cabello13},
which include the simplest noncontextuality inequality, namely the  KCBS inequality presented in example \ref{exnewncycle}.
Moreover, either applied to two copies of the exclusivity graph of the CHSH
inequality 
or of a simpler inequality, the Exclusivity principle excludes the so called PR 
boxes  and provides an upper bound to the maximum violation of the CHSH inequality which is 
close to the Tsirelson bound \cite{FSABCLA12, Cabello13} (see appendix \ref{chaptertsirelson}). 
In addition, when applied to the OR product of an infinite number of copies, 
there is strong evidence that the Exclusivity principle singles out the maximum quantum violation of the 
noncontextuality inequalities whose 
exclusivity graph is
the complement of odd cycles on $n \ge 7$ vertices \cite{CDLP13}. 
Indeed, it might be also the case that, when applied to an infinite number of copies,
the Exclusivity principle singles out the Tsirelson bound of the CHSH inequality \cite{FSABCLA12,Cabello13}.

Another evidence of the strength of the Exclusivity principle was recently found by Yan \cite{Yan13}. 
By exploiting Lemma~1 in \cite{Lovasz79},
Yan has proven that, {\em if all correlations predicted by quantum theory for an experiment with exclusivity graph $G$ are 
reachable in nature,} then the Exclusivity principle singles out the {\em maximum} value of the correlations produced by 
an experiment whose exclusivity graph is the complement of $G$, denoted as $\overline{G}$.

We recently proved three stronger consequences of the E principle \cite{ATC14}. 
The Exclusivity principle singles out the \emph{entire set of quantum correlations} 
associated to any exclusivity graph assuming the set of quantum correlations for the complementary graph. 
Moreover,  for self-complementary graphs, the Exclusivity principle, {\em by itself} (i.e., without further assumptions), 
excludes any set of correlations strictly larger than the quantum set. Finally,  for vertex-transitive graphs, 
the Exclusivity principle singles out the maximum value for the quantum correlations assuming only the quantum maximum for 
the complementary graph. 
These results show that the Exclusivity principle goes beyond any other proposed principle towards the objective of 
singling out quantum correlations.

In this chapter we will prove all these results in detail. In section \ref{sectionexclusivityprinciple} we review the 
 noncontextuality inequalities under consideration, the definition of the exclusivity principle and other important concepts.
 In section \ref{sectionpentagon} we explain how the principle applied to two copies of the pentagon singles out
 the quantum maximum for this graph. In section \ref{sectioncomplement} we show how the principle can be used to
 connect the set of quantum correlations for $G$ and $\overline{G}$, and how this connection is sufficient for 
 ruling out any distribution outside the quantum set in many important cases. In \ref{sectionoperations} we show
 that something similar can be done with graph operations other then complementation and as a consequence
 we prove that the exclusivity principle explains the quantum maximum for all vertex transitive graphs with
 $10$ vertices, except two. 
 We end with our final remarks
 in \ref{sectionfinallovasz}.
Consequences of the E principle under Bell-scenario restrictions are outside the scope of the present thesis (and chapter), but a small introduction can be found in section 
\ref{sectionlocalorthogonality}.

\section{\textsf{The Exclusivity Principle}}
\label{sectionexclusivityprinciple}

First, let us briefly review some of the definitions and concepts introduced in section \ref{sectionexclusivityinequalities}.
We start with an exclusivity  graph $G=(V,E)$.
Each vertex $i$ of $G$ corresponds to 
a transformation $T_i \in \mathcal{T}$ in a physical system and two vertices are connected by an edge if they are exclusive, 
that is, if they can be two different outcomes of 
the same measurement. For a given state of the system, there is a probability $p_i$ associated to each 
vertex $i \in V$. We collect all these probabilities in a vector $p \in \mathbb{R}^{|V|}$. The set of possible vectors depends 
on 
the physical theory used to describe the system and we will see how the Exclusivity principle (principle \ref{defiexclusivityprinciple}) constrains this set.

Could this principle be the reason for quantum theory not be more noncontextual?
Can it explain the quantum maximum for noncontextuality inequalities? It is not clear what happens in general, but for a special class of inequalities (or graphs)
many results supporting a positive answer have been found.
We will apply the E principle for sums of the type

\begin{equation}\label{eqNCIneq}
S_G = \sum _{i\in V}  p_i,
\end{equation}
that is, we set $\gamma_i=1$ for all $i$ in definition \ref{defiexclusivityinequality}.
For non-contextual distributions we know that 
 \begin{equation}\label{eqNC}
S_G  \stackrel{\mbox{\tiny{NC}}}{\leq} \alpha\left(G\right),
\end{equation}
while for quantum distributions we have
\begin{equation}\label{QIneq}
S_G \stackrel{\mbox{\tiny{Q}}}{\leq} \vartheta\left(G\right),
\end{equation}
where $\vartheta(G)$ is the Lov\'asz number of $G$.

The first question is if the Exclusivity principle is capable of explaining the quantum bound $\vartheta(G)$.
For many different cases, a lot of them with special importance for the study of contextuality, this is indeed the case.
A much more ambitious question is 
if this principle is enough to single out the set of quantum distributions and not just the quantum maximum. Again, we
are able to exhibit  a important family of graphs for which this is true.

\section{\textsf{The Pentagon}}
\label{sectionpentagon}

The Exclusivity principle singles out the quantum maximum for the  simplest noncontextuality inequality.

\begin{teo}[Cabello, 2013]
For $G=C_5$, the maximum value for $S_G$ allowed by theories satisfying the Exclusivity principle is $\sqrt{5}$, which is also 
the maximum for quantum distributions. 
\end{teo}

\begin{dem}
Let $\{e_i\}$ and $\{e'_i\}$ be two sets of $5$ events with exclusivity graph $G$  as shown in figure 
\ref{figpentagoncomplement}, such that $e_i$ and $e'_i$ are independent. 
 \begin{figure}[!h]
  \centering
  \includegraphics[scale=0.4]{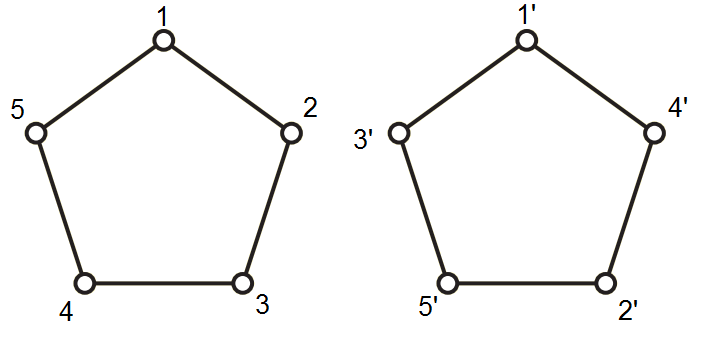}\\
  \caption{Exclusivity graphs of the sets of events $e_i$ and $e'_i$.}\label{figpentagoncomplement}
\end{figure}

Define the event $f_i=e_i \wedge e'_i$ which is true if and only if both $e_i$ and $e'_i$ are true.
Note that the exclusivity graph of the events $\{f_i\}$ is the complete graph on $5$ vertices because $\{f_i\}$ is a set of
pairwise mutually exclusive events.

Since $e_i$ and $e'_i$ are independent $p(f_i)= p(e_i) p(e'_i)$.
Using the Exclusivity principle we have
$$\sum_i p(f_i) = \sum_i p(e_i)p(e'_i) \leq 1.$$
Using the symmetry of the pentagon, we can assume (see lemma \ref{lemmatransitive} below) that the maximum is reached when all the probabilities are the same, 
that is
$$p(e_i)=p(e'_i)=P,  \  \ \forall \  \ i\in V$$
Hence we have
$$\sum_i P^2 = 5P^2 \leq 1$$
which implies that
$$P\leq \frac{1}{\sqrt{5}}.$$
Now, if we substitute this value into equation \eqref{eqNCIneq} for $S_G$ we have
$$S_G=\sum_i P_i \leq \sqrt{5}.$$
\end{dem}

\section{\textsf{The exclusivity principle forbids sets of correlations larger than the quantum set}}
\label{sectioncomplement}


The idea used in the previous section to derive the quantum bound for the pentagon using the 
Exclusivity principle can be applied 
 to show the there is a connection between the set of quantum distributions for $G$
and $\overline{G}$. 
Yan first used it in reference \cite{Yan13}, where he proves the following:

\begin{teo}[Yan, 2013]
\label{teoyan}
 Given the set of quantum distributions for $\overline{G}$, the $E$ principle singles out the quantum maximum for 
 $G$.
\end{teo}

\begin{dem}
 Let $\{e_i\}$ be a set of $n$ events with exclusivity graph $G$ and $\{e'_i\}$ be a set of $n$ events with exclusivity 
graph $\overline{G}$, such that $e_i$ and $e'_i$ are independent.
Define the event $f_i= e_i \wedge e'_i$ which is true if and only if both $e_i$ and $e'_i$ are true.
Note that the exclusivity graph of the events $\{f_i\}$ is the complete graph on $n$ vertices because $\{f_i\}$ is a set of
pairwise mutually exclusive events. The Exclusivity principle implies that
 $$\sum_i p(f_i) = \sum_i p(e_i)p(e'_i) \leq 1.$$

Suppose that the distribution $p(e'_i)$ is given by
\be p(e'_i)=|\braket{\psi}{v_i}|^2. \label{eqquantumgbar} \ee
Then 
$$1 \geq\sum_i p(f_i) = \sum_i p(e_i)p(e'_i)= \sum_i p(e_i)\left|\braket{\psi}{v_i}\right|^2,$$ 
and hence
$$\sum_i p(e_i)\min_i\left[\left|\braket{\psi}{v_i}\right|^2\right] \leq \sum_i p(e_i)\left|\braket{\psi}{v_i}\right|^2 \leq 1$$
which implies that 
$$\sum_i p(e_i)\leq \max_i\frac{1}{\left|\braket{\psi}{v_i}\right|^2}.$$
This inequality should hold for any normalized $\ket{\psi}$ and any orthogonal representation  $\{\ket{v_i}\}$, and hence
$$\sum_i p(e_i) \leq \min_{\ket{\psi}, \ket{v_i}}\max_i\frac{1}{\left|\braket{\psi}{v_i}\right|^2}.$$
The right-hand side is just the Lov\'asz number of $G$ (see \cite{Lovasz79, Knuth93}). Hence, we conclude
that if all quantum distributions given by equation \eqref{eqquantumgbar} can be reached
and if the Exclusivity principle holds, the maximum value of
$S_G$ can not exceed the quantum bound.
\end{dem}

Let us  show that making the same assumptions of the previous theorem, it is possible
not only to single out the quantum maximum but also the entire set of quantum correlations.

\begin{prop}[Amaral, Terra Cunha, Cabello, 2014]
\label{teoexclusivityquantumset}
Given the quantum set $ \mathcal{E}_Q(\overline{G})$, the Exclusivity principle singles out the quantum set $\mathcal{E}_Q(G)$.
\end{prop}


\begin{dem} 
Let $\{e_i\}$ be a set of $n$ events with exclusivity graph $G$ and $\{f_i\}$ be a set of $n$ events with exclusivity 
graph $\overline{G}$, such that $e_i$ and $f_i$ are independent.
Define the event $g_i$ which is true if and only if both $e_i$ and $f_i$ are true, $g_i = e_i\wedge f_i$.
Note that the exclusivity graph of the events $\{g_i\}$ is the complete graph on $n$ vertices because $\{g_i\}$ is a set of
pairwise mutually exclusive events.

Since $e_i$ and $f_i$ are independent $p(g_i)= P_i \bar{P}_i$, where $P_i=p\left(e_i\right)$ and $\bar{P}_i=p\left(f_i\right)$.
Using the Exclusivity principle we have
\begin{equation}
 \label{eqeprinciple1}
 \sum_i P_i \bar{P}_i \stackrel{\mbox{\tiny{E}}}{\leq} 1.
\end{equation}
Now we use corollary 3.4 and  theorem 3.5 in  reference \cite{GLS86}:

\begin{teo}
\label{teoTH}
The set  $TH(G)$ can be written in the following ways:
\begin{equation}
  TH(G)=\left\{P \in \mathbb{R}^n; P_i \geq 0, \vartheta(\overline{G},P) \leq 1\right\},
	 \label{eqTHab}
\end{equation}
where
\begin{equation}
\vartheta(\overline{G},P)= \max\left\{\sum_i P_i \bar{P}_i; \bar{P} \in  TH(\overline{G})\right\},
\end{equation}
and
\be TH(G)=\left\{P \in \mathbb{R}^n; \  P_i=\left|\braket{\psi}{v_i}\right|^2,  \small{ \braket{\psi}{\psi}=1,
 \{\ket{v_i}\} \ \mbox{orthonormal representation for} \ \overline{G}}\}\right\}. \label{eqTHq}\ee
\end{teo}

Equation \eqref{eqTHab} implies that, for  a given $P$, equation \eqref{eqeprinciple1} will be satisfied for all $P'$ if and only if $P$ belongs to $TH(G)$.  
Equation \eqref{eqTHq} shows that $TH(G)=\mathcal{E}_Q(G)$. Then we conclude that if the set of allowed distributions for $\overline{G}$ is 
$TH(\overline{G})= \mathcal{E}_Q(\overline{G})$, theorem \ref{teoTH} implies that the distributions in
$G$ allowed by the Exclusivity principle belong to $\mathcal{E}_ Q(G)$.
\end{dem}

Physically, the proof above can be interpreted as follows: assuming that nature allows all quantum distributions for $\overline{G}$, the Exclusivity principle {\em singles out the quantum distributions for} $G$.

Proposition \ref{teoexclusivityquantumset} does not imply that the Exclusivity principle, by itself, singles out the quantum correlations for $G$, since we have assumed quantum theory for $\overline{G}$.
Nonetheless, it is remarkable that the Exclusivity principle connects the correlations of two, a priori, completely different experiments on two completely
different quantum systems. For example, if $G$ is the $n$-cycle $C_n$ with $n$ odd, the tests of the maximum quantum violation of the corresponding
noncontextuality inequalities require systems of dimension $3$ \cite{CSW10,CDLP13,LSW11,AQBTC13}. However, the tests of the maximum quantum violation of the noncontextuality inequalities with exclusivity graph $\overline{C_n}$ require systems of dimension that grows with $n$ \cite{CDLP13}. Similarly, while two qubits are enough for a test of the maximum quantum violation of the CHSH inequality 
(see appendix \ref{chapternonlocality}), the complementary test is a noncontextuality inequality (not a Bell inequality) that requires
a system of, at least, dimension $5$ \cite{Cabello13b}.

An important consequence of proposition \ref{teoexclusivityquantumset} is that the larger the quantum set of $G$, 
the smaller the quantum set for $\overline{G}$, since each probability allowed for $G$ becomes
a restriction on the possible probabilities for $\overline{G}$. Such duality gets stronger when $G$
is a self-complementary graph.

\begin{prop}[Amaral, Terra Cunha, Cabello, 2014]
\label{propselfcomp}
If $G$ is a self-complementary graph, the Exclusivity principle, by itself, excludes any set of probability distributions strictly larger than the quantum set.
\end{prop}

\begin{dem}

 Let $X$ be a set of distributions containing $ \mathcal{E}_Q(G)$ and let $P \in X \setminus \mathcal{E}_Q(G)$. By Theorem~1, 
there is at least one
$\overline{P} \in  \mathcal{E}_Q\left(\overline{G}\right)$ such that
\begin{equation}
\sum_{i \in V(G)}P_i\overline{P}_i > 1,
\label{equationself}
\end{equation}
which is in contradiction with the Exclusivity principle.
Since $G$ is self-complementary, after a permutation on the entries given by the isomorphism between $G$ and $\overline{G}$,
$\overline{P}$ becomes an element of $ \mathcal{E}_Q(G)$ and hence $P$ and $\bar{P}$ belong to $X$.  
Expression \eqref{equationself}
implies that this set is not allowed by the Exclusivity principle.  
\end{dem}

The fact that the Exclusivity principle is  sufficient for pinning down the quantum correlations as the maximal 
set of correlations for any self-complementary graph, given that the entire quantum set is possible, means that the Exclusivity principle is able to single out the 
quantum correlations for a large number of  nonequivalent noncontextuality inequalities, including the KCBS one. 
In contrast, 
neither information causality, nor macroscopic locality, nor local orthogonality
have been able to single out the set of quantum correlations in any Bell inequality.

The hypothesis in theorem \ref{teoyan} can be weakened for vertex transitive graphs. Instead of assuming
the entire set of quantum correlations for $\overline{G}$, the same result can be proven, given only the 
quantum maximum for $\overline{G}$.
The exclusivity graphs of many interesting inequalities including CHSH \cite{CHSH69}, KCBS \cite{KCBS08}, the $n-$cycle inequalities \cite{CSW10,CDLP13,LSW11,AQBTC13},
and the antihole inequalities \cite{CDLP13} are vertex transitive. A graph is vertex transitive if for any pair $u,v \in V(G)$ there is $\phi \in \mbox{Aut}(G)$ such that $v=\phi(u)$, where $\mbox{Aut}(G)$ is the group of
automorphisms of $G$ (\textit{i.e.}, the permutations $\psi$ of the set of vertices such that $u, v \in V(G)$ are adjacent if and only if $\psi(u), \psi(v)$ are adjacent).

\begin{prop}[Amaral, Terra Cunha, Cabello, 2014]
\label{teotransitive}
If $G$ is a vertex-transitive graph on $n$ vertices, 
given the quantum maximum for $\overline{G}$, the Exclusivity principle singles 
out the quantum maximum for $G$.
\end{prop}


A sequence of three lemmas proves the result. First we prove that the quantum maximum for $S$ is assumed at a symmetric configuration. Then we prove that the product of the quantum maxima for $G$ and 
$\overline{G}$ is bounded from above by the number of vertices of $G$, and the same from below.

\begin{lemma}
\label{lemmatransitive}
If $G$ is a vertex-transitive graph, then the quantum maximum for $S = \sum_i P_i$ is attained at the 
constant distribution $P_i = p_{\textrm{max}}$.
\end{lemma}

\begin{dem}
Let $P= \left(p(e_1),p(e_2),\ldots, p(e_n)\right)$ be a distribution reaching the maximum.
Given an automorphism of $G$, $\phi \in \mbox{Aut}(G)$, consider the distribution $P_{\phi}$ defined as $p_{\phi}(e_i)= p(\phi(e_i))$.
This is  a valid quantum distribution, also reaching the maximum for $S$.
Define the distribution
\begin{equation}
 Q=\frac{1}{A} \sum_{\phi \in \mbox{Aut}(G)} P_{\phi},
\end{equation}
where $A = \# \mbox{Aut}(G)$. This distribution also reaches the maximum for $S$.
Since $G$ is vertex transitive, given any two vertices of $G$, $e_i$ and $e_j$, there is an automorphism $\psi$ such that $\psi(e_i)= e_j$. Then,
\begin{eqnarray}
q(e_j)&=& q(\psi(e_i)) \nonumber \\
 &=& \frac1A \sum_{\phi \in \mbox{Aut}(G)} p_{\phi}(\psi(e_i)) \nonumber \\
 &=& \frac1A \sum_{\phi \in \mbox{Aut}(G)} p\left( \phi \circ \psi(e_i)\right) \nonumber \\
 &=& \frac1A \sum_{\phi' \in \mbox{Aut}(G)} p_{\phi'}(e_i) \nonumber \\
 &=& q(e_i).
\end{eqnarray}
\end{dem}

\begin{lemma}
\label{lemmamaxima}
If $G$ is a  vertex-transitive graph on $n$ vertices, then the Exclusivity principle implies that the quantum maxima for $S(G)$ and for $S(\overline{G})$ obey 
\begin{equation}\label{eq:QuantumMaxima}
M_Q\!\left({G}\right) M_Q\!\left(\overline{G}\right) \stackrel{\mbox{\tiny{E}}}{\leq} n.
\end{equation}

\end{lemma}

\begin{dem}
Lemma \ref{lemmatransitive} applies for both, $G$ and $\overline{G}$, giving $n p_{\textrm{max}} = M_Q\!\left(G\right)$ and $n \bar{p}_{\textrm{max}} = M_Q\!\left(\overline{G}\right)$.
Inequality \eqref{eqeprinciple1} for these extremal distributions reads
\begin{equation}
n\,p_{\textrm{max}}\,\bar{p}_{\textrm{max}} \stackrel{\mbox{\tiny{E}}}{\leq} 1,
\end{equation}
which proves the result.
\end{dem}

\begin{lemma}
If $G$ is a  vertex-transitive graph on $n$ vertices, then 
\be M_Q\!\left({G}\right) M_Q\!\left(\overline{G}\right) {\geq}\; n.\ee
\end{lemma}

\begin{dem}
When we recall that the graph approach identify the quantum maximum with the Lov\'asz number, as proven in 
theorem \ref{teoqboundexclusivity}, we have that
\begin{subequations}
\begin{eqnarray}
\vartheta({G})&=&M_Q\left({G}\right) \nonumber,\\
\vartheta(\overline{G})&=&M_Q\left(\overline{G}\right),
\end{eqnarray}
\end{subequations}
and since
 for vertex-transitive graphs $\vartheta(G)\; \vartheta(\overline{G}) \geq n$ (Lemma~23 in reference \cite{Knuth93}), 
 the lemma follows.
%
\end{dem}

Proposition \ref{teotransitive} opens the door to experimentally discard higher-than-quantum correlations. 
Specifically, lemma \ref{lemmamaxima} implies that we can test if the maximum value
of correlations with exclusivity graph $G$ goes beyond its quantum maximum without violating the Exclusivity principle by performing an independent
experiment testing correlations with exclusivity graph $\overline{G}$ and experimentally reaching its quantum maximum \cite{Cabello13b}.
A violation of the quantum bound for $\overline{G}$ in any laboratory would imply the impossibility of reaching the quantum maximum for $G$ in any other laboratory.

\section{\textsf{Other graph operations}}
\label{sectionoperations}

We have seen in the previous section that using the operation of complementation and the Exclusivity principle, we are able to
explain the quantum bound and the quantum set of distributions for many different noncontextuality inequalities.
In a joint work with Ad\'an Cabello, we study if something similar is possible using
other graph operations.


\subsection{\textsf{Direct cosum of $G'$ and $G''$}}

\begin{defi}
Given two graphs $G'$ and $G''$ we define the \emph{direct cosum}  $G$ of $G'$ and $G''$ as the graph with 
$V(G)=V(G') \sqcup V(G'')$  and such that
$(u,v) \in E(G)$ iff $(u,v) \in E(G')$, or $(u,v) \in E(G'')$, or $u \in V(G')$ and $v \in V(G")$.
\end{defi}

 This operation applied to two copies of $C_5$ is illustrated\footnote{For $C_5$,
 this operation is equivalent to applying the duplication defined is 
subsection \ref{subsectiontwinning} and complementation, but this is not true in general. For  general graphs 
$G'$ and $G''$, $G=\overline{\overline{G'}+ \overline{G''}}$,
where the direct sum of graphs is defined by the disjoint union of vertices and edges.} in figure \ref{figcosumpentagon}.

The result below is a well-known fact and can be found on reference \cite{Knuth93}, but we repeat it here to reinforce  the 
 connections with quantum theory.

\begin{figure}[!h]
\centering
 \includegraphics[scale=0.6]{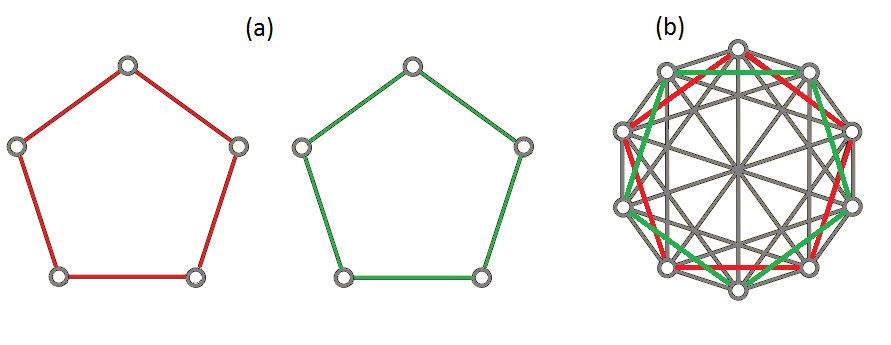}
\caption{Two copies of the pentagon  (a) and their direct cosum (b), the circulant graph $C_{10}(1,2,3,5)$. In (b), one of the copies is colored in red, the other copy in green and the edges connecting the vertices of one copy to the 
other are gray.
\label{figcosumpentagon}}
\end{figure} 

\begin{lemma}
\label{lemmadirectsum}
 $\vartheta\left(G\right)=\max\left\{\vartheta\left(G'\right), \vartheta\left(G''\right) \right\}$.
\end{lemma}

\begin{dem}

  Let $\{\ket{v_i}$\} be an  orthonormal representation for $G$ and $\ket{\psi}$ be a unit vector in the same vector space.
 Every vertex of $G'$ is exclusive to all vertices of $G''$, which
 means that the vectors of the orthonormal representation for $G$ generate a subspace $V'$
orthogonal to the
subspace $V''$ generated by the vectors of the orthonormal representation for $G''$. Because of this, we can decompose $\ket{\psi}$ as a sum of
 two orthogonal vectors:

$$\ket{\psi}= a\ket{\psi'} + b\ket{\psi''}, \  \  \ \ket{\psi'} \in V', \ \ \ \ket{\psi''} \in V'', \ \ \ |a|^2+|b|^2=1.$$

With these definitions we have
$$\sum_{i\in G}|\braket{\psi}{v_i}|^2=|a|^2\left(\sum_{i\in G'}|\braket{\psi'}{v_i}|^2\right) + |b|^2\left(\sum_{i\in G''}|\braket{\psi''}{v_i}|^2\right)$$
 and then
$$\vartheta\left(G\right) \leq \max\left\{\vartheta\left(G'\right), \vartheta\left(G''\right) \right\}.$$

Suppose $\max\left\{\vartheta\left(G'\right), \vartheta\left(G''\right) \right\}=\vartheta\left(G'\right)$. 
Let $\left\{\ket{v'_i}\right\}$ be a Lov\'asz optimal representation for $G'$ and $\ket{\psi}$ the unit vector 
achieving $\vartheta(G')$. Let $\left\{\ket{v{_i}''}\right\}$ be any
Lov\'az representation for $G''$. The set of vectors $\left\{\ket{v'_i}\oplus 0, 0 \oplus \ket{v_{i}''}\right\}$ is an optimal
Lov\'asz representation for $G$ and the unit vector
$\ket{\psi} \oplus 0$ achieves the upper bound.

\end{dem}

\begin{cor}
\label{cordirectsum}
 If the $E$ principle rules out violations above quantum maximum for $G$, it also rules out violations above 
 the quantum maximum for its direct cosum with
any other graph $H$ such that $\vartheta(H) \leq \vartheta(G)$. In particular, it rules out violations above the quantum maximum for the direct cosum of $G$ with itself.
\end{cor}

\subsection{\textsf{Twinning, partial twinning and duplication}}
\label{subsectiontwinning}

We can also consider  graphs obtained from two copies of $G$ by adding some of the edges between the vertices of each 
copy of $G$ but not all of them.
One of
this graphs is the graph $T(G)$ obtained if we consider two copies of $G$ with the same
labeling and join the vertices of one of the copies with the exclusive vertices of the other copy. 
Figure \ref{figtwinningpentagon} shows this operation applied to the pentagon.
We call this operation \emph{twinning}, since the graph associated to $T(G)$ is the one obtained by 
 twinning all the vertices of $G$.

\begin{teo}
\label{thetatwinning}
$\vartheta[T(G)] = 2 \vartheta(G).$

\end{teo}

\begin{dem}

  The upper
bound  $\vartheta[T(G)]\leq 2\vartheta(G)$ comes from the fact that the each copy of $G$ is an induced subgraph
of $T(G)$ and this implies that every orthonormal representation for the
 twinning  includes an orthonormal representation for each copy of $G$. Equality is reached   since given an optimal orthonormal representation 
  $|\psi\rangle$,
$\{|v_i\rangle\}_{i=1}^n$ for $G$, the vectors $|\psi\rangle$, $\{|v_i\rangle\}_{i=1}^{2n}$ with $|v_i\rangle=|v_{i+n}\rangle$ 
form an optimal orthonormal representation  for $T(G)$.

\end{dem}

The same holds true for any graph obtained from $T(G)$ by removing edges between the two copies of $G$.  
We call these graphs \emph{partial twinnings} of $G$. This follows from  the lemma below.

\begin{lemma}[The second sandwich lemma]
\label{lemmasecondsd}
 If $G_1 = (V,E_1)$ and $G_2=(V,E_2)$, with $E_2 \subset E_1$ and $\vartheta\de{G_1} = \vartheta\de{G_2} = \vartheta$, 
 then, for any $G'=(V,E)$ such that $E_2 \subset E \subset E_1$, $\vartheta\left(G'\right) = \vartheta$.
\end{lemma}

\begin{dem}
 Let  $|\psi^1\rangle$, $\{|v_i^1\rangle\}$ be an optimal orthogonal representation for $G_1$. It is also an 
 orthogonal representation for $G'$, wich implies that $\vartheta(G') \geq \vartheta$.  Let  $|\psi\rangle$, $\{|v_i\rangle\}$ be an optimal
 orthogonal representation for $G'$. It is also an orthonormal representation for $G_2$, which implies that
 $\vartheta \geq \vartheta(G')$.
  
\end{dem}

\begin{cor}
\label{corpartialtwinning}
 If $G'$ is a partial twinning of $G$ then $\vartheta(G')=2\vartheta(G)$.
\end{cor}

\begin{dem}
 We apply the second sandwich lemma \ref{lemmasecondsd} with $G_1=T(G)$ and  $G_2$ the graph obtained by disjoint
 union of two copies of $G$.
\end{dem}

Figure \ref{figtwinningpentagon} (a) shows the twinning of $C_5$. Partial twinnings
of $C_5$ can be obtained by removing any of the ten edges present in  figure \ref{figtwinningpentagon} (a) and absent in
figure \ref{figtwinningpentagon} (c). 
Figure \ref{figtwinningpentagon} (b) is 
just a particular case of this.

From theorem \ref{thetatwinning}  and corollary \ref{corpartialtwinning}, we have:

\begin{cor}
\label{cortwinning}
 If the Exclusivity principle singles out the quantum maximum for a graph $G$,
 it also singles out the quantum maximum for its twinning and all its partial twinnings.
\end{cor}

The extreme case of partial twinning  presented in figure \ref{figtwinningpentagon} (c) is also called the direct sum of 
$G$ with itself \cite{Knuth93}. 
We call this operation \emph{duplication}\footnote{Although the term \emph{duplication} is sometimes 
used to refer to a different graph operation than the one we 
define here, we choose this term because its physical interpretation: for exclusivity graph, 
the duplication, as defined above, represents two independent realizations of
the same experiment.} of $G$.
 We can apply this same operation on
two different graphs $G'$ and $G''$, obtaining a graph $G$ with  $v(G')+v(G'')$ vertices and such that 
$u \sim v$ in $G$ if and only if either $u \sim v$ in $G'$ or $u \sim v$ in $G''$. 
Clearly $\vartheta(G)= \vartheta(G') + \vartheta(G'')$, 
and we also have the trivial result that if the Exclusivity principle singles out the quantum
maximum for $G'$ and $G''$ it also singles out the quantum maximum for $G$.

\begin{figure}[!h]
\centering
 \includegraphics[scale=0.6]{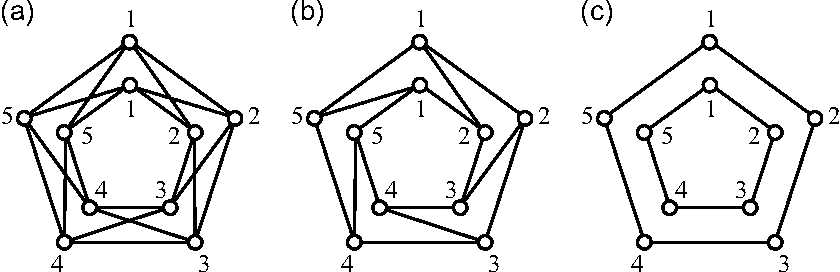}
\caption{(a) The twinning of $C_5$, the circulant graph $Ci_{10}(2,3)$. (b) A partial twinning of $C_5$, the circulant graph $Ci_{10}(2,5)$. (c) The duplication of $C_5$, the circulant graph $Ci_{10}(2)$.
\label{figtwinningpentagon}}
\end{figure}

\subsection{\textsf{Vertex-transitive graphs obtained from $C_5$}}
\label{Sec3}

Applying the operations above to $C_5$, for which the Exclusivity principle singles out the quantum maximum,  
and using the results from previous sections we can explain the quantum maximum for almost all 
vertex-transitive graphs with $10$ vertices.

Among the vertex-transitive graphs on $10$ vertices, only eight have $\vartheta(G) > \alpha(G)$, the circulant graphs 
$Ci_{10}(1,2,3,5)$, $Ci_{10}(1,4)$, $Ci_{10}(2,5)$, $Ci_{10}(2,3,5)$, 
$Ci_{10}(1,2,3)$, $Ci_{10}(1,2)$, and $Ci_{10}(1,2,5)$ and the Johnson graph $J(5,2)$ 
\cite{wikicirculantgraph, wikijohnsongraph}. 


\begin{prop}[Amaral and Cabello]
\label{teo10vertices}
The quantum maximum for the graphs $J(5,2)$, $Ci_{10}(1,2,3,5)$, $Ci_{10}(1,4)$, $Ci_{10}(2,5)$, $Ci_{10}(2,3,5)$ and 
$Ci_{10}(1,2,3)$ is the maximum value  allowed by the $E$ principle.

\end{prop}

\begin{dem}
Since $\vartheta(J(5,2))=\alpha^{*}(J(5,2))$, the Exclusivity principle by itself
explains the quantum maximum for this graph. The graph $Ci_{10}(1,2,3,5)$ is the direct cosum of $C_5$ with itself, 
$Ci_{10}(1,4)$ is the twinning of $C_5$, $Ci_{10}(2,5)$ is a partial twinning of $C_5$,
$Ci_{10}(2,3,5)$ is the complement of $Ci_{10}(1,4)$, and $Ci_{10}(1,2,3)$ is the complement of $Ci_{10}(2,5)$.
Hence, the result follows from proposition \ref{teotransitive} and corollaries \ref{cordirectsum} and \ref{cortwinning}.
\end{dem}

\begin{figure}[!h]
\centering
 \includegraphics[scale=0.7]{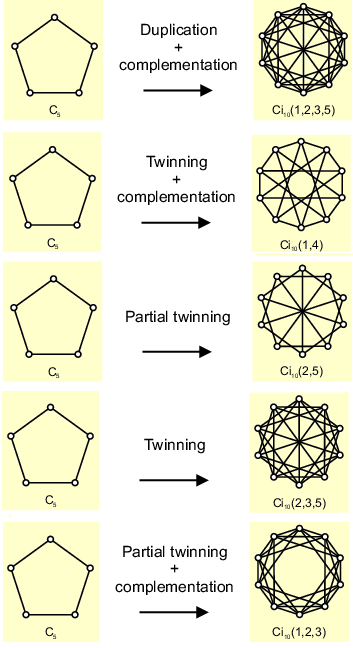}
\caption{Vertex transitive graphs  of theorem \ref{teo10vertices}.}
\label{fig10vertices}
\end{figure}


\section{\textsf{Final Remarks}}
\label{sectionfinallovasz}

In this chapter, we have shown that the Exclusivity principle is able to single out the quantum maximum and even the entire
set of quantum distributions in many different situations.  The results found so far are listed below.

\begin{enumerate}
 \item The Exclusivity principle directly explains the quantum maximum for all graphs with
 $\vartheta(G)=\alpha^*(G)$ \cite{CSW10};

\item Given the set of quantum distributions for $\overline{G}$, the Exclusivity principle explains the 
entire set of quantum correlations for $G$, as shown in proposition \ref{teoexclusivityquantumset} \cite{ATC14};

\item The Exclusivity  principle, applied to two copies of the graph, explains the entire set of 
quantum correlations for self-complementary graphs, including the pentagon, the simplest graph exhibiting 
quantum contextuality, as shown in proposition \ref{propselfcomp} \cite{Cabello13, ATC14};

\item Given the quantum maximum for $\overline{G}$, the Exclusivity principle explains the quantum maximum for any 
vertex-transitive graph $G$, as shown in proposition \ref{teotransitive}  \cite{ATC14};

\item The Exclusivity principle explains the quantum maximum for all vertex-transitive graphs with $10$ vertices, 
except $Ci_{10}(1,2)$ and $Ci_{10}(1,2,5)$,  as shown in proposition \ref{teo10vertices};

\item Either applied to two copies of the exclusivity graph of the  CHSH inequality \cite{FSABCLA12}
or of a simpler inequality \cite{Cabello13}, the E principle excludes Popescu-Rohrlich nonlocal boxes and provides 
an upper bound to the maximum violation of the CHSH inequality which is close to the Tsirelson bound (see Appendix 
\ref{chaptertsirelson};
 
 \item The Exclusivity principle rules out all extremal non-quantum distributions in the $(2,2,d)$
 Bell scenarios \cite{FSABCLA12};
 
 \item When applied to the OR product of an infinite number of copies, 
 there is strong numerical  evidence that the E principle singles out the maximum quantum violation of the
 noncontextuality inequalities whose exclusivity graph is
the complement of odd cycles on $n \ge 7$ vertices \cite{CDLP13}. Indeed, 
it might be also the case that, when applied to an infinite number of copies, 
the Exclusivity principle singles out the Tsirelson bound of the CHSH inequality \cite{FSABCLA12,Cabello13}.
 
\end{enumerate}

\begin{figure}[!h]
\centering
 \includegraphics[scale=0.5]{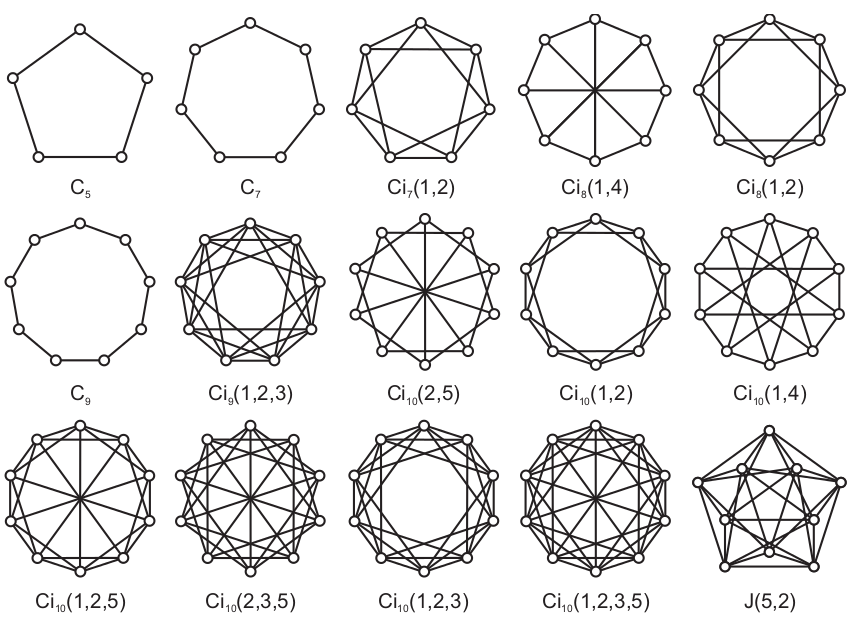}
\caption{Vertex-transitive graphs  with 10 vertices or less.}
\label{figtransitive}
\end{figure}

The simplest vertex-transitive graphs are shown in figure \ref{figtransitive}. The strengh of
the Exclusivity principle can be very well exemplified if we analyze what it predicts for those graphs.
For $G=C_5$, the Exclusivity
principle explains the entire set of quantum distributions. For $C_7$ and $C_9$, there are strong 
numerical evidences that 
it explains the quantum maximum\footnote{A. Cabello, private communication.}. If this is indeed the case, we can also explain the quantum maximum for 
$Ci_7(1,2)=\overline{C_7}$ and $Ci_9(1,2,3)=\overline{C_9}$. It might also be the case that the Exclusivity principle explains
the quantum maximum for  $Ci_8(1,4)$, the exclusivity graph of the CHSH inequality, and if this conjecture is true, it will
also explain the quantum maximum for $Ci_8(1,2)=\overline{Ci_8(1,4)}$. 


\chapter*{\textsf{To conclude, or not to conclude?}}
\addcontentsline{toc}{chapter}{To conclude, or not to conclude?}

This thesis is devoted to a mathematical presentation of some results in the quest for a principle that explains 
quantum contextuality.

The first two chapters are devoted to  setting of the ground in which we work. We define the generalized probability theories
we use to describe a physical system and discuss how contextuality arises naturally in this framework.  
We demand that the Exclusivity Principle be satisfied by all distributions. An open question, we would be happy to answer soon, 
is if there is a set of axioms we could impose on these theories that can guarantee that the E principle holds and still
be compatible with quantum theory.

The original results of the author and collaborators are the focus of chapter 
\ref{chapterlovasz}. In section \ref{sectioncomplement}, we describe the three main results of reference \cite{ATC14}.
Our first result shows that the E principle singles out the
set of the quantum correlations associated to any exclusivity graph assuming the set of quantum correlations
for the complementary graph. This result goes beyond
the one presented by Yan in \cite{Yan13}, since using the same
assumptions we have shown that the E principle singles
out the entire set of quantum correlations and not just
its maximum.

Our second result states that for self-complementary graphs, the E principle, by itself, excludes any set 
of correlations strictly larger than the
quantum set.This shows that the power of the E principle for singling out quantum correlations goes beyond
the power of any previously proposed principle. While
previous principles cannot rule out the existence of sets
of distributions strictly larger than the quantum set in any single scenario,
our results proves that this is indeed the case for many interesting 
ones, including the famous and important KCBS scenario.

Finally, we have shown that, assuming only the maximum for the complementary graph, the E principle singles out 
the quantum maximum for vertex-transitive graphs. This allows
 experimental tests discarding higher-than-quantum
distributions for this kind of dual experiment.  Interestingly, the CHSH Bell inequality is one of these cases.

Section \ref{sectionoperations} 
 is devoted to unpublished results concerning  graph operations other than complementation.
We use these  operations  to connect the quantum maximum  of different
graphs. With these connections, once we prove that the E principle singles out the quantum maximum for one graph,
we are able to conclude that it also does for many others. Using this idea with the pentagon we 
  show that the exclusivity principle explains the quantum maximum for all vertex-transitive graphs with
 $10$ vertices, except two. 
 If the E principle explains the quantum bound for one of them, the result of Yan \cite{Yan13} proves that the 
 E principle also explain the quantum bound for the other.

All these results still do not prove that the E principle
is \emph{the} principle for quantum correlations. However, what
is clear at this point is that the E principle has a surprising and unprecedented power for explaining many puzzling
predictions of quantum theory.

We have many plans for the near future. One of our priorities is to conclude our work with the graphs with $10$ vertices,
explaining the quantum bound for the remaining two, a problem that has been puzzling us for a long time.
We want to continue our search for the families with increasingly large $\frac{\vartheta}{\alpha}$ and find connections of 
this value with applications. We believe that there is a connection between this ratio and advantage of quantum strategies
over classical strategies in a game. The little story of the quantum gambler of subsection \ref{subsectiongambler} is an 
example, but we would like to find more sophisticated situations. We also believe that there may be a connection between
this ratio (or some other quantifier of contextuality) with amplification of randomness.

In summary, this thesis closes with some answers, and many questions.

\appendix
\chapter{\textsf{The impossibility of non-contextual hidden variable models}}
\label{chaptercontextuality}

In this chapter we will present a number of proofs of the impossibility of certain hidden-variable models
aiming to complete  quantum theory.  
We will show that with some very reasonable extra assumptions on these models, we get
a contradiction with the predictions of quantum theory. 

The first one to present such a proof was von Neumann, and we will discuss his
result in section \ref{sectionvonNeumann}. Several further developments were made, which culminated with the proof of 
the Bell-Kochen-Specker theorem,
which states the impossibility of \emph{noncontextual} hidden-variable models compatible with quantum theory.
We give a proof of this theorem using a lemma  by Gleason in section \ref{sectiongleason}, and the
Kochen-Specker original proof in section \ref{sectionks}. We present other simple proofs in section 
\ref{sectionothers}. A contextual hidden-variable model is given
 in section \ref{sectioncontextualmodel}. 



\vspace{-1em}

\section{\textsf{von Neumann}}
\label{sectionvonNeumann}

Von Neumann was the first to rigorously establish a mathematical formulation for quantum theory, 
published in his 1932 work \emph{Mathematische Grundlagen der Quantenmechanik}, and later 
translated to English in 1955 \cite{vonNeumann55}.
His rigorous approach permitted him also to challenge the ideas of completion of quantum theory.

He derived the quantum formula \eqref{eqmeanvalue}
$$\left\langle O\right\rangle = \tr\left(\rho O\right)$$
 for 
the expectation value of a measurement
from a few general assumptions about the expectation-value function. Then, from this formula we can prove that there is
no
dispersion-free state, and hence that hidden-variable models compatible with quantum theory are impossible. Although one of his assumptions was wrong, as we explain later, his result was
a landmark in  foundations of physics, since he opened the door for a series of papers disproving
the impossibility of this kind of completion.

\subsection{\textsf{von Neumann's assumptions}}

Given a specific type of system in a probability theory, every state defines an expectation-value function,
according to definition \ref{defiexpectation}:
$$\langle \ \rangle: \mathcal{M} \longrightarrow \mathbb{R}$$
where $\mathcal{M}$ stands for the set of measurements in the model. Instead on focusing on the possible states, von Neumann was interested in the properties of these functions, and stated a number of requirements
he believed were natural impositions on them.

 \begin{defi} An expectation value function $\langle \ \rangle: \mathcal{M}\longrightarrow \mathbb{R}$ is 
 \emph{dispersion-free}
 if 
 \be \langle \mathrm{M}^2  \rangle= \langle \mathrm{M}  \rangle^2.\ee
 for every measurement $\mathrm{M}$ allowed in the model.
 \end{defi}
 
 Dispersion-free functions are the ones that come from states in which the values of all measurements have definite
 values, that is, for every $\mathrm{M}$, one of the outcomes has probability one.

 \begin{defi} An expectation value function $\langle \ \rangle:\mathcal{M}\longrightarrow \mathbb{R}$ is called \emph{pure} if
 \be \langle \ \rangle = p\langle  \ \rangle' +(1-p)\langle \ \rangle'', \ 0< p< 1,\ee
 implies that $\langle \ \rangle=\langle \ \rangle'=\langle \ \rangle''$.
\end{defi}

Pure functions are the ones that can not be written as a convex sum of others and $\left\langle \ \right\rangle$ is pure
iff the state that defines it is a pure state of
the system. Every dispersion-free function is pure, but the converse is not aways true. For example, 
in quantum theory, pure functions are
the ones defined by one-dimensional projectors, while there is no dispersion-free function. In a hidden-variable model, the
two notions coincide.

In quantum theory, every measurement $\mathrm{M}$ is associated to an observable, a hermitian operator $O$ acting on the Hilbert space of
the system. Von Neumann's first assumption is that this correspondence is one-to-one and onto.

\begin{as}
\label{asmeasurementoperators}
 There is a bijective correspondence between measurements in a quantum system and hermitian operators acting on
 the Hilbert space of the system.
\end{as}

This is not always the case, since some systems are subjected to certain superselection rules, which forbid some
hermitian operators \cite{wikisuperselection}. Although this is not a general assumption, there are other formulations of 
von Neumann's result that
circumvent this difficulty (see \cite{CFS70} and references therein).

Suppose a given hidden-variable model is provided that completes quantum theory. The states of the system, now given by
quantum state plus hidden-variable,
define
expectation-value functions acting on the set of measurements in the system, which is, by assumption
\ref{asmeasurementoperators},
the set of hermitian operators acting on the Hilbert space $\mathcal{H}$ 
of the system  $O\left(\mathcal{H}\right)$. Then, every
state in the theory is associated with a expectation value function 
$$\langle \ \rangle: O\left(\mathcal{H}\right)\longrightarrow \mathbb{R}.$$

The next step in von Neumann's approach  was to impose a few assumptions on the functions $\langle \ \rangle$ that he believed 
 to be valid if these functions came from expectation values in a given state 
of a real physical system.

\begin{as}
\label{asvonneumann}
\begin{enumerate}
 \item If $\mathrm{M}$ is by nature non-negative, $\left\langle\mathrm{M}\right\rangle \geq 0$;
 
 \item If measurement $\mathrm{M}_1$ is associated to observable $O_1$ and
 $\mathrm{M}_2$ is associated to observable $O_2$, we can define 
 measurement $\mathrm{M}_1+\mathrm{M}_2$ 
and it is associated to observable $O_1+O_2$;
 
 \item If $\mathrm{M}_1, \mathrm{M}_2, \ldots$ are arbitrary measurements
 
 $$\left\langle a_1\mathrm{M}_1 + a_2\mathrm{M}_2 + \ldots\right\rangle =
 a_1\left\langle \mathrm{M}_1\right\rangle + a_2\left\langle\mathrm{M}_2 \right\rangle + \ldots$$
 that is, all expectation value functions are linear;\label{asvonneumannad}
 
 \item If measurement $\mathrm{M}$ is associated to observable $O$ and $f:\mathbb{R}\longrightarrow \mathbb{R}$
 is any real function\footnote{Measurement $f(\mathrm{M})$ is defined using the following rule: measure $\mathrm{M}$
 and apply $f$ to the outcome obtained. Observable $f(O)$ can be defined easily if we write $O$ is spectral decomposition.
 Let $O=\sum_i a_i\ket{v_i}\bra{v_i}$, where $\{\ket{v_i}\}$ is an orthonormal basis for the corresponding vector space.
 Then $f(O)=\sum_i f(a_i)\ket{v_i}\bra{v_i}$.}, the measurement $f\left(\mathrm{M}\right)$ is associated to observable $f(O)$.

\end{enumerate}
\end{as}

\begin{teo}
 Under assumptions \ref{asmeasurementoperators} and \ref{asvonneumann}, the expectation value functions in any theory completing quantum theory
 are given by
 \be \left\langle \mathrm{M}\right\rangle = \mbox{Tr}\left(O\rho\right),\ee
 where $O$ is the observable corresponding to measurement $\mathrm{M}$ and $\rho$ is a density operator
 that depends only on the function $\langle \  \rangle$ (and not on the particular measurement $\mathrm{M}$).
\end{teo}

This result implies that, as long as we impose all items of assumption \ref{asvonneumann} and 
\ref{asmeasurementoperators}, we can not 
circumvent the quantum rule for expectation 
values. As we already know, the pure functions of this form are the ones for which the associated density operator
is a one-dimensional projector $P$ and these functions only give dispersion-free expectation value for 
a small number of measurements, namely, the ones for which the subspace in which $P$ projects is an 
eigenspace of the associated observable. This in turn implies that there is no dispersion-free function, proving the 
impossibility of hidden-variable models compatible with quantum theory.

von Neumann's theorem had the support of many important physicists. For a long time,  
it was generally believed to demonstrate that no deterministic theory reproducing the statistical quantum
predictions  was possible. In 1966, 
J. Bell published a paper with some serious criticism to one of the requirements made for the expectation-value 
functions \cite{Bell66}.  
von Neumann required
them to be linear, which is the case for quantum theory, but there is no physical reason to impose this property for more
general theories. In fact, as von Neumann point out himself in reference \cite{vonNeumann55}, the sum of measurements 
$a_1\mathrm{M}_1 + a_2\mathrm{M}_2 + \ldots$ is completely meaningless when the measurements involved are 
not compatible, since there is no way of constructing, in general, the corresponding experimental set up to implement it.
Thus Bell argued that in the case of incompatible measurements, it is not reasonable to require that the expectation values
necessarily reflect the observables' algebraic relationships.

It is a special property of quantum theory that the sum of the corresponding observables corresponds to another 
allowed measurement (as long as assumption \ref{asmeasurementoperators} is valid), and the fact that the 
expectation value is linear  is a consequence of the mathematical 
rules of quantum theory and is not enforced by any general physical law.
In fact, it is not difficult to provide a hidden-variable model
agreeing with quantum theory for a qubit, which does not satisfy linearity of expectation values.

\begin{ex}[An example of hidden-variable model]
In reference \cite{Bell66}, Bell showed an example of a hidden-variable model for a qubit,
agreeing with quantum theory but violating von Neumann's assumption of linearity. Let  $A$ be an operator acting on $\mathbb{C}^2$.
Since the Pauli matrices $\sigma_i$ and the identity $I$ form a basis to the real vector space of $4 \times 4$ hermitian operators we 
can always write $A$ in the form
$$A=a_0I+a_1\sigma_x + a_2\sigma_y + a_3\sigma_z,$$
where $a_i \in \mathbb{R}$.

If we set $\ket{a}=(a_1,a_2,a_3)$, the eigenvalues of $A$, and hence the possible values of $v(A)$, can be written as
$$v(A)=a_0 \pm \norm{a}.$$
Let $\ket{\phi} \in \mathbb{C}^2$  and  $\ket{n}$ be the point on the Bloch sphere corresponding to  $\ket{\phi}$. Then, we have
$$\left\langle A \right\rangle=\sand{\phi}{A}{\phi}=a_0 + \braket{a}{n}.$$

Together with the quantum state $\ket{\phi}$, we will use  another vector $\ket{m}$ in the Bloch sphere to
represent the state of the system. This new vector plays the role of hidden variable in the model.
 The complete state of the system is then given by the pair $\de{\ket{\phi},\ket{m}}$, which specifies definite
 outcomes for every projective measurement according to the rule:
$$\left\{\begin{array}{cc}
v\de{A}=a_0 +\norm{a} & \mbox{if} \ \de{\ket{m}+\ket{n}}\cdot \ket{a}\geq 0,\\
v\de{A}=a_0 -\norm{a} & \mbox{if} \  \de{\ket{m}+\ket{n}}\cdot \ket{a}< 0,
\end{array}\right.$$
in which $v\de{A}$ is the value assigned to $A$ when the system is in the state $\de{\ket{\phi},\ket{m}}$.

It is not difficult to show that this model gives the quantum statistics when we average over the hidden variable
$\ket{m}$ using the uniform measure on the sphere $S^2$. Indeed,
$$\int_{S^2} v\de{A}\, d\ket{m}= \langle A \rangle, \quad \forall \, \ket{\phi}.$$
 
\end{ex}

\subsection{\textsf{Functionally closed sets and von Neumann's theorem}}

We can conclude from von Neumann's result  that it is not possible to reproduce the quantum statistics with hidden-variable models
that provide definite outcomes for all observables and at the same time give rise to linear expectation-value functions.
When dealing with hidden-variable models, the assumption that all measurements have well defined values is mandatory, and 
hence we are obligated to give up from the linearity assumption. At least from the mathematical point of view, it 
might be interesting to do the opposite  \cite{ZF98}.

Given a  quantum state $\rho$ of a system with Hilbert space $\mathcal{H}$, we will now try to
solve the following task: 

\begin{center}
Specify an extra variable and a set of observables for which it is possible to 
assign  definite values, in such a way that the quantum predictions for $\rho$ are
recovered when we average over all possible values of the extra variable.
\end{center}

von Neumann's 
result shows that this set can not be the entire set of operators acting in $\mathcal{H}$, if we assume linearity
of the expectation-value functions.

Let $D(\varrho)$ be the set of all definite-valued operators for a state $\varrho$ in some theory, where $\varrho$ corresponds to
quantum state $\rho$ and possibly an extra variable.
The operators one may include in this set depend
on what we use as a description of the state of the system. For example, if  the state is described accordingly 
only to quantum rules (that is, if there is no extra variable), an observable $O$ is in $D(\varrho)$ if and only if
the support of $\rho$ is included in one of the eigenspaces of $O$.
If the state of the system is provided by a hidden-variable model compatible with quantum theory,  $D(\varrho)$ includes
all hermitian operators acting on $\mathcal{H}$. What structure can we assume, \emph{a priori}, for the set $D(\varrho)$? 

To prove his theorem, von Neumann made two assumptions about this set when the states are given in a hidden-variable model:

\begin{enumerate}
 \item For every state $\varrho$ in the model, $D(\varrho)$ contains all observables acting on $\mathcal{H}$;\label{asDall}
 \item The value assigned to each measurement reflects the observables algebraic structure. This is the content of item
 \ref{asvonneumannad} of assumption \ref{asvonneumann}.\label{asDalgebraic}
\end{enumerate}

The criticism made to von Neumann's result is directed mainly to item number \ref{asDalgebraic}. Of course, since he was 
interested in ruling out hidden-variable models, item number \ref{asDall} was mandatory. When we demand both to be true at
the same time, we reach a contradiction. Bell found a way out von Neumann's impossibility proof
by trowing away requirement \ref{asDalgebraic}. 
We can do the same  giving up of item  \ref{asDall}
instead of item \ref{asDalgebraic}.

\begin{defi}
 We say that $A$ is  $*-$\emph{closed} if any hermitian  function\footnote{A hermitian function defined in the set of operators acting on a Hilbert space is a map that takes hermitian operators
to hermitian operators.} of operators in $A$
 is also in $A$.
\end{defi}

We will assume from now on that the set $D(\rho)$ is $*-$closed. 


\begin{defi}
 Let $A$ be a $*-$closed set of hermitian operators. A \emph{functional valuation} in $A$ is a map
 \begin{eqnarray*}
 \left\langle \  \right\rangle : A&\longrightarrow&  \mathbb{R}\\
 O&\longmapsto&\left\langle O\right\rangle
 \end{eqnarray*}
  which satisfies
  $$\lim_{n\rightarrow\infty} \left\langle F_n \right\rangle =\left\langle F\right\rangle$$
  whenever the sequence $F_n$ converges strongly\footnote{If $F_n (x) \to F(x)$ for all $x$ in $\mathcal{H}$,  we say that the sequence of operators $F_n$  converges strongly to $F$.}   to $F$.
\end{defi}

This is a much stronger assumption than what von Neumann demands from his expectation-value functions. 
Von Neumann assumed these functions respect \emph{linear} relationships among the operators, while here
we demand that these functions respect \emph{arbitrary} functional relationships among the operators.

\begin{teo}
 Let $D$ be a  $*-$closed set of definite-valued operators, $d$ the set of projectors contained in $D$ and $\rho$ a 
 density matrix. The following are equivalent:
 \begin{enumerate}
  \item There is a probability measure $\mu$ defined in the set of all functional valuations  
  $$\left\langle \ \right\rangle : D\longrightarrow  \mathbb{R}$$
  such that for all set of compatible operators $O_1, \ldots, O_n \in D$ 
  $$p\left(o_1, \ldots, o_n|O_1, \ldots, O_n\right)=\mu\left(\left\{\left\langle \ \right\rangle; 
\  \left\langle   O_i\right\rangle=o_i 
 \ \forall i\right\}\right)$$
  where $p\left(o_1, \ldots, o_n|O_1, \ldots, O_n\right)$ is the probability of obtaining outcome $o_i$ in a measurement 
  of $O_i$ in state $\rho$.
  
  \item $D$ is a $I$-quasiBoolean algebra, where $I=\{P \in d; P\rho=0\}$.
 \end{enumerate}

\end{teo}

This means that when $D$ is a $I$-quasiBoolean algebra it is possible to attribute definite values to its elements in such a way that
we recover the quantum predictions when averaging over all possible valuations. Moreover, this attribution is made in
such a way that all functional relations among the observables are preserved \cite{ZF98}.

This shows that there is another way around von Neumann's result. Instead of questioning, like Bell did, the requirement of
linearity of the definite values attributed to the measurements, we drop the assumption that all observable must receive a 
definite value. Then the theorem above shows that we can actually strengthen the assumption of linearity, requiring that 
all functional relations be preserved, and we still can recover the quantum statistics.

We may ask now what if this result has any physical interest. Clearly it can not be used to rule out hidden-variable
theories, since this requires that all measurements  have definite values. Nonetheless, 
this result is connected to a  family of realist interpretations of quantum theory. Each of them   supplies a 
\emph{rule of definite-value ascription},
which picks out, from the set of all observables of a quantum system, the subset of definite-valued observables.
This family is known as \emph{modal interpretations of quantum theory} \cite{modal}.

\section{\textsf{Gleason's lemma}}
\label{sectiongleason}

In reference \cite{Gleason57}, Gleason proves his famous theorem, a mathematical result which is of particular importance 
for the 
field of quantum logic. It proves that the  quantum rule for calculating the probability of obtaining  specific results of
a given measurement
follows naturally from the structure  of events in a real or complex Hilbert space.  Although Gleason's main result is 
motivated by a problem in  foundations of quantum theory, his objective had in principle 
nothing to do with hidden variables, which are not even mentioned in his paper. Nevertheless, his work was of huge
importance to discard the possibility of certain hidden-variable models and its free of certain drawbacks 
present in von Neumann's assumptions.

Gleason's main interest was to determine all measures on the set of subspaces of a Hilbert space. 

\begin{defi}
 A \emph{measure} in the set {\Fontlukas S} of subspaces of a Hilbert space $\mathcal{H}$ is a function
 \be \mu: \mbox{\Fontlukas S} \longrightarrow [0,1]\ee
 such that $\mu\left(\mathcal{H}\right)=1$ and such that 
 if $\{S_1, \ldots, S_n\}$ is a collection of mutually orthogonal subspaces spanning the subspace $S$
 \be \mu\left(S\right)=\sum_{i=1}^{n}\mu\left(S_i\right).\ee
\end{defi}
 
 \begin{ex}
  To every density operator acting on $\mathcal{H}$ corresponds a measure $\mu_\rho$ in {\Fontlukas S} defined by
  \be \mu_\rho\left(S\right)=\mbox{Tr}\left(\rho P_S\right)\label{eqmeasures}\ee
where $P_S$ is the projector onto $S$.
    \end{ex}

    Gleason's main result states  that all measures on {\Fontlukas S} are of the form \ref{eqmeasures}, if the dimension of
    $\mathcal{H}$ is at least three.

    \begin{defi}
     A \emph{frame function of weight} $W$ for a Hilbert space $\mathcal{H}$ is a real-valued function 
     \be f: \mathcal{E} \longrightarrow \mathbb{R} \ee
     where $\mathcal{E}$ is the unit sphere in $\mathcal{H}$, such that if $x_1, \ldots, x_n$ is a an orthonormal basis
     for $\mathcal{H}$ then
     $$ \sum_i f(x_i)=W.$$
    \end{defi}

   Given a non-negative frame function with weight $W=1$, we can define a measure on {\Fontlukas S}. For every 
  one-dimensional subspace $S$ of $\mathcal{H}$, we define $\mu(P)=f(x)$, where $P$ is the projector over $S$
  and $\ket{x}$ is a unit vector belonging to $S$.
    
    \begin{defi}
     A frame function is said to be \emph{regular} if there exists a hermitian operator $T$ acting on
     $\mathcal{H}$ such that 
     $$f(x)= \left\langle x|T|x\right\rangle$$
     for all $x \ \in \ \mathcal{E}$.
    \end{defi}

    Before stating his main theorem, Gleason proves several intermediate lemmas, among which is the following:
    
    \begin{lemma}
    \label{lemmagleason}
     Every non-negative frame function on either a real or complex Hilbert space of dimension at least three is regular.
    \end{lemma}
    
    As a consequence of this lemma, we have Gleason's main result:

    \begin{teo}
     Let $\mu$ be a measure on the set  {\Fontlukas S} of subspaces of a Hilbert space $\mathcal{H}$ of dimension at least
     three. Then there exists a density matrix $\rho$ such that $\mu=\mu_{\rho}.$
     \label{propgleason}
    \end{teo}

    The consequences of Gleason's theorem to the foundations of quantum theory appear clearly  if one notice
    that we can interpret the measure defined not on the set of subspaces, but on the set of corresponding
    orthogonal projectors. Every projector acting on $\mathcal{H}$ corresponds to an outcome of a measurement in the 
    corresponding quantum system, and hence a measure on {\Fontlukas S} defines a way of calculating the probabilities of
    these outcomes. What theorem \ref{propgleason} states is that the only way of defining these probabilities 
   consistently is through the quantum rule using density matrices.
    
    This is certainly a really interesting fact, but for us the most important statement in Gleason's paper is lemma 
    \ref{lemmagleason}. This result implies that all measures on {\Fontlukas S} are \emph{continuous}, and
   this discards the possibility of certain hidden-variable models.
   
\subsection{\textsf{Using Gleason's Lemma to discard hidden-variable models}}   
   
Let $\lambda$ define a dispersion-free state in a hidden-variable model compatible with quantum theory describing
a system whose associated Hilbert space has dimension at least three. Then,  every one-dimensional projector
$P$ has a well defined outcome for $\lambda$ and hence we can define a measure
\be \mu_{\lambda}: \mathcal{E} \longrightarrow \{0,1\}\ee
that takes each vector in  $\mathcal{E}$  to the value associated to the projector in this direction by $\lambda$.
As a consequence of lemma \ref{lemmagleason}, this measure is continuous and hence it has to be a constant function. 

To see 
that this is really the case, we can translate the problem of assigning values to the points of the sphere to a 
problem of coloring the sphere: if the value associated to an one-dimensional projector is $1$, we
paint the corresponding unit vectors in red; if the associated value is $0$, we paint the vectors in green.
Suppose now that there are two vectors with different colors. Then, if we choose a path between the corresponding 
points in the sphere, we have to change abruptly from red to green somewhere in the way from one point to the other.
Hence, the association can not be done continuously if we use both colors.

Since all associations are constant and we
know that, given a pure quantum state,  there is at least one unidimensional projector
with definite outcome $1$. We conclude that  for all states in the hidden-variable model and for all
one-dimensional projectors the  associated definite value is $1$. This clearly can not reproduce the statistics of quantum theory.

At first sight, one may think that the argument above puts an end to the discussion on the possibility
of hidden-variable models completing quantum theory: it just can not be done. Although very compelling,  there
is one extra assumption on the kind of hidden-variable considered that was not explicitly mentioned. This extra assumption seemed
so natural that one may not even realize it is there. Hence, the reasoning above is not enough to discard all kinds of hidden 
variable models. It proves only that  \emph{noncontextual} models are ruled out. 

\subsection{\textsf{The ``hidden'' assumption of noncontextuality}}

The implicit assumption made in the preceding argument is  such that the hidden-variable models considered are not general
enough, and 
hence the argument can not be used to rule out completely the possibility of completing quantum theory.
It was tacitly assumed that the measurement of an observable must yield the same outcome, regardless of what other
compatible measurements can be made simultaneously. This is the hypothesis of noncontextuality discussed in 
section \ref{sectioncontextuality}.

With these observations, we can conclude as a corollary of Gleason's lemma the following result:

\begin{teo}[Kochen-Specker]
\label{teoKS}
 There is no noncontextual hidden-variable model compatible with quantum theory.
\end{teo}

Although this result follows from Gleason's lemma, as  we proved above,  this fact was noticed only after 
it was proved by other means by Kochen and Specker. The advantage of Kochen and Specker proof is that,
contrary to Gleason's lemma,
it uses only a finite number of projectors.

\section{\textsf{Kochen and Specker's proof}}
\label{sectionks}

Suppose  a  hidden-variable model completing quantum theory is given. 
If we fix a quantum state for the system and if we also fix the hidden variable, all observables are assigned
a definite value. We will denote this value for observable $O$ by $v(O)$. We will deal only with observables whose associated
operators are one dimensional projectors, since they are enough to get a contradiction and prove the desired result.

The fact that the hidden-variable models must be compatible with quantum theory, the value $v(P)$ assigned to a projector 
$P$ must be one of its eigenvalues, and hence we have
\be v(P) \in \{0, 1\}. \label{eqvalueP}\ee

We also require that the assignment $v$ preserves the algebraic relations among \emph{compatible} operators, and hence, 
if $P_1, \ldots, P_n$ are orthogonal projectors such that $\sum_iP_i=I$ we have

\be \sum_iv(P_i)=1 \label{eqvaluecompatible}\ee

This means that whenever a set of vectors $\ket{\phi_i}$ is a basis for $\mathcal{H}$, $v(P_i)=1$ for one, and only one $i$,
 where  $P_i=\ket{\phi_i}\bra{\phi_i}$ is the corresponding projector. 
 
 Although $v$ comes form a hidden-variable model, and hence is defined in the set of observables in a quantum system,
 we will use the fact that we are restricted to the set of one dimensional projectors and consider $v$ as function assigning values to
 either the one dimensional projectors acting on $\mathcal{H}$ or   unit vectors in $\mathcal{H}$. 
 If $P=\ket{\phi}\bra{\phi}$, the value of $v$ in both 
 $P$ and $\ket{v}$ is the same 
 $$v(P)=v\left(\ket{\phi}\right).$$
 
The idea behind Kochen and Specker's proof is to find a set of vectors in such a way that is impossible to assign definite values to 
the corresponding projectors 
obeying \eqref{eqvalueP} and \eqref{eqvaluecompatible}. This proves the impossibility of noncontextual hidden-variable models
completing quantum theory.

\begin{defi}
 A \emph{definite prediction set} of vectors (DPS) is a set $A=\{r_1, \ldots, r_n\}$ of unit vectors in a Hilbert space
 $\mathcal{H}$ such that at least for one choice of assignment for some $r_i$ the value of some other $r_j$ is
 determined by \eqref{eqvalueP} and \eqref{eqvaluecompatible}.
\end{defi}

Such a set may be represented with a graph, usually called Kochen-Specker diagram. The vertices of the graph correspond to 
the vectors in the set and two vertices are connected by an edge if the corresponding vectors are orthogonal. In this 
representation, the problem of assigning values to the projectors can be translated into a problem of coloring the vertices of 
the graph. If a hidden-variable model assigns value $1$ to the  projector  we paint the 
corresponding vertex in red. If the model assigns value $0$ we paint the vertex in green. Notice that the painting is 
independent of other compatible measurements performed simultaneously, which is the assumption of noncontextuality of 
the model.

Equation \eqref{eqvaluecompatible} implies a rule for the coloring: in a set of mutually orthogonal vectors, at most one
can be red; if a set of vectors is a orthogonal basis for $\mathcal{H}$, one, and only one of them is red.

The DPS used in Kochen Specker proof is composed of three dimensional vectors, with associated  diagram shown
in figure \ref{figKS8}.  Such a set is called a KS-8 set.

\begin{figure}[h]
 \includegraphics[scale=0.15]{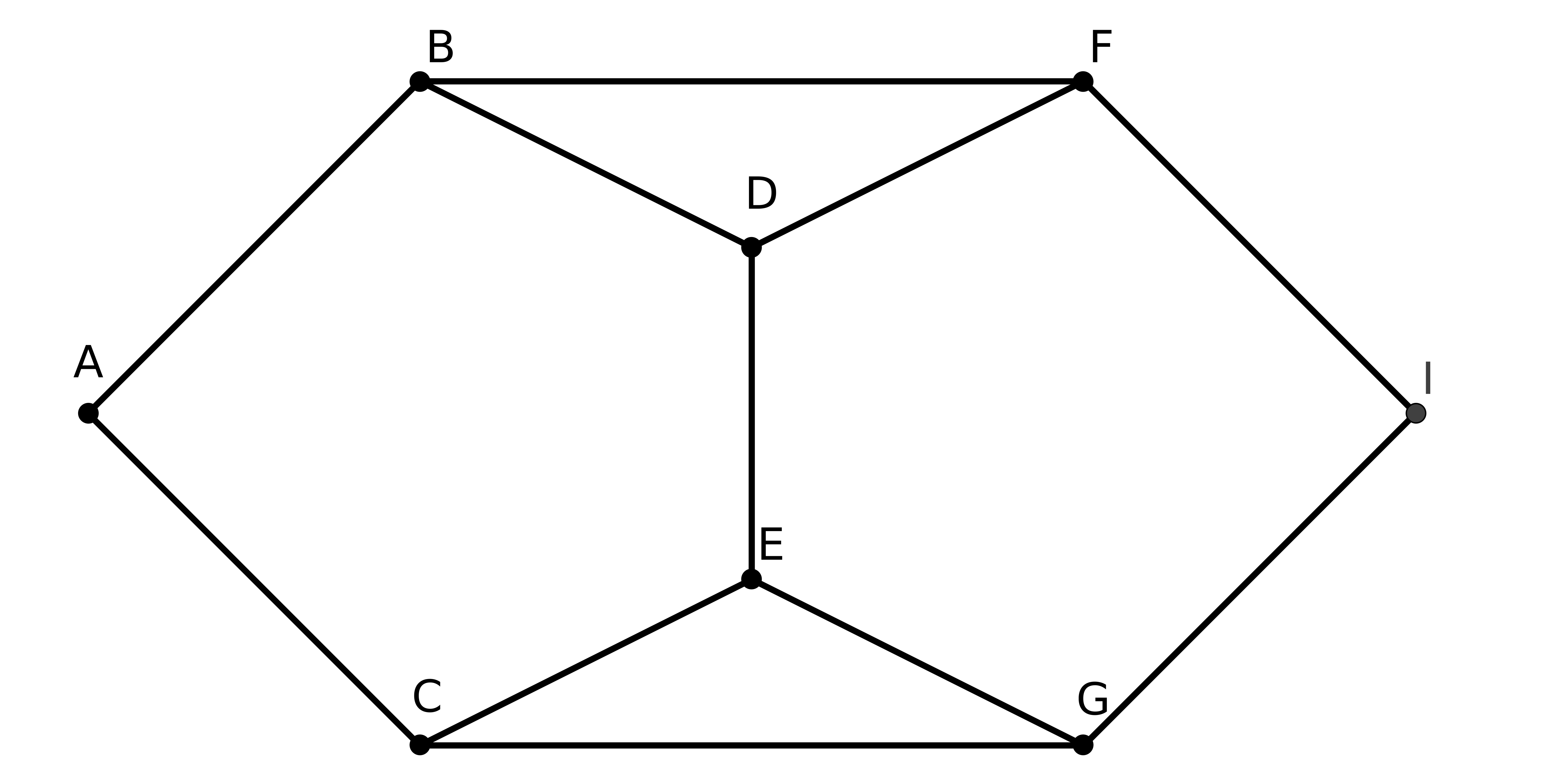}
 \caption{The set KS-8, a DPS used in the original proof of Kochen-Specker theorem.}
 \label{figKS8}
\end{figure}

\begin{teo}
 The set KS-8 is a DPS.
\end{teo}

\begin{proof}
If vector $A$ is red, $B$ and $C$ must necessarily be green. If $H$ is red, $F$ and $G$ are
necessarily green. Since the vectors belong to a  three dimensional space, $D$ and $E$ are necessarily red, which is a contradiction since 
$D$ and $E$ can not be red at the same time. Hence,
$$A=1 \  \Longrightarrow  \  H=0.$$
\end{proof}

A KS-8 can be constructed using the following vectors in three-dimensional space:

$$\begin{array}{ll}
   A=\left(\begin{array}{ccc} 1&0&0\end{array}\right)&E=\left(\begin{array}{ccc} 0&\cos(\beta)&\sin(\beta)\end{array}\right)\\
   B=\left(\begin{array}{ccc} 0&\cos(\alpha)&\sin(\alpha)\end{array}\right)&F=\left(\begin{array}{ccc} \cot(\phi)&1&-\cot(\beta)\end{array}\right)\\
   C=\left(\begin{array}{ccc} \cot(\phi)&1&\cot(\alpha)\end{array}\right)&G=\left(\begin{array}{ccc} \tan(\phi)\cosec(\beta)&-\sin(\beta)&\cos(\beta)\end{array}\right)\\
   D=\left(\begin{array}{ccc} \tan(\phi)\cosec(\alpha)&-\sin(\alpha)&\cos(\alpha)\end{array}\right)&H=\left(\begin{array}{ccc} \sin(\phi)&-\cos(\phi)&0\end{array}\right).
  \end{array}$$

Adding  two more vectors we get another DPS, called KS-10, whose diagram is shown in figure \ref{figKS10}.
In a KS-10, if  $A$ is red, $J$ must necessarily be red. In fact, $v(A)=1 \ \Rightarrow \ v(I)=0\ \mbox{and} \ v(H)=0$. 
Since 
every time we have three mutually orthogonal vectors one of them must be assigned the value $1$, we have $v(J)=1$.
This set is obtained if we use the vectors in KS-8 plus $I=\left(\begin{array}{ccc} 0&0&1\end{array}\right)$ and 
$J=\left(\begin{array}{ccc} \cos(\phi)&\sin(\phi)&0\end{array}\right)$.

\begin{figure}[h]
 \includegraphics[scale=0.7]{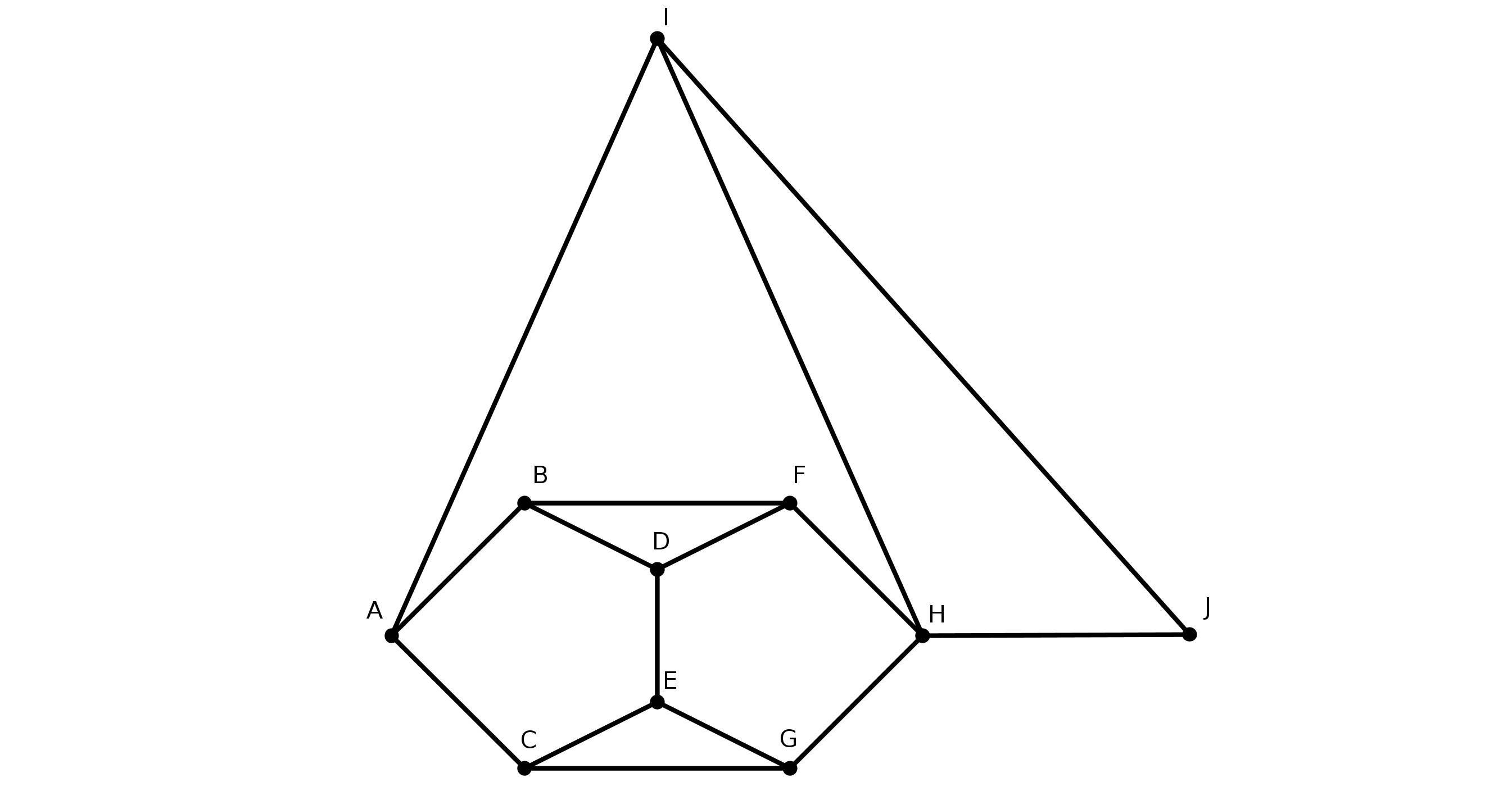}
 \caption{The set KS-10, a DPS used in the original proof of the Kochen-Specker theorem.}
 \label{figKS10}
\end{figure}

\begin{defi}
 A  set  of vectors $A=\{r_1, \ldots, r_n\}$ is called  a \emph{partially no-colorable set} (PNS) if there is at least
 one choice of  assignment to some $r_i$ that makes the assignment of values to the other vectors according to rules
 \eqref{eqvalueP} and \eqref{eqvaluecompatible} impossible.
\end{defi}

To get a PNS we concatenate five diagrams like KS-10, which results in a set of vectors with 
Kochen-Specker diagram as in figure \ref{figKS42}, called KS-42. For such a set,
the assignment of value $1$ to $A$ is impossible.
In fact, $$v(A)=1 \Rightarrow v(A_1)=1 \Rightarrow v(A_2)=1\Rightarrow v(A_3)=1\Rightarrow v(A_4)=1 \Rightarrow v(J)=1,$$ 
but $A$ and $J$ are orthogonal and hence can not be both red.

\begin{center}
\begin{figure}[h]
\centering
 \includegraphics[scale=0.5]{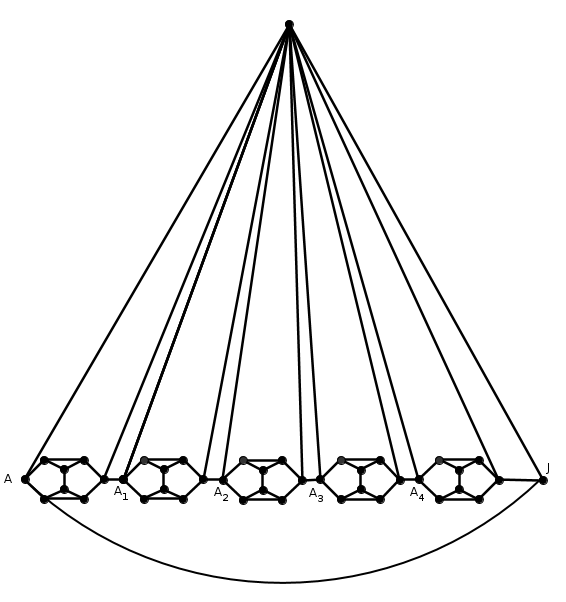}
 \caption{The set KS-42, a PNS used in the original proof of the Kochen-Specker theorem.}
 \label{figKS42}
\end{figure}
\end{center}

\begin{defi}
 A set of vectors is called a \emph{totally non-colorable set} (TNS) if it is impossible to assign definite values to 
 all vectors according to rules \eqref{eqvalueP} and \eqref{eqvaluecompatible}.
\end{defi}

A TNS provides a proof of the Kochen-Specker theorem \ref{teoKS}. 
In fact, a hidden-variable model compatible with quantum theory must assign values to 
all projectors (or equivalently, to the corresponding unit vectors) in such a way that equations
\eqref{eqvalueP} and \eqref{eqvaluecompatible} must be obeyed. Hence, if we find a TNS 
we prove that noncontextual hidden-variable models compatible with quantum theory are impossible.

The sphere in any Hilbert space with dimension at least three is  a TNS, as we have proven 
as a corollary of Gleason's lemma. Using three KS-42 sets we can build a TNS with a finite number of 
vectors in dimension three, simplifying the proof of theorem \ref{teoKS}. This set is shown in figure 
\ref{figKS117}.

A set of vectors with Kochen-Specker diagram as in figure \ref{figKS117} is called KS-117. This is the set used by
Kochen and Specker in their proof of theorem \ref{teoKS}.

\begin{figure}[h]
\centering
 \includegraphics[scale=0.8]{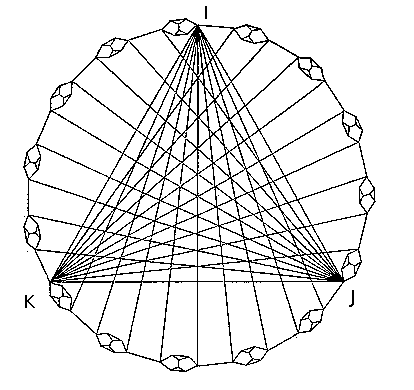}
 \caption{The set KS-117, a TNS used in the original proof of the Kochen-Specker theorem.}
 \label{figKS117}
\end{figure}

\begin{teo}
 It is impossible to assign definite values to the vectors of a KS-117 set according to equations
 \eqref{eqvalueP} and \ref{eqvaluecompatible}.
\end{teo}

\begin{dem} The proof  is quite simple. We just have to notice that the vectors
$I$, $J$ and $K$ can not be assigned the value $1$, since they are the first vector of a KS-42 set. But they are
mutually orthogonal, and hence one of them should be $1$ according to equation \eqref{eqvaluecompatible}.
\end{dem}

The hard part of the proof is to show that there is a set of vectors in a Hilbert space of dimension 
three with this Kochen-Specker diagram. The details can be found in references \cite{KS67,Cabello96}.

\section{\textsf{Other additive proofs of the Kochen-Specker theorem}}
\label{sectionothers}

\subsection{\textsf{P-33}}

One of the simplest proofs of the Kochen-Specker theorem uses a TNS with $33$ vectors in a Hilbert space
of dimension three \cite{Peres91}.  This TNS is known as P-33.

To simplify the notation, let $m=-1$  and $s=\sqrt{2}$. The vectors in P-33 are
$$\de{1,0,0},\ \ \de{0,1, 1}, \ \ \de{0,1, s}, \ \  \de{s,1,1}, $$
$$\de{0,m, 1}, \ \ \de{0, m, s}, \ \ \de{s,m,1}, \ \ \de{s,m,m},$$
and all others obtained from these by relevant permutations of the coordinates. By relevant we mean any permutation that 
generates a vector in a different one dimensional subspace, since what is important for the proof is the projector on the 
one dimensional subspace and not the vector itself.


The set above has an important property: it is invariant under permutations of the axis and  by
a change of orientation in each axis. This allows us to assign value $1$ to some vectors arbitrarily,
since a different choice is equivalent to this one by an operation that leaves P-33 invariant.

The table below shows the proof that  P-33 is a TNS. To simplify the notation even further, we drop the 
parenthesis in the notation of a vector and use just $abc$ to represent the vector $\de{a,b,c}$. In the table, the vectors in
each line are mutually orthogonal. 
The vectors in the first column are assigned the value $1$, and hence the other  vectors in the same line are
assigned the value $0$.  The assignment of $1$ to the vector in the first column is explained in the last column.

\vspace{2em}

\small{\begin{tabular}[h]{rrrrrl}
Trio&   & \vline& Vectors  &   \hspace{-1em} $\bot$  to the $1^\circ$ \vline &  Explanation \\
$\mathbf{001}$ & $100$& $010$ \vline &  $110$ &  $1m0$ \vline &  Arbitrary choice of axis $z$\\
$\mathbf{101}$& $m01$&$010$ \vline& & \vline&Arbitrary choice of orientation in axis $x$\\
$\mathbf{011}$& $0m1$&$100$ \vline& & \vline &Arbitrary choice of orientation in axis $y$\\
$\mathbf{1ms}$& $m1s$&$110$ \vline&$s0m$ &$0s1$ \vline&Arbitrary choice between $x$ and $y$\\
$\mathbf{10s}$& $s0m$&$010$ \vline&$smm$ & \vline& $2^\circ$ and $3^\circ$ are zero\\
$\mathbf{s11}$& $01m$&$smm$ \vline&$m0s$ & \vline& $2^\circ$ and $3^\circ$ are zero\\
$\mathbf{s01}$& $010$&$10s$ \vline&$mms$ & \vline& $2^\circ$ and $3^\circ$ are zero\\
$\mathbf{11s}$& $1m0$&$11s$ \vline&$0sm$ & \vline& $2^\circ$ and $3^\circ$ are zero\\
$\mathbf{01s}$& $100$&$0sm$ \vline&$1s1$ & \vline& $2^\circ$ and $3^\circ$ are zero\\
$\mathbf{1s1}$& $10m$&$0sm$ \vline&$msm$ & \vline& $2^\circ$ and $3^\circ$ are zero\\
$\mathbf{100}$& $0s1$&$01s$ \vline & &  \vline &CONTRADICTION.\\
\end{tabular}}

\vspace{2em}

We get a contradiction in the last line: we have to assign value $1$ to $100$, but it is already assigned 
value $0$ in the first line.

In the table we used only $25$ vectors, but we can not discard the other $8$ because we need them to repeat the argument
with different choices of the first vector in the first four lines. If we use only the $25$ vectors that appear in the table
we would not have a set  invariant under permutations of the axis  and by
change of orientation in each axis, and the set of vectors would not be a TNS.

\subsection{\textsf{Cabello's proof with 18 vectors}}
\label{sectioncabello18}

In 1996, another simple proof of the KS theorem with 18 vectors in a four  dimensional space was found by Cabello \emph{et. al.}
\cite{Cabello96a}. It
was the world record at the time. The TNS in this proof is shown in figure \ref{figcabello18}. Once more, we drop the 
brackets in the vectors to simplify the notation and use $m=-1$.

\begin{figure}[h]
\centering
 \includegraphics[scale=0.43]{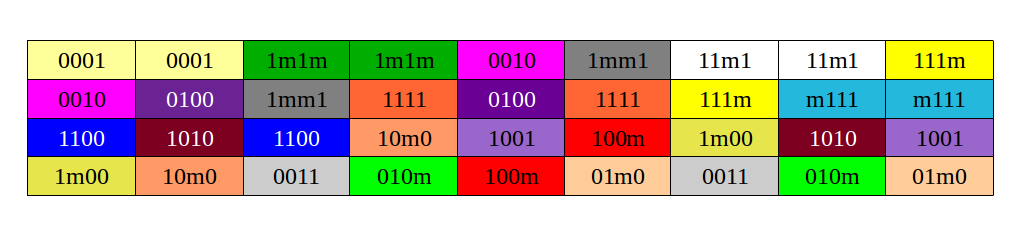}
 \caption{The set used in Cabello's  proof of the Kochen-Specker theorem using 18 vectors.}
 \label{figcabello18}
\end{figure}

In the table bellow, the vectors in each column are orthogonal. Cells that contain the same vector have the same color.
Since we have nine columns, nine different cells, and only nine, can be assigned the value $1$, one for each column. If the 
assignment is noncontextual, cells with the same color must be assigned the same value. To see the contradiction, we just 
notice that  the number of cells with the same color is 2, and hence the number of cells assigned the value 1 must be even.

\subsection{\textsf{The simplest proof of the Kochen-Specker theorem}}

Any TNS shown above provides a proof for the Kochen-Specker theorem
and the impossibility of noncontextual hidden
variable models is established. Nevertheless, from a physical point of view, there is still a lot of work to be done. 
The validity of the theorem should be experimentally verified, and hence people started to work on 
experimental implementations of such proofs \cite{TKLSK13}.

The need of an experimental verification of this result  is what makes the improvement made by Kochen and Specker's original
proof so important: in 
Gleason's proof, we need an infinite number of vectors to reach a contradiction, and this, of course, makes any experimental
test of the result impossible. In the original proof of Kochen and Specker the set of vectors used is finite, but it is really 
big. Any experimental arrangement involving 117 measurements is really hard to implement with small error. 

Many proofs where derived after Kochen-Specker work, with the objective of simplifying the TNS used.
Among the additive proofs (those relying on equation \eqref{eqvaluecompatible}),
the proof presented in section \ref{sectioncabello18} is still the world record for smallest number of vectors in the set.
But a proof with few vectors is not necessarily the simplest proof for an experimentalist. The number of different measurement
setups is related to the number of contexts, and hence it might be better in some situations to seek for a 
set with the smallest number of contexts. In this sense, the simplest proof known was presented in references \cite{LBPC13}.
The 21 vectors used are shown in figure \ref{figcabello21a}. The Kochen-Specker diagram of this set is shown in figure
\ref{figcabello21}.

\begin{figure}[h]
\centering
 \includegraphics[scale=0.43]{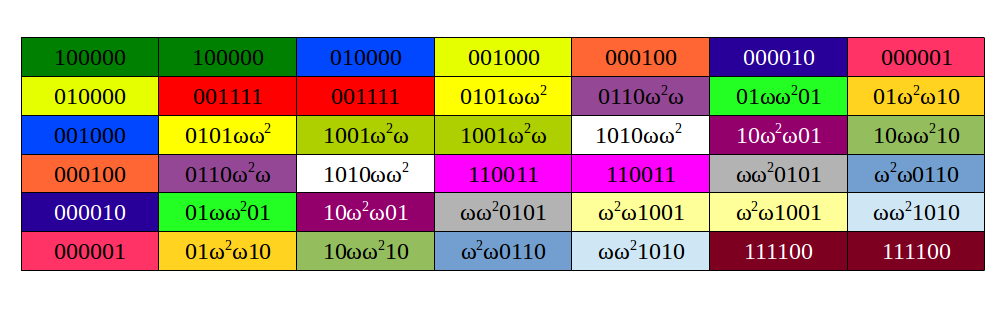}
 \caption{The set used in the simplest  proof of the Kochen-Specker theorem using 21 vectors and 7 contexts.}
 \label{figcabello21a}
\end{figure}

\begin{figure}[h]
\centering
 \includegraphics[scale=0.5]{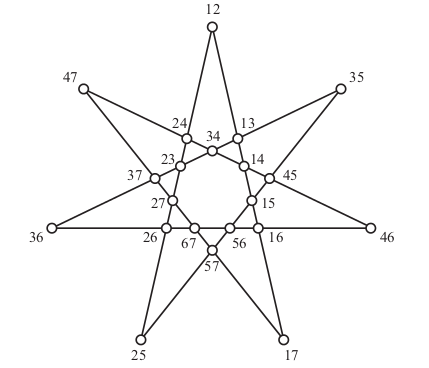}
 \caption{The Kochen-Specker diagram of the TNS. Vector labeled by $ij$ is the vector common to  $i$-th and $j$-th basis.}
 \label{figcabello21}
\end{figure}

In the table of figure \ref{figcabello21a}, the vectors in each column are orthogonal. 
Cells that contain the same vector have the same color.
Since we have seven columns, seven different cells, and only seven, can be assigned the value $1$, one for each column. If the 
assignment is noncontextual, cells with the same color must be assigned the same value. To see the contradiction, once more
we
notice that  the number of cells with the same color is 2, and hence the number of cells assigned the value $1$ must be even.

\subsection{\textsf{Multiplicative proofs of the Kochen-Specker theorem}}

In the previous proofs of the Kochen-Specker theorem, we have used the sum of compatible operators and the fact that
the values assigned by a hidden-variable  model to the observables should obey the same linear relations the 
corresponding operators did. More generally, we can assume that, for compatible operators, the validity of 
$$f(A_1,\ldots, A_n)=0$$
implies that
$$f(v(A_1),\ldots,v( A_n))=0,$$
for any function $f$.

This allows the construction of proofs of the Kochen-Specker theorem with different functions $f$. Examples of
such  proofs are the multiplicative ones we will discuss bellow. In this kind of argument, we use the fact that a set of
compatible operators obey the relation 
$$A_1\times \ldots \times  A_n=B$$
to impose the condition 
\be v(A_1)\times \ldots \times v( A_n) =v(B). \label{eqKSmultiplicative}\ee

\subsubsection{\textsf{The Peres Mermin square}}

A simple multiplicative proof of the Kochen-Specker theorem uses the set of operators known as the 
\emph{Peres-Mermim square} \cite{Mermin90, Peres90}:
\begin{equation}
\begin{array}{ccc}
A_1=\sigma_x \otimes I & A_2=I \otimes  \sigma_x &A_3=\sigma_x \otimes \sigma_x\\
A_4=I\otimes \sigma_y & A_5= \sigma_y \otimes I&A_6=\sigma_y \otimes \sigma_y\\
A_7=\sigma_x \otimes \sigma_y & A_8=\sigma_y \otimes  \sigma_x &A_9=\sigma_z \otimes \sigma_z.\\
\end{array}
\end{equation}

It is not possible to assign definite values $v(A_i)$ to all of these observables in such a way that 
the value assigned to each operator is one of its eigenvalues and \eqref{eqKSmultiplicative} is satisfied.
This happens because this set of operators has the following properties:

\begin{enumerate}
\item The three operator in each line and in each column are compatible;

\item \label{prodop} The product of the operators in the last column is $-I$. 
The product of the operators in the other columns and in all lines is $I$.
    \end{enumerate}

    Using equation \eqref{eqKSmultiplicative}, we have
  \begin{eqnarray}
  P_1= & v(A_1)v(A_2)v(A_3)&=1 \nonumber\\
  P_2= & v(A_4)v(A_5)v(A_6)&=1 \nonumber\\
  P_3= & v(A_7)v(A_8)v(A_9)&=1 \nonumber\\
  P_4= & v(A_1)v(A_4)v(A_7)&=1 \nonumber\\
  P_5= & v(A_2)v(A_5)v(A_8)&=1 \nonumber\\
   P_6= & v(A_3)v(A_6)v(A_9)&=-1
   \label{eqperesmermim}
       \end{eqnarray}

 and hence
  $$1=P_1P_2P_3=P_4P_5P_6=-1$$
  which is a contradiction. This proves that the Peres-Mermim square provides a multiplicative proof of the 
  Kochen-Specker theorem. The assumption of noncontextuality appears clearly in equations \eqref{eqperesmermim} since we
  assumed that each observable has the same value independently if it is measured together with the other compatible
  observables appearing in the same line or in the same column.

 \subsubsection{\textsf{A simple proof in dimension 8}}

Another simple multiplicative proof of the Kochen-Specker theorem is provided by the set of vectors


$$\begin{array}{ll}
A_1=\sigma_y \otimes I \otimes I& A_2=\sigma_x\otimes \sigma_x\otimes \sigma_x \\
A_3=\sigma_y\otimes \sigma_y\otimes \sigma_x &
A_4=\sigma_y\otimes \sigma_x\otimes \sigma_y  \\ A_5=\sigma_x\otimes \sigma_y\otimes \sigma_y & A_6=I \otimes I \otimes \sigma_x  \\
 A_7=I \otimes I \otimes \sigma_y & A_8= \sigma_x\otimes I \otimes I\\
 A_9=I \otimes  \sigma_y \otimes I  & A_{10}=I \otimes  \sigma_x\otimes I.
 \end{array}$$

 The contradiction we get when we assign definite values to these observables is easily understood if we arrange them in a
 star, as shown in figure \ref{figmultiplicative}. The operators are arranged in five lines with four operators each
: $A_1A_3A_6A_9$, $A_1A_4A_7A_{10}$, $A_2A_3A_4A_5$, $A_2A_6A_8A_{10}$ and $A_5A_7A_8A_9$. 
The following properties hold:
 
 \begin{figure}[h]
 \centering
 \includegraphics[scale=0.8]{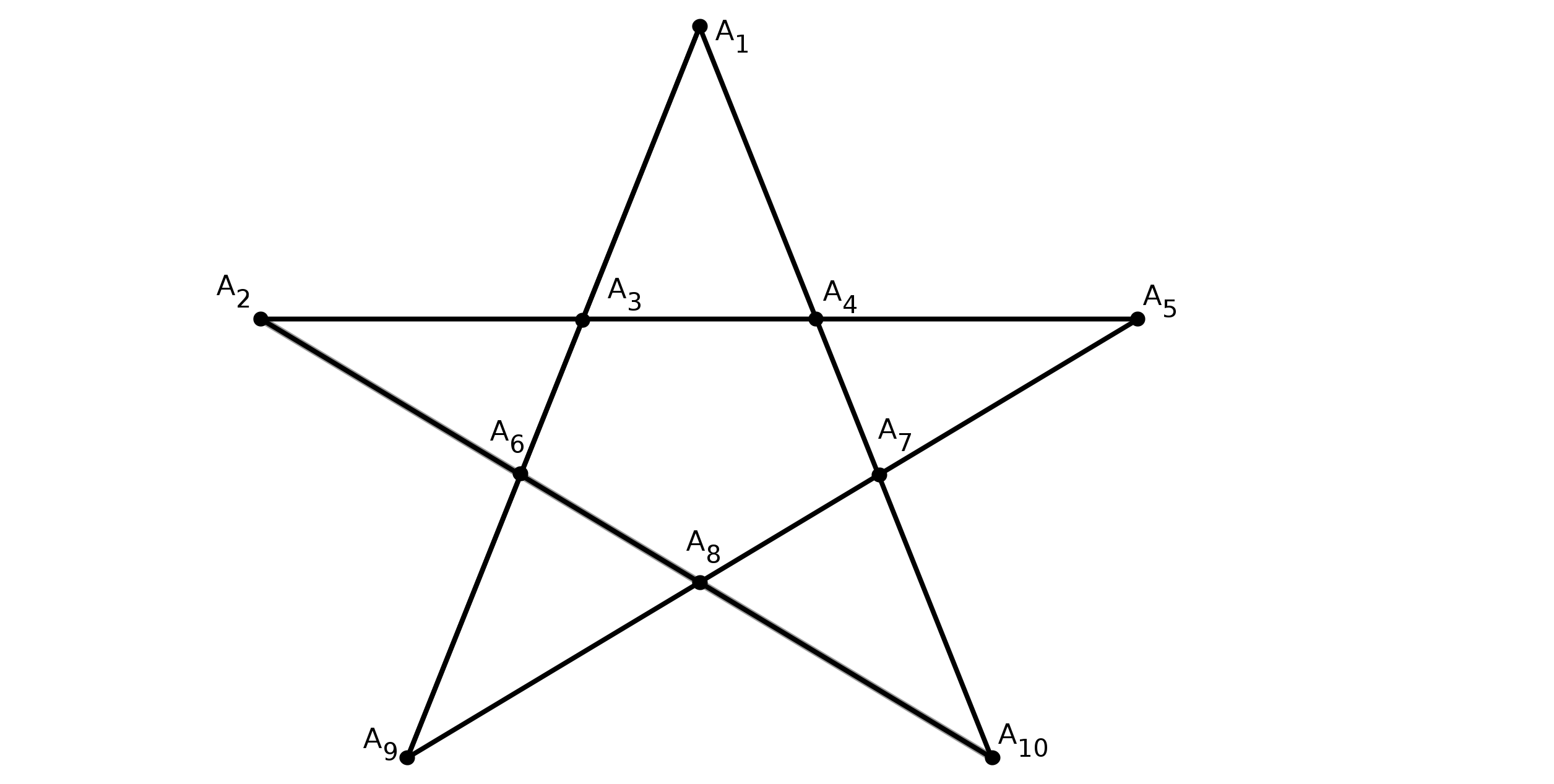}
 \label{figmultiplicative}
 \caption{Observables providing a proof of Kochen-Specker theorem in dimension 8.}
\end{figure}

\begin{enumerate}
\item The observables in each line are compatible;
\item \label{prodop2} The product of the observables that appear in the horizontal line $A_2A_3A_4A_5$ is  $-I$; the 
product of the observables in every other line is $I$.
\end{enumerate}
This properties implies that the values assigned by a hidden-variable model must obey
\begin{subequations}
\begin{eqnarray}
 P_1=&v(A_1)v(A_3)v(A_6)v(A_9)&=1,\\
  P_2=& v(A_1)v(A_4)v(A_7)v(A_{10})&=1,\\
  P_3=&  v(A_2)v(A_6)v(A_8)v(A_{10})&=1,\\
  P_4=&   v(A_5)v(A_7)v(A_8)v(A_9)&=1,\\
  P_5=&    v(A_2)v(A_3)v(A_4)v(A_5)&=-1.
  \end{eqnarray}
\end{subequations}
This leads to a contradiction, since the validity of the equations above would imply
$$-1=P_1P_2P_3P_4P_5= \prod_iv(A_i)^2 = 1.$$

\section{\textsf{A contextual hidden-variable model}}
\label{sectioncontextualmodel}

The Kochen-Specker theorem forbids noncontextual hidden-variable models, but it is possible to complete quantum theory
in order to give definite values for all projective measurements, as long as we  drop the assumption of noncontextuality.
An example of such a model is provide by Bell in reference \cite{Bell66}.

To define a hidden-variable model it suffices to define the values $v(P)$ attributed to the projectors $P$.
This happens because every hermitian operator can be written as a linear combination of compatible projectors
$$A=\sum_i \lambda_i P_{\phi_i},$$
in which $\lambda_i$ is the eigenvalue of $A$ corresponding to eigenvector $\ket{\phi_i}$.
As we can choose the $\ket{\phi}$ mutually orthogonal,  we can assume that $[P_{\phi_i}, P_{\phi_j}]=0$ and hence they are
mutually compatible. Since the assignment $v$ must preserve the linear relationships between compatible vectors, we have
$$v(A)=\sum_{i}\lambda_i v(P_{\phi_i}).$$

 Suppose an experimental arrangement performs the measurement of the observables represented by the projectors
 $P_{\phi_1}, \ldots, P_{\phi_n}$. Let us define the numbers $a_i \in \mathbb{R}$ such that the expectation values of the 
 $P_{\phi_1}, \ldots, P_{\phi_n}$ are $a_1, a_2-a_1, a_3-a_2,\ldots, a_n -a_{n-1}$, respectively.
 As hidden variables we will use a real number between zero and one.
 The value associated to projector $P$ if the value of the hidden variable is $\lambda$ is

$$\left\{\begin{array}{cc}
v(P_{\phi_i})=1 & \mbox{if}\  a_{i-1} <\lambda \leq a_i,\\
v(P_{\phi_i})=0 & \mbox{otherwise.}
\end{array}\right.$$

Notice that the value of each  $a_i$ depends on the entire set of projectors being measured. Hence the value of 
$v(P_{\phi_i})$ does not depend just on the quantum state of the system and the hidden variable $\lambda$, it
depends also on which other projectors are being measured with $P_{\phi_i}$. This means that this is
a contextual hidden-variable model.

To show that this model agrees with the quantum predictions, we notice that
$$\langle P_{\phi_i} \rangle = \int_0^1 v(P_{\phi_i})d\lambda = a_i-a_{i-1}.$$

This model is quite artificial, but it is important conceptually to show that the hypothesis of noncontextuality
in the Kochen-Specker theorem is essential to discard the possibility of hidden-variable models. It shows that the
completion of quantum theory is possible, and brings hope for those who doubt the fact that 
nature could be intrinsically probabilistic. But one important remark must be made. Hidden-variable theories were first imagined
by people who believed  that the world could not behave in such a counter-intuitive manner. The main point was
to recover the notion we have in classical theory that every measurement has a definite outcome, that \emph{exists} prior 
to the measurement and is only \emph{revealed} when the measurement is performed. If we choose to keep this line of thought,
the Kochen-Specker theorem forces contextuality on our theories, which is also a really intriguing feature, not present in 
classical theories. Hence, if quantum theory is really correct, and so far there is no reason to believe it is not, we
have to accept the fact that things are a bit weird and our intuition, modeled by our experience with classical systems, 
can not be applied to explain its phenomena.

There exist also other state-independent proofs with a smaller number of observables.
In reference \ref{YO12} the authors present a proof of the Kochen-Specker theorem with 13 vectors. 
The idea of the proof is quite different from the additive and multiplicative proofs we have shown above. It is based on the violation of   an experimentally testable inequality involving only 13 observables that is satisfied
by all non-contextual  models while being violated by all qutrit states.

\section{\textsf{Final Remarks}}
\label{sectionfinalcontextuality}

In the classical description of physical systems, probabilities come from our lack of 
knowledge about the past history of the system, or due to practical problems that come when we deal with a huge number
of particles at the same time. Every system has well defined values for all physical quantities, that are merely revealed by
the measurements. The impossibility of accessing these values was believed to be 
a technological and practical issue and not a fundamental
limit imposed by nature on the information we can gain when interacting with a system.

This reasoning can not be applied to quantum theory.   Since the development of its  modern mathematical formulation
 in the 1920's, this
intrinsic probabilistic behavior has been puzzling physicists and philosophers of science, experts and non-experts
all around the world. Is it a flaw on the mathematical structure of the theory? Would it be possible to complete
quantum theory in order to predict with certainty the outcomes of each measurement and still recover the quantum statistics?

In this chapter we have shown that if this completion is required to be \emph{noncontextual}, it is not possible.
The first attempt was made by von Neumann in 1932.  He showed that under some assumptions, the expectation-value functions
in the hidden-variable models should obey the quantum rule, and hence could not be dispersion-free. His argument, though,
 discards only a very restrict class of hidden-variable models, since he made the strong assumption that expectation-value
 functions should reproduce the algebraic relations among the observables, even if the observables are not compatible. 
 Although this is the case for quantum theory, we can not justify this assumption physically and hence we should not impose 
 it on our models. In fact, a simple hidden-variable model for a qubit system is provided by Bell as
 a counter example to von Neumann's result.
 
 More successful results appeared with the work of Kochen and Specker. Their main theorem states that
 for system with dimension 3 or higher, noncontextual hidden-variable models recovering the quantum statistics are
 not possible. The noncontextuality assumption requires that the value assigned to a measurement does not depend
 on other compatible measurements performed together.
 The same result can be proven with the help of Gleason's lemma, with the drawback that the number of vectors in 
 the proof is infinite.
 After Kochen and Specker's original proof, many others have been derived. The advantage of these proofs is that they are 
much simpler then the first and hence may be more suitable for experimental implementations. We have
discussed some of these proofs above, but many more are known. We refer to \cite{Cabello96, TKLSK13} for a more details.

 As shown by Bell, 
it is possible to construct a contextual hidden-variable model for any set of measurements in  any dimension. Although this
model is quite artificial, it proves that the  assumption of noncontextuality is crucial in the Kochen-Specker theorem.

In summary, what we learn with this result is that to reconcile the quantum formalism with the notion of well defined physical
properties of classical intuition, we must accept contextuality, which is also a very counter-intuitive  property.
How could the value of one physical quantity depend on what other properties are jointly measured?
The Kochen-Specker theorem implies that there is no way out: the mathematical description of
quantum systems does not agree with the classical  idea of pre-defined physical quantities.
\chapter{\textsf{Non-locality}}
\label{chapternonlocality}

Historically, the discussion of nonlocality in quantum theory preceded the discussion about its noncontextual character.
It started around 1935, when Einsten, Podolsky and Rosen noticed that the way of thinking of classical physics does not apply 
directly to quantum systems \cite{EPR35}. They started one of the greatest debates in foundations of physics and philosophy of science 
in general, that is still fruitful nowadays.

The classical world consists of objects with precise physical attributes: position, mass, velocity, orientation, charge, 
etc. This is how physicists were used to think for centuries. Their job was to understand the connection between
these attributes and create mathematical objects that mimic these relations.
A theory build  for this purpose would be considered satisfactory
if every relevant physical attribute has a counterpart in the theory and if the relations and results predicted by this
correspondence agree with what is observed in real situations.

This line of thought led many scientist, including Einsten, Podolsky and Rosen, to conjecture
the existence of a more complete theory behind the quantum formalism.  The intrinsic probabilistic
character of quantum measurements should be the result of the lack of knowledge about the past history of the system
and  a more adequate theory should be conceived that predicted all these results with certainty. 

This is the same reasoning that we used to conjecture the existence of hidden-variable models completing quantum theory.
Einsten, Podolsky and Rosen belied that such a model would be possible. In this chapter we prove that under the assumption 
of \emph{locality}, this kind of model does not exist.

The first one to provide a proof of the impossibility of \emph{local hidden-variable models} was John Bell, in 1964 \cite{Bell64}.
He demonstrated that if the statistics of joint measurements on a pair of two qubits in the singlet state were 
given by a hidden-variable model, a linear inequality involving the corresponding probabilities should be satisfied.
A simple choice of measurements leads to a violation of this inequality, and hence the model can not reproduce
the quantum statistics. 

Many similar inequalities were derived since Bell's work. Because of his pioneer paper, any inequality
derived under the assumption of a local hidden-variable model is called \emph{Bell inequality}.
Quantum theory violates these inequalities in  many situations. Besides the insight given in foundations of 
quantum theory, those violations are also connected to many interesting applications.

The pioneer paper of Einsten, Podolsky and Rosen is discussed in section \ref{sectionEPR}. Hidden-variable models
are introduced in section \ref{sectionhiddenvariables} and Bell's proof of the impossibility of such models 
in section \ref{sectionBell}. Other proofs based on Bell inequalities are presented in section \ref{sectionBellinequalities}
and its connection with convex geometry in \ref{sectionconvex}. We finish with out final remarks in section
\ref{sectionfinalnonlocality}

\section{\textsf{The EPR paradox}}
\label{sectionEPR}

Einsten, Podolsky and Rosen published in 1935 one of the most important and cited papers in quantum information theory and
also in foundations of quantum mechanics. In their letter, entitled ``\emph{Can Quantum-Mechanical description of Physical 
Reality Be Considered Complete?}'' \cite{EPR35}, the authors argue that in order to a physical theory be considered complete,
every quantity with physical
reality has to be predicted with certainty by the theory. As we know, non-commuting observables in quantum theory can never
have definite values simultaneously, and hence, we must accept one of two possible situations: either quantum
theory does not provide a complete description of nature or two non-commuting observables can not both have
physical reality. They present arguments discarding the second option, and hence they believed that quantum theory
could not be considered complete.

According to EPR, when we analyze the success of a theory, we  must ask two questions: 

\begin{enumerate}
 \item Is the theory correct? \label{enteocorrect}
 \item Is the theory complete? \label{enteocomplete}
\end{enumerate}

The answer to question number \ref{enteocorrect} is 'yes' if the predictions of the theory agree with all data available 
from experimentation in real physical systems. Of course, it is always possible that a theory considered correct be at some 
point contradicted with more modern and advanced experimental setups, and if this happens physicist should seek for different
theories capable of describing the new results.
At that time, as it is nowadays, the answer of this question for quantum theory is 'yes'.

The concept of a complete theory is more delicate and it is not easy to define. EPR argued that any reasonable definition for
completeness must end in a concept for which the following condition is necessary:

\begin{center}
 \emph{``Every element of the physical reality must have a counterpart in the physical theory.''}
\end{center}

The concept of physical reality is also delicate, but they provide a condition they consider to be sufficient for
a physical quantity to be called an \emph{element of reality}:

\begin{center}
 \emph{``If, without any way disturbing a system, we can predict with certainty the value of a physical quantity, then there 
 exists an element of physical reality corresponding to this physical quantity.''}
\end{center}

For them, a physical theory can only be considered satisfactory if it is both correct and complete.

In classical theory, once we have full information about the system, that is, if we have a pure state, all measurements have
definite values. Therefore, every quantity corresponds to an element of reality and classical theory is complete.
In the other hand, quantum theory does not predict the outcomes of every measurement even if the system is in a pure state.
This can only be done if the state is an eigenvector of the corresponding operator and hence two non-commuting operators 
can not have both definite values in every state.

Consider, for example, the quantum system of one qubit. If a qubit is  in state $\ket{\psi}=\ket{0}$, 
we can predict that a measurement of the observable
$\sigma_z$ will have certainly outcome $0$. If instead we measure $\sigma_x$, each possible outcome occurs
with equal probability.

These observation and EPR assumptions lead to the conclusion that one of two conditions must 
hold:

\begin{enumerate}
 \item Quantum theory is not complete;\label{enqtcomplete}
 \item Two non-commuting observables can not represent elements of reality at the same time. \label{enreality}
\end{enumerate}

In fact, if quantum theory was complete and both observables corresponded to elements of reality, both should have
definite values predicted by the theory for all pure states, which is certainly not possible.

Lets see now how EPR discard option \ref{enreality}. Suppose we have a pair of quantum system that have interacted in the past
in composite state $\ket{\Psi}$.
Suppose we want to measure to observables $M$ and $N$ in the first system and let $\{\ket{u_1}, \ldots, \ket{u_m}\}$ and
$\{\ket{v_1}, \ldots, \ket{v_n}\}$ be the eigenvectors of $M$ and $N$, respectively. Then we can decompose
$\ket{\Psi}$ in two different ways:
\begin{eqnarray}
\ket{\Psi}&=&\sum_{i=1}^m \ket{u_i}\otimes \ket{\mu_i} \nonumber\\
\ket{\Psi}&=&\sum_{i=1}^n \ket{v_i} \otimes \ket{\nu_i}
\end{eqnarray}
where $\ket{\mu_i}$ and $\ket{\nu_i}$ are pure states for the second system. Suppose now that measurement $M$ was performed in
the first system. The state of the composite system after the measurement, if outcome $i$ was obtained is 
$\ket{u_i}\otimes \ket{\mu_i}$,
and the second system can be described by the state $\ket{\mu_i}$. On the other hand, if measurement $N$ was performed in
the first system, the state of the composite system after the measurement, if outcome $i$ was obtained, is 
$\ket{v_i}\otimes \ket{\nu_i}$, and the second system is left in state $\ket{\nu_i}$.

Now EPR argument that since nothing was done in the second system, the physical reality of this system is the same 
for both options, and hence $\ket{\mu_i}$ and $\ket{\nu_i}$ describe the same physical reality.

Suppose now that the vectors $\ket{\mu_i}$ are eigenvectors of an observable $M'$ in the second system and the vectors
$\ket{\nu_i}$ are eigenvectors of an observable $N'$ in the second system, not commuting with $M'$. This can
be the case in some situations, as we show in example \ref{exsinglet} below. If we measure $M$ in the first system, we can predict
with certainty the outcome of $M'$ in the second system, \emph{without disturbing the second system}, since we have not 
interacted with it at any point during the measurement. On the other hand, if we measure $N$ in the first system, we can
predict with certainty the outcome of $N'$ in the second system, again without disturbing it. Hence,  both 
$M'$ and $N'$ must correspond to elements of reality, which in turn proves that condition \ref{enreality} is not true.
Thus, we have no option but to accept the fact that quantum theory is not complete.

\begin{ex}[The Singlet]
\label{exsinglet}
The state of two qbits given by
\be\ket{\Psi_-}=\frac{\ket{01}-\ket{10}}{\sqrt{2}},\label{eqsinglet1}\ee
called the \emph{singlet}, can be used to exemplify the situation mentioned above. This state can be also written as
\be\ket{\Psi_-}=\frac{\ket{+-}-\ket{-+}}{\sqrt{2}}.\label{eqsinglet2}\ee
If we use equation \eqref{eqsinglet1}, we see that a measurement of $\sigma_z$ in the first qubit allows the prediction 
of the result of the same measurement in the second qubit. In the other hand, if we use equation \eqref{eqsinglet2},
we see that a measurement of $\sigma_x$ in the first qubit allows the prediction of the result of the same measurement
in the second qubit. 
\end{ex}

EPR's discussion on physical reality is  based on Newtonian (classical) mechanics, which is suitable only
to describe the motion of macroscopic objects.
The study of the motion of bodies is an ancient one, making classical mechanics one of the oldest and largest subjects in 
science. It is also the physical theory that describes most of the phenomena we deal with in our daily life and hence it is
not surprising that  our intuition is guided by this way of thinking. EPR go even further, using this ideas as impositions
of what we should call physical reality.
This line of thought is not necessarily valid for quantum systems, as we already discussed in chapter \ref{chapterncinequalities}
 and appendix
\ref{chaptercontextuality}.

The debate in EPR's paper is of great importance both from
the physical as well as the philosophical point of view. This issue deserves a much more deep analysis then the one presented 
here and many people have devoted their time to investigate it. See \cite{EPRstantford} and references 
therein for more detailed discussion on the subject.

\section{\textsf{Local Hidden-Variable Models}}
\label{sectionhiddenvariables}

If quantum theory is not complete, we should seek for other theories that assign definite outcomes
for all measurements and at the same time, agree with all quantum predictions. We continue with the same nomenclature
used in chapter \ref{chapterncinequalities} and call such theories \emph{hidden-variables models} compatible with quantum
theory. EPR believed in the existence of such theories. We have already  proved that under the assumption of noncontextuality,
these theories can not exist. In section
\ref{sectionBell} we prove that under the assumption of \emph{locality}, these models also do not exist.

The hypothesis of locality is crucial in EPR's argument. It states
that physical processes occurring at one place should have no immediate effect on the other location.
This appears to be a reasonable assumption to make, as it is 
a consequence of special relativity, which states that information can never be transmitted faster than the speed of light.
This assumption is explicit in their argument, since they assume that the measurement performed on the second particle does
not influence the first one. EPR's assumption is generally referred to as \emph{local realism}, 
as it is the combination of the principle of locality with the \emph{realistic} assumption 
that all systems must objectively have a pre-existing value for any possible measurement before the measurement is made.

The assumption of local realism has an immediate consequence on the probability distribution describing the measurements
performed in a composite system. If we assume this condition, a complete description of the system has to
give predefinite values for all measurements in all subsystems and at the same time the value obtained in one subsystem
can not
depend on the measurement performed on any other subsystem. 

Within this perspective, any uncertainty on the outcomes of each measurement comes from the fact that the previous history
of the composite system is not known. With the locality assumption, any correlation among the results of the measurements is a consequence
of the past interaction among the parties. Let $\lambda$ be a set of variables describing the past history of the composite system.
They play the role of hidden variables in a hidden-variable model.
Once these variables are known, there is no correlation between the outcomes in each subsystem, as a consequence, the 
statistics of the experiment can be written as

\be p(a_1, \ldots, a_n|A_1,\ldots, A_n)= \sum_{\lambda} p(\lambda) p(a_1|A_1,\lambda)\times \cdots p(a_n|A_n,\lambda), \label{eqlocaldistributions} \ee
where $p(a_1, \ldots, a_n|A_1,\ldots, A_n)$ is the probability of getting the set of outcomes $a_1, \ldots, a_n$ 
when measurement $A_i$ 
is performed on part $i$, $p(a_i|A_i,\lambda)$ is the probability of getting $a_i$ 
in measurement $A_i$ in the $i$-th subsystem given the past history $\lambda$, and $p(\lambda)$ is the probability
distribution on the hidden variable $\lambda$.

Equation \eqref{eqlocaldistributions} provides a mathematical way of verifying if the statistics of a given
experiment is consistent with the assumption of local realism. If this is the case, it should be possible to write
the probability distribution in the form given by this equation. In the next section we prove that this is not always
possible if the statistics is obtained from quantum systems.


\section{\textsf{Bell's proof of the impossibility of hidden variables compatible with quantum theory}}
\label{sectionBell}

Suppose we have a pair of qubits in the singlet state. Any measurement on one qubit with possible
outcomes $\pm1$ can be written on the form
$$R=\vec{r} \cdot \vec{\sigma}=r_1 \sigma_x + r_2\sigma_y +r_3 \sigma_z,$$
where $\vec{r}=(r_1,r_2,r_3)$ is a unit real vector.

Let us suppose also that a given hidden-variable model provides definite values for the measurements performed in each qubit.
If this model satisfies the locality assumption, the value of such a measurement performed on one of the qubits depends only
on the vector $\vec{r}$ and on the hidden variable $\lambda$. We will denote this value by $v_i(\vec{r}, \lambda)$,
where $i=1,2$ denotes the qubit on which the measurement is performed.

Since the qubits are in the singlet state, the results are anti-correlated if the same measurement is made in both qubits.
Hence, 
$$v_1(\vec{r}, \lambda)=-v_2(\vec{r}, \lambda).$$
Also the quantum expectation value for the measurement of
$R=\vec{r} \cdot \vec{\sigma}$ in the first qubit and $S=\vec{s} \cdot \vec{\sigma}$ in the second qubit is equal to
$$\langle RS\rangle_Q = -\vec{r} \cdot \vec{s}$$
and it must agree with the expectation value calculated using the hidden-variable model, which is 
$$\langle RS\rangle=\sum_\lambda p(\lambda) v_1(\vec{r}, \lambda)v_2(\vec{s}, \lambda)=
-\sum_\lambda p(\lambda) v_1(\vec{r}, \lambda)v_1(\vec{s}, \lambda).$$


It follows that for any other measurement $T=\vec{t}\cdot \vec{\sigma}$ we have
\begin{eqnarray}
\langle RS\rangle-\langle RT\rangle&=& -\sum_\lambda p(\lambda)[ v_1(\vec{r}, \lambda)v_1(\vec{s}, \lambda)-
v_1(\vec{r}, \lambda)v_1(\vec{t}, \lambda)]\\
&=&\sum_\lambda p(\lambda) v_1(\vec{r}, \lambda)v_1(\vec{s}, \lambda)[
v_1(\vec{s}, \lambda)v_1(\vec{t}, \lambda)-1]
\end{eqnarray}
and hence
$$|\langle RS \rangle - \langle RT\rangle| \leq 1 + \langle ST \rangle.$$

If the hidden-variable model agrees with the quantum prediction, we have that
\be|\vec{r} \cdot \vec{s} - \vec{r}\cdot \vec{t}| \leq 1 +\vec{s}\cdot \vec{t}\label{eqBellinequality}\ee
an inequality that must hold for every choice of $\vec{r}, \vec{s}$ and $\vec{t}$.

Now, if we choose $\vec{r}=\vec{s}=-\vec{t}$ the left hand side of the inequality is equal to $2$, while the right
hand side is equal to $0$, which is a contradiction with inequality \eqref{eqBellinequality}. This proves that 
the conclusions obtained with the assumption of local realism do not agree with  quantum theory.

\begin{teo}
 There is no local hidden-variable model compatible with quantum theory.
\end{teo}

\section{\textsf{Bell Inequalities}}
\label{sectionBellinequalities}

There are many other linear inequalities which can be obtained assuming the hypothesis of local realism 
that are violated in some experimental situations involving quantum systems.
All of these inequalities are called \emph{Bell inequalities}, named after Bell's pioneer discovery, inequality
\eqref{eqBellinequality}.  It is possible to find a huge number
of non-equivalent Bell inequalities in the literature and work has been devoted to create a database to collect and organize 
all these examples \ref{RBG14}.

\subsection{\textsf{The CHSH inequality}}

The most famous and also the simplest Bell inequality was derived by Clauser, Horne, Shimony and Holt  
\cite{CHSH69}. This inequality is known  as CHSH inequality.

In the corresponding experimental scenario, there are four measurements available in  a bipartite system,  two measurements 
in each subsystem. 
Each measurement
has two possible outcomes, which we denote by $\pm 1$.

Let us denote the measurements in the first subsystem by $A_1$ and $A_2$ and the measurements in the second subsystem by 
$B_1$, $B_2$. Given a choice of measurement in each subsystem, $p(a,b|A_i,B_j)$ will denote the joint probability of
having outcome $a$ in the first  subsystem and $b$ in the second subsystem. The expectation value of the joint measurement of
$A_i$ and $B_j$ is 
$$\langle A_i B_j\rangle = p(11|A_iB_j) +  p(-1-1|A_iB_j) - p(-11|A_iB_j) - p(1-1|A_iB_j).$$

Consider now that the outcomes of $A_i$ and $B_j$ are given by a local hidden-variable model. Then we have
$$p(a,b|A_i,B_j)=\sum_{\lambda} p(\lambda) p(a|A_i,\lambda)p(b|B_j,\lambda).$$
All probability vectors of this form can be written as convex combination of the ones assigning definite values to each 
measurement locally. 
We will focus first in those distributions. The definite values
assigned to each measurement by the model will be denoted by $v(A_i)$ and $v(B_j)$. In this case we have
\be\langle A_i B_j\rangle = v(A_i)v(B_j). \label{eqexpectationdefined}\ee

Now consider the sum 
\be S_{CHSH}=\langle A_1 B_1\rangle +\langle A_1 B_2\rangle + \langle A_2 B_1\rangle - \langle A_2 B_2\rangle.\ee
If these values are given by equation \eqref{eqexpectationdefined}, we have

\begin{eqnarray}
 S_{CHSH}&=&v(A_1)v(B_1) +v(A_1)v(B_2)+ v(A_2)v(B_1)- v(A_2)v(B_2) \nonumber\\
 &=&v(A_1)(v(B_1)+v(B_2))-v(A_2)(v(B_1)-v(B_2)).
 \end{eqnarray}
Since the possible outcomes are $\pm 1$ it follows that $S_{CHSH}$ is either $2$ or $-2$. Taking convex combinations
of these distributions we conclude that  if some distribution is given by a local hidden
variable model we have
\be -2 \leq \langle A_1 B_1\rangle +\langle A_1 B_2\rangle + \langle A_2 B_1\rangle - \langle A_2 B_2\rangle \leq 2.
\label{eqCHSH}\ee
The second inequality is the famous CHSH inequality.
 
Now we see what can happen if we use a quantum system. 

\begin{ex}
\label{exCHSHviolation}
Consider again the singlet state $\ket{\Psi_-}$ and the measurements
$A_1=\sigma_z$, $A_2=\sigma_x$, $B_1=\frac{-\sigma_x-\sigma_z}{2}$ and $B_2=\frac{-\sigma_x+\sigma_z}{2}$.
In this case we have $S_{CHSH}=2\sqrt{2}$, which violates the local bound of $2$ given by the CHSH inequality
\eqref{eqCHSH}. This is the maximum value obtained with quantum distributions, and this bound is called the Tsirelson bound
for the CHSH inequality \cite{Cirelson80}.

\end{ex}

\section{\textsf{Bell inequalities and convex geometry}}
\label{sectionconvex}

We can define more precisely the scenario we are working with, 
in a similar way as was done in section \ref{sectionsheaftheory}.
Once more we start with a set of possible measurements $X$, and the main difference from what was done before is that now
 we assume that the system is composed of $n$ different spatial separated subsystems. The set $X$ is then divided into
 various distinct subsets $X_1, X_2, \ldots, X_n$, where $X_i$ is the set of measurements available for party $i$. In this case, 
compatibility is guaranteed by the spatial separation among the parties, and all contexts are of the form
$$C=\{M_1,M_2,\ldots, M_n\}, \ \ M_i \in X_i.$$
Scenarios with these extra restrictions are called Bell scenarios. The particular  case in which all parties have
each one $m$ measurements available, each measurement with $o$ possible outcomes, is denoted by $(n,m,o)$.

The vertices of the compatibility hypergraph of a Bell scenario can be split in the $n$ disjoint subsets $X_i$. Each edge has
one, and only one element of each $X_i$. In the bipartite case $n=2$, this graph is the complete bipartite graph
$G=(X_1,X_2)$.

The probability distributions for Bell scenarios can be denoted in a simple way. Given a context $C=\{M_1,M_2,\ldots, M_n\}$,
$$p(m_1,m_2, \ldots , m_n|M_1,M_2,\ldots, M_n)$$
will denote the probability of the set of outcomes $m_1,m_2, \ldots , m_n$ when each measurement $M_i$ is 
performed in party $i$.

The no-disturbance property in this case is a very reasonable restriction to make. It is a consequence of the assumption 
that the measurements performed in one site do not affect any other instantaneously, since no information can travel faster 
then the speed of light. In this context, this property is referred to as the \emph{no-signaling} condition.
The set of no-signaling distributions $\mathcal{N}$ is a polytope, since it is defined by a finite set of linear inequalities.

The noncontextual distributions of a Bell scenario are exactly the ones for which a local hidden-variable model
can be constructed. 

\begin{defi}
 A probability distribution $p$ for a Bell scenario is called \emph{local} if it can be written in the form
 $$p(m_1,m_2, \ldots , m_n|M_1,M_2,\ldots, M_n)=\sum_{\lambda} p(\lambda)\prod_{i=1}^np(m_i|M_i, \lambda)$$
 where $p(\lambda)$ is a probability distribution in the hidden variable $\lambda$.
\end{defi}

Since the set of local distributions $\mathcal{L}$ is the convex hull of a finite set, it is a polytope. The 
H-descriptions of this polytope correspond  to  a finite set of Bell inequalities providing necessary
and sufficient conditions for membership in $\mathcal{L}$.

\begin{defi}
A \emph{Bell inequality} is a linear inequality
$$S=\sum\gamma_{m_1,m_2, \ldots,m_n |M_1,M_2, \ldots,M_n} p(m_1,m_2, \ldots,m_n |M_1,M_2, \ldots,M_n ) \leq b$$
where  $\gamma_{m_1,m_2, \ldots,m_n |M_1,M_2, \ldots,M_n}$ and $b$ are real numbers, which is satisfied by 
all 
classical distributions and violated by some nonlocal distribution. A  \emph{tight Bell
inequality} is a linear inequality defining a non-trivial facet of the local polytope $\mathcal{L}$.
\end{defi}

In general  quantum distributions do not satisfy all Bell inequalities, as we saw in example \ref{exCHSHviolation}.
This behavior is often referred to as \emph{quantum nonlocality}.
The maximal quantum value for $S$ is called the Tsirelson bound for the inequality \cite{Cirelson80}. 

\section{\textsf{Final Remarks}}
\label{sectionfinalnonlocality}

In this chapter we have shown once more that under very reasonable circumstances, a completion of quantum theory by a 
hidden-variable model is not possible. The impossibility proofs are based on multipartite scenarios and rely
on the fact that, according to special relativity, information can not travel faster then light. This 
restriction imposes the condition that what is done in one party can not instantaneously affect any other, and hence that
our hidden-variable models have to be \emph{local}.

The first impossibility proof in this situation was provided by John Bell \cite{Bell64}, who derived an inequality 
for the expectation values of joint measurements in a pair of qubits in the singlet state that should be valid if
those were given by a hidden-variable model. This inequality is not always valid for quantum distributions, what proves that
these models can not reproduce the statics of quantum theory for this state.

After Bell's work many other inequalities satisfied by local hidden-variable models and violated by some
quantum distributions were derived. 
The simplest and also most famous is the CHSH inequality \cite{CHSH69}. 
Violations of Bell inequalities prove that the assumption of local realism is incompatible with quantum theory.
Locality and realism are features of classical theory, properties of our daily life experience, that can not be applied at 
the same time in the description of quantum systems. 
There is  huge amount of work on the subject, both
in the aim of finding new inequalities and finding applications for different types of inequalities
(see \cite{BCPSW13} and references therein).

There are also many experimental implementations leading to violation of a Bell inequality \cite{wikibelltests}.
The first one was performed in 1972 by Stuart J. Freedman and John F. Clauser \cite{FC72}.
Modern experiments are very precise, but unfortunately none of them is able to fulfill all requirements necessary to 
actually eliminate
the possibility of hidden-variable models describing the system involved according to our classical conceptions.
The failures in these experiments are generally called \emph{loopholes} \cite{wikiloopholes}.

The most common of these failures are the \emph{detection loophole} and the \emph{locality loophole}.
The detection loophole comes from the fact that all detectors (or measurement devices) are imperfect:
a portion of the systems prepared are always lost before they are detected. Hence, the data obtained in the experiment is 
incomplete. It is possible that  this missing data  creates the illusion of a violation of the inequality, while if 
we take into account the lost events in the statistics we would have a local distribution.

The locality loophole appears because in some implementations is not possible to guarantee that the subsystems are sufficiently
far apart from each other. We need to make sure that what happens in one laboratory does not affect the results in the other.
To do that we have to assure that the process of choosing a measurement, performing it and getting an outcome is completed
before any signal can travel from one site to the other. 
The first time it was done was in 1981, when Alain Aspect and collaborators performed the pioneer experiment
of violation of the CHSH inequality \cite{ADR82}. This experiment does not eliminate the detection loophole. Since that time, many 
improvements were made. The photon is the first experimental system for which all main experimental loopholes have been 
surmounted, albeit presently only in separate experiments \cite{GMRWKBLCGNUZ13, CMACGLMSZNBLGK13}. We believe that 
a loophole free implementation will soon be achieved.

\chapter{\textsf{What  explains the Tsirelson bound?}}
\label{chaptertsirelson}

Quantum probability distributions may exhibit nonlocality, a feature
that is revealed by the violation of a Bell inequality. In most cases it is possible
to find distributions that violate this inequalities \emph{more} then the quantum distributions. What is the physical
explanation for that? Why isn't quantum theory \emph{more} nonlocal then it is? For a given scenario, what distinguishes the 
set of quantum probability distribution from others obtained with general probability theories? In this chapter we discuss the various physical
principles proposed to answer this question.

In section \ref{sectionnosignaling}, we show that the no-signaling principle, implied by the relativistic
imposition that no signal can travel faster then the speed of light, is not enough to rule out violations higher then the 
Tsirelson bound. Nonetheless, the existence of some of these distributions has implausible consequences for communication
complexity, which we examine in section \ref{sectionimplausible}. 

The principle of 
\emph{Information Causality},  which states that the information gain that one can get about the data of a spatially distant
observer by using all his local
resources and $m$ classical
 bits sent to him by this observer is at most $m$ bits. It is a generalization of the no-signaling principle, which is just 
 Information Causality with $m=0$. This principle is satisfied by quantum distributions, but discards many others outside the 
 quantum set, as we will see in section \ref{sectioninformationcausality}.
 
 The principle of \emph{Macroscopic Locality}, subject of section 
 \ref{sectionmacroscopiclocality}, states that a any physical theory should recover the classical 
results  when we measure a large number of systems and our devices are not capable of identifying
individual particles. It is not equivalent to Information Causality and it is also known that it can not recover the 
quantum set.    Nonetheless, it is a  reasonable property we should expect from any alternative to quantum theory.
 
In section  \ref{sectionmultipartiteprinciples}, we show that no bipartite principle is capable of ruling out some 
non-quantum distributions. This proves that intrinsically multipartite principles must be found. The first
one is the principle of Local Orthogonality, the Exclusivity principle applied to Bell scenarios. 
It can be used to rule out many non-quantum distribution, including some of the distributions that can not be
ruled out by any bipartite principle. This principle and some implications are discussed in section 
\ref{sectionlocalorthogonality}. We finish this appendix in section  \ref{sectionfinaltsirelson} with our
final remarks.

\section{\textsf{No-signaling}}
\label{sectionnosignaling}

We have seen in appendix \ref{chapternonlocality} that relativistic causality is a reasonable imposition to make on the 
acceptable probability distributions in a Bell scenario. This  restriction is a consequence of 
special relativity theory, which states that no signal can travel faster then the speed of light.
Quantum theory does not violate this principle, but more general 
probabilistic theories might. In 1993, Popescu and Rorlich proposed to take non-locality as the quantum principle
and analyze what this assumption, together with relativistic causality, would imply.

We consider once again a bipartite scenario where each subsystem is far away from the other. Relativistic causality
implies that if no signal was sent from one party to the other, one of the parties can get no information about 
the measurements applied  in the other party nor about the results obtained. The mathematical consequence of this
assumption is that the distribution must obey the following principle:

\begin{prin}[The no-signaling principle]
 Probability distributions in a Bell scenario satisfy
\begin{eqnarray}
 \sum_{a_2}P\left(a_1,a_2|x_1,x_2\right)&=&P\left(a_1|x_1\right);\nonumber\\
\sum_{a_1}P\left(a_1,a_2|x_1,x_2\right)&=& P\left(a_2|x_2\right),
\end{eqnarray}
where $x_1$ is a measurement in party one with possible outputs $a_1$ and $x_2$ is a measurement in party two with possible outputs $a_2$.
\end{prin}
These distributions are called \emph{no-signaling}.

We want to see now what are the consequences of taking non-locality and relativistic causality as fundamental axioms. Would that be enough to
single out the set of quantum distributions? Is quantum theory the only one exhibiting non-locality while preserving relativistic
causality?

Let us see what happens with the CHSH inequality
$$S_{CHSH}=\langle A_1 B_1\rangle +\langle A_1 B_2\rangle + \langle A_2 B_1\rangle - \langle A_2 B_2\rangle\leq 2.$$
The  quantum maximum is $2\sqrt{2}$, although the algebraic maximum is $4$. What physical principle prevents 
quantum distributions from reaching  the algebraic maximal? What singles out the bound of $2\sqrt{2}$?
Is it relativistic causality?

Popescu and Rorlich found a simple example that shows that the no-signaling restriction
is not enough to rule out non-quantum correlations. The distribution is known as PR box.

\begin{ex}[PR-box]
Suppose that in a bipartite system one party  can measure $A_1$ and $A_2$ and the other $B_1$ and $B_2$, each with possible
outcomes $\pm 1$. Consider the distribution in the table below:
 
\begin{center}
\begin{tabular}[h]{c|c|c|c|c}
&$(1,1)$&$(1,-1)$&$(-1,1)$&$(-1,-1)$\\ \hline
$11$&$0.5$&$0$&$0$&$0.5$\\ \hline
$12$&$0.5$&$0$&$0$&$0.5$\\ \hline
$21$&$0.5$&$0$&$0$&$0.5$\\ \hline
$22$&$0$&$0.5$&$0.5$&$0$

\end{tabular}
\end{center}
where the number in column $ab$ and line $ij$ is the probability of outcome $a$ for measurement $A_i$ and 
outcome $b$ for measurement $B_j$. This distribution is no-signaling, but it reaches the algebraic maximum  for  
CHSH inequality. 
\end{ex}

The  PR boxes shows that relativistic causality is not enough to distinguish quantum theory from more
general ones. Impossibility of being represented by local hidden variable models is a property of a broad class 
of no-signaling theories. Although they satisfy the no-signaling principle,
the existence of such boxes   would imply  many unreasonable  consequences.

\section{\textsf{Implausible consequences of superstrong non-locality}}
\label{sectionimplausible}

Violations above the quantum threshold are often called superstrong non-locality. The PR box is a simple example 
of a distribution exhibiting this feature. In this section we will show that the existence of this kind of distribution
leads to implausible consequences for the theory of communication complexity, which describes how much communication is 
needed between two parties to
evaluate a distributed function $f$ \cite{V05, BBLMTU6}.

\begin{defi}
 A \emph{distributed function} is a Boolean function
\begin{eqnarray}
 f : \{0,1\}^n \times \{0,1\}^n &\rightarrow &\{0,1\}\nonumber \\
 (x,y)&\longmapsto & f(x,y)
 \label{eqdistributed}
\end{eqnarray}
where the strings $x$ and $y$ are in possession of spatial separated parties, Alice and Bob,
that must communicate in order to compute 
$f$. 
\label{defidistributed}
\end{defi}

By communicating with each other one bit at a time according to some preestablished protocol, they have
to compute the value of $f(x,y)$ in such a way  that at least one of them knows the value at the end of the protocol.
Let $n_f(x,y)$ denote the minimum number of bits exchanged between them in order to accomplish this task.
This number does not depend only on $f$, it may depend also on the resources available for both parties. Once
the resources are fixed, we can define the communication complexity of $f$.

\begin{defi}
 Given the resources shared between the parties, the \emph{communication complexity} of the distributed function $f$ is
 \be c(f)= \max_{x,y}n_f(x,y),\ee
the maximum is taken over all pairs  $(x,y) \in \{0,1\}^n \times \{0,1\}^n$.
\end{defi}

For some functions $f$,  the protocols using quantum systems can be more efficient then the ones assuming only classical
correlations between the parties. Hence, the communication complexity can decrease in the presence of entanglement. 
In other cases, such as for the function
$$Ip(x,y)=\sum_i x_i y_i$$
the communication complexity is  effectively not affected when the parties share quantum correlated systems.
Our purpose in this section is to prove that if the parties shared systems correlated according to the distribution of
a PR box, the communication complexity is reduced to one bit for all distributed functions of the form 
\eqref{eqdistributed}.

First, we will see that this is the case when  $f=Ip$. Suppose that the parties share
at least $n$ PR boxes. In box $i$  Alice will perform measurement $A_{x_i}$, getting outcome $a_i$,
and Bob will perform  measurement $B_{y_j}$, getting outcome $b_i$. The PR box distribution is such that
for all $i$, $a_i + b_i= x_iy_i$ where all sums and products are taken modulo two. Hence, we have
$$Ip(x,y)=\sum_i x_i y_i=\sum_i(a_i+b_i)=\sum_ia_i + \sum_i b_i.$$
The strings $a_i$ and $b_i$ are computed locally and this step does not require any communication between the parties.
After those strings where obtained, Alice, for example, computes $\sum_ia_i$ locally and then sends the resulting
bit to Bob, which is now able to evaluate $f(x,y)$.

The same thing happens for all other $f$. This happens because any function of the form \eqref{eqdistributed} can 
be written as a composition of $Ip$ and local polynomials in $x$ and $y$.

\begin{prop}
 Let $f$ be a distributed function, given according to definition \ref{defidistributed}. There are polynomial
 functions $P_i:\{0,1\}^n\rightarrow \{0,1\}$ and $Q_i:\{0,1\}^n\rightarrow \{0,1\}$ such that
 \be f(x,y)=\sum_i P_i(x)Q_i(y).\ee
\end{prop}

The functions $P_i$ and $Q_i$ depend only on $f$, and hence the strings $w_i=P_i(x)$ and $z_i=Q_i(x)$ can be computed locally
by each party, without any communication. After that they can apply the protocol above to compute
$Ip(w,z)$, and hence compute $f$ with only one bit of communication.

The notion of communication complexity in the presence of PR box is then meaningless, since all functions require only one bit 
to be exchanged in order to compute it. Although this does not contradicts any physical principle, this fact 
does contradict our experiences that certain computational
tasks are harder than other ones. It has been shown that trivial communication complexity can 
be achieved with a violations strictly less than $4$, but it is still not clear if the Tsirelson bound for 
the CHSH inequality
is a critical value that separates trivial
from nontrivial communication complexity. If this is indeed the case, non-triviality of communication complexity
would be a principle singling out the quantum bound.

\section{\textsf{Information Causality}}
\label{sectioninformationcausality}

Information Causality, proposed in reference \cite{PPKSWZ09},
is a generalization of the no-signaling principle. It is respected by both classical and quantum 
theories and violated by some non-quantum distributions. 
Suppose Alice posses some previously assembled data, unknown to some other party, Bob. She is allowed to send 
only classical bits to him. Information Causality states that:
\begin{prin}
 \emph{The information gain that Bob can reach about Alice's data by using all his local resources and $m$ classical
 bits sent by her is at most $m$ bits.}
\end{prin}
The no-signaling condition is just Information Causality with $m=0$. 

Consider now the following task:
Alice receives a bit string $\vec{a}=\left(a_0, a_1, . . . , a_N\right)$ and Bob
receives   $b \in \left(0, 1, \ldots,N \right)$. He is asked to give
the value of Alice's $b$th bit   after receiving from her $m$
classical bits.
If Information Causality is respected, he's information about $\vec{a}$ is
at most $m$ bits.

A good definition of \emph{he's information about her string} would be the mutual information
between the string $\vec{a}$ and everything that Bob has, namely, the $m$-bit message $\vec{x}$ and his party $B$ of 
all presheared correlation, $I(\vec{a}:\vec{x}, B)$.  Information causality would imply $I(\vec{a}:b,\vec{x}, B)\leq m.$
The problem with this definition is that it is not theory-independent: mutual information has to be defined using 
specific objects of the underlying theory and it is not clear if this definition can be done consistently for all
theories, nor whether such definition is unique \cite{BBCLSSWW10}.

Letting aside the problem of defining mutual information, we will show that \emph{if} such a  definition can be made in
a way that 
 three elementary properties are satisfied, the principle of Information Causality holds and we can find a simple necessary 
condition independent of the theory for this principle to be satisfied.

To derive this necessary condition we will need the quantity $I$ defined below, which quantifies the efficiency of Alice and 
Bob's strategy to achieve their goal. Let $\beta$ be Bob's output. Then
\be I=\sum_i I(a_i:\beta|b=i)\ee
where $I(a_i:\beta|b=i)$ is the Shannon mutual information between $a_i$ and $\beta$, given that $b=i$.

\begin{teo}
\label{teoic}
 Suppose that for a given theory a notion of mutual information $I(A:B)$ can be defined and that the following rules are
 satisfied:
 \begin{enumerate}
  \item[I.] Consistency: If the subsystems $A$ and $B$ are classical, $I(A:B)$ coincides with Shannon's mutual information;
  \item[II.] Data processing inequality: Acting on one of the parties locally by any transformation allowed by the theory does 
  not increase the mutual information  $I(A:B)$. More formally, let $S_B$ be the state space of subsystem $B$ and 
   $T:S_B\rightarrow S_B$  any transformation allowed by the theory in this subsystem. Then
   $$I(A:B) \geq I(A:T(B)).$$
   \item[III.] Chain rule: It is possible to define a conditional mutual information $I(A:B|C)$ in such a way that 
   $$I(A:B,C) = I(A:C)+I(A:B|C).$$
 \end{enumerate}
 Then it is possible to prove that
 \begin{enumerate}
  \item The theory satisfies Information Causality; \label{itemic1}
  \item  $I(\vec{a}:\vec{x},B) \leq I$. \label{itemic2}
 \end{enumerate}

\end{teo}

It follows from item \ref{itemic1} that both classical and quantum theories satisfy Information Causality. In classical
theory we use Shannon's mutual information and in quantum theory the mutual information coming from von Neumman's entropy.
For both of them the three requirements of theorem \ref{teoic} are fulfilled.

From item \ref{itemic2} we get the following necessary condition for Information Causality in Alice and Bob's protocol:
\be I \leq m.\ee

The parameter $I$ is easier to work with  because it does not depend on the underlying probabilistic theory. 
It depends solely on 
the input and output bits of their protocol. This condition allows us to prove  that if Alice and Bob 
share PR boxes, Information Causality can be violated.

This violation can be achieved if they use a scheme known as the van Dam's protocol. This is
the simplest situation in which Information Causality can 
be violated. Alice receives two bits $(a_0, a_1)$ and is allowed to send only one of them to Bob. Alice uses $x=a_0 + a_1$ as
input of her part of the PR box and obtains outcome $a$. She sends the bit $m=a_0 + a$ to Bob. He will use as the input 
of his part of the PR box the bit $y=k$, which is $0$ if he wants to learn the value of $a_0$ and $1$ if he wants to learn 
the value of $a_1$. He gets output $b$. As we already mention, for the PR box  inputs and outputs are related according to 
the rule $xy=a + b$ and hence we have: 
\begin{eqnarray}
(a_0+a_1)k=a+b&=&a_0 + m +b\nonumber\\
b+m&=&a_0 + (a_0+a_1)k
\end{eqnarray}
Now, if $k=0$, $b+m=a_0$ and if $k=1$, $b+m=a_1$. Hence, if Bob sums his output of the PR box with Alice's message he 
gets the right value of the bit he had to guess with certainty. With this protocol he has access to two bits of information 
about her data with a message of only one bit, clearly violating Information Causality. 

It is also possible to prove a much more stronger result \cite{PPKSWZ09}.

\begin{teo}
 If Alice and Bob can share distributions violating the CHSH inequality above the Tsirelson bound, they can violate
 Information Causality.
\end{teo}

The idea behind the proof is the following: first, we note that any distribution can be brought into a simple form  where
the local outcomes have a uniform distribution and the joint distributions satisfies
\be p(a+ b= xy)=\frac{1+E}{2}\label{eqnsboxes}\ee
where $0 \leq E \leq 1$. The case $E=1$ corresponds to the PR box and $E=0$ to completely uncorrelated bits.
This transformation can be done locally and does not change the value of $S_{CHSH}.$ The classical bound is violated if
$E >\frac{1}{2}$ and the quantum threshold becomes $E = \frac{1}{\sqrt{2}}.$
Whenever $E >\frac{1}{\sqrt{2}}$ we get a violation of Information Causality.

In the protocol used to obtain this violation, Alice receives $N=2^n$ bits and Bob  receives a list with $n$ bits to
inform him which of her bits he has to guess. She is allowed to send one bit to him. 
Using a chain of preestablished systems correlated according to 
 equation \eqref{eqnsboxes}, they can apply a protocol for which the probability of Bob guessing correctly the bit $a_k$ is
$$p_k=\frac{1}{2}(1+E^n).$$
Information Causality condition is violated as soon as $I>1$
 and this happens if $2E^2>1$ and $n$ is large enough \cite{PPKSWZ09}. 
 This proves that whenever the distribution violates CHSH above
 the Tsirelson bound we can use it to implement a protocol violating Information Causality.
 
 This result  connects the Tsirelson bound with a compelling physical principle.
 However, here are also non-quantum 
 distributions that lie under the quantum threshold and hence are not exclude by the previous argument. 
 It is still not known if Information Causality singles out  entire the set of quantum distributions. 
 A partial answer
 was provided a few months after Information Causality's first paper was released \cite{ABPS09}.
 
The authors present two families distributions which ca be written in the form
\be PR_{\alpha, \beta}=\alpha PR + \beta B +(1-\alpha-\beta)I,\ee
where $I$ is the uniform uncorrelated distribution and  
$PR$ is the usual PR box. In the first family, $B$ is one of the non-local boxes given by
\be P_{NL}^{\mu \nu \sigma}=\left\{\begin{array}{cc}
                   \frac{1}{2} & \mbox{if} \ \ a+ b =xy+ \mu x + \nu y + \sigma\\
                   0& \mbox{otherwise}.
                  \end{array}\right.,\ee
with $\nu \mu \sigma$ any sequence of bits except $000$ and $001$. The distribution $PR_{\alpha, \beta}$ will be quantum iff
$$\alpha^2+\beta^2 \leq 1$$
which is a necessary and sufficient condition for Information Causality to be satisfied if 
$\nu \mu \sigma= 010, \ 011, \ 100$ or $101$. Hence, in this slice of the no-signaling polytope, 
\emph{Information Causality 
singles out the boundary of the set of quantum distributions.} This is shown is figure \ref{figic} (a).

The condition for Information Causality in the case $\nu \mu \sigma= 111$ is 
$$\alpha \leq \frac{1}{2},$$
which gives the quantum maximum value for CHSH. Hence, this protocol can not discard non-quantum 
boxes below the Tsirelson bound. This is shown if figure \ref{figic} (b).
It is not known if these boxes  violate Information Causality in this 
slice of the no-signaling polytope. 

In the second family, $B$ is one of the local boxes given by
\be P_{L}^{\mu \nu \sigma \tau}=\left\{\begin{array}{cc}
                   1 & \mbox{if} \ \ a=\mu x+ \nu, \ \ b=\sigma y + \tau\\
                   0& \mbox{otherwise}.
                  \end{array}\right.,\ee
with $\mu\sigma+ \nu + \tau=0$. For these distributions, Information Causality is violated iff
$$(\alpha + \beta)^2+\alpha^2> 1.$$
This inequality does not coincide with the criteria for quantumness. For this family it is possible to exclude several
non-quantum correlations below the Tsirelson bound, but with the strategy used, it is not possible to reach
the quantum boundary. This is shown in figure \ref{figic} (c).

\begin{figure}[!h]
  \centering
  \includegraphics[scale=0.2]{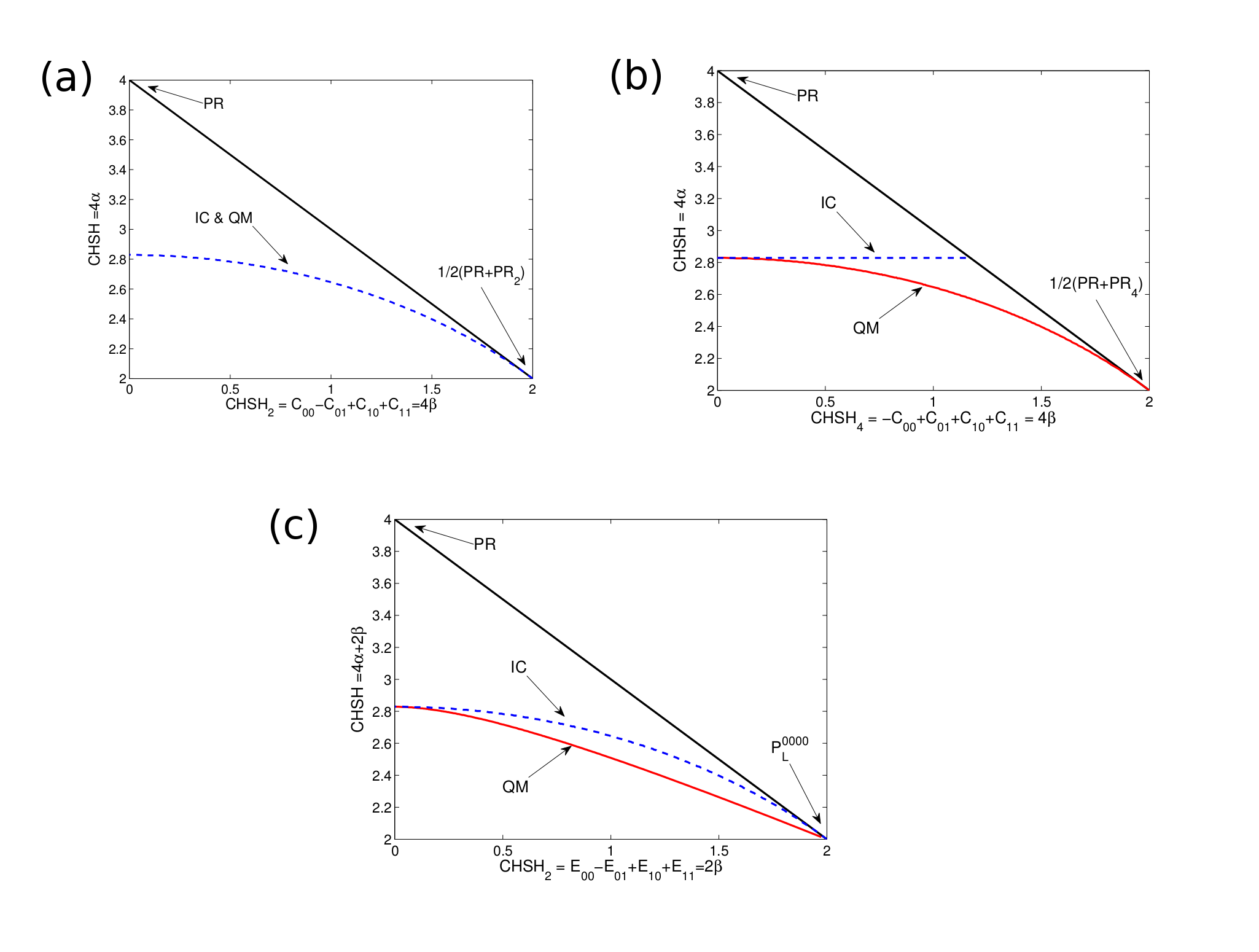}\\
  \label{figic}
  \caption{(a) In this slice of the no-signaling polytope, the principle of Information Causality singles out the 
  boundary of the quantum set. (b) In this slice, the same protocol is not able to explain the boundary of the quantum set.
  (c) In this slice, the same protocol gets close to the quantum boundary. This image was taken from reference \cite{ABPS09}.}
\end{figure}

Alice and Bob's game can be generalized to alphabets with more then two elements \cite{CSS10}.
Instead of giving Alice a string of bits,
she will now receive a string of \emph{dits}, a random variable with $d$ possible outcomes. Her message will also be changed. She 
is now allowed to send Bob $m$ dits. Their goal remains the same: Bob receives a position $y$ in Alice's string and he has to
guess the dit she has in that specific position. The efficiency of their protocol can be measured by the quantity
$$I=\sum_{k=0}^{n} I(a_k: b_k|y=k)$$
where $I(a_k: b_k|y=k)$ is the mutual information between Alice's $k$th dit $a_k$ and Bob's guess $b_k$, given that he was
asked to guess her dit in position $k$. Information Causality will be violated as soon as
$$I > m\log_2d.$$

Let us focus in the case where Alice receives a string of two dits $\vec{a}=(a_0, a_1)$, 
$a_i \in \{0, 1, \ldots, d-1\}$. Bob receives a bit $y$ that tells him if he has to guess the first or the second dit in
Alice's string. Since she only sends him one dit, Information Causality requires that $I=\log_2d$.
If Alice and Bob share the  no-signaling distribution with $d$ inputs in Alice's side, $2$ inputs in Bob's side and $d$ outputs 
in both sides given by
$$PR_d(ab|xy)=\left\{\begin{array}{cc}
                     \frac{1}{d}&\mbox{if} \ xy=(b-a) \ \mbox{mod} \ d\\
                     0& \mbox{otherwise}
                    \end{array}\right.,$$
there is a protocol in which Information Causality is violated.

As inputs of the $PR_d$ box, Alice uses $x=(a_1-a_0) \ \mbox{mod} \ d$ and Bob uses $y$. She gets output $a$ and he gets 
output $b$. Alice send the message $m= (a-a_0)\ \mbox{mod} \ d$. Bob, in possession of $m$ will make his guess 
$g=(b-m) \ \mbox{mod} \ d= (b-a + a_0) \ \mbox{mod} \ d$. Given that the inputs and outputs are correlated
 according to the rule $xy=(b-a) \ \mbox{mod} \ d$, we have 
 $$g=[(a_1-a_0)y +a_0] \ \mbox{mod} \ d$$
 which is equal to $a_0$ if $y=0$ and equal to $a_1$ if $y=1$.
 
 Therefore, using this protocol, Bob can guess any of her bits with certainty. This means that
 $$I=2\log_2d,$$
 clearly violating Information Causality.
 
 We can also see what happens when we use noisy boxes of the type
 $$PR_d(E)=EPR_d +(1-E)I.$$
There is a protocol using nested boxes of this kind that achieves success probability of 
$$P=\frac{(d-1)E^n +1}{d}$$
where $n$ is the number of boxes used.

Figure \ref{figinformationcausality} shows the critical value of $E$ beyond which Information Causality ceases to be violated. For $d=2$ we return to the case discussed
previously  and we have that for values of $E$ above  $\frac{1}{\sqrt{2}}$ Information Causality is violated. This is also 
the bound  for quantum  distributions. 
For $d>2$ the situation becomes richer. The quantum bound is no longer known and 
the critical value in which Information Causality ceases to be violated can be smaller then $\frac{1}{\sqrt{2}}$.

\begin{figure}[!h]
  \centering
  \includegraphics[scale=0.5]{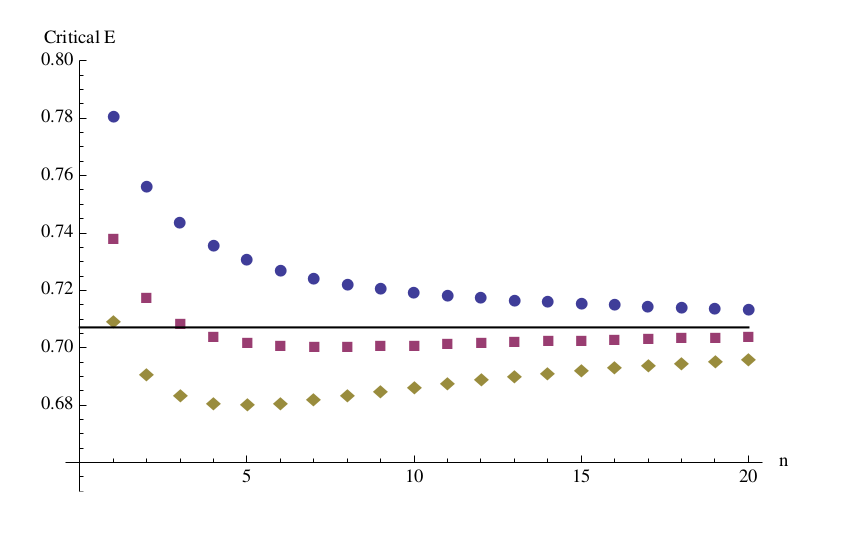}\\
  \label{figinformationcausality}
  \caption{Critical level of noise $E$ for which Information Causality ceases to be violated, as a function of the number
  of boxes used and for different values of $d$ ($d=2$, blue dots, $d = 5$ purple squares, $d = 10$
green diamonds). The solid line corresponds to the Macroscopic Locality bound (see section \ref{sectionmacroscopiclocality}).
This image was taken from reference \cite{CSS10}.}
\end{figure}

\section{\textsf{Macroscopic non-locality}}
\label{sectionmacroscopiclocality}

The motivation for the definition of Macroscopic locality is not to identify the principle behind quantum theory, but rather to understand how to go 
beyond it \cite{NW09}. One of the most important problems of current research in theoretical physics is to reconcile quantum theory and 
general relativity and a first step towards this goal is to derive general results that should apply to 
any theory satisfying a set of reasonable requirements. Macroscopic Locality may be one of them. The idea behind this 
principle is that any such theory should recover the classical 
results  when we measure a large number of equally prepared system and our devices are not capable of identifying
individual particles.

In the kind of experiments we have considered so far, two parties Alice and Bob share individual particles correlated 
according to some distribution $p(ab|xy)$, where, as usual,  $x$ and $y$ label the possible measurements 
and  $a$ and $b$ the possible outcomes in Alice's and Bob's
side, respectively. We refer to this kind of experiment as a \emph{microscopic experiment}.

In a \emph{macroscopic experiment}, 
Alice and Bob share a huge number $N\gg 1$ of pairs of particles correlated according to the distribution  $p(ab|xy)$.
They will not interact with a single particle but with a beam of them and hence they will not be able to 
address them individually and any operation they perform will be applied to all the particles in the beam at the same time.

After Alice and Bob perform some measurement in their particles, each beam will be divided in a number of different beams,
each one corresponding to one possible outcome of that measurement. In this scenario, the probabilities are no longer 
important and the \emph{intensities} of each beam will describe the results of the experiment. If Alice
measures $x$, we will denote the intensity of the beam corresponding to outcome $a$ by $I^x_a$ and analogously for Bob.

\begin{prin}[Macroscopic Locality]
The distribution of intensities
$p\left(I^x_a, I^y_b\right)$ Alice and Bob  
observe  admits a local hidden variable model. This is equivalent of saying that there is a global distribution
$$p(I^x_a,I_{x_1}^{a_1}, \ldots, I_{x_m}^{a_m},I^y_b, I_{y_1}^{b_1}, \ldots, I_{y_n}^{b_n})$$ 
such that 
\be p\left(I^x_a, I^y_b\right)=\int p\left(I^x_a,I_{x_1}^{a_1}, \ldots, I_{x_m}^{a_m},I^y_b, I_{y_1}^{b_1}, \ldots, I_{y_n}^{b_n}\right)
\prod_{i,j}dI_{x_i}^{a_i}I_{y_j}^{b_j}.\ee
\end{prin}

\begin{figure}[h]
  \centering
  \includegraphics[scale=0.6]{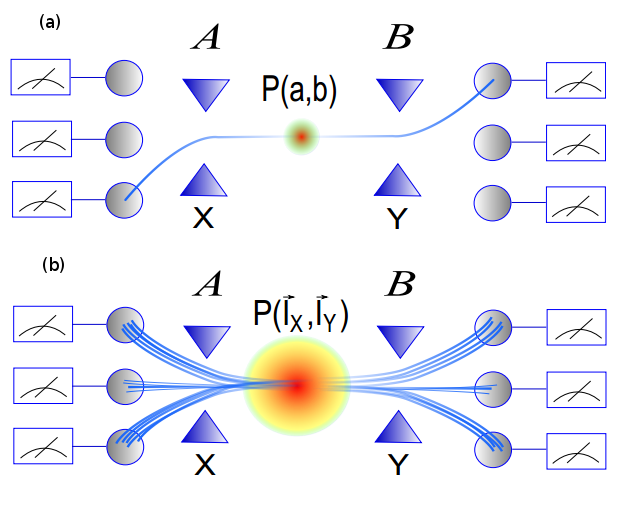}\\
  \caption{(a) A microscopic experiment. (b) A macroscopic experiment. Image taken from reference \cite{NW09}}
\end{figure}

Clearly the intensities are related to the distribution $p(ab|xy)$. With this correspondence written
explicitly, it is possible to identify the set of no-signaling distributions satisfying Macroscopic Locality.
This set is very similar to the set of quantum distributions, but it is not identical.

\begin{teo}
\label{teoml1}
 The set of macroscopic local non-signaling distributions is equal to the set $Q_1$ introduced in reference \cite{NPA08}.
\end{teo}

This set is the first set in a hierarchy  of conditions necessarily satisfied by any distribution $p(ab|xy)$ obtained with 
a quantum system. It can be numerically  characterized via semidefinite programming. 
By definition $Q \subset Q_1$ and even in the simplest case of each part with two measurements with
two outcomes they are not the same, although they are extremely close. 

Although Macroscopic Locality  is not able to single out the set of quantum distributions even in the simplest scenario,
it does single out the Tsirelson bound for the CHSH inequality.

\begin{teo}
\label{teoml2}
 The maximum value for $S_{CHSH}$ for macroscopic local no-signaling theories is equal to the Tsirelson bound
 $2\sqrt{2}$.
\end{teo}

Theorem \ref{teoml1} implies that   if Macroscopic Locality and no-signaling are fundamental properties of nature, 
the set of allowed distributions has to be contained in $Q_1$. If these axioms are enough to pin down the set of allowed 
distributions, they must come from a non-quantum theory. On the other side, theorem \ref{teoml2} shows that in
the same circumstances a violation 
of CHSH inequality above the Tsirelson bound is not possible. The similarities between $Q^1$ and the quantum set decrease,
though, if we increase the number of measurements available for Alice and Bob and the number of possible outcomes
for each measurement. It is possible then that macroscopic local distributions violate some Bell inequality above
the Tsirelson bound. This observation opens the door for finding non-quantum distributions using Bell-like
experimental scenarios.


\subsection{\textsf{Macroscopically local correlations can violate Information Causality}}

In section \ref{sectioninformationcausality} we showed that if Alice and Bob share a large number of bipartite system correlated according to the distribution
$$PR_d(E)=E\left(PR_d\right) +(1-E)I$$
they can apply a nested protocol to violate Information Causality as long as $E$ is above a certain threshold, that depends on the number 
of shared distributions used in the protocol and also on $d$.

When $d$ is equal to $2$ it is clear that whenever $E$ is above $\frac{1}{\sqrt{2}},$ both Information Causality and Macroscopic Locality
are not valid anymore. It is also known that  
$E\leq \frac{1}{\sqrt{2}}$ is a necessary and sufficient condition for the distribution to be quantum.

The situation $d>2$ is much more complex. In this case we do not know what is the condition on $E$ for quantumness
of the distribution. The condition for Macroscopic Locality remains the same, at least up to $d=5$
: the distribution $PR_d(E)$ will violate Macroscopic Locality iff
$E>\frac{1}{\sqrt{2}}$. The critical values for Information Causality, as we already mention, depends strongly on $d$. 
Figure \ref{figinformationcausality} shows the critical values for different values of $d$, as a function of the 
number of boxes available, and also the critical value for Macroscopic Locality.

This observation allows us to prove that some macroscopic local distributions can violate Information Causality.
For example, for $d=5$, the distribution
$$PR_5(E)=\frac{1}{\sqrt{2}}PR_5 +\left(1-\frac{1}{\sqrt{2}}\right)I$$
is macroscopic local but can be used to violate Information Causality.

Therefore, Information Causality and Macroscopic Locality are not equivalent.  Macroscopic Locality was proposed not as a principle capable of singling out quantum
distributions but rather as a desirable axiom of any alternative to quantum theory. The fact that macroscopic local
distribution violate Information Causality shows that if the principle of Information Causality is also a fundamental property of any non-quantum
theory, then the set of  distributions  it allows in some scenarios
has to be smaller then the set of macroscopic local distributions \cite{CSS10}.

\section{\textsf{Quantum correlations require multipartite information principles}}
\label{sectionmultipartiteprinciples}

So far we have seen four different principles proposed to explain quantum nonlocality: no-signaling, non-triviality 
of communication complexity, Information Causality and Macroscopic Locality. Although very fruitful in 
many senses,  these requirements suffer from a common drawback. All of them are based in a bipartite situation in
which two spatially separated parties share a pair of correlated system described according to some probability
distribution.

We can come up with much more  interesting situations. Instead of a bipartite scenario, we can imagine now a
$n$-partite
system shared among $n>2$ spatially separated parties. What physical principles explain the set of quantum distribution in 
in a general situation?

There is a trivial way of applying the bipartite requirements we have studied before to distributions in a multipartite 
scenario. We can consider the situation in which Alice holds $k$ of these subsystems and Bob the $n-k$ left and apply
the bipartite principles to the distribution obtained in this way. We may conjecture that applying some of these
principles to all possible bipartitions we would be able to single out the set of quantum distributions also on the 
multipartite scenario. Unfortunately this is not the case \cite{GWAN11}.

The problem is that there are some non-quantum multipartite distributions that behave exactly like local distributions
for every possible bipartition. One example of such distributions are found in the set of tripartite distributions 
admitting a \emph{time-ordered bilocal model} \cite{PBS11, GWAN12} .

Let $p(a_1a_2a_3|x_1x_2x_3)$ denote the probability of getting  outcomes $a_1,a_2$ and $a_3$, respectively, when the first part 
applies measurement $x_1$,  the second part applies measurement $x_2$ and the third part applies measurement $x_3$. 

\begin{defi}
 We say that the distribution $p(a_1a_2a_3|x_1x_2x_3)$
 admits a \emph{time-ordered bilocal model} (TOBL) if it can be written in the form
\begin{eqnarray}
 p(a_1a_2a_3|x_1x_2x_3)&=&\sum_{\lambda} p_{\lambda}^{i|jk} p_{j\rightarrow k}(a_ja_k|x_jx_k)\\
 &=&\sum_{\lambda} p_{\lambda}^{i|jk} p_{j\leftarrow k}(a_ja_k|x_jx_k).
 \end{eqnarray}
 for $(i,j,k)=(1,2,3), (2,3,1), (3,1,2).$  The distributions $ p_{j\rightarrow k}(a_ja_k|x_jx_k)$ and 
 $p_{j\leftarrow k}(a_ja_k|x_jx_k)$ are allowed to be signaling in at most one direction,
 as indicated by the arrow.
\end{defi}
These models have a very clear operational meaning. Let us consider first the case $(i,j,k)=(1,2,3)$. This case
corresponds to the bipartition $1|23$: the first subsystem is with Alice and the other two are with Bob.
Equation 
$$p(a_1a_2a_3|x_1x_2x_3)=\sum_{\lambda} p_{\lambda}^{1|23}p(a_1|x_1) p_{2\rightarrow 3}(a_2a_3|x_2x_3)$$
means that under this bipartition, the distribution admits a local hidden variable model, $\lambda$ being
the hidden variable. The fact that $p_{2\rightarrow 3}(a_2a_3|x_2x_3)$ may be signaling is not an issue here because
systems $2$ and $3$ are now seen as one, and hence the notion of signaling makes no sense.

Since $(i,j,k)$ can vary over all possible permutations, the same will happen for the other bipartitions $2|13$ and 
$3|12$. This implies that whenever we consider bipartition of a TOBL distribution, the bipartite distribution
obtained will be local. This remains true if we concatenate any number of them under 
\emph{wiring}, which is the most general operation we can apply to this set of distributions \cite{ABLPSV09}.
This implies that it can not violate any principle mentioned above.

The important observation is that there are TOBL distributions that are not quantum. This can be seen with the help of
a famous Bell inequality for the $(3,2,2)$ scenario, known as \emph{Guess Your Neighbor's Input} inequality:
$$p(000|000)+p(110|011)+p(011|101)+p(101|110) \leq 1.$$
For this inequality the quantum bound is also $1$, that is, there is no quantum violation in this case. The maximal
value obtained with TOBL distributions is $\frac{7}{6}$, which proves the existence of TOBL distributions outside
the quantum set. 

Another example is provided in reference \cite{YCATS12}. The authors
study violations of the principle of Information causality  in the presence of extremal no-signaling distributions
on a tripartite 
scenario. They  prove that distribution can not be discarded by any bipartite physical principle.

Hence, neither the bipartite principles already proposed so far nor any other that may be proposed in the future will
be able to single out the set of quantum distributions in the multipartite scenario because none of them is capable
of ruling out 
the TOBL distributions. 
This result implies that  intrinsically multipartite principles are required to fully understand the set of quantum 
distributions in more complicated situations.

\section{\textsf{Local orthogonality: the exclusivity principle for Bell scenarios}}
\label{sectionlocalorthogonality}

Unlike all other principles we have mentioned previously in this appendix, the Exclusivity principle can be applied 
directly to all Bell scenarios, including the ones with multiple parties. In this situation, the principle
is commonly referred to as the principle of \emph{Local Orthogonality} \cite{FSABCLA12}.

Suppose a composite system is shared among $n$ spatially separated parties. In each party an experimentalist can 
apply $m$ measurements with $d$ possible outcomes. The possible events in this scenario are of the form
$$(a_0, a_1, \ldots a_n | x_0, x_1, \ldots , x_n) $$
where $x_i$ stands for the measurement performed in party $i$ and $a_i$ for the corresponding outcome.

\begin{defi}
 Two  events  
 $$e_1 = (a_0, a_1, \ldots a_n | x_0, x_1, \ldots , x_n) \  \mbox{and} 
 \  e_2 = (a'_0, a'_1, \ldots a'_n | x'_0, x'_1, \ldots , x'_n)$$
 are \emph{exclusive} or \emph{locally orthogonal} if they involve different outputs
of the same measurement by (at least) one party:
$$ x_i=x'_i \  \mbox{and} \ a_i \neq a'_i.$$
A collection of events $\{e_i\}$ is \emph{locally orthogonal} if the events are pairwise
locally orthogonal.
 \end{defi}

As before, the Exclusivity principle demands that if a set of events  $\{e_i\}$ is \emph{locally orthogonal}
\be\sum_i p(e_i) \leq 1. \label{eqloinequality}\ee
Such an inequality is called an \emph{orthogonality inequality}.

The set of  distributions that satisfy all LO inequalities in this scenario is denoted by $\mathcal{LO}^1$.
As shown in \cite{CSW10}, for bipartite scenario this set is equal to the set of no-signaling distributions, but this 
equivalence is no longer valid for more parties. Already in the $(3,2,2)$ scenario no-signaling and $\mathcal{LO}^1$ are no longer
equal. All orthogonality inequalities  in this case are equivalent under local operations to the Guess Your Neighbor Input inequality
$$p(000|000)+p(110|011)+p(011|101)+p(101|110) \leq 1$$
for which the maximal no-signaling violation is equal to $\frac{4}{3}$. Numerical data suggests that the gap between the two
sets increase with the number of parties, but already for $n=5$ the problem becomes intractable due to the huge size of 
the exclusivity graph.

Violations of Local Orthogonality can exhibit \emph{activation} effects.
A larger distribution coming from several copies of $p \in \mathcal{LO}^1$ does not necessarily satisfies
Local Orthogonality. Consider $k$ copies of a $n$-partite system with distribution $p$, distributed among $kn$ parties, each party having access
to only one subsystem of one of the copies. If the resulting distribution $p_k$ satisfies all Local Orthogonality inequalities 
for the $(kn,m,d)$ scenario we say that $p$ belongs to the set $\mathcal{LO}^k$. We denote by $\mathcal{LO}^{\infty}$ 
the set of distribution in the $(n,m,d)$ scenario that belong to $\mathcal{LO}^k$ for all $k$.

To see what are the consequences of imposing the Local Orthogonality principle, we have to characterize the sets $\mathcal{LO}^k$, what
requires that we  identify all 
Local Orthogonality inequalities for a given scenario. As we have already seen, this is a hard problem, equivalent to finding all 
maximal cliques of 
the exclusivity graph of the scenario.

At first sight, it seems that Local Orthogonality would not capable of ruling out non-quantum distributions in the bipartite scenario 
because of the equivalence
between this principle and no-signaling, but this is not the case. Due to the activation effects,
imposing Local Orthogonality in the multipartite level leads to detection of non-quantumness even for the bipartite case.

Already for the simplest scenario $(2,2,2)$, Local Orthogonality is able to rule out the  PR box if we use two copies of this distribution.
Suppose that parties $1$ and $2$ are in possession of one of the copies and parties $3$ and $4$ are in possession of the other 
copy. Then, the value of the sum 
$$p(0000|0000) + p(1110|0011) + p(0011|0110) + p(1101|1101) + p(0111|1101)$$
is equal to $\frac{5}{4}$, while Local Orthogonality demands this value to be less or equal then $1$. The same reasoning allows us to 
rule out other distributions obtained from the PR box by adding noise. Consider the family of distributions given by
$$PR(\alpha)=\alpha PR +(1-\alpha)I$$
where $I$ is the distribution where all parties are independent and the probabilities for all measurements are uniform.
Two copies of $PR(\alpha)$ violate Local Orthogonality for all $\alpha > 0.72$. This value is close to the quantum bound of 
$\alpha = \frac{1}{\sqrt{2}}\approx 0.707$.

Local Orthogonality also rules out all extremal distributions also  in the $(2,2,d)$ scenario. This 
happens because we can use them to simulate a PR box, perfectly with one copy if $d$ is even and arbitrarily well with 
sufficiently many copies if $d$ is odd.

Local Orthogonality is very successful in the bipartite case as it rules out many distribution and gets close to the Tsirelson bound.
But it is for $n>2$  that we expect it to perform better then the previous principles,
since its definition is intrinsically
multipartite.  It is possible to prove that all extremal distributions in the $(3,2,2)$ scenario lie outside 
$\mathcal{LO}^1$ or $\mathcal{LO}^2$. The distributions used in section \ref{sectionmultipartiteprinciples}
as examples of non-quantum violations that satisfy
all bipartite principles are also ruled out by Local Orthogonality, since they violate the Guess Your Neighbor Input
inequality. Local Orthogonality rules out distributions where all
other known principles fail.

\section{Final Remarks}
\label{sectionfinaltsirelson}

An important problem in Physics is to understand what kind of correlations can be observed between measurements
conducted in spatially separated physical systems that have interacted in the past. Quantum theory predicts
stronger correlations then the ones that can be obtained with classical systems, which leads to
violations of Bell inequalities. At least mathematically, there is room 
for more: quantum systems do not reach the algebraic maximum violation of several Bell inequalities , which can be
reached only with some 
non-quantum distributions, obtained using  more general probabilistic theories.
Why we do not observe these stronger correlations in nature? Is there any physical principle that 
forbids probability distributions outside the quantum set?

No-signaling is certainly a property we should impose on the distributions in order to discard the unphysical 
ones, but it is not enough to single out the quantum set. The no-signaling distribution of a PR box
can reach the algebraic maximum
of the CHSH inequality, while the Tsirelson bound lies below this value. Nonetheless, the existence of 
such distributions would have strange consequences in the field of communication complexity. If
the parties are allowed to share an arbitrary number of PR boxes,  any
distributed function would require only one bit of communication between the parties to be computed,
making the notion of communication complexity meaningless. Although this does not contradict any principle, it goes against
our experience that some problems are harder to solve then others. Although trivial communication complexity was found
 with  violations strictly less than $4$,  it is still not clear if the Tsirelson bound for 
the CHSH inequality
is a critical value that separates trivial
from nontrivial communication complexity. 

Information Causality is a principle
with a information theoretic motivation. It can also be used to discard several non-quantum distributions.
For the CHSH inequality it is
known that any violation above the Tsirelson bound also violates Information Causality.  In more sophisticated situations,
it is known that this principle can rule out many non-quantum distribution, but it is not known whether 
if we can relate this
to the Tsirelson bound of more complicated inequalities nor if it singles out the entire set of quantum distributions.
 It remains an open question whether this whole zoo of nonlocality can be derived from information causality.
 
 Information Causality was also used to derive limits on Hardy's non-locality \cite{Hardy93}. It has been shown that any generalized 
 probability theory which gives completely random results for local dichotomic observable, 
 can provide Hardy's non-local correlation and satisfy Information Causality at the same time \cite{AKRRJ10, GRKR10}.
 Nevertheless, there are some restrictions imposed by quantum theory that can not be explained by the 
 considered Information causality condition.
 
 The principle of Macroscopic Locality is a reasonable property we should expect from any physical theory, since any such theory should recover the classical 
results  when the number of particles goes to infinity. The set of macroscopic local correlation is not equal to 
the quantum set. They are close for the $(2,2,2)$ scenario, but the similarities decrease
 if we increase the number of measurements available or the number of possible outcomes
for each measurement. Though this principle can not recover the quantum set, it may help us to understand how to
derive generalizations of quantum theory and reconcile it  with
general relativity.

Although Macroscopic Locality  is not able to single out the set of quantum distributions even in the simplest scenario,
it does single out the Tsirelson bound for the CHSH inequality. It is still an open problem  to prove
that macroscopic local distributions violate some Bell inequality above
the Tsirelson bound. 

This principle was also used to derive quantum Bell
inequalities, linear inequalities that provide necessary conditions for a distribution to be quantum \cite{YNSV11}.
The method is applicable to all bipartite
scenarios. Such inequalities
provide  analytical approximations to the quantum
set, which are difficult to find in general.

Although the principles above are very fruitful in many different situations, they are not enough to explain the set of 
quantum distributions in scenarios with more than two parties. Some non-quantum distributions in a tripartite scenario 
have been found that behave like classical distributions for all possible bipartitions.  This implies that, in 
order to explain the quantum set in more complicated scenarios, intrinsically multipartite principles must be used.

The only multipartite principle proposed so far is Local Orthogonality, the Exclusivity principle applied to Bell scenarios.
 Local Orthogonality is very successful  in the bipartite case as it rules out the extremal boxes in the 
 $(2,2,d)$ scenario for any $d$, as well as many others for $d=2$, approaching the Tsirelson bound for the CHSH inequality.
For $n>2$   we expect it to perform better then the previous principles.  
It is possible to prove that all extremal distributions in the $(3,2,2)$ scenario violate Local Orthogonality with one or 
two copies.
Some non-quantum distributions that satisfy
all bipartite principles are also ruled out by Local Orthogonality.

The difficulty in proving the consequences of this principle to other scenarios lie in the fact that the 
exclusivity graph becomes intractable when we increase the number of parties, measurements or outcomes. 
This makes any computational calculation impossible. Nevertheless,
Local Orthogonality rules out distributions where all
other known principles fail. This corroborates  the conjecture that the Exclusivity principle is the fundamental principle 
that singles out the set of  quantum distributions.


\backmatter
\bibliographystyle{alpha}
\bibliography{biblio1}

\end{document}